\documentclass[reqno]{amsart}
\usepackage{fullpage,amsfonts,amssymb,amsthm,mathrsfs,graphicx,color,framed,esint,stackrel,amsmath}
\usepackage{cancel,bm}
\usepackage{tikz}
\usetikzlibrary{decorations.markings}
\usepackage{booktabs,longtable}
\usepackage[initials]{amsrefs}

\DeclareMathAlphabet\mathbfcal{OMS}{cmsy}{b}{n}

\usepackage[colorlinks=true, pdfstartview=FitV, linkcolor=blue, citecolor=blue, urlcolor=blue]{hyperref}
\definecolor{ao}{rgb}{0.0, 0.4, 0.0}

\numberwithin{equation}{section}

\renewcommand{\d}{{\mathrm d}}

\newcommand{\im}{\mathrm{i}}
\newcommand{\e}{\mathrm{e}}

\def\tr{\mathop{\mathrm{tr}}\limits}

\newtheorem{theo}{Theorem}[section]

\newtheorem{lem}[theo]{Lemma}
\newtheorem{rem}[theo]{Remark}
\newtheorem{problem}[theo]{Riemann-Hilbert Problem}
\newtheorem{remark}[theo]{Remark}
\newtheorem{prop}[theo]{Proposition} 
\newtheorem{cor}[theo]{Corollary} 
\newtheorem{definition}[theo]{Definition}

\def\H{{\bf H}}
\def\X{{\bf X}}

\def\A{{\bf A}}
\def\B{{\bf B}}

\usepackage{mdframed}

\begin{document}

\title[Edge level spacings in the eGinUE]{The complex elliptic Ginibre ensemble at weak non-Hermiticity: edge spacing distributions}

\author{Thomas Bothner}
\address{School of Mathematics, University of Bristol, Fry Building, Woodland Road, Bristol, BS8 1UG, United Kingdom}
\email{thomas.bothner@bristol.ac.uk}

\author{Alex Little}
\address{School of Mathematics, University of Bristol, Fry Building, Woodland Road, Bristol, BS8 1UG, United Kingdom}
\email{al17344@bristol.ac.uk}
\date{\today}

\keywords{Complex elliptic Ginibre ensemble, Fredholm determinants, extreme value statistics, integro-differential Painlev\'e functions, Tracy-Widom and Gumbel distributions, Riemann-Hilbert problem, nonlinear steepest descent method.}

\subjclass[2020]{Primary 45B05; Secondary 47B35, 35Q15, 30E25, 37J35, 70H06, 34E05}
\thanks{The authors are grateful to G. Akemann for stimulating discussions. This work is supported by the Engineering and Physical Sciences Research Council through grant EP/T013893/2.}

\begin{abstract}
The focus of this paper is on the distribution function of the rightmost eigenvalue for the complex elliptic Ginibre ensemble in the limit of weak non-Hermiticity. We show how the limiting distribution function can be expressed in terms of an integro-differential Painlev\'e-II function and how the same captures the non-trivial transition between Poisson and Airy point process extreme value statistics as the degree of non-Hermiticity decreases. Our most explicit new results concern the tail asymptotics of the limiting distribution function. For the right tail we compute the leading order asymptotics uniformly in the degree of non-Hermiticity, for the left tail we compute it close to Hermiticity.
\end{abstract}

\maketitle


\section{Introduction and statement of results}

Large complex systems connected at random tend to undergo a sharp transition in stability as the systems' connectance increases, cf. \cite{GA,M}: such systems are stable below a critical threshold and suddenly become unstable above it.  For instance, consider a system with $n$ variables (populations of $n$ interacting species, say) which may obey a rather complicated set of nonlinear, albeit autonomous, first-order differential equations as time $t$ evolves. To determine the stability of any of its equilibria, one Taylor-expands in the vicinity of such an equilibrium point and thus one must analyze a linear equation of the form
\begin{equation}\label{fi1}
	\dot{\bf x}=\A{\bf x},\ \ \ \ (\,\dot{}\,)=\frac{\d}{\d t}.
\end{equation}
In an ecological model, ${\bf x}$ would be the $n\times 1$ column vector of populations $x_j$ and the $n\times n$ matrix $\A$ has entries $A_{jk}$ that encode the impact of species $k$ on species $j$ near the equilibrium. However, given that interactions are seldom known precisely, it is reasonable to assume that $\A$ will be a random matrix without symmetries. Indeed, following \cite{GA}, one supposes that each of the $n$ species, once perturbed from equilibrium, will return to it after some time, i.e. one chooses $A_{jj}=-1$, and thus the damping time is set to unity. Next it is assumed that the mutual interactions are equally likely to be positive or negative, with absolute magnitude chosen independently and identically from some distribution. In short, ${\bf A}$ in \eqref{fi1} takes the form
\begin{equation*}
	\A=-\,\mathbb{I}+\B,
\end{equation*}
where $\B\in\mathbb{C}^{n\times n}$ is a non-Hermitian random matrix with i.i.d. entries and $\mathbb{I}\in\mathbb{C}^{n\times n}$ the identity matrix. Clearly, system \eqref{fi1} is stable if, and only if, all eigenvalues of $\A$ have negative real parts. Moreover, the growth rate of any solution to \eqref{fi1} as $t\rightarrow+\infty$ is determined by the rightmost eigenvalue of $\A$, i.e. by the eigenvalue of $\A$ with largest real part. Thus, motivated by the simple yet widely applicable model \eqref{fi1}, cf. \cite{ASS,AGBTAM,AT,RA,SCS}, we readily appreciate the necessity to accurately understand, and describe, the statistical behavior of the rightmost eigenvalue for a given family of random matrices. The present paper contributes to this field from the viewpoint of integrable systems theory while analyzing the rightmost eigenvalue of a distinguished \textit{interpolating} non-Hermitian random matrix model. Before we define the interpolating ensemble in question, we shall briefly focus on the two classical ensembles which are captured by the interpolation, namely the Gaussian Unitary Ensemble (GUE) and the complex Ginibre Ensemble (GinUE).
\subsection{Two sides of a coin} The limiting behavior of the rightmost eigenvalue in the following two ensembles is a rather classical topic, at least for the GUE, cf. \cite{F1}.
\begin{definition}[Porter \cite{Po}, 1965] We say that a random complex Hermitian matrix $\X\in\mathbb{C}^{n\times n}$ belongs to the \textnormal{GUE} if its diagonal elements and upper triangular elements are independently chosen with pdfs
\begin{equation*}
	\frac{1}{\sqrt{\pi}}\e^{-x_{jj}^2}\ \ \ \ \textnormal{and}\ \ \ \ \frac{2}{\pi}\e^{-2|x_{jk}|^2},
\end{equation*}
respectively.
\end{definition}

\begin{definition}[Ginibre \cite{Gi}, 1965] We say that a random complex matrix $\X\in\mathbb{C}^{n\times n}$ belongs to the \textnormal{GinUE} if its entries are independently chosen with pdfs
\begin{equation*}
	\frac{1}{\pi}\e^{-|x_{jk}|^2}.
\end{equation*}
\end{definition}

Indeed, for any $\X\in\textnormal{GUE}$ (where rightmost $=$ largest) with spectrum $\{\lambda_j(\X)\}_{j=1}^n\subset\mathbb{R}$, the rightmost eigenvalue is Tracy-Widom distributed as $n\rightarrow\infty$. Namely, see \cite{F2,TW},
\begin{equation}\label{fi2}
	\lim_{n\rightarrow\infty}\mathbb{P}\left(\max_{j=1,\ldots,n}\Re\lambda_j(\X)\leq\sqrt{2n}+\frac{t}{\sqrt{2}n^{\frac{1}{6}}}\right)=\exp\left[-\int_t^{\infty}(s-t)\big(q(s)\big)^2\,\d s\right],\ \ \ t\in\mathbb{R},
\end{equation}
with $q=q(t)$ uniquely determined by the following two constraints,
\begin{equation}\label{fi3}
	\frac{\d^2q}{\d t^2}=tq+2q^3,\ \ \ \ \ \ \ q(t)\sim\textnormal{Ai}(t)\ \ \textnormal{as}\ t\rightarrow+\infty.
\end{equation}
Equivalently, \eqref{fi2} equals the Fredholm determinant of the Airy integral operator with kernel
\begin{equation}\label{fi4}
	K_{\textnormal{Ai}}(x,y):=\int_0^{\infty}\textnormal{Ai}(x+z)\textnormal{Ai}(z+y)\,\d z,
\end{equation}
acting on $L^2(t,\infty)$, defined in terms of the Airy function $w=\textnormal{Ai}(z)$, cf. \cite[$9.2.2$]{NIST}. On the other hand, for any $\X\in\textnormal{GinUE}$ with spectrum $\{\lambda_j(\X)\}_{j=1}^n\subset\mathbb{C}$, the rightmost eigenvalue is Gumbel distributed as $n\rightarrow\infty$. Namely, see \cite{CESX},
\begin{equation}\label{fi5}
	\lim_{n\rightarrow\infty}\mathbb{P}\left(\max_{j=1,\ldots,n}\Re\lambda_j(\X)\leq\sqrt{n}+\sqrt{\frac{\gamma_n}{4}}+\frac{t}{\sqrt{4\gamma_n}}\right)=\e^{-\e^{-t}},\ \ \ t\in\mathbb{R},
\end{equation}
with $\gamma_n=\frac{1}{2}(\ln n-5\ln\ln n-\ln(2\pi^4))$. Equivalently, \eqref{fi5} equals the Fredholm determinant of the integral operator with kernel
\begin{equation}\label{fi6}
	K_{\textnormal{ext}}(x,y):=\e^{-x}\begin{cases}1,&x=y\\ 0,&x\neq y\end{cases},
\end{equation}
acting on $L^2(t,\infty)$.
\begin{rem} The limit \eqref{fi5} for the \textnormal{GinUE} is a very recent result, as indicated. A more classical law deals with the spectral radius 
\begin{equation*}
	\max_{j=1,\ldots,n}|\lambda_j(\X)|
\end{equation*}
in the same ensemble, see the work by Rider \cite{Ri}: the spectral radius follows the same limiting law \eqref{fi5}, after correctly modifying $\gamma_n$. Furthermore, probabilistic universality is expected to underwrite both, \eqref{fi2} and \eqref{fi5}: the fluctuations of the rightmost eigenvalue are known to obey Tracy-Widom statistics not only in the \textnormal{GUE}, but in a larger class of complex Wigner matrices, cf.  \cite{Sos}, and the Gumbel fluctuations \eqref{fi5} are expected to hold for more general non-Hermitian complex random matrices with independent, identically distributed entries, see \cite{CESX2} for recent progress in this direction.
\end{rem}
In this paper we study a matrix model that interpolates on the level of its rightmost eigenvalue between the edge laws \eqref{fi2},\eqref{fi5}. One of our first objectives will be to determine the underlying integrable system, i.e. we will first answer the following question: which integrable system can interpolate between the Painlev\'e-II dynamical system \eqref{fi3}, that underpins \eqref{fi2}, and the much more elementary dynamical system that holds up the Gumbel distribution function \eqref{fi5}? Equivalently, which integrable system can interpolate between Airy and Poisson point process extreme value statistics?
\subsection{On the edge of a coin} Admittedly, there are several matrix models that transition from Hermitian to non-Hermitian random matrix theory statistics, see for instance \cite{FS1,ACV}, but we shall only focus on one of them: the complex elliptic Ginibre Ensemble (eGinUE).
\begin{definition}[Girko \cite{Gir}, 1985] We say that a random complex matrix $\X\in\mathbb{C}^{n\times n}$ belongs to the \textnormal{eGinUE} if it is of the form
\begin{equation}\label{fi7}
	{\bf X}=\sqrt{\frac{1+\tau}{2}}\H_1+\im\sqrt{\frac{1-\tau}{2}}\H_2,\ \ \ \tau\in[0,1],
\end{equation}
with two independent $\H_1,\H_2\in\textnormal{GUE}$.
\end{definition}
Evidently, varying $\tau$ from $\tau=1$ to $\tau=0$ in \eqref{fi7} allows us to move from an element in the $\textnormal{GUE}$ to an element in the $\textnormal{GinUE}$. But what does this mean for the eigenvalue scaling limits? On one hand, the global one to be precise, it is known \cite{SCSS} that the empirical spectral distribution
\begin{equation*}
	\mu_{\X}(s,t):=\frac{1}{n}\#\big\{1\leq j\leq n:\ \ \Re\lambda_j(\X)\leq s,\ \ \Im\lambda_j(\X)\leq t\big\},\ \ \ (s,t)\in\mathbb{R}^2,
\end{equation*}
of any properly normalized $\X\in\textnormal{eGinUE}$ with spectrum $\{\lambda_j(\X)\}_{j=1}^n\subset\mathbb{C}$ converges almost surely to the uniform distribution on the ellipse
\begin{equation*}
	\big\{z\in\mathbb{C}:\ \left(\frac{\Re z}{1+\tau}\right)^2+\left(\frac{\Im z}{1-\tau}\right)^2<1\big\},
\end{equation*}
see Figure \ref{figure1} below for a visualization. We note that the same almost sure convergence for the empirical spectral distribution holds true in a larger class of elliptic models with non-Gaussian entries, cf. \cite{Gir,Gir2}. 
\begin{center}
\begin{figure}[tbh]
\resizebox{0.3\textwidth}{!}{\includegraphics{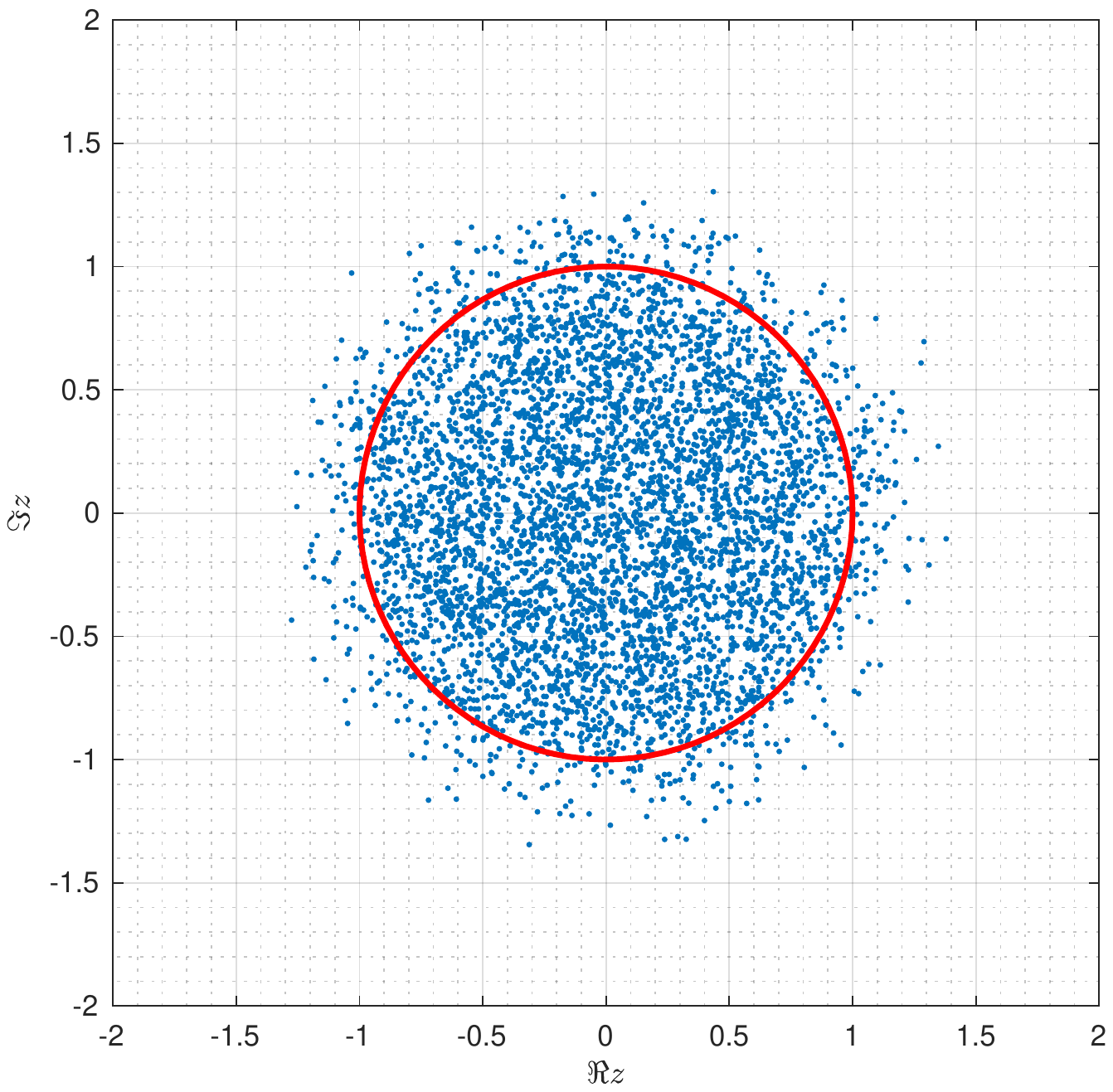}}\ \ \ \ \ \resizebox{0.3\textwidth}{!}{\includegraphics{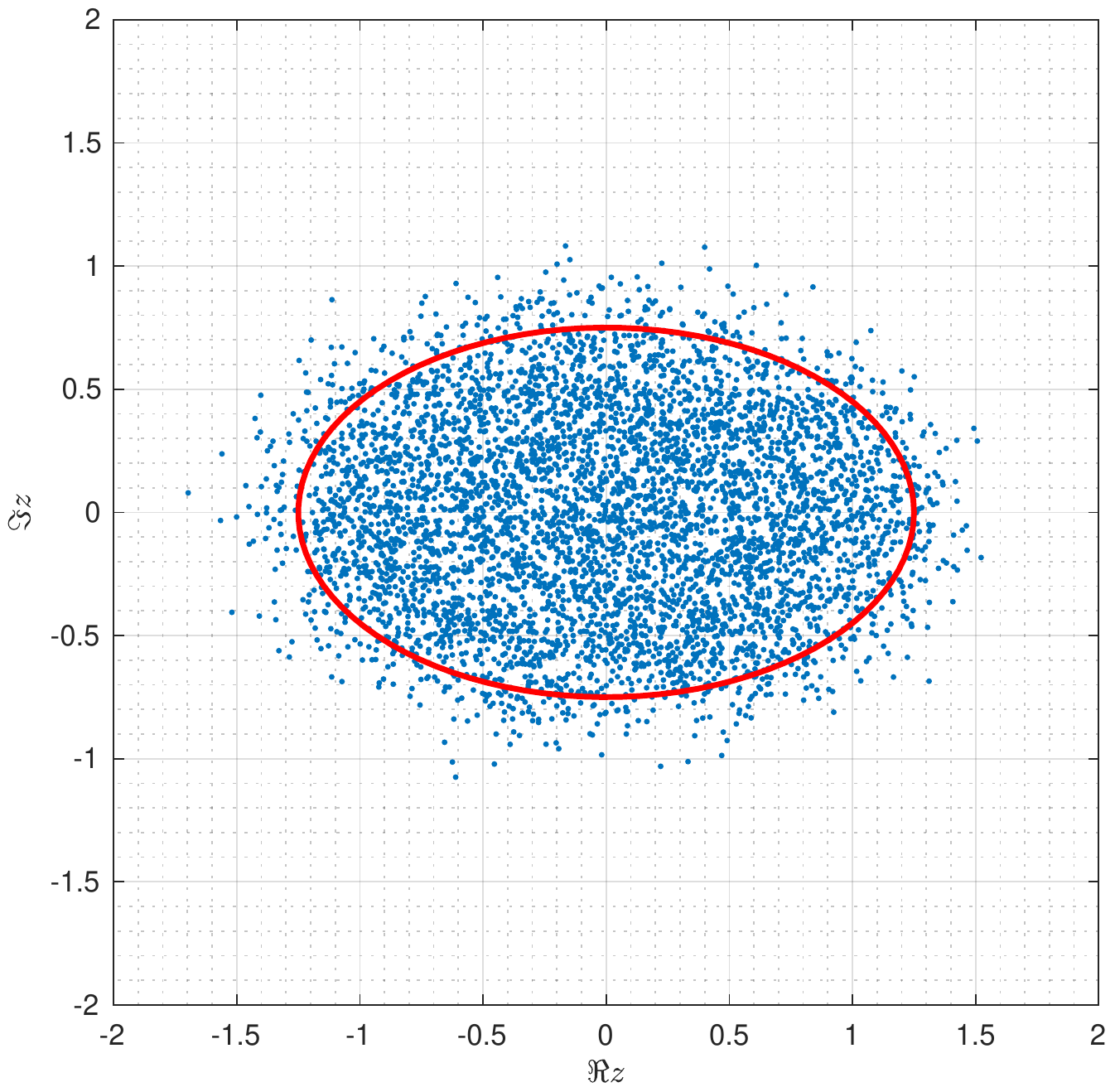}}\ \ \ \ \ \resizebox{0.3\textwidth}{!}{\includegraphics{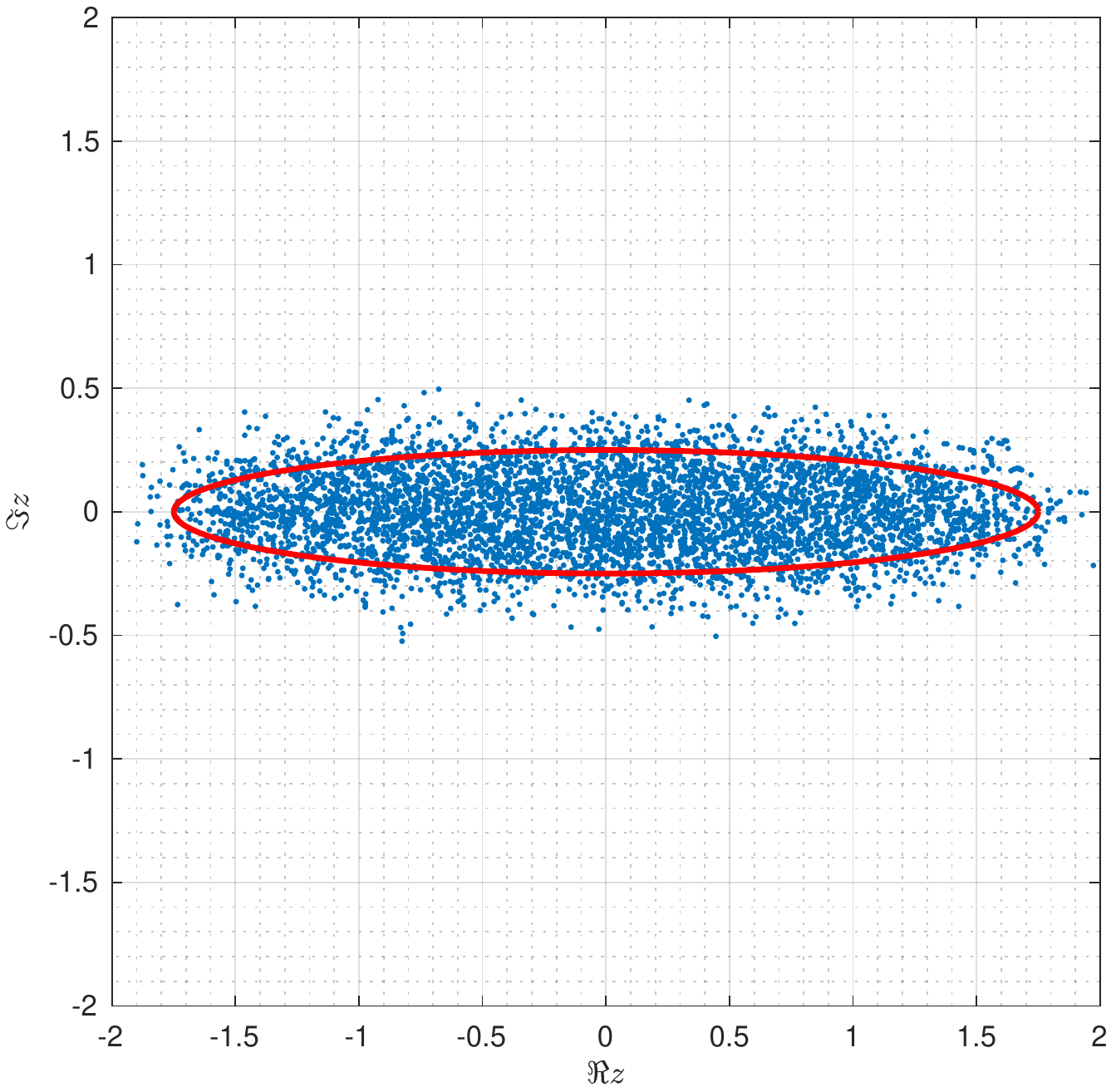}}
\caption{The elliptic law for $500$ complex (normalized) elliptic Ginibre matrices of size $10\times 10$ in comparison with the ellipse boundary. We plot $\tau=0,0.25,0.75$ from left to right.}
\label{figure1}
\end{figure}
\end{center}

On the local hand, we emphasize that the eigenvalues of any $\X\in\textnormal{eGinUE}$ form a determinantal point process in the plane $\mathbb{C}\simeq\mathbb{R}^2$, see \cite{DGIL}. Thus the model's correlation functions can be computed as finite-sized determinants evaluated on a scalar-valued kernel, here a kernel constructed in terms of planar Hermite polynomials. Using said integrable structure of the model, one can then compute the limiting local eigenvalue correlations in a very explicit fashion. Namely, as $n\rightarrow\infty$ and $1-\tau>0$ uniformly in $n$, these coincide with those in the GinUE, compare \cite{FSK}. Since we aim to interpolate between \eqref{fi2} and \eqref{fi5}, we will not consider the same scaling limit, coined the limit of \textit{strong non-Hermiticity}. Instead we will focus on the limit of \textit{weak non-Hermiticity} as pioneered by Fyodorov, Khoruzhenko and Sommers \cite{FKS}. In this limit, one lets $n\rightarrow\infty$ and simultaneously $\tau\uparrow 1$ in a meaningful way. More precisely, on the level of the rightmost eigenvalue, the limit of weak non-Hermiticity was rigorously studied by Bender \cite{Ben}. Setting
\begin{equation*}
	\sigma_n:=n^{\frac{1}{6}}\sqrt{1-\tau_n}>0,\ \ \ \ \ \ \ \ \ (\tau_n)_{n=1}^{\infty}\subset[0,1):\ \tau_n\uparrow 1,
\end{equation*}
he proved two main results, cf. \cite[Theorem $2.3$,$2.5$]{Ben}: first, the existence of centering and scaling constants $a_n,b_n,c_n\in\mathbb{R}$ such that the (rightmost edge) rescaled point process $\{(x_j(\X),y_j(\X))\}_{j=1}^n\subset\mathbb{R}^2$, where
\begin{equation}\label{fi7a}
	\Re\lambda_j(\X)\mapsto x_j(\X)=\frac{\Re\lambda_j(\X)-c_n}{a_n},\ \ \ \ \ \ \ \Im\lambda_j(\X)\mapsto y_j(\X)=\frac{\Im\lambda_j(\X)}{b_n},
\end{equation}
are constructed in terms of the eigenvalue point process $\{\lambda_j(\X)\}_{j=1}^n\equiv\{(\Re\lambda_j(\X),\Im\lambda_j(\X)\}_{j=1}^n\subset\mathbb{R}^2$ of any $\X\in\textnormal{eGinUE}$, converges weakly
\begin{enumerate}
	\item[(1a)] to a Poisson point process on $\mathbb{R}^2$ provided $\sigma_n\rightarrow\infty$ as $n\rightarrow\infty$;
	\item[(1b)] to an interpolating Airy point process on $\mathbb{R}^2$ provided $\sigma_n\rightarrow\sigma\in[0,\infty)$ as $n\rightarrow\infty$.\footnote{See also the recent \cite[Theorem I.$11$]{ADM} for this part.}
\end{enumerate}
Both limiting point processes are determinantal point processes on $\mathbb{R}^2$, and thus characterizable in terms of a scalar-valued kernel. For the Poisson point process in (1a) this kernel equals
\begin{equation}\label{fi8}
	K_{\textnormal{P}}(z_1,z_2):=\frac{1}{\sqrt{\pi}}\e^{-x_1-y_1^2}\begin{cases}1,&z_1=z_2\\ 0,&z_1\neq z_2\end{cases},\ \ \ \ \ \ z_k=(x_k,y_k)\in\mathbb{R}^2,
\end{equation}
and it generalizes \eqref{fi6}. For the interpolating Airy point process in (1b) the kernel equals, see \cite[$(2.7)$]{Ben}, \cite[$(14)$]{AB} and \eqref{k1} below,
\begin{align}\label{fi9}
	K_{\textnormal{Ai}}^{\sigma}&(z_1,z_2)=\frac{1}{\sqrt{\pi}}\exp\left[-\frac{1}{2}(y_1^2+y_2^2)+\frac{1}{2}\sigma^2(x_1+\im \sigma y_1+x_2-\im \sigma y_2)+\frac{1}{6}\sigma^6\right]\nonumber\\
	&\times\int_0^{\infty}\e^{s\sigma^2}\textnormal{Ai}\left(x_1+\im \sigma y_1+\frac{1}{4}\sigma^4+s\right)\textnormal{Ai}\left(x_2-\im \sigma y_2+\frac{1}{4}\sigma^4+s\right)\d s,\ \ \ \ \ z_k=(x_k,y_k)\in\mathbb{R}^2,
\end{align}
which generalizes \eqref{fi4}. Moreover, the determinantal point process on $\mathbb{R}^2$ with correlation kernel \eqref{fi9} interpolates, after centering and scaling, with varying $\sigma\in[0,\infty]$ between the $\mathbb{R}^2$-modified Airy point process with correlation kernel
\begin{equation*}
	M_{\textnormal{Ai}}(z_1,z_2):=\frac{1}{\sqrt{\pi}}\e^{-\frac{1}{2}(y_1^2+y_2^2)}K_{\textnormal{Ai}}(x_1,x_2),\ \ \ z_k=(x_k,y_k)\in\mathbb{R}^2,
\end{equation*}
when $\sigma=0$ and the Poisson point process \eqref{fi8} when $\sigma=\infty$. Second, Bender proved that the rescaled point process $\{(x_j(\X),y_j(\X))\}_{j=1}^n$, see \eqref{fi7a} with appropriate $a_n,b_n,c_n$, has a rightmost particle almost surely and its distribution function converges pointwise 
\begin{enumerate}
	\item[(2a)] to the Fredholm determinant of the integral operator $K_{\textnormal{P}}$ with kernel \eqref{fi8} acting on $L^2((t,\infty)\times\mathbb{R})$ (and so to the Gumbel distribution function \eqref{fi5}) if $\sigma_n\rightarrow\infty$ as $n\rightarrow\infty$;
	\item[(2b)] to the Fredholm determinant of the integral operator $K_{\textnormal{Ai}}^{\sigma}$ with kernel \eqref{fi9} acting on $L^2((t,\infty)\times\mathbb{R})$ if $\sigma_n\rightarrow\sigma\in[0,\infty)$ as $n\rightarrow\infty$.
\end{enumerate}
Consequently the Fredholm determinant of $K_{\textnormal{Ai}}^{\sigma}$ on $L^2((t,\infty)\times\mathbb{R})$ interpolates with varying $\sigma$ between the right hand sides in \eqref{fi2} and \eqref{fi5}. Our first result identifies the integrable system that underwrites the same interpolation. 
\subsection{Main results} Let $\mathbb{R}\ni t\mapsto F_{\sigma}(t)$ denote the Fredholm determinant of $K_{\textnormal{Ai}}^{\sigma}$ on $L^2((t,\infty)\times\mathbb{R})$, equivalently the limiting distribution function of the rightmost eigenvalue in the eGinUE, in the limit of weak non-Hermiticity. By Lemma \ref{lem1}, the same function is well-defined for all $\sigma\geq 0$ and it relates to Painlev\'e special function theory in the following fashion:
\begin{theo}\label{fimain1} For all $(t,\sigma)\in\mathbb{R}\times(0,\infty)$,
\begin{equation}\label{fi10}
	F_{\sigma}(t)=\exp\left[-\int_t^{\infty}(s-t)\left\{\int_{-\infty}^{\infty}\big(p_{\sigma}(s,y)\big)^2\,\d\nu_{\sigma}(y)\right\}\d s\right],\ \ \ \ \ \frac{\d\nu_{\sigma}}{\d\lambda}(\lambda)=\frac{1}{\sigma\sqrt{\pi}}\exp\left[-\Big(\frac{\lambda}{\sigma}\Big)^2\right]
\end{equation}
where $p_{\sigma}=p_{\sigma}(t,y)$ solves the integro-differential Painlev\'e-II equation
\begin{equation}\label{fi11}
	\frac{\partial^2}{\partial t^2}p_{\sigma}(t,y)=\left[t+y+2\int_{-\infty}^{\infty}\big(p_{\sigma}(t,\lambda)\big)^2\d\nu_{\sigma}(\lambda)\right]p_{\sigma}(t,y),\ \ \ (t,y)\in\mathbb{R}^2,
\end{equation}
with boundary constraint
\begin{equation}\label{fi11a}
	p_{\sigma}(t,y)\sim\textnormal{Ai}(t+y)\ \ \textnormal{as}\ \ t\rightarrow+\infty\ \ \textnormal{pointwise in}\ (y,\sigma)\in\mathbb{R}\times(0,\infty).
\end{equation}
\end{theo}
Evidently, \eqref{fi10} constitutes a generalization of the Tracy-Widom representation formula \eqref{fi2}, and \eqref{fi11} a generalization of its underlying Hastings-McLeod Painlev\'e-II dynamical system, cf. \cite{HM}. Similar formul\ae\, which involve integro-differential Painlev\'e-II functions have appeared over the past decade in the context of the narrow wedge solution of the KPZ equation \cite{ACQ} and in the context of finite-temperature free fermionic models \cite{DDMS,Kra,BCT}, albeit with different weights $\frac{\d}{\d\lambda}\nu_{\sigma}$ in both instances. We will discuss some of these weight choices later on. To the best of our knowledge, identity \eqref{fi10} marks the first appearance of an integro-differential Painlev\'e-II function in a distribution function of random matrix theory, although the analysis of the grand-canonical version of the Moshe-Neuberger-Shapiro matrix model in \cite{Joh,LW} could have produced an integro-differential Painlev\'e-II connection as well, but Painlev\'e was not mentioned in loc. cit.\footnote{In a different context that focuses on multiplicative statistics of a class of unitarily invariant random matrix models, the recent work \cite{GS} has also uncovered an integro-differential Painlev\'e-II connection.}
\begin{rem} If one focuses solely on combinations of delta point masses for the weight $\frac{\d}{\d\lambda}\nu_{\sigma}$, noting that $\sigma\downarrow 0$ produces a delta mass at zero in \eqref{fi10}, then \eqref{fi10} and \eqref{fi11} have appeared in several other papers which also include higher-order generalizations of \eqref{fi11}. We refer the interested reader to \cite{CCG,BBW,LMS,KZ}, although those studies won't play a role in the following.
\end{rem}
\begin{rem} Next to the integro-differential Painlev\'e-II connection in \eqref{fi10},\eqref{fi11}, the limiting distribution function $F_{\sigma}(t)$ also relates to KdV. Indeed, combining our Theorem \ref{fimain1} with \cite[Theorem $1.3$]{CCR}, we see that the function
\begin{equation*}
	Q(X,T):=F_{\sigma}(t)\Big|_{\sigma=T^{\frac{2}{3}},\ t=-XT^{-\frac{1}{3}}},\ \ \ \ T>0,\ X\in\mathbb{R},
\end{equation*}
is such that $U(X,T):=\frac{X}{2T}+\frac{\partial^2}{\partial X^2}\ln Q(X,T)$ solves the KdV equation
\begin{equation*}
	\frac{\partial U}{\partial T}+2U\frac{\partial U}{\partial X}+\frac{1}{6}\frac{\partial^3U}{\partial X^3}=0.
\end{equation*}
\end{rem}
\begin{rem} Introducing the potential function
\begin{equation*}
	q_{\sigma}(t):=2\int_{-\infty}^{\infty}\big(p_{\sigma}(t,\lambda)\big)^2\d\nu_{\sigma}(\lambda),
\end{equation*}
in the style of the Deift-Trubowitz trace formula \cite[page $124$]{DT}, one finds from \eqref{fi11} the Stark equation
\begin{equation*}
	\left[-\frac{\partial^2}{\partial t^2}+t+q_{\sigma}(t)\right]p_{\sigma}(t,y)=-yp_{\sigma}(t,y),\ \ \ p_{\sigma}(t,y)\sim\textnormal{Ai}(t+y),\ \ \ t\rightarrow+\infty.
\end{equation*}
\end{rem}
One of the first advantages of the representation formula \eqref{fi10},\eqref{fi11} is that it allows us to identify the Fredholm determinant of $K_{\textnormal{Ai}}^{\sigma}$ on $L^2((t,\infty)\times\mathbb{R})$ with a simpler Fredholm determinant, using the framework of \cite[Section $5$]{ACQ}. Namely, if $N_{\sigma}$ denotes the integral operator on $L^2(t,\infty)$ with kernel
\begin{equation}\label{fi12}
	N_{\sigma}(x,y):=\int_{-\infty}^{\infty}\Phi\Big(\frac{z}{\sigma}\Big)\textnormal{Ai}(x+z)\textnormal{Ai}(z+y)\,\d z,\ \ \ \ \ \Phi(z):=\frac{1}{\sqrt{\pi}}\int_{-\infty}^z\e^{-x^2}\,\d x=1-\frac{1}{2}\textnormal{erfc}(z),\ \ z\in\mathbb{C},
\end{equation}
using the complementary error function $w=\textnormal{erfc}(z)$, cf. \cite[$7.2.2$]{NIST}, then \eqref{fi10},\eqref{fi11} show that $F_{\sigma}(t)$ is simply the Fredholm determinant of the operator $N_{\sigma}$ with \textit{finite-temperature} Airy kernel \eqref{fi12} on $L^2(t,\infty)$. In short, Theorem \ref{fimain1} encodes a correspondence between the operator $K_{\textnormal{Ai}}^{\sigma}$ on $L^2((t,\infty)\times\mathbb{R})$ and the operator $N_{\sigma}$ on $L^2(t,\infty)$. The same correspondence is very useful in the scaling analysis of $\max_{j=1,\ldots,n}\Re\lambda_j(\X)$ for $\X\in\textnormal{eGinUE}$! Namely, it first yields a streamlined derivation of the aforementioned degenerations of the point process on $\mathbb{R}^2$ determined by \eqref{fi9} when $\sigma\downarrow 0$ and $\sigma\rightarrow+\infty$, on the level of the rightmost particle distribution function. The details are as follows:
\begin{cor}[{\cite[Theorem $2.3$]{Ben}}]\label{corfi1} Let $\sigma>1$ and set
\begin{equation*}
	a_{\sigma}:=\frac{\sigma}{\sqrt{6\ln\sigma}},\ \ \ \ \ \ \  c_{\sigma}:=a_{\sigma}\left(3\ln\sigma-\frac{5}{4}\ln(6\ln\sigma)-\ln(2\pi)\right).
\end{equation*}
Then for any fixed $t\in\mathbb{R}$,
\begin{equation}\label{fi12a}
	\lim_{\sigma\rightarrow\infty}F_{\sigma}(c_{\sigma}+a_{\sigma}t)=\sum_{n=0}^{\infty}\frac{(-1)^n}{n!}\int_{(t,\infty)^n}\det\big(K_{\textnormal{ext}}(x_i,x_j)\big)_{i,j=1}^n\,\d x_1\cdots\d x_n
	\stackrel{\eqref{fi6}}{=}\e^{-\e^{-t}},
\end{equation}
in terms of the integral operator $K_{\textnormal{ext}}$ with kernel \eqref{fi6}.
\end{cor}
\begin{cor}[{\cite[Theorem $2.5$(ii)]{Ben}}]\label{corfi2} For any fixed $t\in\mathbb{R}$,
\begin{equation}\label{fi13}
	\lim_{\sigma\downarrow 0}F_{\sigma}(t)=\sum_{n=0}^{\infty}\frac{(-1)^n}{n!}\int_{(t,\infty)^n}\det\big(K_{\textnormal{Ai}}(x_i,x_j)\big)_{i,j=1}^n\,\d x_1\cdots\d x_n
	\stackrel[\eqref{fi2}]{\eqref{fi4}}{=}\exp\left[-\int_t^{\infty}(s-t)\big(q(s)\big)^2\d s\right]
\end{equation}
in terms of the Airy integral operator $K_{\textnormal{Ai}}$ with kernel \eqref{fi4} and the Hastings-McLeod Painlev\'e-II transcendent $q=q(t)$ in \eqref{fi3}.
\end{cor}
Observe that Corollaries \ref{corfi1} and \ref{corfi2} capture the transition from Poisson to Airy point process extreme value statistics for the rightmost eigenvalue as the degree of non-Hermiticity decreases from $\sigma=\infty$ to $\sigma=0$. The limit \eqref{fi13} is easy to derive even when working with $K_{\textnormal{Ai}}^{\sigma}$ on $L^2((t,\infty)\times\mathbb{R})$. However, \eqref{fi12a} is much harder and its derivation made up a significant part of the workings in \cite{Ben}. This was because of Bender's use of double contour integral formul\ae\,which we can bypass all together when working with $N_{\sigma}$ on $L^2(t,\infty)$ instead of $K_{\textnormal{Ai}}^{\sigma}$ on $L^2((t,\infty)\times\mathbb{R})$. 
\begin{rem} Other probabilistic models that encode transitions from Poisson to Airy point process extreme values statistics can be found in \cite{Joh,LW}, on the Moshe-Neuberger-Shapiro model, in \cite{DDMS}, on non-interacting fermions at finite temperature, and in \cite{BB} on periodic Schur processes.
\end{rem}
Still, \eqref{fi12a} and \eqref{fi13} had been found in \cite{Ben} without knowing Theorem \ref{fimain1}. This brings us to the second, and more important, advantage of the representation formula \eqref{fi10},\eqref{fi11}: seeing that the integral operator $N_{\sigma}$ is related to an \textit{integrable operator} in the sense of \cite{IIKS}, we can use the correspondence between $K_{\textnormal{Ai}}^{\sigma}$ and $N_{\sigma}$ to derive detailed tail expansions for the random variable $\max_{j=1,\ldots,n}\Re\lambda_j(\X)$ with $\X\in\textnormal{eGinUE}$. Such expansions, even for the right $t\rightarrow+\infty$ tail, are by no means trivial due to the presence of $\sigma>0$. Indeed, the trace norm of $N_{\sigma}$ on $L^2(t,\infty)$ equals
\begin{equation*}
	\|N_{\sigma}\|_1=\int_t^{\infty}\int_{-\infty}^{\infty}\Phi\Big(\frac{z}{\sigma}\Big)\big(\textnormal{Ai}(x+z)\big)^2\,\d z\,\d x=\int_{-\infty}^{\infty}\Phi\Big(\frac{z-t}{\sigma}\Big)K_{\textnormal{Ai}}(z,z)\,\d z,
\end{equation*}
and so with \cite[$7.8.3$]{NIST}, for all $t\in\mathbb{R}$ and $\sigma\geq 0$,
\begin{equation*}
	\|N_{\sigma}\|_1\geq\frac{1}{2}\|K_{\textnormal{Ai}}\|_1+\frac{\sigma}{2\sqrt{\pi}}\int_0^{\infty}\frac{\e^{-u^2}}{1+u}\left[\int_{t-u\sigma}^{\infty}\big(\textnormal{Ai}(y)\big)^2\d y\right]\d u.
\end{equation*}
Thus $\|N_{\sigma}\|_1$ becomes unbounded as $t\rightarrow-\infty$ for all $\sigma\geq 0$ and as $t\rightarrow+\infty$ once $\sigma\geq ct$ with $c>0$, indicating that simple norm estimates for $N_{\sigma}$ and its powers cannot yield sufficient control over $F_{\sigma}(t)$ in the same asymptotic regimes. To circumvent this issue we shall exploit the integrable structure of $N_{\sigma}$ and employ a Riemann-Hilbert nonlinear steepest descent approach \cite{DZ} in the derivation of tail expansions for $F_{\sigma}(t)$ as $t\rightarrow\pm\infty$. The same analysis leads to the following results, first dealing with the $t\rightarrow+\infty$ tail:
\begin{theo}\label{fimain2} For any $\epsilon\in(0,1)$, there exist $c=c(\epsilon)>0$ and $t_0=t_0(\epsilon)>0$ so that
\begin{equation}\label{fi14}
	F_{\sigma}(t)=1-A(t,\sigma)\e^{-B(t,\sigma)}\big(1+r(t,\sigma)\big),
\end{equation}
for $t\geq t_0$ and $0\leq\sigma\leq t^{\epsilon}$. The functions $A,B$ equal
\begin{equation}\label{fi15}
	\begin{cases}\displaystyle B(t,\sigma):=\frac{4}{3}t^{\frac{3}{2}}\bigg(1+\frac{\sigma^4}{4t}\bigg)^{\frac{3}{2}}-t\sigma^2-\frac{\sigma^6}{6}\bigskip\\
	\displaystyle A(t,\sigma):=\frac{1}{2\pi t^{\frac{3}{2}}}\bigg(\sqrt{4+\frac{\sigma^4}{t}}-\frac{\sigma^2}{\sqrt{t}}\bigg)^{-\frac{5}{2}}\bigg(4+\frac{\sigma^4}{t}\bigg)^{-\frac{1}{4}}\\
	\end{cases},\ \ \ \ t>0,\ \sigma\geq 0,
\end{equation}
the error term $r(t,\sigma)$ is differentiable with respect to $t$ and satisfies
\begin{equation*}
	\big|r(t,\sigma)\big|\leq c\begin{cases}t^{-\frac{3}{2}},&\epsilon\in(0,\frac{1}{4}]\smallskip\\ t^{-1+\epsilon},&\epsilon\in(\frac{1}{4},1)
	\end{cases}\ \ \ \ \ \ \ \forall\,t\geq t_0,\ \ 0\leq\sigma\leq t^{\epsilon}.
\end{equation*}
\end{theo}
For any fixed $\sigma\geq 0$, Taylor-expansion at $t=\infty$ in \eqref{fi15} yields
\begin{equation}\label{fi16}
	B(t,\sigma)=\frac{4}{3}t^{\frac{3}{2}}+\frac{1}{2}\sigma^4\sqrt{t}-t\sigma^2-\frac{\sigma^6}{6}+\mathcal{O}\big(t^{-\frac{1}{2}}\big),\ \ \ \ \ A(t,\sigma)=\frac{1}{16\pi t^{\frac{3}{2}}}\Big(1+\mathcal{O}\big(t^{-\frac{1}{2}}\big)\Big),
\end{equation}
and thus, by taking $\sigma\downarrow 0$, \eqref{fi14} reproduces the known right tail behavior of the Tracy-Widom distribution \eqref{fi2},
\begin{equation}\label{fi17}
	\exp\left[-\int_t^{\infty}(s-t)\big(q(s)\big)^2\d s\right]=1-\frac{1}{16\pi t^{\frac{3}{2}}}\exp\left[-\frac{4}{3}t^{\frac{3}{2}}\right]\big(1+o(1)\big),\ \ \ \ t\rightarrow+\infty,
\end{equation}
see for instance \cite[$(1.17)$]{BDT}. Still, \eqref{fi16} shows that one immediately observes the influence of $\sigma>0$ in the leading order $t\rightarrow+\infty$ asymptotics of $1-F_{\sigma}(t)$. Also, \eqref{fi15} indicates another qualitative change once $\sigma^4\propto t$, for then we cannot expand as in \eqref{fi16} but must use the below equivalent representations for $A,B$,
\begin{equation}\label{fi17a}
	B(t,\sigma)=\frac{\sigma^6}{6}\bigg(1+\frac{4t}{\sigma^4}\bigg)^{\frac{3}{2}}-t\sigma^2-\frac{\sigma^6}{6},\ \ \ \ \ A(t,\sigma)=\frac{\sigma^4}{64\pi t^{\frac{5}{2}}}\bigg(1+\sqrt{1+\frac{4t}{\sigma^4}}\bigg)^{\frac{5}{2}}\bigg(1+\frac{4t}{\sigma^4}\bigg)^{-\frac{1}{4}}.
\end{equation}
These show, in particular, that as $t\rightarrow+\infty$ and $t^{\frac{1}{4}}<\sigma\leq t^{\epsilon}$ with $\epsilon\in(\frac{1}{4},1)$,
\begin{equation*}
	1-F_{\sigma}(t)=\exp\left[-\frac{t-c_{\sigma}}{a_{\sigma}}+\textnormal{remainder}\right],
\end{equation*}
using $a_{\sigma},c_{\sigma}$ as in Corollary \ref{corfi1} and with a non-negative remainder term that is growing in the same regime. In other words, we see the leading order right tail of the Gumbel distribution
\begin{equation}\label{fi18}
	\e^{-\e^{-t}}=1-\e^{-t}\big(1+o(1)\big),\ \ \ t\rightarrow+\infty,
\end{equation}
emerging in \eqref{fi14}, as predicted by \eqref{fi12a}, but Theorem \ref{fimain2} is not sufficient to capture the full crossover between \eqref{fi17} and \eqref{fi18}, simply because $\sigma$ is constrained by $t^{\epsilon},\epsilon\in(0,1)$ from above in \eqref{fi14}. Instead, the complete crossover between \eqref{fi17} and \eqref{fi18} will be achieved by combining Theorem \ref{fimain2} with the following result:
\begin{theo}\label{fimain3} There exist $t_0,\sigma_0\geq 1$ and $c>0$ so that
\begin{equation}\label{fi19}
	F_{\sigma}(t)=\exp\left[\sigma^{\frac{3}{2}}C\Big(\frac{t}{\sigma}\Big)+\frac{1}{4}\int_{\frac{t}{\sigma}}^{\infty}\left\{\frac{\d}{\d u}D(u)\right\}^2\d u\right]\big(1+r(t,\sigma)\big),
\end{equation}
for $t\geq t_0$ and $\sigma\geq\sigma_0$. Here,
\begin{equation}\label{rej3}
	C(x):=\frac{1}{\pi}\int_0^{\infty}\sqrt{y}\,\ln\Phi(x+y)\,\d y,\ \ \ \ \ \ D(x):=\frac{1}{\pi}\int_0^{\infty}\frac{1}{\sqrt{y}}\ln\Phi(x+y)\,\d y,\ \ \ \ \ x\in\mathbb{R},
\end{equation}
are defined in terms of  $\Phi=\Phi(z)$, see \eqref{fi12}, and the error term $r(t,\sigma)$, which is differentiable with respect to $(t,\sigma)$, satisfies the estimate
\begin{equation*}
	\big|r(t,\sigma)\big|\leq\frac{c}{\min\{t,\sigma\}}\ \ \ \ \ \ \ \ \forall\,t\geq t_0,\ \sigma\geq\sigma_0.
\end{equation*}
\end{theo}
Given that $t,\sigma$ are uncoupled in Theorem \ref{fimain3}, expansion \eqref{fi14} and \eqref{fi19} together capture the $t\rightarrow+\infty$ asymptotics of $F_{\sigma}(t)$ uniformly in $\sigma\geq 0$, i.e. uniformly in the degree of non-Hermiticity.
\begin{rem}\label{KPZrem} The asymptotic analysis of families of finite-temperature Airy-kernel Fredholm determinants has recently attracted interest in mathematics as evidenced by  \cite{CC,CCR,ChCR}, foremost due to its connection to the narrow wedge KPZ solution. In a nutshell, the asymptotic problem consists in studying the Fredholm determinant on $L^2(t,\infty)$ of the integral operator with kernel
\begin{equation}\label{fi19a}
	F_{\alpha}(x,y)=\int_{-\infty}^{\infty}\textnormal{Ai}(x+z)\textnormal{Ai}(z+y)w(z,\sigma)\,\d z,
\end{equation}
 as $t\rightarrow\pm\infty$, noting that the weight $w$ depends on an auxiliary parameter $\sigma>0$, which has a significant impact on the asymptotics. The abovementioned three papers have achieved the same as $t\rightarrow-\infty$ for some weight families, see our discussion right before Theorem \ref{fimain4} below, and some values of $\sigma$. Our Theorems \ref{fimain2} and \ref{fimain3} mark the first rigorous study on the $t\rightarrow+\infty$ asymptotics of finite-temperature Airy-kernel Fredholm determinants and we achieve it for all allowed values of the auxiliary parameter, albeit only for one weight.\footnote{In theoretical physics, right tail asymptotics for finite-temperature Airy kernel determinants with Fermi weights have been computed in \cite{DMRS,D,KD}. We thank Tom Claeys for drawing our attention to the same literature.}
\end{rem}

 By the asymptotic properties of the complementary error function, see \cite[$7.12.1$]{NIST}, and by Laplace's method \cite[Chapter $3$, Theorem $8.1$]{Olv},
 \begin{equation*}
 	C(\omega)=\frac{1}{\pi}\int_0^{\infty}\sqrt{y}\,\ln\Phi(\omega+y)\,\d y=-\frac{\omega^{-\frac{5}{2}}}{8\pi\sqrt{2}}\,\e^{-\omega^2}\Big(1+\mathcal{O}\big(\omega^{-2}\big)\Big),\ \ \ \ \omega\rightarrow+\infty,
\end{equation*}
followed by
\begin{equation*}
	\frac{\d}{\d\omega}D(\omega)=\frac{\omega^{-\frac{1}{2}}}{\pi\sqrt{2}}\,\e^{-\omega^2}\Big(1+\mathcal{O}\big(\omega^{-2}\big)\Big),\ \ \ \ \ \omega\rightarrow+\infty,
\end{equation*}
so that back in \eqref{fi19},
\begin{equation}\label{fi20}
	F_{\sigma}(t)=\exp\left[-\frac{\sigma^{\frac{3}{2}}\omega^{-\frac{5}{2}}}{8\pi\sqrt{2}}\,\e^{-\omega^2}\Big(1+\mathcal{O}\big(\omega^{-2}\big)\Big)\right]\big(1+r(t,\sigma)\big)\ \ \ \ \forall\,\sigma\geq \sigma_0\ \ \ \textnormal{with}\ \ \omega:=\frac{t}{\sigma}\rightarrow+\infty.
\end{equation}
Note that \eqref{fi20} holds true in particular as $t\rightarrow+\infty$ and $t^{\frac{1}{4}+\delta}\leq\sigma\leq t^{\epsilon}$ with $\epsilon\in(\frac{1}{4}+\delta,1)$ and $\delta\in(0,\frac{3}{4})$ and so from \eqref{fi19} via \eqref{fi20}, in the same regime,
\begin{equation}\label{fi21}
	1-F_{\sigma}(t)\stackrel{\eqref{fi19}}{=}\frac{\sigma^4t^{-\frac{5}{2}}}{8\pi\sqrt{2}}\exp\left[-\Big(\frac{t}{\sigma}\Big)^2\right]\big(1+o(1)\big).
\end{equation}
Observe that \eqref{fi21} matches onto the expansion of $1-F_{\sigma}(t)$ coming from \eqref{fi14} when using the aforementioned equivalent representations for $A,B$. Namely, when $t\rightarrow+\infty$ and $t^{\frac{1}{4}+\delta}\leq\sigma\leq t^{\epsilon}$ with $\epsilon\in(\frac{1}{4}+\delta,1)$ and $\delta\in(0,\frac{3}{4})$, then
\begin{equation*}
	B(t,\sigma)=\Big(\frac{t}{\sigma}\Big)^2-\frac{2t^3}{3\sigma^6}+\mathcal{O}\big(t^4\sigma^{-10}\big),\ \ \ \ \ \ \ \ A(t,\sigma)=\frac{\sigma^4t^{-\frac{5}{2}}}{8\pi\sqrt{2}}\big(1+o(1)\big),
\end{equation*}
and thus, in the same regime with $\delta\in(\frac{1}{4},\frac{3}{4})$,
\begin{equation}\label{fi22}
	1-F_{\sigma}(t)\stackrel{\eqref{fi14}}{=}\frac{\sigma^4t^{-\frac{5}{2}}}{8\pi\sqrt{2}}\exp\left[-\Big(\frac{t}{\sigma}\Big)^2\right]\big(1+o(1)\big),
\end{equation}
which matches onto \eqref{fi21}. In short, Theorem \ref{fimain2} and \ref{fimain3} are consistent as evidenced by \eqref{fi21} and \eqref{fi22}. Furthermore, \eqref{fi19} matches onto the Gumbel right tail behavior \eqref{fi18}, for with \eqref{fi20}, when $t,\sigma\rightarrow+\infty$ using $a_{\sigma},c_{\sigma}$ as in Corollary \ref{corfi1},
\begin{equation*}
	F_{\sigma}(c_{\sigma}+a_{\sigma}t)=\exp\left[-\e^{-t}\Big(1+\mathcal{O}\Big(\frac{t^2}{\ln\sigma}\Big)+\mathcal{O}\Big(\frac{t\ln\ln\sigma}{\ln\sigma}\Big)+o(1)\Big)\right]\big(1+o(1)\big),
\end{equation*}
and so we find indeed the Gumbel right tail behavior \eqref{fi18} as soon $\frac{t^2}{\ln\sigma}\rightarrow 0$ and $\frac{t\ln\ln\sigma}{\ln\sigma}\rightarrow 0$, which are admissible constraints for $t,\sigma\rightarrow+\infty$ in Theorem \ref{fimain3}. In summary, Theorem \ref{fimain2} and \ref{fimain3} together capture the full crossover between \eqref{fi17} and \eqref{fi18}. This completes our discussion of the same two results.\bigskip

The $t\rightarrow-\infty$ asymptotics of $F_{\sigma}(t)$ concern the left tail behavior of the random variable $\max_{j=1,\ldots,n}\Re\lambda_j(\X)$ for $\X\in\textnormal{eGinUE}$. Once more, the same asymptotics are non-trivial and very sensitive to $\sigma>0$ for, by Corollary \ref{corfi1} and \ref{corfi2}, we now interpolate between the left tail behavior of the Tracy-Widom distribution
\begin{equation}\label{fi23}
	\exp\left[-\int_t^{\infty}(s-t)\big(q(s)\big)^2\d s\right]=\exp\left[\frac{t^3}{12}-\frac{1}{8}\ln|t|+\frac{1}{24}\ln 2+\zeta'(-1)\right]\big(1+o(1)\big),\ \ \ \ t\rightarrow-\infty,
\end{equation}
see for instance \cite{TW,BBD,DIK} with the Riemann zeta function $w=\zeta(z)$, and the left tail of the Gumbel distribution
\begin{equation}\label{fi24}
	\e^{-\e^{-t}}=\e^{-\e^{-t}}\big(1+o(1)\big),\ \ \ t\rightarrow-\infty.
\end{equation}
In itself, the same crossover constitutes one particular problem in the asymptotic analysis of finite-temperature Airy-kernel Fredholm determinants, see Remark \ref{KPZrem}. More concretely, the $t\rightarrow-\infty$ studies in \cite{CC,CCR,ChCR} of integral operators on $L^2(t,\infty)$ with kernel \eqref{fi19a} have focused on weights of the form $w(z,\sigma)=v(\frac{z}{\sigma})$ where\smallskip 
\begin{enumerate}
	\item[(i)] $v:\mathbb{R}\rightarrow[0,1]$ in \cite{CCR} is assumed to be smooth almost everywhere on $\mathbb{R}$, with at most finitely many jump discontinuities and $v$ approaches zero, resp. one, at $-\infty$, resp. $+\infty$ sufficiently fast.
	\item[(ii)] $v:\mathbb{R}\rightarrow[0,1)$ in \cite{CC,ChCR} is assumed to be such that $1/(1-v)$ extends to an entire function, is non-decreasing and log-convex on $\mathbb{R}$, and $1/(1-v)$ approaches one, resp. $+\infty$, at $-\infty$, resp. $+\infty$ exponentially fast.\smallskip
\end{enumerate}
Class (i) is the least understood one when it comes to the $t\rightarrow-\infty$ asymptotics of the underlying Fredholm determinant, still we note that at the time of this writing a complete $t\rightarrow-\infty$ analysis, uniformly in $\sigma>0$, is outstanding for both classes. The reader will notice that \eqref{fi12} fits into class (i) since 
\begin{equation*}
	\frac{1}{1-v(x)}=\frac{2}{\textnormal{erfc}(x)},\ \ \ x\in\mathbb{R},
\end{equation*}
does \textit{not} extend to an entire function because of the complex zeros of the complementary error function, cf. \cite[$7.13$]{NIST}. Also, $\Phi(x)$ approaches zero at $-\infty$ and tends to unity at $+\infty$, in both cases super-exponentially fast, which creates further obstructions for the workings in \cite{CC,ChCR}. In turn, we do not attempt to describe the full crossover between \eqref{fi23} and \eqref{fi24} in this paper, we only deal with the $t\rightarrow-\infty$ asymptotics of $F_{\sigma}(t)$ for values of $\sigma$ very close to $\sigma=0$, i.e. close to Hermiticity, where the absence of analytic extensions of $v$ does not play a significant role in our analysis and where we can borrow techniques from \cite{CCR}. The more challenging parts, where $\sigma$ is bounded away from zero, are left for a future publication.
\begin{cor}[{\cite[Theorem $1.14$(iii)]{CCR}}]\label{fimain4} Let $c>0$. There exist $t_0,d>0$ so that
\begin{equation}\label{fi25}
	F_{\sigma}(t)=\exp\left[\frac{t^3}{12}-\frac{1}{8}\ln|t|-\int_0^{-t\sqrt{\sigma}}(t\sqrt{\sigma}+x)\left(u(x)-\frac{1}{8x^2}\right)\,\d x+\frac{1}{24}\ln 2+\zeta'(-1)\right]\big(1+r(t,\sigma)\big)
\end{equation}
for all $-t\geq t_0$ and $0<\sigma\leq ct^{-2}$. Here, $\zeta=\zeta(z)$ is the Riemann zeta function \cite[$25.2.1$]{NIST}, the function $u=u(x),x>0$ is uniquely determined through the solution of RHP \ref{genBesselRHP}, see \eqref{app4a},\eqref{app4bb},\eqref{rej1}, it has the boundary behavior
\begin{equation}\label{fi26}
	u(x)=\frac{1}{8x^2}+\frac{1}{2}\int_{-\infty}^{\infty}\big(\chi_{[0,\infty)}(y)-\Phi(y)\big)\,\d y+\mathcal{O}\big(x^2\big),\ \ x\downarrow 0,
\end{equation}
and the error term $r(t,\sigma)$ in \eqref{fi25}, which is differentiable with respect to $(t,\sigma)$, satisfies the estimate
\begin{equation*}
	\big|r(t,\sigma)\big|\leq\frac{d}{|t|}.
\end{equation*}
\end{cor}
Expansion \eqref{fi25} is the analogue \cite[$(1.35)$]{CCR} and we choose to derive it here from first principles in order to keep our paper self-contained. Still, all our steps in the derivation of \eqref{fi25} can be traced back to \cite{CCR}. Note that \eqref{fi25} with \eqref{fi26} reproduces \eqref{fi23} for $\sigma=0$ as expected from \eqref{fi13}, however \eqref{fi25} is far off from capturing the full crossover between \eqref{fi23} and the Gumbel left tail \eqref{fi24}.
\begin{rem}\label{rejrem} The special function $u=u(x)$ in \eqref{fi25} relates to a certain integro-differential generalization of the Painlev\'e-V equation, see the workings in \cite[Section $7$]{CCR} that lead to \cite[Theorem $1.12$]{CCR}.
\end{rem}
\subsection{Methodology and outline of paper}
The first obstruction in deriving an integrable system for the limiting distribution function $F_{\sigma}(t)$ originates from the higher-dimensional integration domain that underlies the operator $K_{\textnormal{Ai}}^{\sigma}$ with kernel \eqref{fi9}. Secondly, the kernel \eqref{fi9} seemingly does not display any of the usual integrable structures common in invariant random matrix models where Painlev\'e connections have been discovered in the past. Thus the techniques of \cite{JMMS,IIKS,TW2,BoDei} are not readily available for the analysis of $F_{\sigma}(t)$. To get around this issue we propose and execute two strategies: first, in Section \ref{sec2}, we set out to simplify the operator traces of $K_{\textnormal{Ai}}^{\sigma}$ via integral identities and complex analytic techniques. These steps change the operator, but leave its Fredholm determinant invariant. What results is the much more convenient representation \eqref{e10} for $F_{\sigma}(t)$ that involves a Hankel composition operator with kernel written in \eqref{e9a}. Consequently, with the Hankel composition structure flushed out, we can then draw inspiration from the recent works \cite{Bo, Kra} on Fredholm determinants of Hankel composition integral operators acting on $L^2(J\subset\mathbb{R})$. The methods of \cite{Bo,Kra} need to be adjusted because of the higher-dimensional integration domain $J\subset\mathbb{R}^2$ that we face in the eGinUE and we carry out all necessary steps in the remainder of Section \ref{sec2}. In particular, we first derive the Tracy-Widom type representation formula \eqref{e16} for $F_{\sigma}(t)$ and afterwards identify the underlying integro-differential dynamical system using a coupled integro-differential system in \eqref{e20}, conserved quantities thereof and lastly a closure relation. The outcome of this explicit and constructive, albeit somewhat lengthy, approach is summarized in Proposition \ref{theo1}. We are then left to simplify the same result and Theorem \ref{fimain1} follows. Afterwards, in Section \ref{sec3}, we offer an alternative and much shorter proof for Theorem \ref{fimain1}. However, the same proof is non-constructive as it relies on a curious and somewhat hard-to-guess operator factorization, compare the proof of Proposition \ref{Sylvprop}, and since it rests on the workings in \cite{ACQ}. Still, since our second proof does not use Hankel composition structures, we expect the same method to be applicable to other Fredholm determinants in the eGinUE where Hankel structures are absent, cf. \cite{AP,AP2}. Moving ahead, in Section \ref{sec4} we state our proofs of Corollary \ref{corfi1} and \ref{corfi2}, using the aforementioned correspondence between $F_{\sigma}(t)$ and the operator $N_{\sigma}$ on $L^2(t,\infty)$ with kernel \eqref{fi12}, a correspondence which we make precise in Corollary \ref{theo2}. In particular, our proof of Corollary \ref{corfi1} relies on a novel single contour integral formula \eqref{f2} that simplifies the proof of \eqref{fi12a} considerably when compared to the workings in \cite{Ben}. Next, we address Theorem \ref{fimain2} and \ref{fimain3} in Section \ref{sec5}. Both results are proven within the framework of the Deift-Zhou nonlinear steepest descent method \cite{DZ} when applied to the master Riemann-Hilbert problem \ref{master}. While the nonlinear steepest descent proof of Theorem \ref{fimain2} is straightforward, the workings leading to Theorem \ref{fimain3} are much more technical and comparable to the degree of difficulty encountered in \cite{CC,CCR,ChCR}. In particular, the proof of Theorem \ref{fimain3} relies on localized differential identities, see \eqref{e52} and \eqref{e52aa}, on contour deformations, a $g$-function transformation in \eqref{e62} and on local Riemann-Hilbert model and error analysis. Although they are technical in nature, we expect our proof methods for Theorem \ref{fimain2} and \ref{fimain3} to be applicable to a larger class of weights beyond $w=\Phi(z)$ in \eqref{fi12} and this can be of interest when dealing, for instance, with the exact short-time height distribution of the one-dimensional KPZ equation, cf. \cite{DMRS,D,KD}. Once the right tail analysis is completed we then briefly address the left tail $t\rightarrow-\infty$ of $F_{\sigma}(t)$ for $\sigma$ close to zero. In this case our workings are heavily inspired by the analysis in \cite{CCR} and we provide a proof of Corollary \ref{fimain4} in Section \ref{sec6} only to keep our paper self-contained. The paper closes out with Appendix \ref{appA} and \ref{appB} which summarize several auxiliary results used in the main body of the text.

\section{First proof of Theorem \ref{fimain1}}\label{sec2}
In this section we state our first proof of Theorem \ref{fimain1}, a proof which is constructive and which relies on trace identities and a refinement of the Hankel composition operator techniques from \cite{Kra,Bo}. These refinements are necessary because of the higher-dimensional integration domain that underpins $K_{\textnormal{Ai}}^{\sigma}$. We begin with two preliminary results, asserting that the distribution function $\mathbb{R}\ni t\mapsto F_{\sigma}(t)$ is well-defined for all $\sigma\geq 0$.
\begin{lem}\label{lem0} Set $J_t:=(t,\infty)\times\mathbb{R}\subset\mathbb{R}^2$ with $t\in\mathbb{R}$ and let $K_{\textnormal{Ai}}^{\sigma}:L^2(J_t)\rightarrow L^2(J_t)$ denote the integral operator with kernel \eqref{fi9}. Then, for any $z_k=(x_k,y_k)\in\mathbb{R}^2$ and any $\sigma\geq 0$,
\begin{equation}\label{k1}
	K_{\textnormal{Ai}}^{\sigma}(z_1,z_2)=\frac{\im}{4\pi^{\frac{5}{2}}}\int_{\gamma_1}\int_{\gamma_2}\frac{\e^{\im(\frac{1}{3}\lambda^3+x_1\lambda)}\e^{\im(\frac{1}{3}\mu^3+x_2\mu)}}{\lambda+\mu}\e^{-\frac{1}{2}(\sigma\lambda+y_1)^2}\e^{-\frac{1}{2}(\sigma\mu-y_2)^2}\,\d\lambda\,\d\mu,
\end{equation}
where $\gamma_k=\mathbb{R}+\im\delta_k$ with $\delta_k>0$ for $k\in\{1,2\}$.
\end{lem}
\begin{proof} Identity \eqref{k1} is part of \cite[Section VI]{AB}: one begins with the right hand side of \eqref{k1} and uses algebra combined with the integral identity $\frac{\im}{\lambda+\mu}=\int_0^{\infty}\e^{\im(\lambda+\mu)s}\,\d s$, valid for $(\lambda,\mu)\in\gamma_1\times\gamma_2$. What results is
\begin{align*}
	\textnormal{RHS}\ &\eqref{k1}=\frac{1}{4\pi^{\frac{5}{2}}}\e^{-\frac{1}{2}(y_1^2+y_2^2)+\frac{1}{2}\sigma^2(w_1+w_2)-\frac{1}{12}\sigma^6}\\
	&\times\int_0^{\infty}\int_{\gamma_1}\int_{\gamma_2}\e^{\frac{\im}{3}(\lambda+\frac{\im}{2}\sigma^2)^3+\im w_1(\lambda+\frac{\im}{2}\sigma^2)+\frac{\im}{4}\lambda\sigma^4+\im s\lambda}\,\e^{\frac{\im}{3}(\mu+\frac{\im}{2}\sigma^2)^3+\im w_2(\mu+\frac{\im}{2}\sigma^2)+\frac{\im}{4}\mu\sigma^4+\im s\mu}\,\d\lambda\,\d\mu\,\d s,
\end{align*}
with the shorthand $w_1:=x_1+\im\sigma y_1$ and $w_2:=x_2-\im\sigma y_2$. Then, with $\lambda+\frac{\im}{2}\sigma^2\mapsto\lambda$ and $\mu+\frac{\im}{2}\sigma^2\mapsto\mu$,
\begin{equation}\label{k2}
	\textnormal{RHS}\ \eqref{k1}=\frac{1}{\sqrt{\pi}}\,\e^{-\frac{1}{2}(y_1^2+y_2^2)+\frac{1}{2}\sigma^2(w_1+w_2)+\frac{1}{6}\sigma^6}
	\int_0^{\infty}\e^{s\sigma^2}\textnormal{Ai}\left(w_1+\frac{\sigma^4}{4}+s\right)\textnormal{Ai}\left(w_2+\frac{\sigma^4}{4}+s\right)\d s,
\end{equation}
having used a well-known contour integral formula for the Airy function along the way, namely \cite[$9.5.4$]{NIST},
\begin{equation}\label{k3}
	\textnormal{Ai}(z)=\frac{1}{2\pi}\int_{\mathbb{R}+\im\delta}\e^{\im(\frac{1}{3}\lambda^3+z\lambda)}\,\d\lambda,\ \ \ \ \ \delta>0,\ \ \ z\in\mathbb{C}.
\end{equation}
The proof of \eqref{k1} is now complete as the right hand side of \eqref{k2} matches \eqref{fi9}.
\end{proof}
The kernel representation \eqref{k1} is useful when answering questions about the operator $K_{\textnormal{Ai}}^{\sigma}$, such as the following one.
\begin{lem}\label{lem1} Let $K_{\textnormal{Ai}}^{\sigma}:L^2(J_t)\rightarrow L^2(J_t)$ denote the integral operator with kernel \eqref{k1}. Then $K_{\textnormal{Ai}}^{\sigma}$ is trace class for any $(t,\sigma)\in\mathbb{R}\times[0,\infty)$.
\end{lem}
\begin{proof} Write again $\frac{\im}{\lambda+\mu}=\int_0^{\infty}\e^{\im(\lambda+\mu)s}\,\d s$ when $(\lambda,\mu)\in\gamma_1\times\gamma_2$ and view \eqref{k1} as composition kernel,
\begin{equation*}
	K_{\textnormal{Ai}}^{\sigma}(z_1,z_2)=\int_0^{\infty}G_{\sigma}(x_1+s,y_1)G_{\sigma}(s+x_2,-y_2)\,\d s,\ \ \ G_{\sigma}(a,b):=\frac{1}{2\pi^{\frac{5}{4}}}\int_{\gamma}\e^{\im(\frac{1}{3}\lambda^3+a\lambda)-\frac{1}{2}(\sigma\lambda+b)^2}\,\d\lambda,
\end{equation*}
with contour $\gamma=\mathbb{R}+\im\delta,\delta>0$. Since $\mathbb{R}^2\ni(a,b)\mapsto G_{\sigma}(a,b)$ is continuous and since it satisfies the global bound
\begin{equation}\label{e2}
	\big|G_{\sigma}(a,b)\big|\leq c_1\e^{-a\delta-c_2b^2},\ \ c_k=c_k(\delta,\sigma)>0,\ \ \ (a,b)\in\mathbb{R}^2,
\end{equation}
we deduce continuity of $\mathbb{R}^4\ni(z_1,z_2)\mapsto K_{\textnormal{Ai}}^{\sigma}(z_1,z_2)$. Moreover, for any $z_1,\ldots,z_n\in J_t$, any $\alpha\in\mathbb{C}^n$ and all $\sigma\geq 0$,
\begin{equation*}
	\sum_{j,k=1}^n\alpha_j\overline{\alpha_k}K_{\textnormal{Ai}}^{\sigma}(z_j,z_k)=\int_0^{\infty}\left|\sum_{j=1}^n\alpha_jG_{\sigma}(x_j+s,y_j)\right|^2\d s\geq 0,
\end{equation*}
having used the conjugation symmetry $\overline{G_{\sigma}(a,b)}=G_{\sigma}(a,-b),(a,b)\in\mathbb{R}^2$. But also, with $\d^2z\equiv\d x\,\d y$,
\begin{equation*}
	\int_{J_t}K_{\textnormal{Ai}}^{\sigma}(z,z)\,\d^2z=\int_t^{\infty}\int_{-\infty}^{\infty}\int_0^{\infty}\big|G_{\sigma}(x+s,y)\big|^2\,\d s\,\d y\,\d x\stackrel{\eqref{e2}}{<}\infty,\ \ \ \forall\,(t,\sigma)\in\mathbb{R}\times[0,\infty),
\end{equation*}
so the operator $K_{\textnormal{Ai}}^{\sigma}$ with kernel \eqref{k1} is trace class on $L^2(J_t)$ by \cite[Theorem $2.12$]{Sim}.
\end{proof}
By the last result, the Fredholm determinant of $K_{\textnormal{Ai}}^{\sigma}$ is well-defined for all $(t,\sigma)\in\mathbb{R}\times[0,\infty)$ and it equals
\begin{equation}\label{e3}
	F_{\sigma}(t):=\prod_{k=1}^{\infty}\big(1-\lambda_k(t,\sigma)\big),
	\ \ \ (t,\sigma)\in\mathbb{R}\times[0,\infty),
\end{equation}
with the non-zero eigenvalues $\lambda_k(t,\sigma)$ of $K_{\textnormal{Ai}}^{\sigma}:L^2(J_t)\rightarrow L^2(J_t)$, multiplicities taken into account. We now start working towards our first proof of Theorem \ref{fimain1} by assembling the trace identities \eqref{e5} summarized in the below result.
\begin{prop}[Trace identities]\label{prop1} Let $\sigma\geq 0$ and consider with $z_k=(x_k,y_k)\in\mathbb{R}^2$ the kernel
\begin{equation}\label{e3a}
	K_{\sigma}(z_1,z_2):=\frac{1}{\sqrt{\pi}}\e^{-\frac{1}{2}y_1^2}K_{\textnormal{Ai}}(x_1+\sigma y_1,x_2+\sigma y_2)\e^{-\frac{1}{2}y_2^2},\ \ \ K_{\textnormal{Ai}}(a,b)\stackrel{\eqref{fi4}}{=}\int_0^{\infty}\textnormal{Ai}(a+s)\textnormal{Ai}(s+b)\,\d s.
\end{equation}
Then the self-adjoint integral operator $K_{\sigma}:L^2(J_t)\rightarrow L^2(J_t)$, with $J_t=(t,\infty)\times\mathbb{R}\subset\mathbb{R}^2$, given by
\begin{equation}\label{e4}
	(K_{\sigma}f)(z):=\int_{J_t}K_{\sigma}(z,w)f(w)\,\d^2w,
\end{equation}
is trace class for any $(t,\sigma)\in\mathbb{R}\times[0,\infty)$ and we have for all $n\in\mathbb{Z}_{\geq 0}$ and $(t,\sigma)\in\mathbb{R}\times[0,\infty)$,
\begin{equation}\label{e5}
	\tr_{L^2(J_t)}K_{\sigma}^n=\tr_{L^2(J_t)}(K_{\textnormal{Ai}}^{\sigma})^n.
\end{equation}
Furthermore, $0\leq K_{\sigma}\leq 1$, $\|K_{\sigma}\|\leq 1$ in operator norm on $L^2(J_t)$ and $I-K_{\sigma}$ is invertible on $L^2(J_t)$ for all $(t,\sigma)\in\mathbb{R}\times[0,\infty)$.
\end{prop}
\begin{proof} The Airy kernel function $(a,b)\mapsto K_{\textnormal{Ai}}(a,b)$ is continuous on $\mathbb{R}^2$, hence $\mathbb{R}^4\ni(z_1,z_2)\mapsto K_{\sigma}(z_1,z_2)$ is continuous. Also, for any $z_1,\ldots,z_n\in J_t$, any $\alpha\in\mathbb{C}^n$ and all $\sigma\geq 0$
\begin{equation*}	
	\sum_{j,k=1}^n\alpha_j\overline{\alpha_k}K_{\sigma}(z_j,z_k)=\int_0^{\infty}\left|\sum_{j=1}^n\frac{\alpha_j}{\pi^{\frac{1}{4}}}\e^{-\frac{1}{2}y_j^2}\textnormal{Ai}(x_j+\sigma y_j+s)\right|^2\d s\geq 0.
\end{equation*}
Since also for all $(t,\sigma)\in\mathbb{R}\times[0,\infty)$,
\begin{equation*}
	\int_{J_t}K_{\sigma}(z,z)\,\d^2z=\frac{1}{\sqrt{\pi}}\int_t^{\infty}\int_{-\infty}^{\infty}e^{-y^2}\left\{\int_0^{\infty}\textnormal{Ai}^2(x+\sigma y+s)\,\d s\right\}\d y\,\d x<\infty
\end{equation*}
by the bound $|\textnormal{Ai}(x)|\leq c(\chi_{(-\infty,1]}(x)+\e^{-\frac{1}{2}x}\chi_{(1,\infty)}(x)),x\in\mathbb{R}$ with $c>0$ and the characteristic function $\chi_A$ of a set $A\subset\mathbb{R}$, cf. \cite[$\S 9.7$(ii)]{NIST}, we deduce that \eqref{e4} is trace class on $L^2(J_t)$ by \cite[Theorem $2.12$]{Sim}. Moving ahead we identify $f(z)\equiv f(x,y)\in L^2(\mathbb{R}^2)$ with $z=(x,y)\in\mathbb{R}^2$. Then, for any $f\in L^2(J_t)$,
\begin{equation*}
	0\leq\langle f,K_{\sigma}f\rangle_{L^2(J_t)}=\frac{1}{\sqrt{\pi}}\int_0^{\infty}\left|\int_t^{\infty}\int_{-\infty}^{\infty}\e^{-\frac{1}{2}y^2}\textnormal{Ai}(x+\sigma y+s)f(x,y)\,\d y\,\d x\right|^2\d s,
\end{equation*}
so $K_{\sigma}\geq 0$. Moreover, with $S_{\sigma}:L^2(\mathbb{R}^2)\rightarrow L^2(\mathbb{R}^2)$ given by $(S_{\sigma}f)(x,y):=f(x+\sigma y,y)$ and with $L_t:L^2(J_t)\rightarrow L^2(J_t),t\in\mathbb{R}$ given by
\begin{equation*}
	(L_tf)(x,y):=\int_t^{\infty}\textnormal{Ai}(x+\lambda)f(\lambda,y)\,\d\lambda,\ \ \ \ (x,y)\in J_t,
\end{equation*}
we have $\|S_{\sigma}\|=1$ in operator norm on $L^2(\mathbb{R}^2)$ and $\|L_t\|\leq 1$ in operator norm on $L^2(J_t)$, see \cite[Lemma $8.3$]{HM}. In fact, seeing the function $L_tf$ as defined on $\mathbb{R}^2$, we have $\lim_{t\rightarrow-\infty}\|L_tf\|=\|f\|$ in norms on $L^2(\mathbb{R}^2)$ by \cite[Lemma $8.2$]{HM}. Hence, denoting with $K_{\sigma}$ the operator on $\mathbb{R}^2$ with kernel $K_{\sigma}(z_1,z_2)$, we have for any $f\in L^2(\mathbb{R}^2)$,
\begin{equation}\label{e6}
	\big(S_{\sigma}GL_{-\infty}P_+ L_{-\infty}S_{\sigma}^{-1}f\big)(x_1,y_1)=\int_{-\infty}^{\infty}\int_{-\infty}^{\infty}K_{\sigma}(z_1,z_2)f(x_2,y_2)\,\d y_2\,\d x_2,
\end{equation}
where $P_+:L^2(\mathbb{R}^2)\rightarrow L^2(J_0)$ is the projection from $L^2(\mathbb{R}^2)$ to $L^2(J_0)$ and $G:L^2(\mathbb{R}^2)\rightarrow L^2(\mathbb{R}^2)$ denotes the operator
\begin{equation*}
	(Gf)(x,y):=\frac{1}{\sqrt{\pi}}\e^{-\frac{1}{2}y^2}\int_{-\infty}^{\infty}\e^{-\frac{1}{2}\mu^2}f(x,\mu)\,\d\mu.
\end{equation*}
Note that $\|G\|\leq 1$ in operator norm on $L^2(\mathbb{R}^2)$, so in the same norm $\|S_{\sigma}GL_{-\infty}P_+L_{-\infty}S_{\sigma}^{-1}\|\leq 1$. Hence, by self-adjointness of $K_{\sigma}$ and \eqref{e6} we have for the operator norm on $L^2(J_t)$,
\begin{equation*}
	\|K_{\sigma}\|=\sup_{\|f\|_{L^2(J_t)}=1}\big|\langle f,K_{\sigma}f\rangle_{L^2(J_t)}\big|\leq 1,
\end{equation*}
and thus also $K_{\sigma}\leq 1$. Next, if $K_{\sigma}f=f$ for some $f\in L^2(J_t)$, not identically zero, then the workings in \eqref{e6} show that $K_{\sigma}f=0$ almost everywhere on $(-\infty,t)\times\mathbb{R}$. The analytic properties of the Airy function and the exponential imply that $(x,y)\mapsto (K_{\sigma}f)(x,y)$ is an entire function of $x$ and $y$ separately. We thus have $K_{\sigma}f=0$ everywhere on $\mathbb{R}^2$ and so $f\equiv 0$, a contradiction. Thus $I-K_{\sigma}$ is invertible on $L^2(J_t)$ for all $(t,\sigma)\in\mathbb{R}\times[0,\infty)$ by the Fredholm Alternative. It now remains to establish \eqref{e5} and we begin with \cite[Theorem $3.9$]{Sim}, compare also Lemma \ref{lem1},
\begin{align}
	&\hspace{1cm}\tr_{L^2(J_t)}(K_{\textnormal{Ai}}^{\sigma})^n=\int_{J_t^n}\left[\prod_{k=1}^nK_{\textnormal{Ai}}^{\sigma}(z_k,z_{k+1})\right]\d^2z_1\cdots\d^2 z_n\label{e7}\\
	\stackrel{\eqref{k1}}{=}&\,\left(\frac{\im}{4\pi^{\frac{5}{2}}}\right)^n\int\limits_{J_t^n}\left[\prod_{k=1}^n\int_{\gamma_1}\int_{\gamma_2}\frac{\e^{\im(\frac{1}{3}\lambda_k^3+x_k\lambda_k)}\e^{\im(\frac{1}{3}\mu_k^3+x_{k+1}\mu_k)}}{\lambda_k+\mu_k}\e^{-\frac{1}{2}(\sigma\lambda_k+y_k)^2}\e^{-\frac{1}{2}(\sigma\mu_k-y_{k+1})^2}\d\lambda_k\,\d\mu_k\right]\d^2z_1\cdots\d^2z_n\nonumber
\end{align}
using throughout the cyclic convention $(x_{n+1},y_{n+1})=z_{n+1}\equiv z_1=(x_1,y_1)$. We proceed by integrating out all $y$-variables in \eqref{e7} via Fubini's theorem while identifying $\lambda_{n+1}\equiv\lambda_1$,
\begin{align}
	\int_{\mathbb{R}^n}&\,\left[\prod_{k=1}^n\e^{-\frac{1}{2}(\sigma\lambda_k+y_k)^2}\e^{-\frac{1}{2}(\sigma\mu_k-y_{k+1})^2}\right]\d y_1\cdots\d y_n=\prod_{k=1}^n\int_{-\infty}^{\infty}\e^{-\frac{1}{2}(\sigma\lambda_{k+1}+y)^2}\e^{-\frac{1}{2}(\sigma\mu_k-y)^2}\d y\nonumber\\
	&\hspace{1cm}=\pi^{\frac{n}{2}}\prod_{k=1}^n\e^{-\frac{1}{4}\sigma^2(\lambda_{k+1}+\mu_k)^2},\ \ \ (\lambda_k,\mu_k)\in\gamma_1\times\gamma_2.\label{e8}
\end{align}
Afterwards we recall the Gaussian integral
\begin{equation*}
	\e^{-z^2}=\frac{1}{\sqrt{\pi}}\int_{-\infty}^{\infty}\e^{-y^2+2\im zy}\,\d y,\ \ \ \ z\in\mathbb{C},
\end{equation*}
and transform \eqref{e8} to, with $y_{n+1}\equiv y_1$,
\begin{align*}
	\int_{\mathbb{R}^n}&\,\left[\prod_{k=1}^n\e^{-\frac{1}{2}(\sigma\lambda_k+y_k)^2}\e^{-\frac{1}{2}(\sigma\mu_k-y_{k+1})^2}\right]\d y_1\cdots\d y_n=\prod_{k=1}^n\int_{-\infty}^{\infty}\e^{-y^2+\im\sigma(\lambda_{k+1}+\mu_k)y}\,\d y\\
	&\hspace{1cm}=\int_{\mathbb{R}^n}\left[\prod_{k=1}^n\e^{-\frac{1}{2}y_k^2-\frac{1}{2}y_{k+1}^2+\im\sigma(\lambda_ky_k+\mu_ky_{k+1})}\right]\d y_1\cdots\d y_n,\ \ \ \ \ \ (\lambda_k,\mu_k)\in\gamma_1\times\gamma_2.
\end{align*}
Consequently, back in \eqref{e7}, by Fubini's theorem,
\begin{align*}
	\tr_{L^2(J_t)}(K_{\textnormal{Ai}}^{\sigma})^n=&\,\left(\frac{\im}{4\pi^{\frac{5}{2}}}\right)^n\int\limits_{J_t^n}\bigg[\prod_{k=1}^n\int_{\gamma_1}\int_{\gamma_2}\e^{-\frac{1}{2}y_k^2}\e^{\im(\frac{1}{3}\lambda_k^3+[x_k+\sigma y_k]\lambda_k)}\e^{\im(\frac{1}{3}\mu_k^3+[x_{k+1}+\sigma y_{k+1}]\mu_k)}\e^{-\frac{1}{2}y_{k+1}^2}\frac{\d\lambda_k\,\d\mu_k}{\lambda_k+\mu_k}\bigg]\\
	&\hspace{1cm}\,\times\d^2z_1\cdots\d^2z_n\\
	=&\,\int_{J_t^n}\left[\prod_{k=1}^n\frac{1}{\sqrt{\pi}}\e^{-\frac{1}{2}y_k^2}\int_0^{\infty}\textnormal{Ai}(x_k+\sigma y_k+s)\textnormal{Ai}(s+x_{k+1}+\sigma y_{k+1})\,\d s\,\e^{-\frac{1}{2}y_{k+1}^2}\right]\d^2z_1\cdots\d^2 z_n\\
	=&\,\int_{J_t^n}\left[\prod_{k=1}^nK_{\sigma}(z_k,z_{k+1})\right]\d^2z_1\cdots\d^2z_n=\tr_{L^2(J_t)}K_{\sigma}^n,
\end{align*}
where we have used the contour integral formula \eqref{k3} for the Airy function in the second equality and the fact that $K_{\sigma}$ is trace class on $L^2(J_t)$ with continuous kernel in the last. The proof is complete.
\end{proof}

In deriving the integro-differential system for $F_{\sigma}(t)$, the representation \eqref{e3}, but now with \eqref{e5} in place, uncovers the following fact: abbreviating
\begin{equation}\label{e9}
	\phi_{\sigma}(z)\equiv\phi_{\sigma}(x,y):=\frac{1}{\pi^{\frac{1}{4}}}\e^{-\frac{1}{2}y^2}\textnormal{Ai}(x+\sigma y),\ \ \ z\equiv(x,y)\in\mathbb{R}^2,\ \sigma\geq 0,
\end{equation}
the kernel of $K_{\sigma}$ in \eqref{e3a} is simply
\begin{equation*}
	K_{\sigma}(z_1,z_2)=\int_0^{\infty}\phi_{\sigma}(x_1+s,y_1)\phi_{\sigma}(s+x_2,y_2)\,\d s,
\end{equation*}
and thus constitutes the kernel of a Hankel composition operator in the first variable. This feature is the key ingredient of our first proof of Theorem \ref{fimain1} as we are now able to draw inspiration from the recent works \cite{Kra,Bo} on Fredholm determinants of Hankel composition operators acting on $L^2(J\subset\mathbb{R})$. The necessary adjustments that account for the higher-dimensional integration domain $J_t\subset\mathbb{R}^2$ are worked out below. 
\subsection{The Hankel method} First we view $K_{\sigma}$ not acting on $L^2(J_t)$ but rather on $L^2(J_0)$ however then with kernel
\begin{equation}\label{e9a}
	K_{t,\sigma}(z_1,z_2):=\frac{1}{\sqrt{\pi}}\e^{-\frac{1}{2}y_1^2}K_{\textnormal{Ai}}(x_1+\sigma y_1+t,x_2+\sigma y_2+t)\e^{-\frac{1}{2}y_2^2},\ \ \ z_k=(x_k,y_k)\in\mathbb{R}^2,
\end{equation}
so that by \eqref{e3}, Proposition \ref{prop1} and \cite[Chapter IV, $(5.11),(6.2)$]{GGK},
\begin{equation}\label{e10}
	F_{\sigma}(t)=\prod_{k=1}^{\infty}\big(1-\mu_k(t,\sigma)\big),\ \ \ (t,\sigma)\in\mathbb{R}\times[0,\infty),
\end{equation}
in terms of the non-zero eigenvalues $\mu_k(t,\sigma)$ of $K_{t,\sigma}:L^2(J_0)\rightarrow L^2(J_0)$, counting multiplicities. In turn we record the below result.
\begin{prop}\label{prop2} For any $(t,\sigma)\in\mathbb{R}\times[0,\infty)$, viewing $\sigma$ as parameter,
\begin{equation}\label{e11}
	\frac{\d^2}{\d t^2}\ln F_{\sigma}(t)=-\int_{-\infty}^{\infty}\big(q_{\sigma}(t,y)\big)^2\,\d y;\ \ \ \ \ \ q_{\sigma}(t,y):=((I-K_{t,\sigma})^{-1}\tau_t\phi_{\sigma})(0,y),\ \ y\in\mathbb{R},
\end{equation}
using \eqref{e9} and where $\tau_t:L^2(\mathbb{R}^2)\rightarrow L^2(\mathbb{R}^2)$ denotes the shift in the first variable, i.e. $(\tau_tf)(x,y):=f(x+t,y)$.
\end{prop}
\begin{proof} By the analytic and asymptotic properties of the Airy function and the exponential, we have $\tau_t\phi_{\sigma}\in L^2(J_0)$ for all $(t,\sigma)\in\mathbb{R}\times[0,\infty)$ and $K_{t,\sigma}(L^2(J_0))\subset W^{1,2}(J_0):=\{f\in L^2(J_0):\,D_1f\in L^2(J_0),\ D_2f\in L^2(J_0)\}$ with the partial derivatives $(D_1f)(x,y):=\frac{\partial}{\partial x}f(x,y)$ and $(D_2f)(x,y):=\frac{\partial}{\partial y}f(x,y)$. Hence, $J_0\ni (t,y)\mapsto q_{\sigma}(t,y)$ is well-defined for all $\sigma\geq 0$, compare Proposition \ref{prop1}, and the map $y\mapsto q_{\sigma}(t,y)$ is in $L^2(\mathbb{R})$ for all $(t,\sigma)\in\mathbb{R}\times[0,\infty)$. In order to establish the equality in \eqref{e11} we begin with the following observation: integrating by parts, $t\mapsto K_{t,\sigma}$ is differentiable on $\mathbb{R}$ with trace class derivative on $L^2(J_0)$ given by
\begin{equation*}
	\frac{\d}{\d t}K_{t,\sigma}=-\tau_t\phi_{\sigma}\otimes\tau_t\phi_{\sigma}.
\end{equation*}
Here we write $\alpha\otimes\beta$ for a general rank one integral operator on $L^2(J_0)$,
\begin{equation*}
	\big((\alpha\otimes\beta)f\big)(z):=\alpha(z)\int_{J_0}\beta(w)f(w)\d^2 w.
\end{equation*} 
Thus, by Jacobi's formula \cite[$(3.12)$]{W},
\begin{equation}\label{e12}
	\frac{\d}{\d t}\ln F_{\sigma}(t)=-\tr_{L^2(J_0)}\left((I-K_{t,\sigma})^{-1}\frac{\d}{\d t}K_{t,\sigma}\right)=\tr_{L^2(J_0)}\big((I-K_{t,\sigma})^{-1}\tau_t\phi_{\sigma}\otimes\tau_t\phi_{\sigma}\big).
\end{equation}
Next we note that $t\mapsto (I-K_{t,\sigma})^{-1}\tau_t\phi_{\sigma}\otimes\tau_t\phi_{\sigma}$ is also differentiable on $\mathbb{R}$ with trace class derivative on $L^2(J_0)$ equal to
\begin{align}
	\frac{\d}{\d t}\big((I-K_{t,\sigma})^{-1}&\,\tau_t\phi_{\sigma}\otimes\tau_t\phi_{\sigma}\big)=(I-K_{t,\sigma})^{-1}\tau_t\phi_{\sigma}\otimes D_1\tau_t\phi_{\sigma}\label{e13}
\\
	&\hspace{1cm}+\Big[(I-K_{t,\sigma})^{-1}D_1\tau_t\phi-(I-K_{t,\sigma})^{-1}(\tau_t\phi_{\sigma}\otimes\tau_t\phi_{\sigma})(I-K_{t,\sigma})^{-1}\tau_t\phi_{\sigma}\Big]\otimes\tau_t\phi_{\sigma}.\nonumber
\end{align}
In computing \eqref{e13} we have used the operator identity
\begin{equation*}
	\frac{\d}{\d t}(I-K_{t,\sigma})^{-1}=(I-K_{t,\sigma})^{-1}\frac{\d K_{t,\sigma}}{\d t}(I-K_{t,\sigma})^{-1},
\end{equation*}
which holds as soon as $I-K_{t,\sigma}$ is invertible (recall Proposition \ref{prop1}) and $t\mapsto K_{t,\sigma}$ differentiable. All together, \eqref{e12} and \eqref{e13} confirm that $\mathbb{R}\ni t\mapsto\ln F_{\sigma}(t)$ is twice differentiable with second derivative equal to 
\begin{align}
	\frac{\d^2}{\d t^2}\ln F_{\sigma}&\,(t)\stackrel{\eqref{e12}}{=}\tr_{L^2(J_0)}\left(\frac{\d}{\d t}\big((I-K_{t,\sigma})^{-1}\tau_t\phi_{\sigma}\otimes\tau_t\phi_{\sigma}\big)\right)\stackrel{\eqref{e13}}{=}\tr_{L^2(J_0)}\big((I-K_{t,\sigma})^{-1}\tau_t\phi_{\sigma}\otimes D_1\tau_t\phi_{\sigma}\big)\label{e14}\\
	&\hspace{1cm}+\tr_{L^2(J_0)}\Big(\Big[(I-K_{t,\sigma})^{-1}D_1\tau_t\phi_{\sigma}-(I-K_{t,\sigma})^{-1}(\tau_t\phi_{\sigma}\otimes\tau_t\phi_{\sigma})(I-K_{t,\sigma})^{-1}\tau_t\phi_{\sigma}\Big]\otimes\tau_t\phi_{\sigma}\Big),\nonumber
\end{align}
where we rely on the analytic and asymptotic properties of the Airy function and exponential in the first equality while interchanging differentiation with trace evaluation. Note that the first term in \eqref{e14} can be simplified by \cite[Theorem $3.9$]{Sim} and integration by parts in the first variable,
\begin{align*}
	\tr_{L^2(J_0)}\big((I-K_{t,\sigma})^{-1}&\,\tau_t\phi_{\sigma}\otimes D_1\tau_t\phi_{\sigma}\big)=\int_0^{\infty}\int_{-\infty}^{\infty}\big((I-K_{t,\sigma})^{-1}\tau_t\phi_{\sigma}\big)(x,y)\frac{1}{\pi^{\frac{1}{4}}}\e^{-\frac{1}{2}y^2}\frac{\partial}{\partial x}\textnormal{Ai}(x+\sigma y+t)\,\d y\,\d x\\
	\stackrel{\eqref{e11}}{=}&\,-\int_{-\infty}^{\infty}q_{\sigma}(t,y)(\tau_t\phi_{\sigma})(0,y)\,\d y-\tr_{L^2(J_0)}\big(D_1(I-K_{t,\sigma})^{-1}\tau_t\phi_{\sigma}\otimes\tau_t\phi_{\sigma}\big).
\end{align*}
Here we have used that $D_1(I-K_{t,\sigma})^{-1}\tau_t\phi_{\sigma}\otimes\tau_t\phi_{\sigma}$ is trace class on $L^2(J_0)$ by the aforementioned mapping property $K_{t,\sigma}(L^2(J_0))\subset W^{1,2}(J_0)$ and combining the last equality with \eqref{e14}, the same simplifies to
\begin{align}
	\frac{\d^2}{\d t^2}\ln F_{\sigma}(t)=-\int_{-\infty}^{\infty}&\,q_{\sigma}(t,y)(\tau_t\phi_{\sigma})(0,y)\,\d y-\tr_{L^2(J_0)}\Big(\Big[D_1(I-K_{t,\sigma})^{-1}\tau_t\phi_{\sigma}-(I-K_{t,\sigma})^{-1}D_1\tau_t\phi_{\sigma}\Big]\otimes\tau_t\phi_{\sigma}\Big)\nonumber\\
	&\,-\tr_{L^2(J_0)}\Big(\Big[(I-K_{t,\sigma})^{-1}(\tau_t\phi_{\sigma}\otimes\tau_t\phi_{\sigma})(I-K_{t,\sigma})^{-1}\tau_t\phi_{\sigma}\Big]\otimes\tau_t\phi_{\sigma}\Big).\label{e15}
\end{align}
Moving ahead, by algebra,
\begin{equation*}
	D_1(I-K_{t,\sigma})^{-1}\tau_t\phi_{\sigma}-(I-K_{t,\sigma})^{-1}D_1\tau_t\phi_{\sigma}=(I-K_{t,\sigma})^{-1}(D_1K_{t,\sigma}-K_{t,\sigma}D_1)(I-K_{t,\sigma})^{-1}\tau_t\phi_{\sigma}\in L^2(J_0),
\end{equation*}
which we can simplify with the help of the identities
\begin{equation*}
	\begin{cases}\displaystyle\frac{\partial}{\partial x_1}K_{t,\sigma}(z_1,z_2)+\frac{\partial}{\partial x_2}K_{t,\sigma}(z_1,z_2)=-(\tau_t\phi_{\sigma}\otimes\tau_t\phi_{\sigma})(z_1,z_2)&\smallskip\\
	\displaystyle \lim_{x_2\rightarrow+\infty}K_{t,\sigma}\big(z_1,(x_2,y_2)\big)=0,\ \ \ \ \ \ z_1\in J_0,\ \ y_2\in\mathbb{R}&
	\end{cases}.
\end{equation*}
What results is,
\begin{align*}
	\big((D_1K_{t,\sigma}-K_{t,\sigma}D_1)&\,(I-K_{t,\sigma})^{-1}\tau_t\phi_{\sigma}\big)(z)=-(\tau_t\phi_{\sigma})(z)\int_{J_0}(\tau_t\phi_{\sigma})(w)\big((I-K_{t,\sigma})^{-1}\tau_t\phi_{\sigma}\big)(w)\,\d^2w\\
	&\,+\int_{-\infty}^{\infty}K_{t,\sigma}\big(z,(0,y)\big)q_{\sigma}(t,y)\,\d y,\ \ \ \ \ z\in J_0,
\end{align*}
with the shorthand $K_{t,\sigma}(z,(0,y))=\lim_{x_2\downarrow 0}K_{t,\sigma}(z,(x_2,y))$, and therefore
\begin{align*}
	\tr_{L^2(J_0)}&\,\Big(\Big[D_1(I-K_{t,\sigma})^{-1}\tau_t\phi_{\sigma}-(I-K_{t,\sigma})^{-1}D_1\tau_t\phi_{\sigma}\Big]\otimes\tau_t\phi_{\sigma}\Big)\\
	&\,=-\tr_{L^2(J_0)}\Big(\Big[(I-K_{t,\sigma})^{-1}(\tau_t\phi_{\sigma}\otimes\tau_t\phi_{\sigma})(I-K_{t,\sigma})^{-1}\tau_t\phi_{\sigma}\Big]\otimes\tau_t\phi_{\sigma}\Big)\\
	&\hspace{1cm}+\int_{J_0}\left[\int_{-\infty}^{\infty}\big((I-K_{t,\sigma})^{-1}K_{t,\sigma}\big)\big(z,(0,y)\big)q_{\sigma}(t,y)\,\d y\right](\tau_t\phi_{\sigma})(z)\,\d^2 z,\ \ \ \ (t,\sigma)\in\mathbb{R}\times[0,\infty).
\end{align*}
All together, \eqref{e15} simplifies to
\begin{equation*}
	\frac{\d^2}{\d t^2}\ln F_{\sigma}(t)=-\int_{-\infty}^{\infty}q_{\sigma}(t,y)\left[(\tau_t\phi_{\sigma})(0,y)+\int_{J_0}\big((I-K_{t,\sigma})^{-1}K_{t,\sigma}\big)\big(z,(0,y)\big)(\tau_t\phi_{\sigma})(z)\,\d^2 z\right]\d y,
\end{equation*}
and this is precisely \eqref{e11} given that $(I-K_{t,\sigma})^{-1}=I+(I-K_{t,\sigma})^{-1}K_{t,\sigma}$ on $L^2(J_0)$ and since $K_{t,\sigma}$ is self-adjoint. Our proof is complete.
\end{proof}
\begin{rem}\label{rem2} The workings in Proposition \ref{prop2} are valid for a wider class of Hankel composition operators on $L^2(J_0)$. In fact, \eqref{e11} generalizes to kernels of the type
\begin{equation*}
	K_t(z_1,z_2)=\int_0^{\infty}\phi(x_1+s+t,y_1)\psi(s+x_2+t,y_2)\,\d s,\ \ \ (z_1,z_2)\in J_0\times J_0
\end{equation*}
with $\phi,\psi$ of the form $\phi(x,y)=\omega_1(y)f(x+\sigma_1(y))$ and $\psi(x,y)=\omega_2(y)g(x+\sigma_2(y))$, subject to suitable regularity and integrability assumptions placed on the one-variable functions $\omega_k,\sigma_k,f,g$ and in turn on the properties of the operator $K_t$. We shall see later that our workings lead to the dynamical system in Proposition \ref{prop3} below, which can be generalized in a similar fashion.
\end{rem}
Observe that \eqref{e11} allows us to derive the following compact expression for the distribution function $F_{\sigma}(t)$. This expression bears a striking resemblance to \cite[$(5.4)$]{ACQ}.
\begin{cor}\label{cor1} For all $(t,\sigma)\in\mathbb{R}\times[0,\infty)$,
\begin{equation}\label{e16}
	F_{\sigma}(t)=\exp\left[-\int_t^{\infty}(s-t)\left\{\int_{-\infty}^{\infty}\big(q_{\sigma}(s,y)\big)^2\,\d y\right\}\d s\right]
\end{equation}
in terms of $q_{\sigma}(t,y)=((I-K_{t,\sigma})^{-1}\tau_t\phi_{\sigma})(0,y)$.
\end{cor}
\begin{proof} Write $q_{\sigma}(t,y)=(\tau_t\phi_{\sigma})(0,y)+((I-K_{t,\sigma})^{-1}K_{t,\sigma}\tau_t\phi_{\sigma})(0,y)$ and note that for all $t>0$ and $\sigma\geq 0$,
\begin{equation}\label{e16a}
	\big|K_{t,\sigma}(z_1,z_2)\big|\leq \frac{t^{-1}}{8\pi^{\frac{3}{2}}}\e^{-\frac{4}{3}t^{\frac{3}{2}}}\e^{-\frac{1}{2}(y_1^2+y_2^2)-\sqrt{t}(x_1+\sigma y_1+x_2+\sigma y_2)},
\end{equation}
uniformly in $(z_1,z_2)\in J_0\times J_0$. Hence, for all $n\in\mathbb{Z}_{\geq 1}$ and all $t>0$ and $\sigma\geq 0$,
\begin{align}
	\big|K_{t,\sigma}^n&\,(z_1,z_{n+1})\big|\leq \left(\frac{t^{-1}}{8\pi^{\frac{3}{2}}}\e^{-\frac{4}{3}t^{\frac{3}{2}}}\right)^n\int_{J_0^{n-1}}\left[\prod_{j=1}^n\e^{-\frac{1}{2}(y_j^2+y_{j+1}^2)-\sqrt{t}(x_j+\sigma y_j+x_{j+1}+\sigma y_{j+1})}\right]\d^2z_2\cdots\d^2z_n\nonumber\\
	=&\,\left(\frac{t^{-1}}{8\pi^{\frac{3}{2}}}\e^{-\frac{4}{3}t^{\frac{3}{2}}}\right)^n\e^{-\frac{1}{2}(y_1^2+y_{n+1}^2)-\sqrt{t}(x_1+\sigma y_1+x_{n+1}+\sigma y_{n+1})}\int_{J_0^{n-1}}\left[\prod_{j=2}^n\e^{-y_j^2-2\sqrt{t}(x_j+\sigma y_j)}\right]\d^2 z_2\cdots\d^2 z_n\nonumber\\
	=&\,\frac{t^{-1}}{8\pi^{\frac{3}{2}}}\e^{-\frac{4}{3}t^{\frac{3}{2}}}\left(\frac{t^{-\frac{3}{2}}}{16\pi}\e^{-\frac{4}{3}t^{\frac{3}{2}}+t\sigma^2}\right)^{n-1}\e^{-\frac{1}{2}(y_1^2+y_{n+1}^2)-\sqrt{t}(x_1+\sigma y_1+x_{n+1}+\sigma y_{n+1})},\label{e17}
\end{align}
uniformly in $(z_1,z_{n+1})\in J_0\times J_0$. In turn, for all $t\geq t_0$ sufficiently large positive and $\sigma\geq 0$ such that $0\leq\sigma^2\leq\frac{4}{3}\sqrt{t}$ by the Neumann series,
\begin{equation*}
	\big|\big((I-K_{t,\sigma})^{-1}K_{t,\sigma}\big)(z_1,z_2)\big|\stackrel{\eqref{e17}}{\leq}\frac{t^{-1}}{8\pi^{\frac{3}{2}}}\e^{-\frac{4}{3}t^{\frac{3}{2}}}\e^{-\frac{1}{2}(y_1^2+y_2^2)-\sqrt{t}(x_1+\sigma y_1+x_2+\sigma y_2)}\left[1+\frac{t^{-\frac{3}{2}}}{16\pi}\e^{-\frac{4}{3}t^{\frac{3}{2}}+t\sigma^2}\right],
\end{equation*}
and thus there exist $c,t_0>0$ so that 
\begin{equation}\label{e18}
	\big|\big((I-K_{t,\sigma})^{-1}K_{t,\sigma}\tau_t\phi_{\sigma}\big)(0,y)\big|\leq c\,t^{-\frac{3}{2}}\e^{-\frac{4}{3}t^{\frac{3}{2}}+\frac{1}{4}t\sigma^2}\left[1+\frac{t^{-\frac{3}{2}}}{16\pi}\e^{-\frac{4}{3}t^{\frac{3}{2}}+t\sigma^2}\right]\e^{-\frac{1}{2}y^2-\sqrt{t}\sigma y},
\end{equation}
for all $t\geq t_0$, any $\sigma\geq 0$ such that $0\leq\sigma^2\leq\frac{4}{3}\sqrt{t}$ and all $y\in\mathbb{R}$. In deriving \eqref{e18} we have used that the Airy function is bounded on $\mathbb{R}$. On the other hand, using the slightly refined bounds, cf. \cite[$\S 9.7$]{NIST},
\begin{equation*}
	\big|\textnormal{Ai}(x)\big|\leq d\,\e^{-\frac{2}{3}x^{\frac{3}{2}}},\ \ x\geq 0;\ \ \ \ \ \ \ \ \ \ \ \ \ \big|\textnormal{Ai}(x)\big|\leq d,\ \ x<0,
\end{equation*}
with universal $d>0$, we find by Laplace's method that there exist $c,t_0>0$ such that
\begin{equation}\label{e19}
	\|(\tau_t\phi_{\sigma})(0,\cdot)\|_{L^2(\mathbb{R})}\leq c\sqrt{\frac{\sigma}{t}}\,\e^{-\frac{t^2}{2\sigma^2}},\ \ \ \ \ \ \forall\,t\geq t_0,\ \ \sigma\geq 0:\ 0\leq\sigma^2\leq\frac{4}{3}\sqrt{t}.
\end{equation}
Hence, by Minkowski's inequality, combining \eqref{e18} and \eqref{e19},
\begin{align*}
	\int_{-\infty}^{\infty}\big(q_{\sigma}(t,y)\big)^2\,\d y=\|q_{\sigma}(t,\cdot)\|_{L^2(\mathbb{R})}^2\leq&\,\big(\|(\tau_t\phi_{\sigma})(0,\cdot)\|_{L^2(\mathbb{R})}+\|((I-K_{t,\sigma})^{-1}K_{t,\sigma}\tau_t\phi_{\sigma})(0,\cdot)\|_{L^2(\mathbb{R})}\big)^2\\
	\leq&\,c\,t^{-3}\e^{-\frac{2}{3}t^{\frac{3}{2}}},\ \ \ \ \ \ \forall\,t\geq t_0,\ \ \sigma\geq 0:\ 0\leq\sigma^2\leq\frac{4}{3}\sqrt{t}.
\end{align*}
The last estimate shows that the right hand side in \eqref{e11} decays super exponentially fast to zero as $t\rightarrow+\infty$ for any fixed $\sigma\geq 0$, i.e. we can integrate the same expression twice, for any $(t,\sigma)\in\mathbb{R}\times[0,\infty)$,\footnote{Note that $F_{\sigma}(\infty)=1$ for all $\sigma\geq 0$ by Hadamard's inequality and the bound for $K_{t,\sigma}(z_1,z_2)$ stated in \eqref{e16a}.}
\begin{equation*}
	\ln F_{\sigma}(t)=-\int_t^{\infty}\left(\int_{\lambda}^{\infty}\left\{\int_{-\infty}^{\infty}\big(q_{\sigma}(s,y)\big)^2\,\d y\right\}\d s\right)\d\lambda=-\int_t^{\infty}(s-t)\left\{\int_{-\infty}^{\infty}\big(q_{\sigma}(s,y)\big)^2\,\d y\right\}\d s,
\end{equation*}
and this completes the proof of \eqref{e16}.
\end{proof}
We now proceed to characterize $q_{\sigma}(t,y)$ through a dynamical system. This will be achieved in two steps and we begin with the integro-differential ODE system written in \eqref{e20}.
\begin{prop}\label{prop3}
Let $n\in\mathbb{Z}_{\geq 0}$ and introduce for any $(t,y,\sigma)\in\mathbb{R}^2\times[0,\infty)$,
\begin{equation*}
	p_{\sigma}^{[n]}(t):=\tr_{L^2(J_0)}\big((I-K_{t,\sigma})^{-1}D_1^n\tau_t\phi_{\sigma}\otimes\tau_t\phi_{\sigma}\big),\ \ \ \ \ \ \ 
	q_{\sigma}^{[n]}(t,y):=\big((I-K_{t,\sigma})^{-1}D_1^n\tau_t\phi_{\sigma}\big)(0,y),
\end{equation*}
where $(D_1^nf)(x,y):=\frac{\partial^n}{\partial x^n}f(x,y)$. Then $\big\{p_{\sigma}^{[n]},q_{\sigma}^{[n]}\big\}_{n=0}^{\infty}$ satisfy the coupled integro-differential system
\begin{equation}\label{e20}
	\frac{\d}{\d t}p_{\sigma}^{[n]}(t)=-\int_{-\infty}^{\infty}q_{\sigma}^{[n]}(t,y)q_{\sigma}^{[0]}(t,y)\,\d y,\ \ \ \ \ \ \ \frac{\partial}{\partial t}q_{\sigma}^{[n]}(t,y)=q_{\sigma}^{[n+1]}(t,y)-q_{\sigma}^{[0]}(t,y)p_{\sigma}^{[n]}(t),\ \ n\in\mathbb{Z}_{\geq 0},
\end{equation}
and 
\begin{equation*}
	I^{[n]}_{\sigma}(t):=p_{\sigma}^{[n+1]}(t)+(-1)^np_{\sigma}^{[n+1]}(t)+\sum_{k=0}^n(-1)^k\left[\int_{-\infty}^{\infty}q_{\sigma}^{[k]}(t,y)q_{\sigma}^{[n-k]}(t,y)\,\d y-p_{\sigma}^{[k]}(t)p_{\sigma}^{[n-k]}(t)\right],\,n\in\mathbb{Z}_{\geq 0}
\end{equation*}
constitutes a conserved quantity for the same system, i.e. subject to \eqref{e20} we have $\frac{\d}{\d t}I^{[n]}_{\sigma}(t)=0$ for all $(t,\sigma)\in\mathbb{R}\times[0,\infty)$ and $n\in\mathbb{Z}_{\geq 0}$.
\end{prop}
\begin{proof}
Note that $\mathbb{R}\ni t\mapsto p_{\sigma}^{[n]}(t)$ and $\mathbb{R}^2\ni (t,y)\mapsto q_{\sigma}^{[n]}(t,y)$ are well-defined for all $\sigma\geq 0$ since $D_1^n\tau_t\phi_{\sigma}\in L^2(J_0)$ by the properties of the Airy function and the exponential. More is true, viewed as functions of $t\in\mathbb{R}$, they are differentiable: indeed, the second equation in \eqref{e20} is straightforward,
\begin{align*}
	&\hspace{1cm}\frac{\partial}{\partial t}q_{\sigma}^{[n]}(t,y)=\Big(\left\{\frac{\d}{\d t}(I-K_{t,\sigma})^{-1}\right\}D_1^n\tau_t\phi_{\sigma}\Big)(0,y)+\left((I-K_{t,\sigma})^{-1}\frac{\partial}{\partial t}D_1^n\tau_t\phi_{\sigma}\right)(0,y)\\
	=&\,-\big((I-K_{t,\sigma})^{-1}(\tau_t\phi_{\sigma}\otimes\tau_t\phi_{\sigma})(I-K_{t,\sigma})^{-1}D_1^n\tau_t\phi_{\sigma}\big)(0,y)+q_{\sigma}^{[n+1]}(t,y)
	=-q_{\sigma}^{[0]}(t,y)p_{\sigma}^{[n]}(t)+q_{\sigma}^{[n+1]}(t,y),
\end{align*}
where we relied again on $\frac{\d}{\d t}(I-K_{t,\sigma})^{-1}=-(I-K_{t,\sigma})^{-1}(\tau_t\phi_{\sigma}\otimes\tau_t\phi_{\sigma})(I-K_{t,\sigma})^{-1}$. The first equation on the other hand follows from a generalization of the workings in Proposition \ref{prop2} below \eqref{e12}. Indeed,
\begin{align}
	\frac{\d}{\d t}\big((I-K_{t,\sigma})^{-1}D_1^n&\,\tau_t\phi_{\sigma}\otimes\tau_t\phi_{\sigma}\big)=(I-K_{t,\sigma})^{-1}D_1^n\tau_t\phi_{\sigma}\otimes D_1\tau_t\phi_{\sigma}\label{e21}\\
	&\,+\Big[(I-K_{t,\sigma})^{-1}D_1^{n+1}\tau_t\phi-(I-K_{t,\sigma})^{-1}(\tau_t\phi_{\sigma}\otimes\tau_t\phi_{\sigma})(I-K_{t,\sigma})^{-1}D_1^n\tau_t\phi_{\sigma}\Big]\otimes\tau_t\phi_{\sigma},\nonumber
\end{align}
and therefore, based on the analytic and asymptotic properties of the Airy function and exponential while interchanging differentiation and trace evaluation,
\begin{align*}
	\frac{\d}{\d t}p_{\sigma}^{[n]}(t)=&\,\tr_{L^2(J_0)}\left(\frac{\d}{\d t}\big((I-K_{t,\sigma})^{-1}D_1^n\tau_t\phi_{\sigma}\otimes\tau_t\phi_{\sigma}\big)\right)\stackrel{\eqref{e21}}{=}\tr_{L^2(J_0)}\big((I-K_{t,\sigma})^{-1}D_1^n\tau_t\phi_{\sigma}\otimes D_1\tau_t\phi_{\sigma}\big)\\
	&\,+\tr_{L^2(J_0)}\Big(\Big[(I-K_{t,\sigma})^{-1}D_1^{n+1}\tau_t\phi-(I-K_{t,\sigma})^{-1}(\tau_t\phi_{\sigma}\otimes\tau_t\phi_{\sigma})(I-K_{t,\sigma})^{-1}D_1^n\tau_t\phi_{\sigma}\Big]\otimes\tau_t\phi_{\sigma}\Big),
\end{align*}
which generalizes \eqref{e14}. Now use \cite[Theorem $3.9$]{Sim} and integrate by parts in the first summand,
\begin{align*}
	\tr_{L^2(J_0)}\big((I-K_{t,\sigma})^{-1}D_1^n\tau_t\phi_{\sigma}\otimes D_1\tau_t\phi_{\sigma}\big)=-\int_{-\infty}^{\infty}&\,q_{\sigma}^{[n]}(t,y)(\tau_t\phi_{\sigma})(0,y)\,\d y\\
	&\,-\tr_{L^2(J_0)}\big(D_1(I-K_{t,\sigma})^{-1}D_1^n\tau_t\phi_{\sigma}\otimes\tau_t\phi_{\sigma}\big).
\end{align*}
Note that $D_1(I-K_{t,\sigma})^{-1}D_1^n\tau_t\phi_{\sigma}\otimes\tau_t\phi_{\sigma}$ is trace class on $L^2(J_0)$ given $D_1^n\tau_t\phi_{\sigma}\in L^2(J_0)$. Consequently,
\begin{align}
	\frac{\d}{\d t}&\,p_{\sigma}^{[n]}(t)=-\int_{-\infty}^{\infty}q_{\sigma}^{[n]}(t,y)(\tau_t\phi_{\sigma})(0,y)\,\d y-\tr_{L^2(J_0)}\Big(\Big[(I-K_{t,\sigma})^{-1}(\tau_t\phi_{\sigma}\otimes\tau_t\phi_{\sigma})(I-K_{t,\sigma})^{-1}D_1^n\tau_t\phi_{\sigma}\Big]\nonumber\\
	&\,\otimes\tau_t\phi_{\sigma}\Big)-\tr_{L^2(J_0)}\Big(\Big[D_1(I-K_{t,\sigma})^{-1}D_1^n\tau_t\phi_{\sigma}-(I-K_{t,\sigma})^{-1}D_1^{n+1}\tau_t\phi_{\sigma}\Big]-\otimes\tau_t\phi_{\sigma}\Big),\label{e22}
\end{align}
where, by algebra,
\begin{align*}
	D_1(I-K_{t,\sigma})^{-1}D_1^n\tau_t\phi_{\sigma}\,-&\,(I-K_{t,\sigma})^{-1}D_1^{n+1}\tau_t\phi_{\sigma}\\
	&\,=(I-K_{t,\sigma})^{-1}(D_1K_{t,\sigma}-K_{t,\sigma}D_1)(I-K_{t,\sigma})^{-1}D_1^n\tau_t\phi_{\sigma}\in L^2(J_0).
\end{align*}
Simplifying the same combination as in our proof of Proposition \ref{prop2} we then obtain
\begin{align*}
	\tr_{L^2(J_0)}&\,\Big(\Big[D_1(I-K_{t,\sigma})^{-1}D_1^n\tau_t\phi_{\sigma}-(I-K_{t,\sigma})^{-1}D_1^{n+1}\tau_t\phi_{\sigma}\Big]\otimes\tau_t\phi_{\sigma}\Big)\\
	=&\,-\tr_{L^2(J_0)}\Big(\Big[(I-K_{t,\sigma})^{-1}(\tau_t\phi_{\sigma}\otimes\tau_t\phi_{\sigma})(I-K_{t,\sigma})^{-1}D_1^n\tau_t\phi_{\sigma}\Big]\otimes\tau_t\phi_{\sigma}\Big)\\
	&\hspace{0.85cm}+\int_{J_0}\left[\int_{-\infty}^{\infty}\big((I-K_{t,\sigma})^{-1}K_{t,\sigma}\big)\big(z,(0,y)\big)q_{\sigma}^{[n]}(t,y)\,\d y\right](\tau_t\phi_{\sigma})(z)\,\d^2 z,\ \ \ \ \ \ (t,\sigma)\in\mathbb{R}\times[0,\infty),
\end{align*}
and hence all together for \eqref{e22},
\begin{align*}
	\frac{\d}{\d t}p_{\sigma}^{[n]}(t)=&\,-\int_{-\infty}^{\infty}q_{\sigma}^{[n]}(t,y)\left[(\tau_t\phi_{\sigma})(0,y)+\int_{J_0}\big((I-K_{t,\sigma})^{-1}K_{t,\sigma}\big)\big(z,(0,y)\big)(\tau_t\phi_{\sigma})(z)\,\d^2 z\right]\d y\\
	=&\,-\int_{-\infty}^{\infty}q_{\sigma}^{[n]}(t,y)q_{\sigma}^{[0]}(t,y)\,\d y
\end{align*}
by self-adjointness of $K_{t,\sigma}$ and by the identity $(I-K_{t,\sigma})^{-1}=I+(I-K_{t,\sigma})^{-1}K_{t,\sigma}$. This completes our proof of \eqref{e20}. The outstanding check of $\frac{\d}{\d t}I_{\sigma}^{[n]}=0,t\in\mathbb{R}$ is a simple application of \eqref{e20}.
\end{proof}
In order to close up the open system \eqref{e20} we will make use of the Airy differential equation, i.e. unlike the general workings in Proposition \ref{prop2} and \ref{prop3}, compare Remark \ref{rem2}, the next step in our derivation of a dynamical system for $F_{\sigma}(t)$ crucially depends on the specific Airy function choice in the kernel of $K_{t,\sigma}$, in particular Proposition \ref{prop4} after the next Corollary is only valid for our $K_{t,\sigma}$.
\begin{cor}[Constant of motion]\label{cor2} For any $(t,\sigma)\in\mathbb{R}\times[0,\infty)$,
\begin{equation}\label{e23}
	2p_{\sigma}^{[1]}(t)+\int_{-\infty}^{\infty}\big(q_{\sigma}^{[0]}(t,y)\big)^2\d y-\big(p_{\sigma}^{[0]}(t)\big)^2=0.
\end{equation}
\end{cor}
\begin{proof} The left hand side in \eqref{e23} equals $I_0$ in the notation of Proposition \ref{prop3} and which is known to be $t$-independent. However, by the workings in Corollary \ref{cor1} with $q_{\sigma}^{[0]}(t,y)=q_{\sigma}(t,y)$,
\begin{equation*}
	\lim_{t\rightarrow+\infty}\int_{-\infty}^{\infty}\big(q_{\sigma}^{[0]}(t,y)\big)^2\d y=0,\ \ \ \ \ \ \lim_{t\rightarrow+\infty}p_{\sigma}^{[0]}(t)=0
\end{equation*}
super exponentially fast for any $\sigma\in[0,\infty)$. Likewise, by the coarse bound $|\textnormal{Ai}'(x)|\leq c|x|$ with $c>0$ for any $x\in\mathbb{R}$ and Laplace's method,
\begin{equation*}
	\lim_{t\rightarrow+\infty}p_{\sigma}^{[1]}(t)=0,
\end{equation*}
also super exponentially fast for any fixed $\sigma\in[0,\infty)$, i.e. the claim \eqref{e23} follows.
\end{proof}
\begin{prop}[Closure relation]\label{prop4} For any $(t,y,\sigma)\in\mathbb{R}^2\times[0,\infty)$,
\begin{equation}\label{e24}
	q_{\sigma}^{[2]}(t,y)=(t+\sigma y)q_{\sigma}^{[0]}(t,y)-q_{\sigma}^{[0]}(t,y)p_{\sigma}^{[1]}(t)+q_{\sigma}^{[1]}(t,y)p_{\sigma}^{[0]}(t).
\end{equation}
\end{prop}
\begin{proof} Begin by using the Airy differential equation $w''=zw$ for $w=\textnormal{Ai}(z),z\in\mathbb{C}$ and derive
\begin{align}
	q_{\sigma}^{[2]}(t,y)=&\,\int_{J_0}(I-K_{t,\sigma})^{-1}\big((0,y),z_2\big)(D_1^2\tau_t\phi_{\sigma})(z_2)\,\d^2 z_2\nonumber\\
	=&\,\int_0^{\infty}\int_{-\infty}^{\infty}(I-K_{t,\sigma})^{-1}\big((0,y),(x_2,y_2)\big)(x_2+\sigma y_2+t)(\tau_t\phi_{\sigma})(x_2,y_2)\,\d y_2\,\d x_2\nonumber\\
	=&\,\,tq_{\sigma}^{[0]}(t,y)+\big((I-K_{t,\sigma})^{-1}M_1\tau_t\phi_{\sigma}\big)(0,y)+\sigma\big((I-K_{t,\sigma})^{-1}M_2\tau_t\phi_{\sigma}\big)(0,y)\label{e25}
\end{align}
with the partial multiplications $(M_1f)(x,y):=xf(x,y)$ and $(M_2f)(x,y):=yf(x,y)$. Note that $M_j\tau_t\phi_{\sigma}\in L^2(J_0)$ by the properties of the Airy function and exponential, moreover 
\begin{equation}\label{e26}
	(I-K_{t,\sigma})^{-1}M_k\tau_t\phi_{\sigma}=\big[(I-K_{t,\sigma})^{-1},M_k\big]\tau_t\phi_{\sigma}+M_k(I-K_{t,\sigma})^{-1}\tau_t\phi_{\sigma},\ \ \ \ (t,\sigma)\in\mathbb{R}\times[0,\infty),
\end{equation}
in terms of the operator commutator $[A,B]:=AB-BA$. Moving ahead, we find
\begin{equation*}
	\big((I-K_{t,\sigma})^{-1}M_1\tau_t\phi_{\sigma}\big)(0,y)\stackrel{\eqref{e26}}{=}\Big(\big[(I-K_{t,\sigma})^{-1},M_1\big]\tau_t\phi_{\sigma}\Big)(0,y)
\end{equation*}
and likewise
\begin{equation*}
	\big((I-K_{t,\sigma})^{-1}M_2\tau_t\phi_{\sigma}\big)(0,y)\stackrel{\eqref{e26}}{=}\Big(\big[(I-K_{t,\sigma})^{-1},M_2\big]\tau_t\phi_{\sigma}\Big)(0,y)+yq_{\sigma}^{[0]}(t,y).
\end{equation*}
Hence back in \eqref{e25},
\begin{equation}\label{e27}
	q_{\sigma}^{[2]}(t,y)=(t+\sigma y)q_{\sigma}^{[0]}(t,y)+\Big(\big[(I-K_{t,\sigma})^{-1},M_1+\sigma M_2\big]\tau_t\phi_{\sigma}\Big)(0,y),
\end{equation}
so it remains to simplify the operator commutator. This is straightforward given the well-known Christoffel-Darboux property
\begin{equation*}
	K_{\textnormal{Ai}}(a,b)=\int_0^{\infty}\textnormal{Ai}(a+s)\textnormal{Ai}(s+b)\,\d s=\frac{\textnormal{Ai}(a)\textnormal{Ai}'(b)-\textnormal{Ai}'(a)\textnormal{Ai}(b)}{a-b},
\end{equation*}
and the particular shape of $K_{t,\sigma}(z_1,z_2)$, compare Proposition \ref{prop1}. Precisely,
\begin{align*}
	\Big(\big[&\,(I-K_{t,\sigma})^{-1},M_1+\sigma M_2\big]\tau_t\phi_{\sigma}\Big)(0,y)=\Big((I-K_{t,\sigma})^{-1}\big[K_{t,\sigma},M_1+\sigma M_2\big](I-K_{t,\sigma})^{-1}\tau_t\phi_{\sigma}\Big)(0,y)\\
	&\,=-\int_{J_0}(I-K_{t,\sigma})^{-1}\big((0,y),z_1\big)\Big((\tau_t\phi_{\sigma})(z_1)(D_1\tau_t\phi_{\sigma})(z_2)-(D_1\tau_t\phi_{\sigma})(z_1)(\tau_t\phi_{\sigma})(z_2)\Big)\\
	&\hspace{1cm}\times\big((I-K_{t,\sigma})^{-1}\tau_t\phi_{\sigma}\big)(z_2)\,\d^2z_2\,\d^2z_1=-q_{\sigma}^{[0]}(t,y)p_{\sigma}^{[1]}(t)+q_{\sigma}^{[1]}(t,y)p_{\sigma}^{[0]}(t),\ \ \ (t,y,\sigma)\in\mathbb{R}^2\times[0,\infty),
\end{align*}
and so with \eqref{e27}, precisely \eqref{e24}. The proof is complete.
\end{proof}
We are now one step closer to Theorem \ref{fimain1}, namely we are ready to derive the following auxiliary result. This results constitutes an alternative integro-differential connection for $F_{\sigma}(t)$.
\begin{prop}\label{theo1} For any $(t,\sigma)\in\mathbb{R}\times[0,\infty)$,
\begin{equation}\label{e28}
	F_{\sigma}(t)=\exp\left[-\int_t^{\infty}(s-t)\left\{\int_{-\infty}^{\infty}\big(q_{\sigma}(s,y)\big)^2\,\d y\right\}\d s\right],
\end{equation}
where $q_{\sigma}(t,y)=((I-K_{t,\sigma})^{-1}\tau_t\phi_{\sigma})(0,y)$ solves the integro-differential equation
\begin{equation}\label{e29}
	\frac{\partial^2}{\partial t^2}q_{\sigma}(t,y)=\left[t+\sigma y+2\int_{-\infty}^{\infty}\big(q_{\sigma}(t,\lambda)\big)^2\,\d\lambda\right]q_{\sigma}(t,y),\ \ \ \ (t,y,\sigma)\in\mathbb{R}^2\times[0,\infty),
\end{equation}
subject to the boundary condition
\begin{equation}\label{e30}
	q_{\sigma}(t,y)\sim(\tau_t\phi_{\sigma})(0,y)=\frac{1}{\pi^{\frac{1}{4}}}\e^{-\frac{1}{2}y^2}\textnormal{Ai}(t+\sigma y),
\end{equation}
valid as $t\rightarrow+\infty$, pointwise in $(y,\sigma)\in\mathbb{R}\times[0,\infty)$.
\end{prop}
\begin{proof} We already know \eqref{e28} and \eqref{e30} from the workings in Corollary \ref{cor1}, so it remains to show that $q_{\sigma}(t,y)$ solves \eqref{e29}. To this end $t$-differentiate $q_{\sigma}(t,y)$ twice via \eqref{e20}, then use \eqref{e23} and \eqref{e24},
\begin{align*}
	\frac{\partial^2}{\partial t^2}q_{\sigma}(t,y)&\,\,\equiv\frac{\partial^2}{\partial t^2}q_{\sigma}^{[0]}(t,y)\stackrel{\eqref{e20}}{=}\frac{\partial}{\partial t}\left[q_{\sigma}^{[1]}(t,y)-q_{\sigma}^{[0]}(t,y)p_{\sigma}^{[0]}(t)\right]\stackrel{\eqref{e20}}{=}q_{\sigma}^{[2]}(t,y)-q_{\sigma}^{[0]}(t,y)p_{\sigma}^{[1]}(t)\\
	&\hspace{1.5cm}-\big(q_{\sigma}^{[1]}(t,y)-q_{\sigma}^{[0]}(t,y)p_{\sigma}^{[0]}(t)\big)p_{\sigma}^{[0]}(t)+q_{\sigma}^{[0]}(t,y)\int_{-\infty}^{\infty}\big(q_{\sigma}^{[0]}(t,\lambda)\big)^2\,\d\lambda\\
	&\,\stackrel{\eqref{e24}}{=}(t+\sigma y)q_{\sigma}^{[0]}(t,y)-2q_{\sigma}^{[0]}(t,y)p_{\sigma}^{[1]}(t)+q_{\sigma}^{[0]}(t,y)\big(p_{\sigma}^{[0]}(t)\big)^2+q_{\sigma}^{[0]}(t,y)\int_{-\infty}^{\infty}\big(q_{\sigma}^{[0]}(t,\lambda)\big)^2\d\lambda\\
	&\,\stackrel{\eqref{e23}}{=}(t+\sigma y)q_{\sigma}(t,y)+2q_{\sigma}(t,y)\int_{-\infty}^{\infty}q_{\sigma}^2(t,\lambda)\,\d\lambda.
\end{align*}
The proof of Theorem \ref{theo1} is complete.
\end{proof}
Evidently, equation \eqref{e29} is similar to the class of integro-differential Painlev\'e-II equations studied in \cite[Proposition $5.2$]{ACQ}, although the origin in loc. cit. are Fredholm determinants corresponding to integral operators acting on $L^2(J\subset\mathbb{R})$ and not on $L^2(J_0\subset\mathbb{R}^2)$. We will now establish the precise relation between the Amir-Corwin-Quastel framework and our result for $F_{\sigma}(t)$ in Proposition \ref{theo1}. Once done, we will have completed our first proof of Theorem \ref{fimain1}.
\subsection{Simplifying the result}
We set out to simplify \eqref{e29} and \eqref{e30} and convert the same to \eqref{fi10},\eqref{fi11} and \eqref{fi11a}.
\begin{lem}\label{lem2} Let $(t,\sigma)\in\mathbb{R}\times[0,\infty)$ and
\begin{equation*}
	N_{t,\sigma}(a,b):=\frac{1}{\sqrt{\pi}}\int_0^{\infty}\left[\int_{-\infty}^{\infty}\e^{-y^2}\textnormal{Ai}(a+\sigma y+t+s)\textnormal{Ai}(s+b+\sigma y+t)\,\d y\right]\d s,\ \ \ a,b\in\mathbb{R}.
\end{equation*}
Then the self-adjoint integral operator $N_{t,\sigma}:L^2(0,\infty)\rightarrow L^2(0,\infty)$ given by
\begin{equation*}
	(N_{t,\sigma}f)(a):=\int_0^{\infty}N_{t,\sigma}(a,b)f(b)\,\d b,
\end{equation*}
is trace class for all $(t,\sigma)\in\mathbb{R}\times[0,\infty)$. Moreover, $0\leq N_{t,\sigma}\leq 1,\|N_{t,\sigma}\|\leq 1$ in operator norm on $L^2(0,\infty)$ and $I-N_{t,\sigma}$ is invertible on $L^2(0,\infty)$ for all $(t,\sigma)\in\mathbb{R}\times[0,\infty)$.
\end{lem}
\begin{proof} By the Airy differential equation $w''=zw$ for $w=\textnormal{Ai}(z),z\in\mathbb{C}$, integration by parts yields
\begin{align*}
	(a-b)N_{t,\sigma}(a,b)=&\,\frac{1}{\sqrt{\pi}}\int_{-\infty}^{\infty}\e^{-y^2}\big(\textnormal{Ai}(a+\sigma y+t)\textnormal{Ai}'(b+\sigma y+t)-\textnormal{Ai}'(a+\sigma y+t)\textnormal{Ai}(b+\sigma y+t)\big)\,\d y\\
	=&\,(a-b)\int_{-\infty}^{\infty}\Phi\left(\frac{y}{\sigma}\right)\textnormal{Ai}(a+y+t)\textnormal{Ai}(y+b+t)\,\d y,\ \ \ \ \ \Phi(y):=\frac{1}{\sqrt{\pi}}\int_{-\infty}^y\e^{-x^2}\,\d x
\end{align*}
and therefore the identity, valid for all $(t,\sigma)\in\mathbb{R}\times(0,\infty)$,
\begin{equation}\label{e31}
	N_{t,\sigma}(a,b)=\int_{-\infty}^{\infty}\Phi\left(\frac{y}{\sigma}\right)\textnormal{Ai}(a+y+t)\textnormal{Ai}(y+b+t)\,\d y,\ \ \ \ a,b\in\mathbb{R}.
\end{equation}
Consequently, $N_{t,\sigma}$ is an integral operator on $L^2(0,\infty)$ of the form studied in \cite[Section $9$]{Bo0}, so in particular trace class by \cite[Lemma $9.5$]{Bo0} for all $(t,\sigma)\in\mathbb{R}\times[0,\infty)$ (note that $N_{t,0}(a,b)=K_{\textnormal{Ai}}(a+t,b+t)$). Moreover $0\leq N_{t,\sigma}\leq 1$ and the outstanding claim about the invertibility of $I-N_{t,\sigma}$ were proven in \cite[Proposition $9.6$]{Bo0}. This completes our proof of Lemma \ref{lem2}.
\end{proof}
Note that \eqref{e31} is of the form \cite[$(5.1)$]{ACQ}, so we have by \cite[Proposition $5.2$]{ACQ} for the Fredholm determinant of $N_{t,\sigma}$ on $L^2(0,\infty)$,
\begin{equation}\label{e32}
	\det\big(I-N_{t,\sigma}\upharpoonright_{L^2(0,\infty)}\big)=\exp\left[-\int_t^{\infty}(s-t)\left\{\int_{-\infty}^{\infty}\big(p_{\sigma}(s,\lambda)\big)^2\d\nu_{\sigma}(\lambda)\right\}\d s\right],
\end{equation}
with
\begin{equation*}
	\d\nu_{\sigma}(\lambda):=\frac{1}{\sigma\sqrt{\pi}}\e^{-\lambda^2/\sigma^2}\d\lambda
\end{equation*}
and where $p_{\sigma}(t,\lambda)$ solves the integro-differential Painlev\'e-II equation
\begin{equation}\label{e33}
	\begin{cases}
	\displaystyle\frac{\partial^2}{\partial t^2}p_{\sigma}(t,y)=\left[t+y+2\int_{-\infty}^{\infty}\big(p_{\sigma}(t,\lambda)\big)^2\d\nu_{\sigma}(\lambda)\right]p_{\sigma}(t,y)&\bigskip\\	\displaystyle \hspace{0.6cm}p_{\sigma}(t,y)\sim\textnormal{Ai}(t+y),\ \ \ t\rightarrow+\infty,\ \ (y,\sigma)\in\mathbb{R}\times[0,\infty)&\\
	\end{cases}.
\end{equation}
We will now show that there is a direct correspondence between our \eqref{e28},\eqref{e29},\eqref{e30} and \eqref{e32},\eqref{e33}, without using the results in \cite{ACQ}. This will in turn conclude our first proof of Theorem \ref{fimain1}.
\begin{prop}\label{prop5} Define the one-variable function 
\begin{equation*}
	\mathbb{R}\ni x\mapsto\phi_{t,\sigma}^y(x):=\frac{1}{\pi^{\frac{1}{4}}}\e^{-\frac{1}{2}y^2}\textnormal{Ai}(x+\sigma y+t),\ \ \ (t,y,\sigma)\in\mathbb{R}^2\times[0,\infty).
\end{equation*}
Then for all $n\in\mathbb{Z}_{\geq 0}$ and any $(t,y,\sigma)\in\mathbb{R}^2\times[0,\infty)$, with $K_{t,\sigma}$ as in \eqref{e9a} and $N_{t,\sigma}$ as in \eqref{e31},
\begin{equation}\label{e34}
	\big(K_{t,\sigma}^n\tau_t\phi_{\sigma}\big)(0,y)=\big(N_{t,\sigma}^n\phi_{t,\sigma}^y\big)(0),
\end{equation}
where $\tau_t:L^2(\mathbb{R}^2)\rightarrow L^2(\mathbb{R}^2)$ denotes again the shift in the first variable, i.e. $(\tau_tf)(x,y)=f(x+t,y)$ .
\end{prop}
\begin{proof} There is nothing to prove when $n=0$ given that 
\begin{equation*}
	(\tau_t\phi_{\sigma})(0,y)\stackrel{\eqref{e9}}{=}\frac{1}{\pi^{\frac{1}{4}}}\e^{-\frac{1}{2}y^2}\textnormal{Ai}(\sigma y+t)=\phi_{t,\sigma}^y(0),
\end{equation*}
so we continue with general $n\in\mathbb{Z}_{\geq 1}$. By definition, with $n\in\mathbb{Z}_{\geq 0}$ arbitrary,
\begin{align}\label{e35}
	(K_{t,\sigma}^{n+1}&\,\tau_t\phi_{\sigma})(0,y)=\int_{J_0}K_{t,\sigma}\big((0,y),z_1\big)(K_{t,\sigma}^n\tau_t\phi_{\sigma})(z_1)\,\d^2 z_1\nonumber\\
	=&\,\int_{J_0}K_{t,\sigma}\big((0,y),z_1\big)\int_{J_0}\left[\int_{J_0^{n-2}}\prod_{j=1}^{n-1}K_{t,\sigma}(z_j,z_{j+1})\,\d^2z_2\cdots\d^2z_{n-1}\right](K_{t,\sigma}\tau_t\phi_{\sigma})(z_n)\,\d^2z_n\,\d^2 z_1,
\end{align}
where, with $K_{t,\sigma}(z_1,z_2)=\int_0^{\infty}(\tau_t\phi_{\sigma})(x_1+s,y_1)(\tau_t\phi_{\sigma})(s+x_2,y_2)\,\d s$ and Fubini's theorem,
\begin{align*}
	(K_{t,\sigma}\tau_t\phi_{\sigma})(z_n)=&\,\int_{J_0}K_{t,\sigma}(z_n,w)(\tau_t\phi_{\sigma})(w)\,\d^2 w\\
	=&\,\int_{J_0}\left[\int_0^{\infty}(\tau_t\phi_{\sigma})(x_n+s,y_n)(\tau_t\phi_{\sigma})(s+\lambda,\mu)\,\d s\right](\tau_t\phi_{\sigma})(\lambda,\mu)\,\d\lambda\,\d\mu\\
	=&\,\int_J(\tau_t\phi_{\sigma})(x_n+s,y_n)\left[\int_0^{\infty}(\tau_t\phi_{\sigma})(\lambda,\mu)(\tau_t\phi_{\sigma})(\lambda+s,\mu)\,\d\lambda\right]\,\d s\,\d\mu\\
	=&\,\int_{J_0}K_{t,\sigma}\big((0,\mu),(s,\mu)\big)(\tau_t\phi_{\sigma})(x_n+s,y_n)\,\d s\,\d\mu=\int_0^{\infty}N_{t,\sigma}(0,s)(\tau_t\phi_{\sigma})(x_n+s,y_n)\,\d s.
\end{align*}
Using the last formula back in \eqref{e35} together with the previous composition integral identity for $K_{t,\sigma}(z_j,z_{j+1})$ we obtain after integrating out all variables $x_1,x_2,\ldots,x_n$ over $(0,\infty)$ and afterwards $y_1,y_2,\ldots,y_n$ over $\mathbb{R}$,
\begin{align*}
	(K_{t,\sigma}^{n+1}\tau_t\phi_{\sigma})(0,y)=&\,\int_{\mathbb{R}^n}\int_0^{\infty}\int_{\mathbb{R}_+^{n-1}}\int_0^{\infty}(\tau_t\phi_{\sigma})(\lambda,y)K_{t,\sigma}\big((\lambda,y_1),(s_1,y_1)\big)\left[\prod_{j=1}^{n-2}K_{t,\sigma}\big((s_j,y_{j+1}),(s_{j+1},y_{j+1})\big)\right]\\
	&\hspace{1cm}\times K_{t,\sigma}\big((s_{n-1},y_n),(s,y_n)\big)N_{t,\sigma}(0,s)\,\d s\,(\d s_1\cdots\d s_{n-1})\,\d\lambda\,(\d y_1\cdots\d y_n)\\
	=&\,\int_0^{\infty}\int_{\mathbb{R}_+^{n-1}}\int_0^{\infty}(\tau_t\phi_{\sigma})(\lambda,y)N_{t,\sigma}(\lambda,s_1)\left[\prod_{j=1}^{n-2}N_{t,\sigma}(s_j,s_{j+1})\right]N_{t,\sigma}(s_{n-1},s)N_{t,\sigma}(0,s)\\
	&\hspace{1cm}\times\d s\,(\d s_1\cdots\d s_{n-1})\,\d\lambda.
\end{align*}
Here, $\mathbb{R}_+:=(0,\infty)$. Finally integrate out $s_1,\ldots,s_{n-1}\in(0,\infty)$ via Fubini's theorem and lastly $s\in(0,\infty)$, using the symmetry of $N_{t,\sigma}$,
\begin{align*}
	(K_{t,\sigma}^{n+1}\tau_t\phi_{\sigma})(0,y)=&\,\int_0^{\infty}\int_0^{\infty}(\tau_t\phi_{\sigma})(\lambda,y)N_{t,\sigma}^n(\lambda,s)N_{t,\sigma}(0,s)\,\d s\,\d\lambda=\int_0^{\infty}N_{t,\sigma}^{n+1}(0,\lambda)\phi_{t,\sigma}^y(\lambda)\,\d\lambda\\
	=&\,(N_{t,\sigma}^{n+1}\phi_{t,\sigma}^y)(0),\ \ \ \ (t,y,\sigma)\in\mathbb{R}^2\times[0,\infty).
\end{align*}
The proof of \eqref{e34} is now complete.
\end{proof}
Below we state the central Corollary to Proposition \ref{prop5}.
\begin{cor}\label{cor3} Consider $q_{\sigma}(t,y)$ as defined in Theorem \ref{theo1}. Then
\begin{equation}\label{e36}
	q_{\sigma}\left(t,\frac{y}{\sigma}\right)=\frac{1}{\pi^{\frac{1}{4}}}\e^{-\frac{y^2}{2\sigma^2}}\big((I-N_{t,\sigma})^{-1}\psi_{y+t}\big)(0),\ \ \ \ \ \psi_{y+t}(x):=\textnormal{Ai}(x+y+t),
\end{equation}
valid for all $(t,y,\sigma)\in\mathbb{R}^2\times(0,\infty)$.
\end{cor}
\begin{proof} By the Neumann series, Proposition \ref{prop1} and Lemma \ref{lem2},
\begin{eqnarray*}
	q_{\sigma}\left(t,\frac{y}{\sigma}\right)&=&\lim_{\gamma\uparrow 1}\big((I-\gamma K_{t,\sigma})^{-1}\tau_t\phi_{\sigma}\big)\left(0,\frac{y}{\sigma}\right)=\lim_{\gamma\uparrow 1}\sum_{m=0}^{\infty}\gamma^m(K_{t,\sigma}^m\tau_t\phi_{\sigma})\left(0,\frac{y}{\sigma}\right)\\
	&\stackrel{\eqref{e34}}{=}&\lim_{\gamma\uparrow 1}\sum_{m=0}^{\infty}\gamma^m(N_{t,\sigma}^m\phi_{t,\sigma}^{\frac{y}{\sigma}})(0)=\lim_{\gamma\uparrow 1}\big((I-\gamma N_{t,\sigma})^{-1}\phi_{t,\sigma}^{\frac{y}{\sigma}}\big)(0)=\big((I-N_{t,\sigma})^{-1}\phi_{t,\sigma}^{\frac{y}{\sigma}}\big)(0).
\end{eqnarray*}
However,
\begin{equation*}
	\phi_{t,\sigma}^{\frac{y}{\sigma}}(x)=\frac{1}{\pi^{\frac{1}{4}}}\e^{-\frac{y^2}{2\sigma^2}}\textnormal{Ai}(x+y+t)=\frac{1}{\pi^{\frac{1}{4}}}\e^{-\frac{y^2}{2\sigma^2}}\psi_{y+t}(x),\ \ \ \ (t,y,x,\sigma)\in\mathbb{R}^3\times[0,\infty),
\end{equation*}
which implies \eqref{e36}.
\end{proof}
We have now gathered sufficient information in order to prove Theorem \ref{fimain1}.
\begin{proof}[First proof of Theorem \ref{fimain1}]
By \eqref{e28} and \eqref{e36},
\begin{eqnarray*}
	F_{\sigma}(t)&=&\exp\left[-\int_t^{\infty}(s-t)\left\{\frac{1}{\sigma}\int_{-\infty}^{\infty}\left(q_{\sigma}\left(s,\frac{y}{\sigma}\right)\right)^2\,\d y\right\}\d s\right]\\
	&=&\exp\left[-\int_t^{\infty}(s-t)\left\{\int_{-\infty}^{\infty}\big(p_{\sigma}(s,y)\big)^2\,\d\nu_{\sigma}(y)\right\}\d s\right],\ \ \ \ p_{\sigma}(t,y):=((I-N_{t,\sigma})^{-1}\psi_{y+t})(0),
\end{eqnarray*}
with $\d\nu_{\sigma}(\lambda)=\frac{1}{\sigma\sqrt{\pi}}\e^{-\lambda^2/\sigma^2}\d\lambda$ as in \eqref{e32}. But with \eqref{e36}, i.e. with the identity
\begin{equation*}
	q_{\sigma}\left(t,\frac{y}{\sigma}\right)=\frac{1}{\pi^{\frac{1}{4}}}\e^{-\frac{y^2}{2\sigma^2}}p_{\sigma}(t,y),
\end{equation*}
our previous \eqref{e29},\eqref{e30} becomes exactly \eqref{e33}, i.e. \eqref{fi10},\eqref{fi11} and \eqref{fi11a} are now proven.
\end{proof}
The reader can notice that Theorem \ref{fimain1}, when combined with \cite[Proposition $5.2$]{ACQ}, yields the following useful and remarkable result, as discussed right above Corollary \ref{corfi1}.

\begin{cor}[Dimensional reduction]\label{theo2} Consider $F_{\sigma}(t)$ as written in \eqref{e10}, then for all $(t,\sigma)\in\mathbb{R}\times(0,\infty)$,
\begin{equation}\label{e37}
	F_{\sigma}(t)=\prod_{k=1}^{\infty}\big(1-\nu_k(t,\sigma)\big),
\end{equation}
in terms of the non-zero eigenvalues $\nu_k(t,\sigma)$ of the operator $N_{t,\sigma}$ introduced in Lemma \ref{lem2}, taking multiplicities into account.
\end{cor}
\begin{proof} The integro-differential Painlev\'e-II boundary value problem \eqref{e33} is uniquely solvable, cf. \cite[R$66$]{Bo0}. Thus, by \cite[Proposition $5.2$]{ACQ}, equivalently by \eqref{e32}, and \cite[Chapter IV, Theorem $6.1$]{GGK} we have indeed
\begin{equation*}
	F_{\sigma}(t)=\det\big(I-N_{t,\sigma}\upharpoonright_{L^2(0,\infty)}\big)=\prod_{k=1}^{\infty}\big(1-\nu_k(t,\sigma)\big),
\end{equation*}
as claimed in \eqref{e37}.
\end{proof}

\section{Second proof of Theorem \ref{fimain1}}\label{sec3}
Our second proof of Theorem \ref{fimain1} is much shorter and it consists in establishing \eqref{e37} from first principles, without exploiting Hankel composition structures. With \eqref{e37} in place, Theorem \ref{fimain1} is then a consequence of \cite[Proposition $5.2$]{ACQ} - so unlike our first proof, our second proof of Theorem \ref{fimain1} will rely on the results in \cite{ACQ}.
\begin{prop}\label{Sylvprop} For any $(t,\sigma)\in\mathbb{R}\times[0,\infty)$ and all $n\in\mathbb{Z}_{\geq 1}$,
\begin{equation}\label{e38}
	\tr_{L^2(J_0)}K_{t,\sigma}^n=\tr_{L^2(0,\infty)}N_{t,\sigma}^n.
\end{equation}
\end{prop}
\begin{proof} Let $(t,\sigma)\in\mathbb{R}\times[0,\infty)$ be fixed throughout. The idea is to apply Sylvester's identity \cite[Chapter IV, $(5.9)$]{GGK}, \cite[Corollary $3.8$]{Sim}, i.e. we first set out to factorize $K_{t,\sigma}:L^2(J_0)\rightarrow L^2(J_0)$ with kernel
\begin{equation}\label{e38a}
	K_{t,\sigma}(z_1,z_2)=\frac{1}{\sqrt{\pi}}\e^{-\frac{1}{2}y_1^2-\frac{1}{2}y_2^2}\int_0^{\infty}\textnormal{Ai}(x_1+\sigma y_1+t+s)\textnormal{Ai}(s+x_2+\sigma y_2+t)\,\d s,\ \ z_k=(x_k,y_k).
\end{equation}
Precisely, we define the bounded linear transformations
\begin{equation*}
	S_{\sigma}:L^2(\mathbb{R}^2)\rightarrow L^2(\mathbb{R}^2),\ \ \ \ \ \ \ \ \ \ \ \ \ \ \ \ \ \ \ H:L^2(\mathbb{R}^2)\rightarrow L^2(\mathbb{R}),\ \ \ \ \ \ \ \ \ \ \ \ \ \ \ \ \ \ \epsilon:L^2(J_0)\rightarrow L^2(\mathbb{R}^2),
\end{equation*}
\begin{equation*}
	(S_{\sigma}f)(x,y):=f(x+\sigma y,y),\ \ \ \ \ \ (Hf)(x):=\frac{1}{\pi^{\frac{1}{4}}}\int_{-\infty}^{\infty}f(x,y)\e^{-\frac{1}{2}y^2}\,\d y,\ \ \ \ \ \ (\epsilon f)(x,y):=\chi_{(0,\infty)}(x)f(x,y),
\end{equation*}
where $\chi_A$ denotes the characteristic function of a set $A\subset\mathbb{R}$, and whose Hilbert space adjoints are equal to 
\begin{equation*}
	S_{\sigma}^{\ast}:L^2(\mathbb{R}^2)\rightarrow L^2(\mathbb{R}^2),\ \ \ \ \ \ \ \ \ \ \ \ \ \ \ \ \ \ \ H^{\ast}:L^2(\mathbb{R})\rightarrow L^2(\mathbb{R}^2),\ \ \ \ \ \ \ \ \ \ \ \ \ \ \ \ \ \ \ \ \ \epsilon^{\ast}:L^2(\mathbb{R}^2)\rightarrow L^2(J_0),
\end{equation*}
\begin{equation*}
	(S_{\sigma}^{\ast}f)(x,y)=f(x-\sigma y,y),\ \ \ \ \ \ \ \ \ \ \ \ \ \ \ (H^{\ast}f)(x,y)=\frac{1}{\pi^{\frac{1}{4}}}f(x)\e^{-\frac{1}{2}y^2},\ \ \ \ \ \ \ \ \ \ \ \ \ \ (\epsilon^{\ast} f)(x,y)=f(x,y).\ \ \ 
\end{equation*}
We also require the linear transformation $Q_t:L^2(0,\infty)\rightarrow L^2(\mathbb{R})$ given by
\begin{equation}\label{e38aa}
	(Q_tf)(x):=\int_0^{\infty}\textnormal{Ai}(x+y+t)f(y)\,\d y,
\end{equation}
which is an isometry by \eqref{appC5} below and whose Hilbert space adjoint equals
\begin{equation*}
	Q_t^{\ast}: L^2(\mathbb{R})\rightarrow L^2(0,\infty),\ \ \ \ (Q_t^{\ast}f)(x)=\int_{-\infty}^{\infty}\textnormal{Ai}(x+y+t)f(y)\,\d y.
\end{equation*}
Now define the bounded linear transformation $P_{t,\sigma}:=Q_t^{\ast}HS_{\sigma}^{\ast}\epsilon:L^2(J_0)\rightarrow L^2(0,\infty)$, i.e. we are considering the integral operator from $L^2(J_0)$ to $L^2(0,\infty)$ with kernel 
\begin{equation}\label{e38b}
	P_{t,\sigma}\big(x,(y,z)\big):=\frac{1}{\pi^{\frac{1}{4}}}\textnormal{Ai}(x+y+\sigma z+t)\e^{-\frac{1}{2}z^2},\ \ \ x\in(0,\infty),\ \ (y,z)\in J_0.
\end{equation}
Its Hilbert space adjoint $P_{t,\sigma}^{\ast}=\epsilon^{\ast}S_{\sigma}H^{\ast}Q_t:L^2(0,\infty)\rightarrow L^2(J_0)$ is again an integral operator, this time from $L^2(0,\infty)$ to $L^2(J_0)$ and with kernel
\begin{equation}\label{e38c}
	P_{t,\sigma}^{\ast}\big((x,y),z\big):=\frac{1}{\pi^{\frac{1}{4}}}\textnormal{Ai}(x+\sigma y+z+t)\e^{-\frac{1}{2}y^2}.
\end{equation}
Most importantly, compare \eqref{e38a},\eqref{e38b},\eqref{e38c}, we record the factorization
\begin{equation}\label{e38d}
	K_{t,\sigma}=P_{t,\sigma}^{\ast}P_{t,\sigma}:L^2(J_0)\rightarrow L^2(J_0),
\end{equation}
where $K_{t,\sigma}$ is trace class on $L^2(J_0)$ by Proposition \ref{prop1}. Second, we note that $P_{t,\sigma}P_{t,\sigma}^{\ast}:L^2(0,\infty)\rightarrow L^2(0,\infty)$ has kernel
\begin{align*}
	(P_{t,\sigma}P_{t,\sigma}^{\ast})(a,b)=&\,\int_0^{\infty}\int_{-\infty}^{\infty}P_{t,\sigma}\big(a,(c_1,c_2)\big)P_{t,\sigma}^{\ast}\big((c_1,c_2),b\big)\,\d c_2\d c_1\\
	=&\,\frac{1}{\sqrt{\pi}}\int_0^{\infty}\left[\int_{-\infty}^{\infty}\e^{-c_2^2}\textnormal{Ai}(a+\sigma c_2+t+c_1)\textnormal{Ai}(c_1+b+\sigma c_2+t)\,\d c_2\right]\d c_1,
\end{align*}
and which is equal to $N_{t,\sigma}(a,b)$, see Lemma \ref{lem2}. Consequently, using that $N_{t,\sigma}$ is trace class on $L^2(0,\infty)$, we have by Sylvester's identity \cite[Chapter IV, $(5.9)$]{GGK},
\begin{align*}
	\det\big(I-\gamma K_{t,\sigma}\upharpoonright_{L^2(J_0)}\big)\stackrel{\eqref{e38d}}{=}\det\big(I-\gamma P_{t,\sigma}^{\ast}P_{t,\sigma}\upharpoonright_{L^2(J_0)}\big)
	=\det\big(I-\gamma N_{t,\sigma}\upharpoonright_{L^2(0,\infty)}\big),\ \ 0\leq|\gamma|<1.
\end{align*}
The last equality establishes \eqref{e38} via the Plemelj-Smithies formula \cite[Chapter II, $(3.2)$]{GGK}.
\end{proof}
Equipped with \eqref{e38} we can now prove Theorem \ref{fimain1}.
\begin{proof}[Second proof of Theorem \ref{fimain1}] The trace identities \eqref{e38} establish \eqref{e37} via \eqref{e3}, \eqref{e5} and \cite[Chapter IV, $(5.11),(6.2)$]{GGK}, without exploiting Hankel composition structures. In turn, \eqref{fi10},\eqref{fi11} and \eqref{fi11a} follow now from \cite[Proposition $5.2$]{ACQ} using \eqref{e31}.
\end{proof}

\section{Proof of Corollaries \ref{corfi1} and \ref{corfi2}}\label{sec4}

We begin with a basic estimate for 
\begin{equation}\label{e41}
	\Phi(z)=\frac{1}{\sqrt{\pi}}\int_{-\infty}^z\e^{-x^2}\,\d x=1-\frac{1}{\sqrt{\pi}}\int_z^{\infty}\e^{-x^2}\,\d x=1-\frac{1}{2}\textnormal{erfc}(z),\ \ \ \ \ \Phi(-z)=\frac{1}{2}\textnormal{erfc}(z),\ \ \ z\in\mathbb{C},
\end{equation}
that follows from Mill's ratio in \cite[$7.8.1$]{NIST}.
\begin{lem}\label{lem4} For all $x\in\mathbb{R}$ and $\sigma>0$,
\begin{equation}\label{e42}
	\left|\Phi\left(\frac{x}{\sigma}\right)-\chi_{[0,\infty)}(x)\right|\leq\frac{1}{2}\exp\left[-\frac{x^2}{\sigma^2}\right],
\end{equation}
where $\chi_A$ denotes the characteristic function of a set $A\subset\mathbb{R}$.
\end{lem}
\begin{proof} Fix $\sigma>0$ and let $x\geq 0$, then by \eqref{e41} and \cite[$7.8.1,7.8.2$]{NIST},
\begin{equation*}
	\left|\Phi\left(\frac{x}{\sigma}\right)-1\right|\leq\frac{1}{2}\exp\left[-\frac{x^2}{\sigma^2}\right],\ \ \ \ \ \ \ \ 
	\left|\Phi\left(-\frac{x}{\sigma}\right)\right|\leq\frac{1}{2}\exp\left[-\frac{x^2}{\sigma^2}\right],
\end{equation*}
which proves \eqref{e42}.
\end{proof}
In turn, the proof of Corollary \ref{corfi2} is now straightfoward.
\begin{proof}[Proof of Corollary \ref{corfi2}] First, pointwise in $x,y\in(0,\infty)$ and $t\in\mathbb{R}$,
\begin{equation}\label{e43}
	\lim_{\sigma\downarrow 0}N_{t,\sigma}(x,y)=K_{\textnormal{Ai}}^t(x,y):=\int_0^{\infty}\textnormal{Ai}(x+z+t)\textnormal{Ai}(z+y+t)\,\d z.
\end{equation}
Indeed, the pointwise limit \eqref{e43} is a consequence of \eqref{e31},\eqref{e42} and the boundedness of $K_{\textnormal{Ai}}^t(x,y)$. For \eqref{fi13} we need to establish trace norm convergence on $L^2(0,\infty)$ of $N_{t,\sigma}$, see \eqref{e31},\eqref{e37}, to $K_{\textnormal{Ai}}^t$, see \eqref{e43},\eqref{fi13}. To this end recall that $N_{t,\sigma},K_{\textnormal{Ai}}^t$ are non-negative and that both are self-adjoint, hence by \cite[Theorem $2.20$]{Sim}, we are left to show $N_{t,\sigma}\rightarrow K_{\textnormal{Ai}}^t$ weakly and $\|N_{t,\sigma}\|_1\rightarrow \|K_{\textnormal{Ai}}^t\|_1$ with the trace norm $\|\cdot\|_1$. But this is straightforward given \eqref{e42}, Fubini's theorem, Cauchy-Schwarz inequality and the known analytic and asymptotic properties of the Airy function on $\mathbb{R}$,
\begin{equation*}
	\lim_{\sigma\downarrow 0}\|N_{t,\sigma}\|_1=\lim_{\sigma\downarrow 0}\int_0^{\infty}L_{t,\sigma}(x,x)\,\d x=\int_0^{\infty}K_{\textnormal{Ai}}^t(x,x)\,\d x=\|K_{\textnormal{Ai}}^t\|_1,
\end{equation*}
and, for any $f,g\in L^2(0,\infty)$,
\begin{equation*}
	\lim_{\sigma\downarrow 0}\int_0^{\infty}\int_0^{\infty}N_{t,\sigma}(x,y)f(y)\overline{g(x)}\,\d x\,\d y=\int_0^{\infty}\int_0^{\infty}K_{\textnormal{Ai}}^t(x,y)f(y)\overline{g(x)}\,\d x\,\d y.
\end{equation*}
The proof of \eqref{fi13} is now complete by \cite[Theorem $3.4$]{Sim}.
\end{proof}
The $\sigma\rightarrow+\infty$ degeneration of $F_{\sigma}(t)$ in Corollary \ref{corfi1} is less straightforward and we now state a streamlined proof for it, using a novel contour integral formula for the kernel of $N_{t,\sigma}$, see \eqref{f2} below.
\begin{proof}[Proof of Corollary \ref{corfi1}] Let $\sigma>1$ and set
\begin{equation}\label{e44a}
	a_{\sigma}:=\frac{\sigma}{\sqrt{6\ln\sigma}},\ \ \ \ \ \ \ c_{\sigma}:=a_{\sigma}\left(3\ln\sigma-\frac{5}{4}\ln(6\ln\sigma)-\ln(2\pi)\right).
\end{equation}
We claim that, pointwise in $x,y\in(0,\infty)$ and $t\in\mathbb{R}$,
\begin{equation}\label{e45}
	\lim_{\sigma\rightarrow+\infty}a_{\sigma}N_{a_{\sigma}t,\sigma}(c_{\sigma}+a_{\sigma}x,c_{\sigma}+a_{\sigma}y)=\e^{-t-x}\begin{cases}0,&x\neq y\\ 1,&x=y\end{cases}.
\end{equation}
Indeed, for $(t,\sigma)\in\mathbb{R}\in(0,\infty)$, by \eqref{e31} and Fubini's theorem, for any $a,b\in\mathbb{R}$,
\begin{align}
	&\,N_{t,\sigma}(a,b)=\frac{1}{\sigma\sqrt{\pi}}\int_{-\infty}^{\infty}\int_{-\infty}^y\e^{-(x/\sigma)^2}\textnormal{Ai}(a+y+t)\textnormal{Ai}(y+b+t)\,\d x\,\d y\label{f0}\\
	=&\,\frac{1}{\sigma\sqrt{\pi}}\int_{-\infty}^{\infty}\int_x^{\infty}\e^{-(x/\sigma)^2}\textnormal{Ai}(a+y+t)\textnormal{Ai}(y+b+t)\,\d y\,\d x\stackrel{\eqref{e43}}{=}\frac{1}{\sigma\sqrt{\pi}}\int_{-\infty}^{\infty}\e^{-(x/\sigma)^2}K_{\textnormal{Ai}}^t(a+x,b+x)\,\d x.\nonumber
\end{align}
Hence, using the Gaussian integral
\begin{equation*}
	\e^{-x^2}=\frac{1}{2\sqrt{\pi}}\int_{\Gamma}\e^{-\frac{1}{4}\lambda^2-\im x\lambda}\,\d\lambda,\ \ x\in\mathbb{R},
\end{equation*}
where $\Gamma$ denotes any smooth contour in the upper half-plane oriented from $\infty\e^{\im\alpha}$ to $\infty\e^{\im\beta}$ with $\alpha\in(\frac{3\pi}{4},\pi)$ and $\beta\in(0,\frac{\pi}{4})$, \eqref{f0} becomes
\begin{equation}\label{f1}
	N_{t,\sigma}(a,b)=\frac{1}{2\pi\sigma}\int_{\Gamma}\e^{-\frac{1}{4}\lambda^2}\left[\int_{-\infty}^{\infty}\e^{px}K_{\textnormal{Ai}}^t(a+x,b+x)\,\d x\right]\bigg|_{p=-\im\lambda/\sigma}\d\lambda,
\end{equation}
and we proceed with the evaluation of the Laplace transform of the Airy kernel along the lines of \cite[Lemma $2.6$]{O}. In fact, for any $p\in\mathbb{C}$ with $\Re p>0$, the workings in \cite[Lemma $2.6$]{O} yield
\begin{equation*}
	\int_{-\infty}^{\infty}\e^{px}\textnormal{Ai}(a+x)\textnormal{Ai}(b+x)\,\d x=\frac{p^{-\frac{1}{2}}}{2\sqrt{\pi}}\exp\left[\frac{p^3}{12}-\frac{a+b}{2}p-\frac{(a-b)^2}{4p}\right],\ \ \ \ a,b\in\mathbb{C},
\end{equation*}
with the principal branch for the fractional exponent such that $z^{\alpha}>0$ for $z>0$. Consequently \eqref{f1} equals
\begin{equation}\label{f2}
	N_{t,\sigma}(a,b)=\frac{\e^{\im\frac{3\pi}{4}}}{4\pi}\sqrt{\frac{\sigma}{\pi}}\int_{\Gamma}\exp\left(\im\left[\frac{\lambda^3}{12\sigma^3}+\frac{\lambda}{2\sigma}(a+b+2t)-\frac{\sigma}{4\lambda}(a-b)^2\right]-\frac{\lambda^2}{4}\right)\lambda^{-\frac{3}{2}}\,\d\lambda,\ \ \ a,b\in\mathbb{R}.
\end{equation}
Starting from the seemingly novel contour integral formula \eqref{f2} for the kernel of $N_{t,\sigma}:L^2(0,\infty)\rightarrow L^2(0,\infty)$ we now set out to prove \eqref{e45} and subsequently \eqref{fi12a}. To this end rewrite \eqref{f2} in the form
\begin{equation}\label{f2a}
	N_{t,\sigma}(a,b)=\frac{\e^{\im\frac{3\pi}{4}}}{4\pi}\sqrt{\frac{\sigma}{\pi}}\int_{\Gamma}\e^{f(\lambda;\sigma,t,a)}g(\lambda;\sigma,a-b)\,\d\lambda,\ \ \ \ \ \begin{cases}\displaystyle f(\lambda;\sigma,t,a):=\frac{\im\lambda}{\sigma}(a+t)-\frac{\lambda^2}{4}-\frac{3}{2}\ln\lambda&\bigskip\\
	\displaystyle \,\,\,\,\,g(\lambda;\sigma,a):=\exp\left[\frac{\im\lambda^3}{12\sigma^3}-\frac{\im\lambda}{2\sigma}a-\frac{\im\sigma}{4\lambda}a^2\right]&
	\end{cases}\!\!,
\end{equation}
afterwards choose $\Gamma=\mathbb{R}+\im\delta$ with $\delta>0$ and estimate coarsely by triangle inequality. For any $a,b\in\mathbb{R}$,
\begin{equation*}
	\big| N_{t,\sigma}(a,b)\big|\leq\frac{1}{4\pi}\sqrt{\frac{\sigma}{\pi}}\exp\left[-\frac{\delta}{\sigma}(a+t)+\frac{\delta^2}{4}+\frac{\delta^3}{12\sigma^3}+\frac{\delta}{2\sigma}(a-b)\right]\delta^{-\frac{3}{2}}\int_{-\infty}^{\infty}\e^{-\frac{1}{4}s^2}\exp\left[-\frac{\sigma\delta}{4}\frac{(a-b)^2}{s^2+\delta^2}\right]\d s,
\end{equation*}
where
\begin{equation*}
	\int_0^{\delta}\e^{-\frac{1}{4}s^2}\exp\left[-\frac{\sigma\delta}{4}\frac{(a-b)^2}{s^2+\delta^2}\right]\d s\leq \sqrt{\pi}\exp\left[-\frac{\sigma}{8\delta}(a-b)^2\right].
\end{equation*}
Since also
\begin{equation*}
	\int_{\delta}^{\infty}\e^{-\frac{1}{4}s^2}\exp\left[-\frac{\sigma\delta}{4}\frac{(a-b)^2}{s^2+\delta^2}\right]\d s\leq\int_{\delta}^{\infty}\e^{-\frac{1}{8}s^2}\underbrace{\exp\left[-\frac{s^2}{8}-\frac{\sigma\delta}{8s^2}(a-b)^2\right]}_{\leq\,\exp\left[-\frac{1}{4}\sqrt{\sigma\delta}\,|a-b|\right],\,s\neq 0}\d s\leq\sqrt{2\pi}\exp\left[-\frac{1}{4}\sqrt{\sigma\delta}\,|a-b|\right],
\end{equation*}
we have thus
\begin{equation}\label{f2b}
	\big|N_{t,\sigma}(a,b)\big|\leq\frac{1}{\pi}\sqrt{\frac{\sigma}{2}}\exp\left[-\frac{\delta}{\sigma}(a+t)+\frac{\delta^2}{4}+\frac{\delta^3}{12\sigma^3}+\frac{\delta}{2\sigma}(a-b)\right]\delta^{-\frac{3}{2}}\left(\e^{-\frac{\sigma}{8\delta}(a-b)^2}+\e^{-\frac{1}{4}\sqrt{\sigma\delta}|a-b|}\right).
\end{equation}
The estimate \eqref{f2b} establishes the limiting behavior in \eqref{e45} for $x\neq y$ as soon as
\begin{equation*}
	c_{\sigma}>0\ \ \textnormal{for large}\ \sigma>1\ \textnormal{and}\ \ \frac{a_{\sigma}}{\sigma}\rightarrow 0,\ \ \sqrt{\sigma}a_{\sigma}\rightarrow+\infty\ \  \textnormal{as}\ \ \sigma\rightarrow+\infty,
\end{equation*}
and which is certainly true back in \eqref{e44a}. The precise dependence of the constants $c_{\sigma}$ and $a_{\sigma}$ on $\sigma>1$ originates from the following computation:  with \eqref{e44a} in place and $\delta=\sqrt{6\ln\sigma},\sigma>1$, we obtain from \eqref{f2b} for any $x,y\in\mathbb{R}$ and $(t,\sigma)\in\mathbb{R}\times(1,\infty)$,
\begin{equation*}
	\big|a_{\sigma}N_{a_{\sigma}t,\sigma}(c_{\sigma}+a_{\sigma}x,c_{\sigma}+a_{\sigma}y)\big|\leq 6\e^{-\frac{1}{2}(x+y)-t},
\end{equation*}
and thus by Hadamard's determinant inequality \cite[Chapter IV, $(8.7)$]{GGK},
\begin{equation*}
	\Big|\det\big[a_{\sigma}N_{a_{\sigma}t,\sigma}(c_{\sigma}+a_{\sigma}x_i,c_{\sigma}+a_{\sigma}x_j)\big]_{i,j=1}^n\Big|\leq 6^n\e^{-nt}n^{\frac{n}{2}}\e^{-\frac{1}{2}\sum_{j=1}^nx_j},\ \ \ x_j\geq 0.
\end{equation*}
Consequently, by \eqref{e37}, by the dominated convergence theorem and \cite[Theorem $3.10$]{Sim}, for any fixed $t\in\mathbb{R}$, with $\mathbb{R}_+=(0,\infty)$,
\begin{align*}
	\lim_{\sigma\rightarrow+\infty}&\,F_{\sigma}(c_{\sigma}+a_{\sigma}t)=1+\lim_{\sigma\rightarrow+\infty}\sum_{n=1}^{\infty}\frac{(-1)^n}{n!}\int_{\mathbb{R}_+^n}\det\big[N_{c_{\sigma}+a_{\sigma}t,\sigma}(x_i,x_j)\big]_{i,j=1}^n\d x_1\cdots\d x_n\\
	&=1+\lim_{\sigma\rightarrow+\infty}\sum_{n=1}^{\infty}\frac{(-1)^n}{n!}\int_{\mathbb{R}_+^n}\det\big[a_{\sigma}N_{a_{\sigma}t,\sigma}(c_{\sigma}+a_{\sigma}x_i,c_{\sigma}+a_{\sigma}x_j)\big]_{i,j=1}^n\d x_1\cdots\d x_n\\
	&=1+\sum_{n=1}^{\infty}\frac{(-1)^n}{n!}\int_{\mathbb{R}_+^n}\det\Big[\lim_{\sigma\rightarrow+\infty}a_{\sigma}N_{a_{\sigma}t,\sigma}(c_{\sigma}+a_{\sigma}x_i,c_{\sigma}+a_{\sigma}x_j)\Big]_{i,j=1}^n\d x_1\cdots\d x_n\\
	&\!\!\!\stackrel{\eqref{e45}}{=}1+\sum_{n=1}^{\infty}\frac{(-1)^n}{n!}\int_{\mathbb{R}_+^n}\left(\prod_{i=1}^n\e^{-t-x_i}\right)\d x_1\cdots \d x_n=1+\sum_{n=1}^{\infty}\frac{(-1)^n}{n!}\e^{-nt}=\e^{-\e^{-t}},
\end{align*}
so we are left to establish the pointwise limit in \eqref{e45} for $x=y\in(0,\infty)$. For this we return to \eqref{f2a}, take $a=b=c_{\sigma}+a_{\sigma}x$, replace $t\mapsto a_{\sigma}t$ and choose $\Gamma=\mathbb{R}+\im\lambda_{\ast}$ with
\begin{equation*}
	\lambda_{\ast}:=\frac{1}{\sigma}(c_{\sigma}+a_{\sigma}x+a_{\sigma}t)+\sqrt{\left(\frac{c_{\sigma}+a_{\sigma}x+a_{\sigma}t}{\sigma}\right)^2+3}\,\,>0.
\end{equation*}
Note that $\im\lambda_{\ast}$ constitutes the unique stationary point of $f(\lambda)\equiv f(\lambda;\sigma,a_{\sigma}t,c_{\sigma}+a_{\sigma}x)$ in the upper half-plane, so after parametrizing $\Gamma$ accordingly,
\begin{align}
	&\,\,\,\,a_{\sigma}N_{a_{\sigma}t,\sigma}(c_{\sigma}+a_{\sigma}x,c_{\sigma}+a_{\sigma}x)=\frac{1}{4\pi}\sqrt{\frac{\sigma}{\pi}}\,a_{\sigma}\exp\left[-\frac{\lambda_{\ast}}{\sigma}(c_{\sigma}+a_{\sigma}x+a_{\sigma}t)+\frac{1}{4}\lambda_{\ast}^2-\frac{3}{2}\ln\lambda_{\ast}\right]\label{f4}\\
	&\,\times\int_{-\infty}^{\infty}\exp\left[-\frac{1}{4}s^2+\frac{\im s}{\sigma}(c_{\sigma}+a_{\sigma}x+a_{\sigma}t)-\frac{\im}{2}\lambda_{\ast}s\right]h_{t,\sigma}(s)\,\d s,\  h_{t,\sigma}(s):=\left(1-\frac{\im s}{\lambda_{\ast}}\right)^{-\frac{3}{2}}\exp\left[\frac{\im}{12\sigma^3}(s+\im\lambda_{\ast})^3\right].\nonumber
\end{align}
Next, given that $a_{\sigma}/\sigma\rightarrow 0$ as $\sigma\rightarrow+\infty$ and since $t\in\mathbb{R}$ is fixed, the integral in \eqref{f4} yields after Taylor expansion at $\sigma=\infty$ the following leading order asymptotic behavior
\begin{align*}
	\int_{-\infty}^{\infty}&\,\exp\left[-\frac{1}{4}s^2+\frac{\im s}{\sigma}(c_{\sigma}+a_{\sigma}x+a_{\sigma}t)-\frac{\im}{2}\lambda_{\ast}s\right]h_{t,\sigma}(s)\,\d s\sim\int_{-\infty}^{\infty}\exp\left[-\frac{1}{4}s^2+\frac{\im s}{\sigma} c_{\sigma}-\frac{\im}{2}\lambda_{\ast}s\right]\,\d s\\
	&\hspace{2cm}=2\sqrt{\pi}\exp\left[-\frac{1}{\sigma^2}\left(c_{\sigma}-\frac{1}{2}\lambda_{\ast}\sigma\right)^2\right],\ \ \ \sigma\rightarrow+\infty.
\end{align*}
Consequently, back in \eqref{f4}
\begin{equation}\label{f4a}
	a_{\sigma}N_{a_{\sigma}t,\sigma}(c_{\sigma}+a_{\sigma}x,c_{\sigma}+a_{\sigma}x)\sim\frac{\sqrt{\sigma}}{2\pi}a_{\sigma}\exp\left[-\frac{\lambda_{\ast}}{\sigma}a_{\sigma}(x+t)-\Big(\frac{c_{\sigma}}{\sigma}\Big)^2\right]\lambda_{\ast}^{-\frac{3}{2}},\ \ \sigma\rightarrow+\infty,
\end{equation}
i.e. it remains to insert the explicit expressions for $\lambda_{\ast},c_{\sigma},a_{\sigma}$ and carefully expand the result at $\sigma=+\infty$. In detail, we have as $\sigma\rightarrow\infty$, pointwise in $(x,t)\in(0,\infty)\times\mathbb{R}$,
\begin{equation*}
	\frac{1}{\sigma}(c_{\sigma}+a_{\sigma}x+a_{\sigma}t)\stackrel{\eqref{e44a}}{=}\frac{1}{2}\sqrt{6\ln\sigma}+\mathcal{O}\left(\frac{\ln\ln\sigma}{\sqrt{\ln\sigma}}\right),\ \ \ \ \ \lambda_{\ast}=\sqrt{6\ln\sigma}+\mathcal{O}\left(\frac{\ln\ln\sigma}{\sqrt{\ln\sigma}}\right),
\end{equation*}
\begin{equation*}
	\Big(\frac{c_{\sigma}}{\sigma}\Big)^2\stackrel{\eqref{e44a}}{=}\frac{3}{2}\ln\sigma-\frac{5}{4}\ln(6\ln\sigma)-\ln(2\pi)+\mathcal{O}\left(\frac{\ln^2\ln\sigma}{\ln\sigma}\right),
\end{equation*}
and so \eqref{f4a} yields at once the outstanding pointwise limit in \eqref{e45}. The proof of Corollary \ref{corfi1} is thus complete.
\end{proof}


\section{Proof of Theorem \ref{fimain2} and \ref{fimain3}}\label{sec5}

In this section we carry out the right tail $t\rightarrow+\infty$ asymptotic analysis of $F_{\sigma}(t)$, relying on the insight that $N_{t,\sigma}:L^2(0,\infty)\rightarrow L^2(0,\infty)$ as defined in Lemma \ref{lem2} and \eqref{e31} relates to an integrable operator in the sense of \cite{IIKS}. See \cite[Proposition $5.1$]{ACQ} for essentially the same result as written in \eqref{e39},\eqref{e40} below. 
\begin{lem}\label{lem3} Let $(t,\sigma)\in\mathbb{R}\times(0,\infty)$ and
\begin{equation}\label{e39}
	M_{t,\sigma}(x,y):=\sqrt{\Phi\left(\frac{x}{\sigma}\right)}K_{\textnormal{Ai}}(x+t,y+t)\sqrt{\Phi\left(\frac{y}{\sigma}\right)},\ \ \ x,y\in\mathbb{R},
\end{equation}
with $\Phi(z)=1-\frac{1}{2}\textnormal{erfc}(z)$ and the Airy kernel $K_{\textnormal{Ai}}$, see \eqref{fi4}. Then the self-adjoint integral operator $M_{t,\sigma}:L^2(\mathbb{R})\rightarrow L^2(\mathbb{R})$ given by
\begin{equation*}
	(M_{t,\sigma}f)(x):=\int_{-\infty}^{\infty}M_{t,\sigma}(x,y)f(y)\,\d y,
\end{equation*}
is trace class for all $(t,\sigma)\in\mathbb{R}\times(0,\infty)$ and we have
\begin{equation}\label{e40}
	F_{\sigma}(t)=\det\big(I-M_{t,\sigma}\upharpoonright_{L^2(\mathbb{R})}\big),
	\ \ \ \ (t,\sigma)\in\mathbb{R}\times(0,\infty).
\end{equation}
\end{lem}
\begin{proof} Define the linear transformations $A_{t,\sigma}:L^2(0,\infty)\rightarrow L^2(\mathbb{R})$ and $B_{t,\sigma}:L^2(\mathbb{R})\rightarrow L^2(0,\infty)$ with
\begin{equation*}
	\begin{cases}
	\displaystyle (A_{t,\sigma}f)(x):=\int_0^{\infty}\sqrt{\Phi\left(\frac{x}{\sigma}\right)}\textnormal{Ai}(x+y+t)f(y)\,\d y&\smallskip\\
	\displaystyle (B_{t,\sigma}g)(x):=\int_{-\infty}^{\infty}\textnormal{Ai}(x+y+t)\sqrt{\Phi\left(\frac{y}{\sigma}\right)}g(y)\,\d y&
	\end{cases}.
\end{equation*}
Now apply \cite[Lemma $2.2$]{BCT} and conclude that these are Hilbert-Schmidt transformations on their respective spaces for all $(t,\sigma)\in\mathbb{R}\times(0,\infty)$. Consequently, Sylvester's identity applies,
\begin{equation*}
	\det\big(I-B_{t,\sigma}A_{t,\sigma}\upharpoonright_{L^2(0,\infty)}\big)=\det\big(I-A_{t,\sigma}B_{t,\sigma}\upharpoonright_{L^2(\mathbb{R})}\big),\ \ \ \ \ \ \ (t,\sigma)\in\mathbb{R}\times(0,\infty),
\end{equation*}
and it simply remains to check that the kernel of $B_{t,\sigma}A_{t,\sigma}$ matches the kernel of $N_{t,\sigma}$ and likewise the kernel of $A_{t,\sigma}B_{t,\sigma}$ the kernel of $M_{t,\sigma}$. The proof of \eqref{e40} is complete given that all kernel functions are continuous.
\end{proof}

Having established \eqref{e39} and \eqref{e40} the following general Riemann-Hilbert theory of integrable operators is now at our disposal, cf. \cite{IIKS}. First, by a change of variables,
\begin{equation}\label{f4b}
	F_{\sigma}(t)\stackrel{\eqref{e40}}{=}\det\big(I-M_{t,\sigma}\upharpoonright_{L^2(\mathbb{R})}\big)=\det\big(I-\overline{M}_{t,\sigma}\upharpoonright_{L^2(\mathbb{R})}\big),\ \ \ (t,\sigma)\in\mathbb{R}\times(0,\infty),
\end{equation}
where $\overline{M}_{t,\sigma}:L^2(\mathbb{R})\rightarrow L^2(\mathbb{R})$ is trace class with kernel
\begin{equation*}
	\overline{M}_{t,\sigma}(x,y):=\sigma\sqrt{\Phi(x)}K_{\textnormal{Ai}}(x\sigma+t,y\sigma+t)\sqrt{\Phi(y)}=\frac{f_1(x)f_2(y)-f_2(x)f_1(y)}{x-y},\ \ \ \ x,y\in\mathbb{R},
\end{equation*}
having used the abbreviations $f_1(x):=\sqrt{\Phi(x)}\textnormal{Ai}(x\sigma+t)$ and $f_2(x):=\sqrt{\Phi(x)}\textnormal{Ai}'(x\sigma+t)$. In turn, we consider the below Riemann-Hilbert problem (RHP).
\begin{problem}[Master problem]\label{master} Fix $(t,\sigma)\in\mathbb{R}\times(0,\infty)$. Now determine ${\bf X}(z)={\bf X}(z;t,\sigma)\in\mathbb{C}^{2\times 2}$ such that
\begin{enumerate}
	\item[(1)] ${\bf X}(z)$ is analytic for $z\in\mathbb{C}\setminus\mathbb{R}$ and extends continuously to the closed upper and lower half-plane.
	\item[(2)] The continuous limiting values ${\bf X}_{\pm}(z):=\lim_{\epsilon\downarrow 0}{\bf X}(z\pm\im\epsilon)$ on $\mathbb{R}\ni z$ satisfy the jump condition
	\begin{equation}\label{e47}
		{\bf X}_+(z)={\bf X}_-(z)\left\{\mathbb{I}-2\pi\im\,\Phi(z)\begin{bmatrix}\textnormal{Ai}(z\sigma+t)\textnormal{Ai}'(z\sigma+t) & -(\textnormal{Ai}(z\sigma+t))^2\\ (\textnormal{Ai}'(z\sigma+t))^2& -\textnormal{Ai}(z\sigma+t)\textnormal{Ai}'(z\sigma+t)\end{bmatrix}\right\},\ \ z\in\mathbb{R}.
	\end{equation}
	with $\Phi(z)$ as in \eqref{e41}.
	\item[(3)] As $z\rightarrow\infty$,
	\begin{equation}\label{e48}
		{\bf X}(z)=\mathbb{I}+{\bf X}_1z^{-1}+{\bf X}_2z^{-2}+\mathcal{O}\big(z^{-3}\big),\ \ \ \ \ \ {\bf X}_k={\bf X}_k(t,\sigma)=\big[X_k^{mn}(t,\sigma)\big]_{m,n=1}^2.
	\end{equation}
\end{enumerate}
\end{problem}
The following result about the solvability of RHP \ref{master} and its relation to $F_{\sigma}(t)$ in \eqref{f4b} is standard.
\begin{prop}\label{prop6} The RHP \ref{master} is uniquely solvable for any $(t,\sigma)\in\mathbb{R}\times(0,\infty)$. Moreover, the kernel of $R_{t,\sigma}=(I-\overline{M}_{t,\sigma})^{-1}-I$ equals
\begin{equation}\label{e49}
	R_{t,\sigma}(x,y)=\frac{F_1(x)F_2(y)-F_2(x)F_1(y)}{x-y};\ \ \ \ \ \ F_k(z):=\big((I-\overline{M}_{t,\sigma})^{-1}f_k\big)(z),\ \ z\in\mathbb{R},
\end{equation}
where $F_k$ is given in terms of the solution of RHP \ref{master} by 
\begin{equation*}
	\begin{bmatrix} F_1(z)\\ F_2(z)\end{bmatrix}={\bf X}_{\pm}(z)\begin{bmatrix}f_1(z)\\ f_2(z)\end{bmatrix},\ \ \ z\in\mathbb{R},
\end{equation*}
independently of the choice of limiting values on $\mathbb{R}$. Conversely, the solution of RHP \ref{master} is expressible in terms of the $F_k$ through the Cauchy integral
\begin{equation}\label{e51}
	{\bf X}(z)=\mathbb{I}-\int_{-\infty}^{\infty}\begin{bmatrix}F_1(\lambda)f_2(\lambda) & -F_1(\lambda)f_1(\lambda)\\ F_2(\lambda)f_2(\lambda) & -F_2(\lambda)f_1(\lambda)\end{bmatrix}\frac{\d\lambda}{\lambda-z},\ \ \ z\in\mathbb{C}\setminus\mathbb{R},
\end{equation}
and we have the differential identities
\begin{equation}\label{e52}
	\frac{\partial}{\partial t}\ln F_{\sigma}(t)=-\sigma X_1^{12}(t,\sigma),\ \ \ \ \frac{\partial}{\partial\sigma}\ln F_{\sigma}(t)=-X_1^{21}(t,\sigma)-tX_1^{12}(t,\sigma)-2\sigma X_2^{12}(t,\sigma),
\end{equation}
followed by
\begin{equation}\label{e52aa}
	\frac{\d}{\d\alpha}\ln F_{\alpha}(\alpha)=-X_1^{21}(\alpha,\alpha)-2\alpha X_1^{12}(\alpha,\alpha)-2\alpha X_2^{12}(\alpha,\alpha),
\end{equation}
for any $(t,\sigma,\alpha)\in\mathbb{R}\times(0,\infty)\times(0,\infty)$ in terms of the matrix entries $X_k^{mn}(t,\sigma)$ in \eqref{e48}.
\end{prop}
\begin{proof} Recall that $I-\overline{M}_{t,\sigma}$ is invertible on $L^2(\mathbb{R})$ for all $(t,\sigma)\in\mathbb{R}\times(0,\infty)$ by \eqref{f4b},\eqref{e37} and Lemma \ref{lem2}. Thus, \cite{IIKS} yields the unique solvability of RHP \ref{master} as well as the formul\ae\,\eqref{e49} and \eqref{e51}, specialized to the case of our self-adjoint operator $\overline{M}_{t,\sigma}$. Moving along, the $t$-identity in \eqref{e52} is standard: $\mathbb{R}\ni t\mapsto\frac{\partial}{\partial t}\overline{M}_{t,\sigma}$ is trace class on $L^2(\mathbb{R})$ with continuous kernel, 
\begin{equation}\label{e52a}
	\frac{\partial}{\partial t}\overline{M}_{t,\sigma}(x,y)=-\sigma f_1(x)f_1(y),
\end{equation}
so
\begin{align*}
	\frac{\partial}{\partial t}\ln F_{\sigma}(t)\stackrel{\eqref{f4b}}{=}-\tr_{L^2(\mathbb{R})}&\,\left((I-\overline{M}_{t,\sigma})^{-1}\frac{\partial\overline{M}_{t,\sigma}}{\partial t}\right)\stackrel{\eqref{e52a}}{=}\sigma\tr_{L^2(\mathbb{R})}\left((I-\overline{M}_{t,\sigma})^{-1}f_1\otimes f_1\right)\\
	&\stackrel{\eqref{e49}}=\sigma\int_{-\infty}^{\infty}F_1(\lambda)f_1(\lambda)\,\d\lambda\stackrel[\eqref{e51}]{\eqref{e48}}{=}-\sigma X_1^{12}(t,\sigma).
\end{align*}
The $\sigma$-identity in \eqref{e52} is slightly more involved: $\mathbb{R}_+\ni\sigma\mapsto\frac{\partial}{\partial\sigma}\overline{M}_{t,\sigma}$ is trace class on $L^2(\mathbb{R})$ with continuous kernel,
\begin{equation}\label{e52b}
	\frac{\partial}{\partial\sigma}\overline{M}_{t,\sigma}(x,y)=f_2(x)f_2(y)-tf_1(x)f_1(y)-\sigma f_1(x)(x+y)f_1(y),
\end{equation}
so
\begin{align*}
	\frac{\partial}{\partial\sigma}\ln F_{\sigma}(t)\stackrel[\eqref{e52b}]{\eqref{f4b}}{=}-\tr_{L^2(\mathbb{R})}&\,\left((I-\overline{M}_{t,\sigma})^{-1}f_2\otimes f_2\right)+t\tr_{L^2(\mathbb{R})}\big((I-\overline{M}_{t,\sigma})^{-1}f_1\otimes f_1\big)\\
	&\hspace{0.5cm}+\sigma\int_{-\infty}^{\infty}\int_{-\infty}^{\infty}(I-\overline{M}_{t,\sigma})^{-1}(\lambda,\mu)f_1(\mu)(\mu+\lambda)f_1(\lambda)\,\d \mu\,\d \lambda,
\end{align*}
and hence by the self-adjointness of $\overline{N}_{t,\sigma}$,
\begin{align*}
	\frac{\partial}{\partial\sigma}\ln F_{\sigma}(t)
	&\stackrel{\eqref{e49}}=-\int_{-\infty}^{\infty}F_2(\lambda)f_2(\lambda)\,\d\lambda+t\int_{-\infty}^{\infty}F_1(\lambda)f_1(\lambda)\,\d\lambda+2\sigma\int_{-\infty}^{\infty}F_1(\lambda)f_1(\lambda)\lambda\,\d\lambda\\
	&\hspace{1.975cm}\stackrel[\eqref{e51}]{\eqref{e48}}{=}-X_1^{21}(t,\sigma)-tX_1^{12}(t,\sigma)-2\sigma X_2^{12}(t,\sigma).
\end{align*}
The derivation of the third identity for $\sigma=t=\alpha$ is similar to the one for the second one and relies on the fact that $\mathbb{R}_+\ni\alpha\mapsto\frac{\d}{\d\alpha}\overline{M}_{\alpha,\alpha}$ is trace class on $L^2(\mathbb{R})$ with continuous kernel
\begin{equation*}
	\frac{\d}{\d\alpha}\overline{M}_{\alpha,\alpha}(x,y)=f_2(x)f_2(y)-2\alpha f_1(x)f_1(y)-\alpha f_1(x)(x+y)f_1(y).
\end{equation*}
This completes our proof of the Proposition.
\end{proof}
With \eqref{e52} in mind, we now investigate RHP \ref{master} asymptotically as $t\rightarrow+\infty$.
\subsection{Asymptotics, part $1$ ($t\rightarrow+\infty$ and $\sigma>0$ such that $t/\sigma\rightarrow+\infty$)} Consider the simple rescaling transformation,
\begin{equation}\label{e53}
	{\bf Y}(z;t,\sigma):={\bf X}\left(\frac{zt}{\sigma};t,\sigma\right),\ \ \ \ z\in\mathbb{C}\setminus\mathbb{R},
\end{equation}
and note that ${\bf Y}(z)$ solves a RHP very similar to RHP \ref{master}, but with rescaled jump matrix on $\mathbb{R}$ and rescaled normalization at $z=\infty$:
\begin{problem}\label{trafo1} Let $(t,\sigma)\in(0,\infty)\times(0,\infty)$. The function ${\bf Y}(z)={\bf Y}(z;t,\sigma)\in\mathbb{C}^{2\times 2}$ defined in \eqref{e53} is uniquely determined by the following properties:
\begin{enumerate}
	\item[(1)] ${\bf Y}(z)$ is analytic for $z\in\mathbb{C}\setminus\mathbb{R}$ and extends continuously to the closed upper and lower half-plane.
	\item[(2)] The continuous limiting values ${\bf Y}_{\pm}(z):=\lim_{\epsilon\downarrow 0}{\bf Y}(z\pm\im\epsilon)$ on $\mathbb{R}\ni z$ satisfy the jump condition ${\bf Y}_+(z)={\bf Y}_-(z){\bf G}_{\bf Y}(z;t,\sigma)$ where the jump matrix ${\bf G}_{\bf Y}(z;t,\sigma)$ is given by
	\begin{equation*}
		{\bf G}_{\bf Y}(z;t,\sigma)=\mathbb{I}-2\pi\im\Phi\left(\frac{zt}{\sigma}\right)\begin{bmatrix}\textnormal{Ai}(t(z+1))\textnormal{Ai}'(t(z+1)) & -(\textnormal{Ai}(t(z+1)))^2\\ (\textnormal{Ai}'(t(z+1)))^2 & -\textnormal{Ai}(t(z+1))\textnormal{Ai}'(t(z+1))\end{bmatrix},\ \ z\in\mathbb{R}.
	\end{equation*}
	\item[(3)] As $z\rightarrow\infty$,
	\begin{equation}\label{e53a}
		{\bf Y}(z)=\mathbb{I}+{\bf Y}_1z^{-1}+{\bf Y}_2z^{-2}+\mathcal{O}\big(z^{-3}\big),\ \ \ \ {\bf Y}_k={\bf Y}_k(t,\sigma)=\left(\frac{\sigma}{t}\right)^k{\bf X}_k(t,\sigma)
	\end{equation}
\end{enumerate}
\end{problem}
By the known analytic and asymptotic properties of the Airy and complementary error function as well as Lemma \ref{lem4}, the jump matrix in RHP \ref{trafo1} satisfies the following norm estimates.
\begin{prop}\label{prop7} There exists $c>0$ such that
\begin{equation*}
	\|{\bf G}_{\bf Y}(\cdot;t,\sigma)-\mathbb{I}\|_{L^{\infty}[0,\infty)}\leq c\sqrt{t}\,\e^{-\frac{4}{3}t^{\frac{3}{2}}},\ \ \ \ \|{\bf G}_{\bf Y}(\cdot;t,\sigma)-\mathbb{I}\|_{L^{\infty}(-\infty,-1]}\leq c\max\left\{1,\frac{\sigma}{\sqrt{t}}\right\}\,\e^{-(t/\sigma)^2},
\end{equation*}
\begin{equation*}
	\|{\bf G}_{\bf Y}(\cdot;t,\sigma)-\mathbb{I}\|_{L^{\infty}[-1,0]}\leq c\sqrt{t}\,\e^{-(t/\sigma)^2},
\end{equation*}
for all $\sigma>0$ and $t\geq 1$. In addition, there exists $c>0$ such that
\begin{equation*}
	\|{\bf G}_{\bf Y}(\cdot;t,\sigma)-\mathbb{I}\|_{L^2[0,\infty)}\leq c\,t^{-\frac{1}{4}}\e^{-\frac{4}{3}t^{\frac{3}{2}}},\ \ \ \ \|{\bf G}_{\bf Y}(\cdot;t,\sigma)-\mathbb{I}\|_{L^2(-\infty,-1]}\leq c\,\sigma t^{-1}\max\left\{1,\frac{\sigma}{\sqrt{t}}\right\}\e^{-(t/\sigma)^2},
\end{equation*}
\begin{equation*}
	\|{\bf G}_{\bf Y}(\cdot;t,\sigma)-\mathbb{I}\|_{L^{2}[-1,0]}\leq c\sqrt{t}\,\e^{-(t/\sigma)^2},
\end{equation*}
for all $\sigma,t>0$.
%
%
%
\end{prop}
\begin{proof} We only use that for all $x>0$ with some numerical constant $c>0$, cf. \cite[Chapter $9$]{NIST},
\begin{equation*}
	|\textnormal{Ai}(x)|\leq c\,x^{-\frac{1}{4}}\e^{-\frac{2}{3}x^{\frac{3}{2}}},\ \ \ |\textnormal{Ai}'(x)|\leq c\,x^{\frac{1}{4}}\e^{-\frac{2}{3}x^{\frac{3}{2}}},\ \ \ |\textnormal{Ai}(-x)|\leq c\,|x|^{-\frac{1}{4}},\ \ \ |\textnormal{Ai}'(-x)|\leq c\,|x|^{\frac{1}{4}}.
\end{equation*}
These combined with \eqref{e42} and the fact that $\Phi(\frac{zt}{\sigma})\exp[-\frac{4}{3}t^{3/2}(z+1)^{3/2}]$ is decreasing in $z\in[-1,0]$ yield at once the above six inequalities.
\end{proof}
By Proposition \ref{prop7}, as soon as $t\rightarrow+\infty$ and $\sigma>0$ is such that $t/\sigma\rightarrow+\infty$, the jump matrix in the ${\bf Y}$-RHP is exponentially close to the identity matrix on all of $\mathbb{R}$. 
Consequently, the general theory in \cite{DZ} ensures asymptotic solvability of the same problem, a fact which is made precise in the following result.
\begin{theo}\label{theo3} There exists $t_0>0$ such that the RHP for ${\bf Y}(z)$ defined in \eqref{e53} is uniquely solvable in $L^2(\mathbb{R})$ for all $t\geq t_0$ and all $\sigma>0$ such that $t\sigma^{-1}\geq t_0$. We can compute the solution of the same problem iteratively via the integral equation
\begin{equation}\label{e54}
	{\bf Y}(z;t,\sigma)=\mathbb{I}+\frac{1}{2\pi\im}\int_{-\infty}^{\infty}{\bf Y}_-(\lambda;t,\sigma)\big({\bf G}_{\bf Y}(\lambda;t,\sigma)-\mathbb{I}\big)\frac{\d\lambda}{\lambda-z},\ \ \ z\in\mathbb{C}\setminus\mathbb{R}
\end{equation}
using that, for all $t\geq t_0$ and $\sigma>0$ such that $t\sigma^{-1}\geq t_0$, with $c>0$,
\begin{equation}\label{e55}
	\|{\bf Y}_-(\cdot;t,\sigma)-\mathbb{I}\|_{L^2(\mathbb{R})}\leq c\max\left\{t^{-\frac{1}{4}}\e^{-\frac{4}{3}t^{\frac{3}{2}}},\sigma t^{-1}\max\left\{1,\frac{\sigma}{\sqrt{t}}\right\}\e^{-(t/\sigma)^2},\sqrt{t}\,\e^{-(t/\sigma)^2}\right\}.
\end{equation}
\end{theo}
Equipped with Theorem \ref{theo3}, we arrive at the below proof of Theorem \ref{fimain2}.
\begin{proof}[Proof of Theorem \ref{fimain2}] By \eqref{e52},\eqref{e53a},\eqref{e54},\eqref{e55}, Cauchy-Schwarz inequality and the explicit form of ${\bf G}_{\bf Y}(z;t,\sigma)$,
\begin{equation}\label{e56}
	\frac{\partial}{\partial t}\ln F_{\sigma}(t)=-\frac{\im t}{2\pi}\int_{-\infty}^{\infty}\big[{\bf Y}_-(\lambda;t,\sigma)\big({\bf G}_{\bf Y}(\lambda;t,\sigma)-\mathbb{I}\big)\big]^{12}\,\d\lambda=\int_{-\infty}^{\infty}\Phi\left(\frac{\lambda}{\sigma}\right)\textnormal{Ai}^2(\lambda+t)\,\d\lambda+r(t,\sigma),
\end{equation}
where the error term $r(t,\sigma)$ satisfies
\begin{equation*}
	\big|r(t,\sigma)\big|\leq t\big(\textnormal{RHS of}\,\eqref{e55}\big)^2.
\end{equation*}
Note that, see \eqref{e31},
\begin{equation*}
	\int_{-\infty}^{\infty}\Phi\left(\frac{y}{\sigma}\right)\textnormal{Ai}^2(y+t)\,\d y=N_{t,\sigma}(0,0)\stackrel{\eqref{f2}}{=}\frac{\e^{\im\frac{3\pi}{4}}}{4\pi}\sqrt{\frac{\sigma}{\pi}}\int_{\Gamma}\exp\left[\im\left(\frac{\lambda^3}{12\sigma^3}+\frac{\lambda t}{\sigma}\right)-\frac{\lambda^2}{4}\right]\lambda^{-\frac{3}{2}}\,\d\lambda,
\end{equation*}
so after changing variables according to the rule $y=\sigma\sqrt{t}\,\mu$,
\begin{equation*}
	\int_{-\infty}^{\infty}\Phi\left(\frac{y}{\sigma}\right)\textnormal{Ai}^2(y+t)\,\d y=\frac{\e^{\im\frac{3\pi}{4}}}{4\pi}\frac{t^{-\frac{1}{4}}}{\sqrt{\pi}}\int_{\Gamma}\e^{-t^{\frac{3}{2}}h(\mu,\tau)}\mu^{-\frac{3}{2}}\,\d\mu;\ \ \ \ \ h(\mu,\tau):=-\im\left(\frac{\mu^3}{12}+\mu\right)+\frac{\tau}{4}\mu^2,
\end{equation*}
where we abbreviate 
\begin{equation*}
	\tau:=\sigma^2/\sqrt{t}.
\end{equation*}
For $\epsilon\in(0,\frac{1}{4}]$, we have that $0<\tau\leq t^{2(\epsilon-\frac{1}{4})}$ is bounded as $t\rightarrow+\infty$, so we can asymptotically evaluate the integral in question by the steepest descent method. In detail, by \cite[Chapter $4$, Theorem $7.1$]{Olv},
\begin{equation*}
	\int_{-\infty}^{\infty}\Phi\left(\frac{y}{\sigma}\right)\textnormal{Ai}^2(y+t)\,\d y=\frac{1}{2\pi t}\frac{(-\tau+\sqrt{4+\tau^2})^{-\frac{3}{2}}}{(4+\tau^2)^{\frac{1}{4}}}\e^{-t^{\frac{3}{2}}h(\mu_{\ast},\tau)}\Big(1+\mathcal{O}\big(t^{-\frac{3}{2}}\big)\Big),\ \ t\rightarrow+\infty,\ \ 0<\sigma\leq t^{\epsilon},
\end{equation*}
for any $\epsilon\in(0,\frac{1}{4}]$ where $\mu_{\ast}:=-\im\tau+\im\sqrt{4+\tau^2}$ is the unique stationary point of $\mu\mapsto h(\mu,\tau)$ in the upper half-plane and which is bounded away from zero. Seeing that by \eqref{fi15},
\begin{equation*}
	t^{\frac{3}{2}}h(\mu_{\ast},\tau)=B(t,\sigma),
\end{equation*}
estimate \eqref{e56} then yields in the same asymptotic regime $t\rightarrow+\infty,0<\sigma\leq t^{\epsilon},\epsilon\in(0,\frac{1}{4}]$,
\begin{equation*}
	\frac{\partial}{\partial t}\ln F_{\sigma}(t)=\frac{1}{2\pi t}\frac{(-\sigma^2t^{-\frac{1}{2}}+\sqrt{4+\sigma^4t^{-1}})^{-\frac{3}{2}}}{(4+\sigma^4t^{-1})^{\frac{1}{4}}}\e^{-B(t,\sigma)}\Big(1+\mathcal{O}\big(t^{-\frac{3}{2}}\big)\Big)+\mathcal{O}\big(\max\Big\{t^{-\frac{1}{2}}\e^{-\frac{8}{3}t^{\frac{3}{2}}},t\e^{-2(t/\sigma)^2}\Big\}\big).
\end{equation*}
However, since $h(\mu_{\ast},\tau)<\frac{8}{3}$ and $h(\mu_{\ast},\tau)<\frac{2}{\tau}$ for all $\tau\geq 0$, we have thus for all $\epsilon\in(0,\frac{1}{4}]$,
\begin{equation}\label{e56a}
	\frac{\partial}{\partial t}\ln F_{\sigma}(t)=\frac{1}{2\pi t}\frac{(-\sigma^2t^{-\frac{1}{2}}+\sqrt{4+\sigma^4t^{-1}})^{-\frac{3}{2}}}{(4+\sigma^4t^{-1})^{\frac{1}{4}}}\e^{-B(t,\sigma)}\Big(1+\mathcal{O}\big(t^{-\frac{3}{2}}\big)\Big),\ \ \ t\rightarrow+\infty,\ \ 0<\sigma\leq t^{\epsilon}.
\end{equation}
Moving ahead, for $\epsilon\in(\frac{1}{4},1)$ we return to the contour integral formula for $N_{t,\sigma}(0,0)$, but no longer rescale. Instead we write
\begin{equation*}
	\int_{-\infty}^{\infty}\Phi\left(\frac{y}{\sigma}\right)\textnormal{Ai}^2(y+t)\,\d y=\frac{\e^{\im\frac{3\pi}{4}}}{4\pi}\sqrt{\frac{\sigma}{\pi}}\int_{\Gamma}\e^{-(t/\sigma)k(y,t,\sigma)}y^{-\frac{3}{2}}\,\d y;\ \ \ k(y,t,\sigma):=-\im\left(y+\frac{y^3}{12\sigma^2t}\right)+\frac{y^2\sigma}{4t},
\end{equation*}
and note that the unique stationary point $y_{\ast}:=-\im\sigma^3+\im\sigma^3\sqrt{1+4\tau^{-2}}$ of $y\mapsto k(y,t,\sigma)$ in the upper half-plane is bounded away from zero. Thus, by the steepest descent method,
\begin{equation*}
	\int_{-\infty}^{\infty}\Phi\left(\frac{y}{\sigma}\right)\textnormal{Ai}^2(y+t)\,\d y=\frac{1}{2\pi}\frac{(-1+\sqrt{1+4\tau^{-2}})^{-\frac{3}{2}}}{\sigma^4(1+4\tau^{-2})^{\frac{1}{4}}}\e^{-(t/\sigma)k(y_{\ast},t,\sigma)}\Big(1+\mathcal{O}\big(\sigma t^{-1}\big)\Big),\ t\rightarrow+\infty,\ t^{\frac{1}{4}}\leq\sigma\leq t^{\epsilon},
\end{equation*}
for any $\epsilon\in(\frac{1}{4},1)$. But since also by \eqref{fi15},
\begin{equation*}
	(t/\sigma)k(y_{\ast},t,\sigma)=B(t,\sigma)=t^{\frac{3}{2}}h(\mu_{\ast},\tau),
\end{equation*}
estimate \eqref{e56} yields by the same reasoning as above,
\begin{equation}\label{e56b}
	\frac{\partial}{\partial t}\ln F_{\sigma}(t)=\frac{1}{2\pi\sigma^4}\frac{(-1+\sqrt{1+4t\sigma^{-4}})^{-\frac{3}{2}}}{(1+4t\sigma^{-4})^{\frac{1}{4}}}\e^{-B(t,\sigma)}\Big(1+\mathcal{O}\big(\sigma t^{-1}\big)\Big),\ \ t\rightarrow+\infty,\ t^{\frac{1}{4}}\leq\sigma\leq t^{\epsilon},
\end{equation}
for any $\epsilon\in(\frac{1}{4},1)$. Observe that the leading orders in the right hand sides of \eqref{e56a} and \eqref{e56b} match, so we are left to integrate both expansions with respect to $t$. Indeed, since $F_{\sigma}(\infty)=1$ for all $\sigma>0$, we have for all $t\geq t_0$,
\begin{equation*}
	\ln F_{\sigma}(t)=-\int_t^{\infty}\frac{\partial}{\partial s}\ln F_{\sigma}(s)\,\d s,
\end{equation*}
and so as $t\rightarrow+\infty$ and $0<\sigma\leq t^{\epsilon}$ with $\epsilon\in(0,\frac{1}{4}]$, using the representation for $B(t,\sigma)$ in \eqref{fi15},
\begin{align*}
	\ln F_{\sigma}(t)&\,\stackrel{\eqref{e56a}}{=}-\frac{1}{2\pi}\int_t^{\infty}\frac{(-\sigma^2s^{-\frac{1}{2}}+\sqrt{4+\sigma^4s^{-1}})^{-\frac{3}{2}}}{s(4+\sigma^4s^{-1})^{\frac{1}{4}}}\e^{-B(s,\sigma)}\Big(1+\mathcal{O}\big(s^{-\frac{3}{2}}\big)\Big)\,\d s\\
	=&\,-\frac{1}{2\pi t^{\frac{3}{2}}}\frac{(-\sigma^2t^{-\frac{1}{2}}+\sqrt{4+\sigma^4t^{-1}})^{-\frac{5}{2}}}{(4+\sigma^4t^{-1})^{\frac{1}{4}}}\e^{-B(t,\sigma)}\Big(1+\mathcal{O}\big(t^{-\frac{3}{2}}\big)\Big)=-A(t,\sigma)\e^{-B(t,\sigma)}\Big(1+\mathcal{O}\big(t^{-\frac{3}{2}}\big)\Big).
\end{align*}
On the other hand, as $t\rightarrow+\infty$ and $t^{\frac{1}{4}}\leq\sigma\leq t^{\epsilon}$ with $\epsilon\in(\frac{1}{4},1)$, using the alternative representation for $B(t,\sigma)$ in \eqref{fi17a},
\begin{align*}
	\ln F_{\sigma}(t)&\,\stackrel{\eqref{e56b}}{=}-\frac{1}{2\pi\sigma^4}\int_t^{\infty}\frac{(-1+\sqrt{1+4s\sigma^{-4}})^{-\frac{3}{2}}}{(1+4s\sigma^{-4})^{\frac{1}{4}}}\e^{-B(s,\sigma)}\Big(1+\mathcal{O}\big(\sigma s^{-1}\big)\Big)\,\d s\\
	=&\,-\frac{1}{2\pi\sigma^6}\frac{(-1+\sqrt{1+4t\sigma^{-4}})^{-\frac{5}{2}}}{(1+4t\sigma^{-4})^{\frac{1}{4}}}\e^{-B(t,\sigma)}\Big(1+\mathcal{O}\big(\sigma t^{-1}\big)\Big)=-A(t,\sigma)\e^{-B(t,\sigma)}\Big(1+\mathcal{O}\big(\sigma t^{-1}\big)\Big).
\end{align*}
Hence, seeing that $A(t,\sigma)\e^{-B(t,\sigma)}\rightarrow 0$ super exponentially fast in the current asymptotic regimes, expansion \eqref{fi14} now follows from exponentiation. Our proof of Theorem \ref{fimain2} is complete.
\end{proof}
\subsection{Asymptotics, part 2 ($t\rightarrow+\infty$ and $\sigma\rightarrow+\infty$)} Consider the transformation,
\begin{equation}\label{e58}
	{\bf Y}(z;t,\sigma,z_0):={\bf X}\left(\frac{z-t}{\sigma};t,\sigma\right){\bf M}^{\textnormal{Ai}}\left(z;z_0\sigma,\frac{\pi}{4}\right),\ \ \ \ z\in\mathbb{C}\setminus\Sigma_{\bf Y},
\end{equation}
with the Airy parametrix ${\bf M}^{\textnormal{Ai}}$ in \eqref{app2}, the contour $\Sigma_{\bf Y}$ as in RHP \ref{trafo2} and a yet to be determined real parameter $z_0$. Using RHP \ref{AiryRHP}, our initial RHP \ref{master} is transformed to the following one:
\begin{problem}\label{trafo2} Let $(t,\sigma,z_0)\in(0,\infty)\times(0,\infty)\times\mathbb{R}$. The function ${\bf Y}(z)={\bf Y}(z;t,\sigma,z_0)\in\mathbb{C}^{2\times 2}$ defined in \eqref{e58} is uniquely determined by the following properties:
\begin{enumerate}
	\item[(1)] ${\bf Y}(z)$ is analytic for $z\in\mathbb{C}\setminus\Sigma_{\bf Y}$ where $\Sigma_{\bf Y}:=\bigcup_{j=1}^4\Gamma_j^{\sigma}\cup\{z_0\sigma\}$ with
	\begin{equation*}
		\Gamma_1^{\sigma}:=(z_0\sigma,\infty),\ \ \ \ \Gamma_3^{\sigma}:=(-\infty,z_0\sigma),\ \ \ \ \Gamma_2^{\sigma}:=\e^{-\im\frac{\pi}{4}}(-\infty,z_0\sigma),\ \ \ \ \Gamma_4^{\sigma}:=\e^{\im\frac{\pi}{4}}(-\infty,z_0\sigma),
	\end{equation*}
	denotes the oriented contour shown in Figure \ref{figuretrafo2}. Moreover, on each connected component of $\mathbb{C}\setminus\Sigma_{\bf Y}$ there is a continuous extension of ${\bf Y}(z)$ to the closure of the same component.
	%
	\item[(2)] The continuous limiting values ${\bf Y}_{\pm}(z)$ on $\Gamma_j^{\sigma}\ni z$ satisfy the jump condition ${\bf Y}_+(z)={\bf Y}_-(z){\bf G}_{\bf Y}(z;t,\sigma,z_0)$ where the jump matrix ${\bf G}_{\bf Y}(z;t,\sigma,z_0)$ is piecewise given by
	\begin{equation*}
		{\bf G}_{\bf Y}(z;t,\sigma,z_0)=\begin{bmatrix}1 & 1-\Phi\big(\frac{z-t}{\sigma}\big)\\ 0 & 1\end{bmatrix},\ \ z\in\Gamma_1^{\sigma};\ \ \ \ {\bf G}_{\bf Y}(z;t,\sigma,z_0)=\begin{bmatrix}1 & 0\\ 1 & 1\end{bmatrix},\ \ z\in\Gamma_2^{\sigma}\cup\Gamma_4^{\sigma};
	\end{equation*}
	and
	\begin{equation*}
		{\bf G}_{\bf Y}(z;t,\sigma,z_0)=\begin{bmatrix}\Phi\big(\frac{z-t}{\sigma}\big) & 1-\Phi\big(\frac{z-t}{\sigma}\big)\smallskip\\ -1-\Phi\big(\frac{z-t}{\sigma}\big) & \Phi\big(\frac{z-t}{\sigma}\big)\end{bmatrix},\ \ z\in\Gamma_3^{\sigma}.
	\end{equation*}
	\begin{figure}[tbh]
	\begin{tikzpicture}[xscale=0.65,yscale=0.65]
	\draw [thick, color=red, decoration={markings, mark=at position 0.25 with {\arrow{>}}}, decoration={markings, mark=at position 0.75 with {\arrow{>}}}, postaction={decorate}] (-5,0) -- (5,0);
\node [below] at (0.75,-0.2) {{\small $z=z_0\sigma$}};
\node [right] at (4.5,0.6) {{\small $\Gamma_1^{\sigma}$}};
\node [left] at (-4.5,0.6) {{\small $\Gamma_3^{\sigma}$}};
\draw [thick, color=red, decoration={markings, mark=at position 0.5 with {\arrow{>}}}, postaction={decorate}] (-4,4) -- (0,0);
\draw [thick, color=red, decoration={markings, mark=at position 0.5 with {\arrow{>}}}, postaction={decorate}] (-4,-4) -- (0,0);
\node [right] at (-3.6,4.4) {{\small $\Gamma_2^{\sigma}$}};
\node [right] at (-3.6,-4.4) {{\small $\Gamma_4^{\sigma}$}};
\end{tikzpicture}
\caption{The oriented jump contours for ${\bf Y}(z)$ in the complex $z$-plane.}
\label{figuretrafo2}
\end{figure}
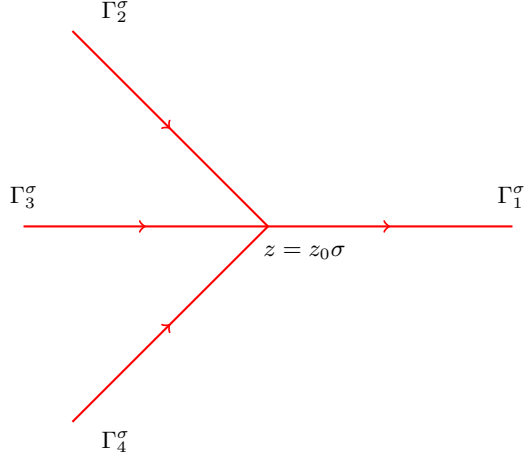
	\item[(3)] ${\bf Y}(z)$ is bounded in a neighbourhood of $z=z_0\sigma$.
	\item[(4)] As $z\rightarrow\infty$, valid in a full neighborhood of infinity off the jump contours,
	\begin{equation}\label{e59}
		{\bf Y}(z)=\Big\{\mathbb{I}+{\bf Y}_1z^{-1}+{\bf Y}_2z^{-2}+\mathcal{O}\big(z^{-3}\big)\Big\}z^{-\frac{1}{4}\sigma_3}\frac{1}{\sqrt{2}}\begin{bmatrix}1 & 1\\ -1 & 1\end{bmatrix}\e^{-\im\frac{\pi}{4}\sigma_3}\e^{-\frac{2}{3}z^{\frac{3}{2}}\sigma_3},
	\end{equation}
	where we choose principal branches for all fractional exponents and the coefficients ${\bf Y}_1,{\bf Y}_2$ equal
	\begin{align*}
		{\bf Y}_1={\bf Y}_1(t,\sigma)=&\,\,\sigma{\bf X}_1(t,\sigma)-\frac{7}{48}\begin{bmatrix}0 & 0\\ 1 & 0\end{bmatrix},\\ 
		{\bf Y}_2={\bf Y}_2(t,\sigma)=&\,\,\sigma^2{\bf X}_2(t,\sigma)+t\sigma{\bf X}_1(t,\sigma)
		-\frac{7\sigma}{48}{\bf X}_1(t,\sigma)\begin{bmatrix}0 & 0\\ 1 & 0\end{bmatrix}+\frac{5}{48}\begin{bmatrix}0 & 1\\ 0 & 0\end{bmatrix}.
	\end{align*}
\end{enumerate}
\end{problem}
Our next move will simplify the jump on $\Gamma_3$ in RHP \ref{trafo2}, while modifying the one on $\Gamma_2\cup\Gamma_4$. In detail, we make use of the matrix factorization
\begin{equation*}
	\begin{bmatrix}\Phi(\zeta) & 1-\Phi(\zeta)\\
	-1-\Phi(\zeta) & \Phi(\zeta)\end{bmatrix}=\underbrace{\begin{bmatrix}1 & 0\\ \Phi(\zeta)(1-\Phi(\zeta))^{-1} & 1\end{bmatrix}}_{=:{\bf M}(\zeta)}\begin{bmatrix}0 & 1-\Phi(\zeta)\\ -(1-\Phi(\zeta))^{-1} & 0\end{bmatrix}\begin{bmatrix}1 & 0\\ \Phi(\zeta)(1-\Phi(\zeta))^{-1} & 1\end{bmatrix},
\end{equation*}
valid for any $\zeta\in\mathbb{R}$ given that $\Phi:\mathbb{R}\rightarrow(0,1)$. In addition, since $1-\Phi(\zeta)=\frac{1}{2}\textnormal{erfc}(\zeta)$, the same factorization is also valid in the closed sector $\textnormal{arg}\,\zeta\in[\frac{3\pi}{4},\frac{5\pi}{4}]$ which does not contain any of the zeros of the complementary error function, see \cite[$\S 7.13$(ii)]{NIST}. Thus, provided we choose 
\begin{equation}\label{rej2}
	\omega:=\frac{t}{\sigma}\geq z_0
\end{equation}
later on, the following transformation is well-defined,
\begin{equation}\label{e60}
	{\bf T}(z;t,\sigma,z_0):={\bf Y}(z\sigma;t,\sigma,z_0)\begin{cases}\mathbb{I},&\textnormal{arg}\,(z-z_0)\in(-\frac{3\pi}{4},\frac{3\pi}{4})\smallskip\\
	{\bf M}^{-1}(z-\omega),&\textnormal{arg}\,(z-z_0)\in(\frac{3\pi}{4},\pi)\smallskip\\
	{\bf M}(z-\omega),&\textnormal{arg}\,(z-z_0)\in(-\pi,-\frac{3\pi}{4})
	\end{cases},\ \ \ \
%
%
%
%
%
\end{equation}
and it transforms RHP \ref{trafo2} to the problem below:
\begin{problem}\label{trafo3} Let $(t,\sigma,z_0)\in(0,\infty)\times(0,\infty)\times\mathbb{R}$ with $z_0\leq\frac{t}{\sigma}=:\omega$. The function ${\bf T}(z)={\bf T}(z;t,\sigma,z_0)\in\mathbb{C}^{2\times 2}$ defined in \eqref{e60} is uniquely determined by the following properties:
\begin{enumerate}
	\item ${\bf T}(z)$ is analytic for $z\in\mathbb{C}\setminus\Sigma_{\bf T}$ where $\Sigma_{\bf T}:=\bigcup_{j=1}^4\Gamma_j\cup\{z_0\}$ with
	\begin{equation*}
		\Gamma_1:=(z_0,\infty),\ \ \ \ \Gamma_3:=(-\infty,z_0),\ \ \ \ \Gamma_2:=\e^{-\im\frac{\pi}{4}}(-\infty,z_0),\ \ \ \ \Gamma_4:=\e^{\im\frac{\pi}{4}}(-\infty,z_0),
	\end{equation*}	
	denotes the oriented contour shown in Figure \ref{figuretrafo2a}. On each connected component of $\mathbb{C}\setminus\Sigma_{\bf T}$ there is a continuous extension of ${\bf T}(z)$ to the closure of the same component. 
		\begin{figure}[tbh]
	\begin{tikzpicture}[xscale=0.65,yscale=0.65]
	\draw [thick, color=red, decoration={markings, mark=at position 0.25 with {\arrow{>}}}, decoration={markings, mark=at position 0.75 with {\arrow{>}}}, postaction={decorate}] (-5,0) -- (5,0);
\node [below] at (0.75,-0.2) {{\small $z=z_0$}};
\node [right] at (4.5,0.6) {{\small $\Gamma_1$}};
\node [left] at (-4.5,0.6) {{\small $\Gamma_3$}};
\draw [thick, color=red, decoration={markings, mark=at position 0.5 with {\arrow{>}}}, postaction={decorate}] (-4,4) -- (0,0);
\draw [thick, color=red, decoration={markings, mark=at position 0.5 with {\arrow{>}}}, postaction={decorate}] (-4,-4) -- (0,0);
\node [right] at (-3.6,4.4) {{\small $\Gamma_2$}};
\node [right] at (-3.6,-4.4) {{\small $\Gamma_4$}};
\end{tikzpicture}
\caption{The oriented jump contours for ${\bf T}(z)$ in the complex $z$-plane.}
\label{figuretrafo2a}
\end{figure}
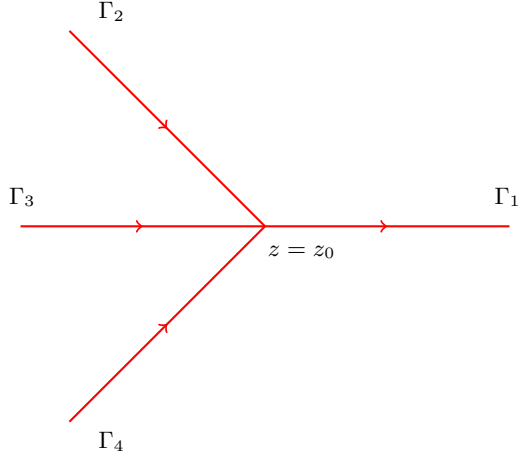
%
	\item The continuous limiting values ${\bf T}_{\pm}(z)$ on $\Gamma_j\ni z$ satisfy ${\bf T}_+(z)={\bf T}_-(z){\bf G}_{\bf T}(z;t,\sigma,z_0)$ where the jump matrix ${\bf G}_{\bf T}(z;t,\sigma,z_0)$ is of the form
	\begin{align*}
		{\bf G}_{\bf T}(z;t,\sigma,z_0)=&\,\begin{bmatrix}1 & 1-\Phi(z-\omega)\smallskip\\ 0 & 1\end{bmatrix},\ \ z\in\Gamma_1;\\
		{\bf G}_{\bf T}(z;t,\sigma,z_0)=&\,\begin{bmatrix}1 & 0\smallskip\\ \big(1-\Phi(z-\omega)\big)^{-1} & 1\end{bmatrix},\ \ z\in\Gamma_2\cup\Gamma_4;\\
		{\bf G}_{\bf T}(z;t,\sigma,z_0)=&\,\begin{bmatrix}0 & 1-\Phi(z-\omega)\smallskip\\ -\big(1-\Phi(z-\omega)\big)^{-1} & 0\end{bmatrix},\ \ z\in\Gamma_3.
	\end{align*}
	\item ${\bf T}(z)$ is bounded in a neighbourhood of $z=z_0$.
	\item As $z\rightarrow\infty$, valid in a full neighborhood of infinity off the jump contours,
	\begin{equation}\label{e61}
		{\bf T}(z)=\Big\{\mathbb{I}+{\bf T}_1z^{-1}+{\bf T}_2z^{-2}+\mathcal{O}\big(z^{-3}\big)\Big\}(z\sigma)^{-\frac{1}{4}\sigma_3}\frac{1}{\sqrt{2}}\begin{bmatrix}1 & 1\\ -1 & 1\end{bmatrix}\e^{-\im\frac{\pi}{4}\sigma_3}\e^{-\frac{2}{3}(z\sigma)^{\frac{3}{2}}\sigma_3},
	\end{equation}
	with principal branches throughout and the coefficients ${\bf T}_1,{\bf T}_2$ equal to, compare \eqref{e59},
	\begin{equation*}
		{\bf T}_1={\bf T}_1(t,\sigma)=\frac{1}{\sigma}{\bf Y}_1(t,\sigma),\ \ \ \ \ \ {\bf T}_2={\bf T}_2(t,\sigma)=\frac{1}{\sigma^2}{\bf Y}_2(t,\sigma).
	\end{equation*}
\end{enumerate}
\end{problem}
\begin{remark}\label{rem6} The analytic properties of ${\bf T}(z)$ listed in RHP \ref{trafo3} need no explanation given RHP \ref{trafo2} and the fact that $1-\Phi(\zeta)$ is zero-free in the sector $\textnormal{arg}\,\zeta\in[\frac{3\pi}{4},\frac{5\pi}{4}]$. As for the asymptotic normalization \eqref{e61}, notice that
\begin{equation*}
	\e^{-\frac{2}{3}(z\sigma)^{\frac{3}{2}}\sigma_3}{\bf M}^{\pm 1}(z-\omega)\e^{\frac{2}{3}(z\sigma)^{\frac{3}{2}}\sigma_3}\rightarrow\mathbb{I},
\end{equation*}
super-exponentially fast as $z\rightarrow\infty$ and $\textnormal{arg}\,(z-z_0)\in(-\pi,-\frac{3\pi}{4})\cup(\frac{3\pi}{4},\pi)$, for any $t,\sigma>0$ such that $z_0\leq\omega<+\infty$, cf. \cite[$\S 7.12$]{NIST}. Hence, modulo the rescaling $z\mapsto z\sigma$, the transformation \eqref{e60} affects the normalization \eqref{e59} not to leading orders.
\end{remark}
The off-diagonal entries in the jump matrix in RHP \ref{trafo3} on $\Gamma_3$, together with the normalization \eqref{e61}, motivate our next step, namely the $g$-function transformation. First, define the scalar-valued function $g:\mathbb{C}\setminus(-\infty,z_0]\rightarrow\mathbb{C}$ as the contour integral
\begin{equation}\label{e62}
	g(z;t,\sigma,z_0):=\int_{z_0}^z\rho(\zeta;t,\sigma,z_0)\,\d\zeta,
\end{equation}
with path of integration in $\mathbb{C}\setminus(-\infty,z_0]$ and where
\begin{equation}\label{e63}
	\rho(z;t,\sigma,z_0):=(z-z_0)^{\frac{1}{2}}\left[1+\frac{1}{2(\pi\sigma)^{\frac{3}{2}}}\int_{-\infty}^{z_0}\frac{\e^{-(\lambda-\omega)^2}}{1-\Phi(\lambda-\omega)}\frac{1}{\sqrt{z_0-\lambda}}\frac{\d\lambda}{\lambda-z}\right],\ \ \ z\in\mathbb{C}\setminus(-\infty,z_0],
\end{equation}
is defined with the principal branch for $(z-z_0)^{\frac{1}{2}}$ such that $(z-z_0)^{\frac{1}{2}}=\sqrt{z-z_0}>0$ when $\mathbb{R}\ni z>z_0$. Second, provided we choose $z_0=z_0(t,\sigma)$ as the unique real-valued solution to the equation
\begin{equation}\label{e64}
	z_0=-\frac{1}{(\pi\sigma)^{\frac{3}{2}}}\int_{-\infty}^{z_0}\frac{\e^{-(\lambda-\omega)^2}}{1-\Phi(\lambda-\omega)}\frac{\d\lambda}{\sqrt{z_0-\lambda}},
\end{equation}
so that in particular $(-\infty,0)\ni z_0<\omega=\frac{t}{\sigma}$ for all $t,\sigma>0$ which makes \eqref{e60} bona fide, the function $g(z)=g(z;t,\sigma,z_0)$ has the following analytic and asymptotic properties.
\begin{prop}\label{prop8} Let $t,\sigma>0$. Consider $g(z)=g(z;t,\sigma,z_0),z\in\mathbb{C}\setminus(-\infty,z_0]$ as defined in \eqref{e62},\eqref{e63} and \eqref{e64}. Then
\begin{enumerate}
	\item $g(z)$ is analytic for $z\in\mathbb{C}\setminus(-\infty,z_0]$ and attains continuous boundary values $g_{\pm}(z)=\lim_{\epsilon\downarrow 0}g(z\pm\im\epsilon)$ on $(-\infty,z_0)\ni z$.
	\item The boundary values $g_{\pm}(z)$ on $(-\infty,z_0)$ satisfy the constraint
	\begin{equation}\label{e65}
		g_+(z)+g_-(z)=\frac{1}{\sigma^{\frac{3}{2}}}\ln\big(1-\Phi(z-\omega)\big)-\frac{1}{\sigma^{\frac{3}{2}}}\ln\big(1-\Phi(z_0-\omega)\big),\ \ z\in(-\infty,z_0).
	\end{equation}
	\item As $z\rightarrow\infty$ with $z\notin(-\infty,z_0]$, for any $t,\sigma>0$,
	\begin{equation}\label{e66}
		g(z)=\frac{2}{3}z^{\frac{3}{2}}+\frac{\ell}{2}+\frac{g_1}{z^{\frac{1}{2}}}+\frac{g_2}{z^{\frac{3}{2}}}+\mathcal{O}\big(z^{-\frac{5}{2}}\big),
	\end{equation}
	choosing the principal branch for $z^{\alpha}:\mathbb{C}\setminus(-\infty,0]\rightarrow\mathbb{C}$ such that $z^{\alpha}>0$ if $z>0$. The $z$-independent coefficients $\ell=\ell(t,\sigma)$ and $g_j=g_j(t,\sigma)$ equal
	\begin{equation}\label{e66a}
		\ell=-\frac{1}{\sigma^{\frac{3}{2}}}\ln\big(1-\Phi(z_0-\omega)\big),\ \ \ g_1=-\frac{1}{(\pi\sigma)^{\frac{3}{2}}}\int_{-\infty}^{z_0}\frac{\e^{-(\lambda-\omega)^2}}{1-\Phi(\lambda-\omega)}\sqrt{z_0-\lambda}\,\d\lambda-\frac{z_0^2}{4},
	\end{equation}
	and
	\begin{equation}\label{e66aa}
		g_2=\frac{1}{3(\pi\sigma)^{\frac{3}{2}}}\int_{-\infty}^{z_0}\frac{\e^{-(\lambda-\omega)^2}}{1-\Phi(\lambda-\omega)}(z_0-\lambda)^{\frac{3}{2}}\,\d\lambda-\frac{z_0}{2(\pi\sigma)^{\frac{3}{2}}}\int_{-\infty}^{z_0}\frac{\e^{-(\lambda-\omega)^2}}{1-\Phi(\lambda-\omega)}\sqrt{z_0-\lambda}\,\d\lambda-\frac{z_0^3}{12}.
	\end{equation}
\end{enumerate}
\end{prop}
\begin{proof} Fix $t,\sigma>0$ throughout. We begin by observing that $(-\infty,0)\ni x\mapsto-\ln(1-\Phi(x))$ is convex, and so 
\begin{equation*}
	-\int_{-\infty}^{z_0}\frac{\e^{-(\lambda-\omega)^2}}{1-\Phi(\lambda-\omega)}\frac{\d\lambda}{\sqrt{z_0-\lambda}}=-\int_{-\infty}^0\frac{\e^{-(\mu-\omega+z_0)^2}}{1-\Phi(\mu-\omega+z_0)}\frac{\d\mu}{\sqrt{-\mu}}<0
\end{equation*}
is a decreasing continuous function in $z_0\in(-\infty,0)$, which vanishes at $z_0=-\infty$. On the other hand, the left hand side in \eqref{e64} is an increasing, continuous function that connects $-\infty$ to $+\infty$ as $z_0$ increases on $\mathbb{R}$. Thus \eqref{e64} indeed has a unique real-valued solution $z_0=z_0(t,\sigma)$, for any $t,\sigma>0$. Next, given the analytic and asymptotic properties of the complementary error function on $\mathbb{R}$, we note that
\begin{equation*}
	\Omega(z):=\int_{-\infty}^{z_0}\frac{\e^{-(\lambda-\omega)^2}}{1-\Phi(\lambda-\omega)}\frac{1}{\sqrt{z_0-\lambda}}\frac{\d\lambda}{\lambda-z}
\end{equation*}
is well-defined for any $z\in\mathbb{C}\setminus(-\infty,z_0]$. Moreover we have the representation, valid for $z\notin(-\infty,z_0]$,
\begin{equation}\label{e67}
	\Omega(z)=\int_{-\infty}^{z_0}\left[\frac{\e^{-(\lambda-\omega)^2}}{1-\Phi(\lambda-\omega)}-\frac{\e^{-(z_0-\omega)^2}}{1-\Phi(z_0-\omega)}\right]\frac{1}{\sqrt{z_0-\lambda}}\frac{\d\lambda}{\lambda-z}-\frac{\e^{-(z_0-\omega)^2}}{1-\Phi(z_0-\omega)}\pi(z-z_0)^{-\frac{1}{2}},
\end{equation}
%
%
where the first integral is an analytic function in $z$ off its integration path, since it is a Cauchy integral of a function which obeys a H\"older condition on the contour, and the second term is analytic save for a branch cut.
Consequently $\Omega:\mathbb{C}\setminus(-\infty,z_0]\rightarrow\mathbb{C}$ is analytic on $\mathbb{C}\setminus(-\infty,z_0]$ and attains continuous boundary values $\Omega_{\pm}(z)$ on $(-\infty,z_0)\ni z$ by \eqref{e67} and \cite[$\S 4.6,5.1$]{Ga}. Moreover, by the Plemelj-Sokhotski formula,
\begin{equation*}
	\Omega_+(z)-\Omega_-(z)=\frac{\e^{-(z-\omega)^2}}{1-\Phi(z-\omega)}\frac{2\pi\im}{\sqrt{z_0-z}},\ \ \ z\in(-\infty,z_0).
\end{equation*}
This reasoning verifies the analyticity of $\mathbb{C}\setminus(-\infty,z_0]\ni z\mapsto\rho(z)=\rho(z;t,\sigma,z_0)$, for any $t,\sigma>0$ and $z_0\in\mathbb{R}$, the existence of its continuous boundary values $\rho_{\pm}(z)=\lim_{\epsilon\downarrow 0}\rho(z\pm\im\epsilon)$ on $(-\infty,z_0)$,
\begin{equation*}
	\rho_{\pm}(z)=\pm\im\sqrt{z_0-z}\left[1+\frac{1}{2(\pi\sigma)^{\frac{3}{2}}}\Omega_{\pm}(z)\right],\ \ \ z\in(-\infty,z_0),
\end{equation*}
and the relation
\begin{equation}\label{e69}
	\rho_+(z)+\rho_-(z)=-\frac{1}{\sqrt{\pi}\sigma^{\frac{3}{2}}}\frac{\e^{-(z-\omega)^2}}{1-\Phi(z-\omega)}=\frac{1}{\sigma^{\frac{3}{2}}}\frac{\d}{\d z}\ln\big(1-\Phi\left(z-\omega\right)\big),\ \ z\in(-\infty,z_0).
\end{equation}
All together, the antiderivative of $z\mapsto\rho(z)$, i.e. $z\mapsto g(z)$ with integration path in $\mathbb{C}\setminus(-\infty,z_0]$, is analytic in $\mathbb{C}\setminus(-\infty,z_0]\ni z$ and has continuous boundary values $g_{\pm}(z)$ on $(-\infty,z_0)$ that satisfy
\begin{equation*}
	g_+(z)+g_-(z)=\int_{z_0}^z\big[\rho_+(\zeta)+\rho_-(\zeta)\big]\,\d\zeta\stackrel{\eqref{e69}}{=}\frac{1}{\sigma^{\frac{3}{2}}}\ln\big(1-\Phi\left(z-\omega\right)\big)-\frac{1}{\sigma^{\frac{3}{2}}}\ln\big(1-\Phi\left(z_0-\omega\right)\big).
\end{equation*}
This establishes property $(1)$ and $(2)$ in Proposition \ref{prop8}. For expansion \eqref{e66}, we begin with the observation
\begin{align*}
	\rho(z)=z^{\frac{1}{2}}\bigg[1-&\,\frac{1}{z}\left\{\frac{z_0}{2}+\rho_0\right\}-\frac{1}{z^2}\left\{\frac{z_0^2}{8}-\frac{1}{2}z_0\rho_0+\rho_1\right\}-\frac{1}{z^3}\left\{\frac{z_0^3}{16}-\frac{1}{8}z_0^2\rho_0-\frac{1}{2}z_0\rho_1+\rho_2\right\}\\
	&\,+\mathcal{O}\big(z^{-4}\big)\bigg];\ \ \ \ \ \ \ \ \ \rho_k:=\frac{1}{2(\pi\sigma)^{\frac{3}{2}}}\int_{-\infty}^{z_0}\frac{\e^{-(\lambda-\omega)^2}}{1-\Phi(\lambda-\omega)}\frac{\lambda^k}{\sqrt{z_0-\lambda}}\,\d\lambda,\ \ k\in\mathbb{Z}_{\geq 0},
\end{align*}
valid as $z\rightarrow\infty$ and $z\notin(-\infty,z_0]$. Thus, with \eqref{e64}, i.e. with $z_0=-2\rho_0$ in place, after simplification,
%
%
%
%
%
%
%
%
%
%
\begin{equation}\label{e70}
	\rho(z)=z^{\frac{1}{2}}\bigg[1-\frac{g_1}{2z^2}-\frac{3g_2}{2z^3}+\mathcal{O}\big(z^{-4}\big)\bigg],
\end{equation}
using $g_j=g_j(t,\sigma)$ as in \eqref{e66a},\eqref{e66aa}, and so after $z$-integration, from \eqref{e62},
\begin{equation}\label{e71}
	g(z)=\frac{2}{3}z^{\frac{3}{2}}+g_0+\frac{g_1}{z^{\frac{1}{2}}}+\frac{g_2}{z^{\frac{3}{2}}}+\mathcal{O}\big(z^{-\frac{5}{2}}\big),\ \ \ \ z\rightarrow\infty,\ \ z\notin(-\infty,z_0],
\end{equation}
where $g_0=g_0(t,\sigma)$, the integration constant, is yet to be determined. However, consistency between \eqref{e71} and \eqref{e65} enforces $g_0=-\frac{1}{2}\sigma^{-\frac{3}{2}}\ln(1-\Phi(z_0-\omega))$ and thus $\ell=\ell(t,\sigma)$ in \eqref{e66} as written in \eqref{e66a}. This completes the proof of property $(3)$ and thus the proof of Proposition \ref{prop8}.
\end{proof}
Equipped with the above properties of $g(z)=g(z;t,\sigma,z_0)$ we now proceed with the aformentioned $g$-function transformation. Define
\begin{equation}\label{e75}
	{\bf S}(z;t,\sigma,z_0):=\begin{bmatrix}1 & 0\\ g_1\sigma^2 & 1\end{bmatrix}{\bf T}(z;t,\sigma,z_0)\e^{s(g(z)-\frac{\ell}{2})\sigma_3},\ \ z\in\mathbb{C}\setminus\Sigma_{\bf T}; \ \ \ s=\sigma^{\frac{3}{2}},
\end{equation}
where we don't indicate the $(t,\sigma,z_0)$-dependence of $g(z),\ell$ and $g_1$, compare \eqref{e62} and \eqref{e66a}. The defining properties of ${\bf S}(z)$ are summarized below:
\begin{problem}\label{trafo4} Let $(t,\sigma)\in(0,\infty)\times(0,\infty)$ and $z_0=z_0(t,\sigma)$ as in \eqref{e64}. The function ${\bf S}(z)={\bf S}(z;t,\sigma,z_0)\in\mathbb{C}^{2\times 2}$ defined in \eqref{e75} is uniquely determined by the following properties:
\begin{enumerate}
	\item ${\bf S}(z)$ is analytic for $z\in\mathbb{C}\setminus\Sigma_{\bf T}$ with $\Sigma_{\bf T}$ defined in RHP \ref{trafo3} and shown in Figure \ref{figuretrafo2a}. On each connected component of $\mathbb{C}\setminus\Sigma_{\bf T}$ there is a continuous extension of ${\bf S}(z)$ to the closure of the same component.
%
%
	\item The continuous limiting values ${\bf S}_{\pm}(z)$ on $\Gamma_j\ni z$ satisfy ${\bf S}_+(z)={\bf S}_-(z){\bf G}_{\bf S}(z;t,\sigma,z_0)$ with the jump matrix ${\bf G}_{\bf S}(z;t,\sigma,z_0)$ equal to
	\begin{align*}
		{\bf G}_{\bf S}(z;t,\sigma,z_0)=&\,\,\begin{bmatrix}1 & (1-\Phi(z-\omega))\e^{-s(2g(z)-\ell)}\smallskip\\ 0 & 1\end{bmatrix},\ \ z\in\Gamma_1;\\
		{\bf G}_{\bf S}(z;t,\sigma,z_0)=&\,\,\begin{bmatrix}1 & 0\smallskip\\ (1-\Phi(z-\omega))^{-1}\e^{s(2g(z)-\ell)} & 1\end{bmatrix},\ \ z\in\Gamma_2\cup\Gamma_4;
	\end{align*}
	and
	\begin{align*}
		{\bf G}_{\bf S}(z;t,\sigma,z_0)=&\,\,\begin{bmatrix}0 & 1\smallskip\\ -1 & 0\end{bmatrix},\ \ z\in\Gamma_3.
	\end{align*}
	\item ${\bf S}(z)$ is bounded in a neighbourhood of $z=z_0$.
	\item As $z\rightarrow\infty$, valid in a full neighborhood of infinity off the jump contours,
	\begin{equation}\label{e76}
		{\bf S}(z)=\Big\{\mathbb{I}+{\bf S}_1z^{-1}+{\bf S}_2z^{-2}+\mathcal{O}\big(z^{-3}\big)\Big\}(z\sigma)^{-\frac{1}{4}\sigma_3}\frac{1}{\sqrt{2}}\begin{bmatrix}1 & 1\\ -1 & 1\end{bmatrix}\e^{-\im\frac{\pi}{4}\sigma_3}
	\end{equation}
	with the principal branch for $z^{\alpha}:\mathbb{C}\setminus(-\infty,0]\rightarrow\mathbb{C}$ and ${\bf S}_1,{\bf S}_2$ equal to, compare \eqref{e61},\eqref{e66},
	\begin{align*}
		{\bf S}_1=&\,{\bf S}_1(t,\sigma)=\begin{bmatrix}1 & 0\\
		g_1\sigma^2 & 1\end{bmatrix}{\bf T}_1(t,\sigma)\begin{bmatrix}1 & 0\\ -g_1\sigma^2 & 1\end{bmatrix}+\begin{bmatrix}\frac{1}{2}g_1^2\sigma^3 & -g_1\sigma\smallskip\\ -g_2\sigma^2+\frac{1}{3}g_1^3\sigma^5 & -\frac{1}{2}g_1^2\sigma^3\end{bmatrix},\\
		{\bf S}_2=&\,{\bf S}_2(t,\sigma)=\begin{bmatrix}1 & 0\\ g_1\sigma^2 & 1\end{bmatrix}{\bf T}_1(t,\sigma)\begin{bmatrix}\frac{1}{2}g_1^2\sigma^3 & -g_1\sigma\smallskip\\ -g_2\sigma^2-\frac{1}{6}g_1^3\sigma^5 & \frac{1}{2}g_1^2\sigma^3\end{bmatrix}+
		\begin{bmatrix}1 & 0\\ g_1\sigma^2 & 1\end{bmatrix}{\bf T}_2(t,\sigma)\begin{bmatrix}1 & 0\\ -g_1\sigma^2 & 1\end{bmatrix}\\
		&\hspace{2cm}+\begin{bmatrix}g_1g_2\sigma^3+\frac{1}{24}g_1^4\sigma^6 & -g_2\sigma-\frac{1}{6}g_1^3\sigma^4\smallskip\\ \ast & -\frac{1}{8}g_1^4\sigma^6\end{bmatrix}.
	\end{align*}
	We use $\ast$ to represent scalar entries, dependent on $g_j(t,\sigma)$ only, that are not needed in the following.
\end{enumerate}
\end{problem}
As seen in property $(2)$ of RHP \ref{trafo4}, certain combinations of $1-\Phi(z-\omega)$ and $\e^{-s(2g(z)-\ell)}$ appear in the jump entries. These turn out to be asymptotically small away from $z=z_0$, made precise below in our next result.
\begin{prop}\label{prop9}
The following two inequalities hold for any $t,\sigma>0$ with $s:=\sigma^{\frac{3}{2}}>0$,
\begin{equation}\label{e66b}
		\bigg|\e^{-s(2g(z)-\ell)}\big(1-\Phi\left(z-\omega\right)\big)\bigg|\leq \e^{-\frac{4}{3}s(z-z_0)^{\frac{3}{2}}},\ \ z\in\Gamma_1;
\end{equation}
\begin{equation}\label{e66c}
		\bigg|\e^{s(2g(z)-\ell)}\big(1-\Phi\left(z-\omega\right)\big)^{-1}\bigg|\leq \e^{-\frac{2}{3}s|z-z_0|^{\frac{3}{2}}+2s|z_0||z-z_0|^{\frac{1}{2}}}
		,\ \ z\in\Gamma_2\cup\Gamma_4.
\end{equation}
Furthermore, for any $t,\sigma>0$, let 
\begin{equation}\label{e66d}
	0<\epsilon_{t,\sigma}:=\frac{1}{3\sqrt{2}}\min\left\{\frac{t}{\sigma},1\right\}>0.
\end{equation}
Then the function
\begin{equation*}
	z\mapsto\eta(z)=\eta(z;t,\sigma,z_0):=1+\frac{1}{2(\pi\sigma)^{\frac{3}{2}}}\int_{-\infty}^{z_0}\left[\frac{\e^{-(\lambda-\omega)^2}}{1-\Phi(\lambda-\omega)}-\frac{\e^{-(z-\omega)^2}}{1-\Phi(z-\omega)}\right]\frac{1}{\sqrt{z_0-\lambda}}\frac{\d\lambda}{\lambda-z}
\end{equation*}
is analytic in the open disk $|z-z_0|<\epsilon_{t,\sigma}$ and satisfies $\eta(z_0)\geq 1$. Additionally, there exists $c>0$, so that for any $t,\sigma\geq 1$,
\begin{equation}\label{e66e}
	\big|\eta(z)-\eta(z_0)\big|\leq c|z-z_0|,\ \ \ \ |z-z_0|<\epsilon_{t,\sigma}.
\end{equation}
\end{prop}
\begin{proof}
We have for $z>z_0$,
\begin{align*}
	\frac{\d}{\d z}\bigg[2&g(z)-\ell-\frac{1}{\sigma^{\frac{3}{2}}}\ln\big(1-\Phi\left(z-\omega\right)\big)\bigg]\stackrel{\eqref{e62}}{=}2\rho(z)+\frac{1}{\sqrt{\pi}\sigma^{\frac{3}{2}}}\frac{\e^{-(z-\omega)^2}}{1-\Phi(z-\omega)}\\
	&\,\stackrel{\eqref{e63}}{=}\sqrt{z-z_0}\left\{2+\frac{1}{(\pi\sigma)^{\frac{3}{2}}}\int_{-\infty}^{z_0}\left[\frac{\e^{-(\lambda-\omega)^2}}{1-\Phi(\lambda-\omega)}-\frac{\e^{-(z-\omega)^2}}{1-\Phi(z-\omega)}\right]\frac{1}{\sqrt{z_0-\lambda}}\frac{\d\lambda}{\lambda-z}\right\}\geq 2\sqrt{z-z_0},
\end{align*}
since $(-\infty,0)\ni x\mapsto-\ln(1-\Phi(x))$ is convex and since
\begin{equation*}
	\sqrt{z-z_0}\int_{-\infty}^{z_0}\frac{1}{\sqrt{z_0-\lambda}}\frac{\d\lambda}{\lambda-z}=-\pi,\ \ \ \ z>z_0.
\end{equation*}
Thus, upon integration of the above inequality, for any $z>z_0$,
\begin{align*}
	2g(z)-\ell-\frac{1}{\sigma^{\frac{3}{2}}}\ln\big(1-&\,\Phi\left(z-\omega\right)\big)\stackrel{\eqref{e62}}{=}\int_{z_0}^z\frac{\d}{\d\zeta}\bigg[2g(\zeta)-\ell-\frac{1}{\sigma^{\frac{3}{2}}}\ln\big(1-\Phi\left(\zeta-\omega\right)\big)\bigg]\,\d\zeta\\
	&\,\geq 2\int_{z_0}^z\sqrt{\zeta-z_0}\,\d\zeta=\frac{4}{3}(z-z_0)^{\frac{3}{2}},
\end{align*}
which implies \eqref{e66b}. For \eqref{e66c} we write instead
\begin{equation}\label{e72}
	\bigg|\e^{s(2g(z)-\ell)}\big(1-\Phi\left(z-\omega\right)\big)^{-1}\bigg|\stackrel{\eqref{e66a}}{=}\big|\e^{2sg(z)}\big|\,\frac{\textnormal{erfc}(z_0-\omega)}{\big|\textnormal{erfc}(z-\omega)\big|}\stackrel{\eqref{appB1}}{\leq}\big|\e^{2sg(z)}\big|,
\end{equation}
and now estimate $\Gamma_2\cup\Gamma_4\ni z\mapsto g(z)$ as follows: since $g(z)=\int_0^{z-z_0}\rho(\zeta+z_0)\,\d\zeta$ we can use \eqref{e63}, Fubini's theorem and the endpoint equation \eqref{e64} to obtain the following exact expression for $g(z)$,
\begin{equation*}
	g(z)=\frac{2}{3}(z-z_0)^{\frac{3}{2}}+z_0(z-z_0)^{\frac{1}{2}}+\frac{1}{(\pi\sigma)^{\frac{3}{2}}}\int_{-\infty}^{z_0}\frac{\e^{-(\lambda-\omega)^2}}{1-\Phi(\lambda-\omega)}\arctan\left[\left(\frac{z-z_0}{z_0-\lambda}\right)^{\frac{1}{2}}\right]\d\lambda,\ z\in\mathbb{C}\setminus(-\infty,z_0],
\end{equation*}
with principal branches for all fractional exponents and also with the principal value for the inverse trigonometric function, cf. \cite[$\S 4.23$(ii)]{NIST}. But for $w=|w|\e^{\pm\im\frac{3\pi}{8}}\in\mathbb{C}$ we have the simple estimate
\begin{equation*}
	\big|\arctan(w)\big|=\left|\int_0^w\frac{\d t}{1+t^2}\right|\leq|w|\sup_{0\leq t\leq|w|}\frac{1}{|1+\e^{\pm\im\frac{3\pi}{4}}t^2|}\leq\sqrt{2}|w|,
\end{equation*}
and so all together, for $z\in\Gamma_2\cup\Gamma_4$,
\begin{equation}\label{e72a}
	\Re g(z)\leq\frac{2}{3}|z-z_0|^{\frac{3}{2}}\cos\Big(\frac{9\pi}{8}\Big)+z_0\left(\cos\Big(\frac{3\pi}{8}\Big)-\sqrt{2}\right)|z-z_0|^{\frac{1}{2}},
\end{equation}
where we have used \eqref{e64} one more time. Estimate \eqref{e72a} yields \eqref{e66c} and completes our derivation of the same inequality.
Next, $z\mapsto\eta(z)$ is well defined for all $z\in\mathbb{C}$ away from the sectors $\textnormal{arg}(z-\frac{t}{\sigma})\in(-\frac{3\pi}{4},-\frac{\pi}{2})\cup(\frac{\pi}{2},\frac{3\pi}{4})$ and hence certainly analytic in any disk centered at $z_0$ which lies in the complement of the same sectors. This is achieved by \eqref{e66d}. Also, again by convexity of $(-\infty,0)\ni x\mapsto-\ln(1-\Phi(x))$, we obtain at once the lower bound $\eta(z_0)\geq 1$, so it remains to derive \eqref{e66e} for any $t,\sigma\geq 1$. For that we estimate the derivative
\begin{equation*}
	\frac{\d}{\d z}\eta(z)=\frac{1}{2(\pi\sigma)^{\frac{3}{2}}}\int_{-\infty}^{z_0-\epsilon_{t,\sigma}}k\left(\lambda,z;\omega\right)\frac{\d\lambda}{\sqrt{z_0-\lambda}}+\frac{1}{2(\pi\sigma)^{\frac{3}{2}}}\int_{z_0-\epsilon_{t,\sigma}}^{z_0}k\left(\lambda,z;\omega\right)\frac{\d\lambda}{\sqrt{z_0-\lambda}}
\end{equation*}
in the disk $|z-z_0|<\epsilon_{t,\sigma}$ when $t,\sigma\geq 1$. Given that
\begin{equation*}
	k\left(\lambda,z;\omega\right):=\frac{\partial}{\partial z}\left\{\left[\frac{\e^{-(\lambda-\omega)^2}}{1-\Phi(\lambda-\omega)}-\frac{\e^{-(z-\omega)^2}}{1-\Phi(z-\omega)}\right]\frac{1}{\lambda-z}\right\},
\end{equation*}
we obtain for $\lambda\in(-\infty,z_0-\epsilon_{t,\sigma})$, by the properties of the functions involved,
\begin{equation*}
	\left|k\left(\lambda,z;\omega\right)\right|\leq\frac{c}{1+|\lambda|},
\end{equation*}
uniformly on $|z-z_0|<\epsilon_{t,\sigma}$ for all $t,\sigma\geq 1$. On the other hand, when $\lambda\in(z_0-\epsilon,z_0)$, in fact more generally for $|\lambda-z|<2\epsilon_{t,\sigma}$ and $|z-z_0|<\epsilon_{t,\sigma}$, by Taylor's theorem,
\begin{align*}
	\left|k\left(\lambda,z;\omega\right)\right|=&\,\frac{1}{|\lambda-z|^2}\left|\frac{\e^{-(\lambda-\omega)^2}}{1-\Phi(\lambda-\omega)}-\frac{\e^{-(z-\omega)^2}}{1-\Phi(z-\omega)}+\sqrt{\pi}\frac{\d^2}{\d\lambda^2}\ln\big(1-\Phi\left(\lambda-\omega\right)\big)\bigg|_{\lambda=z}(\lambda-z)\right|\\
	\leq&\,\frac{c}{\epsilon_{t,\sigma}^2}\max_{|\lambda-z|=2\epsilon_{t,\sigma}}\left|\frac{\e^{-(\lambda-\omega)^2}}{1-\Phi(\lambda-\omega)}\right|\leq\frac{c}{\epsilon_{t,\sigma}^2}.
\end{align*}
Consequently, there exist $c_k>0$ so that for any $t,\sigma\geq 1$,
\begin{equation*}
	\left|\frac{\d}{\d z}\eta(z)\right|\leq c_1+\frac{c_2}{\sigma^{\frac{3}{2}}}\int_{z_0-\epsilon_{t,\sigma}}^{z_0}\frac{\d\lambda}{\epsilon_{t,\sigma}^2\sqrt{z_0-\lambda}}=c_1+\frac{c_3}{(\sigma\epsilon_{t,\sigma})^{\frac{3}{2}}}\stackrel{\eqref{e66d}}{\leq}c_4,
\end{equation*}
valid in the open disk $|z-z_0|<\epsilon_{t,\sigma}$. This yields \eqref{e66e} after integration and  completes our proof of Proposition \ref{prop9}.
\end{proof}

By \eqref{e66b} and \eqref{e66c}, seeing that $sz_0$ is bounded from above because of \eqref{appB4}, ${\bf G}_{\bf S}$ in RHP \ref{trafo4} is aymptotically localized near $z=z_0$ and along the line segment $\Gamma_3$. We therefore proceed with the necessary local analysis and begin with the following standard problem:
\begin{problem}\label{para1} Let $\sigma>0$. Find ${\bf P}^{(\infty)}(z)={\bf P}^{(\infty)}(z;\sigma,z_0)\in\mathbb{C}^{2\times 2}$ such that
\begin{enumerate}
	\item ${\bf P}^{(\infty)}(z)$ is analytic for $z\in\mathbb{C}\setminus(-\infty,z_0]$.
	\item The function ${\bf P}^{(\infty)}(z)$ attains square-integrable limiting values on $(-\infty,z_0]$ and those are related by the jump condition
	\begin{equation*}
		{\bf P}^{(\infty)}_+(z)={\bf P}^{(\infty)}_-(z)\begin{bmatrix}0 & 1\\ -1 & 0\end{bmatrix},\ \ \ z\in(-\infty,z_0).
	\end{equation*}
	\item As $z\rightarrow\infty$,
	\begin{equation}\label{e74a}
		{\bf P}^{(\infty)}(z)=\bigg\{\mathbb{I}+\frac{z_0}{4z}\sigma_3+\frac{z_0^2}{32z^2}\begin{bmatrix}5 & 0\\ 0 & -3\end{bmatrix}+\mathcal{O}\big(z^{-3}\big)\bigg\}(z\sigma)^{-\frac{1}{4}\sigma_3}\frac{1}{\sqrt{2}}\begin{bmatrix}1 & 1\\ -1 & 1\end{bmatrix}\e^{-\im\frac{\pi}{4}\sigma_3},
	\end{equation}
	with the principal branch for $z^{\alpha}:\mathbb{C}\setminus(-\infty,0]\rightarrow\mathbb{C}$ such that $z^{\alpha}>0$ when $z>0$.
\end{enumerate}
\end{problem}
A direct computation verifies that 
\begin{equation}\label{e77}
	{\bf P}^{(\infty)}(z)=(z-z_0)^{-\frac{1}{4}\sigma_3}\sigma^{-\frac{1}{4}\sigma_3}\frac{1}{\sqrt{2}}\begin{bmatrix}1 & 1\\ -1 & 1\end{bmatrix}\e^{-\im\frac{\pi}{4}\sigma_3},\ \ z\in\mathbb{C}\setminus(-\infty,z_0],
\end{equation}
has the properties listed in RHP \ref{para1}. Next, in a vicinity of $z=z_0$ we require a solution to the following model problem.
\begin{problem}\label{para2} Let $t,\sigma\geq 1$ and $z_0=z_0(t,\sigma)$ as well as $\epsilon_{t,\sigma}>0$ as in \eqref{e64} and \eqref{e66d}. Now find ${\bf P}^{(z_0)}(z)={\bf P}^{(z_0)}(z;t,\sigma,z_0)\in\mathbb{C}^{2\times 2}$ such that
\begin{enumerate}
	\item ${\bf P}^{(z_0)}(z)$ is analytic for $z\in\mathbb{D}_{\epsilon_{t,\sigma}}(z_0)\setminus\Sigma_{\bf T}$ with the open disk $\mathbb{D}_r(z_0):=\{z\in\mathbb{C}:\,|z-z_0|<r\}$.
	\item ${\bf P}^{(z_0)}(z)$ has the following local jump behavior, see Figure \ref{figuretrafo2a} for contour orientations,
	\begin{align*}
		{\bf P}_+^{(z_0)}(z)=&\,\,{\bf P}_-^{(z_0)}(z)\begin{bmatrix}1 & (1-\Phi(z-\omega))\e^{-s(2g(z)-\ell)}\smallskip\\ 0 & 1\end{bmatrix},\ \ z\in\Gamma_1\cap\mathbb{D}_{\epsilon_{t,\sigma}}(z_0);\\
		{\bf P}_+^{(z_0)}(z)=&\,\,{\bf P}_-^{(z_0)}(z)\begin{bmatrix}1 & 0\\ (1-\Phi(z-\omega))^{-1}\e^{s(2g(z)-\ell)} & 1\end{bmatrix},\ \ z\in(\Gamma_2\cup\Gamma_4)\cap\mathbb{D}_{\epsilon_{t,\sigma}}(z_0);
	\end{align*}
	followed by
	\begin{align*}
		{\bf P}_+^{(z_0)}(z)=&\,\,{\bf P}_-^{(z_0)}(z)\begin{bmatrix}0 & 1\\ -1 & 0\end{bmatrix},\ \ z\in\Gamma_3\cap\mathbb{D}_{\epsilon_{t,\sigma}}(z_0).
	\end{align*}
	\item ${\bf P}^{(z_0)}(z)$ is bounded in a neighbourhood of $z=z_0$.
	\item As $t,\sigma\rightarrow\infty$, we have the following asymptotic matching between ${\bf P}^{(z_0)}(z)$ and ${\bf P}^{(\infty)}(z)$,
	\begin{align}\label{e78}
		{\bf P}^{(z_0)}(z)\sim\bigg\{\mathbb{I}\,\,+&\,\sum_{m=0}^{\infty}\begin{bmatrix}0 & 0\\ b_{2m+1} & 0\end{bmatrix}\left(\frac{\sigma(z-z_0)}{\zeta(z)}\right)^{\frac{1}{2}}\big(\zeta(z)\big)^{-3m-1}+\sum_{m=0}^{\infty}\begin{bmatrix}0 & a_{2m+1}\\ 0 & 0\end{bmatrix}\nonumber\\
		&\,\times\left(\frac{\sigma(z-z_0)}{\zeta(z)}\right)^{-\frac{1}{2}}\big(\zeta(z)\big)^{-3m-2}+\sum_{m=1}^{\infty}\begin{bmatrix}a_{2m} & 0\\ 0 & b_{2m}\end{bmatrix}\big(\zeta(z)\big)^{-3m}\bigg\}{\bf P}^{(\infty)}(z),
%
%
%
	\end{align}
	which holds for $0<\frac{1}{4}\epsilon_{t,\sigma}\leq|z-z_0|\leq \frac{3}{4}\epsilon_{t,\sigma}<\epsilon_{t,\sigma}$. Here, $\zeta=\zeta(z)$ is defined in \eqref{e79} below and the $(t,\sigma)$-independent coefficients $\{a_m,b_m\}_{m=1}^{\infty}$ are given in RHP \ref{AiryRHP}.
\end{enumerate}
\end{problem}
In order to solve Problem \ref{para2} we consider the change of coordinates
\begin{equation}\label{e79}
	 \zeta(z)=\zeta(z;t,\sigma,z_0):=\left[\frac{3s}{2}\left\{g(z)-\frac{\ell}{2}-\frac{1}{2\sigma^{\frac{3}{2}}}\ln\big(1-\Phi(z-\omega)\big)\right\}\right]^{\frac{2}{3}},\ \ z\in\mathbb{D}_{\epsilon_{t,\sigma}}(z_0);\ \ \ \ s=\sigma^{\frac{3}{2}}>0,
\end{equation}
which is well-defined by Proposition \ref{prop9} and \eqref{e66e} for all $t,\sigma\geq 1$. Indeed, we have
\begin{align*}
	g(z)-\frac{\ell}{2}-\frac{1}{2\sigma^{\frac{3}{2}}}\ln&\,\big(1-\Phi(z-\omega)\big)=\int_{z_0}^z(\xi-z_0)^{\frac{1}{2}}\eta(\xi)\,\d\xi\\
	=&\,\frac{2}{3}\eta(z_0)(z-z_0)^{\frac{3}{2}}\Big\{1+\frac{3}{5}\frac{\eta'(z_0)}{\eta(z_0)}(z-z_0)+\mathcal{O}\big((z-z_0)^2\big)\Big\},\ \ z\in\mathbb{D}_{\epsilon_{t,\sigma}}(z_0),
\end{align*}
for all $t,\sigma\geq 1$ and so $\mathbb{D}_{\epsilon_{t,\sigma}}(z_0)\ni z\mapsto\zeta(z)$ is locally conformal,
\begin{equation}\label{e79a}
	\zeta(z)=\sigma\big(\eta(z_0)\big)^{\frac{2}{3}}(z-z_0)\left[1+\sum_{k=1}^{\infty}\frac{\eta^{(k)}(z_0)}{\eta(z_0)}\frac{3}{2k+3}\frac{(z-z_0)^k}{k!}\right]^{\frac{2}{3}},\ \ z\in\mathbb{D}_{\epsilon_{t,\sigma}}(z_0).
\end{equation}
In turn, the function
\begin{equation}\label{e80}
	{\bf P}^{(z_0)}(z)=\left(\frac{\sigma(z-z_0)}{\zeta(z)}\right)^{-\frac{1}{4}\sigma_3}{\bf M}^{\textnormal{Ai}}\Big(\zeta(z);z_0,\frac{\pi}{4}\Big)\e^{\frac{2}{3}\zeta^{\frac{3}{2}}(z)\sigma_3},\ \ \ z\in\mathbb{D}_{\epsilon_{t,\sigma}}(z_0)\setminus\Sigma_{\bf T},
\end{equation}
defined in terms of the Airy parametrix \eqref{app2}, solves RHP \ref{para2}.
\begin{rem}\label{rem7} In order to verify that \eqref{e80} has the desired properties one needs to recall RHP \ref{AiryRHP} and realize that
\begin{equation*}
	\frac{\sigma(z-z_0)}{\zeta(z)}=\frac{1}{(\eta(z_0))^{\frac{2}{3}}}\Big\{1-\frac{2}{5}\frac{\eta'(z_0)}{\eta(z_0)}(z-z_0)+\mathcal{O}\big((z-z_0)^2\big)\Big\},\ \ z\in\mathbb{D}_{\epsilon_{t,\sigma}}(z_0),
\end{equation*}
is locally analytic. Moreover, \eqref{e78} follows from property $(4)$ in RHP \ref{AiryRHP} and from the estimate, see \eqref{e79a},
\begin{equation*}
	\big|\zeta(z)\big|\geq c\sigma|z-z_0|\geq c\sigma\epsilon_{t,\sigma}=c\min\{t,\sigma\},\ \ \ 0<\frac{1}{4}\epsilon_{t,\sigma}\leq|z-z_0|\leq \frac{3}{4}\epsilon_{t,\sigma},\ \ \ c>0,
\end{equation*}
where we use that $\eta(z_0)\geq 1$ for all $t,\sigma\geq 1$.
\end{rem}
Having completed the necessary local analysis with \eqref{e77} and \eqref{e80} we now move to our final transformation. Define
\begin{equation}\label{e81}
	{\bf R}(z;t,\sigma,z_0):={\bf S}(z;t,\sigma,z_0)\begin{cases}\displaystyle\big({\bf P}^{(z_0)}(z;t,\sigma,z_0)\big)^{-1},&z\in\mathbb{D}_{\frac{1}{2}\epsilon_{t,\sigma}}(z_0)\setminus\Sigma_{\bf T}\bigskip\\
	\displaystyle\big({\bf P}^{(\infty)}(z;\sigma,z_0)\big)^{-1},&z\notin\overline{\mathbb{D}_{\frac{1}{2}\epsilon_{t,\sigma}}(z_0)}\setminus\Sigma_{\bf T}
	\end{cases},
\end{equation}
and recall RHP \ref{trafo4}, \ref{para1} and \ref{para2}. The characterizing properties of ${\bf R}(z)$ are then as follows:
\begin{problem}\label{trafo5} Let $t,\sigma\geq 1$ and $z_0=z_0(t,\sigma),\epsilon_{t,\sigma}$ as in \eqref{e64} and \eqref{e66d}. The function ${\bf R}(z)={\bf R}(z;t,\sigma,z_0)\in\mathbb{C}^{2\times 2}$ defined in \eqref{e81} is uniquely determined by the following properties:
\begin{enumerate}
	\item ${\bf R}(z)$ is analytic for $z\in\mathbb{C}\setminus\Sigma_{\bf R}$ where
	\begin{equation*}
		\Sigma_{\bf R}:=\partial\mathbb{D}_{\frac{1}{2}\epsilon_{t,\sigma}}(z_0)\cup\left[z_0+\frac{1}{2}\epsilon_{t,\sigma},\infty\right)\cup\left\{\e^{-\im\frac{\pi}{4}}(-\infty,z_0]\setminus\mathbb{D}_{\frac{1}{2}\epsilon_{t,\sigma}}(z_0)\right\}\cup\left\{\e^{\im\frac{\pi}{4}}(-\infty,z_0]\setminus\mathbb{D}_{\frac{1}{2}\epsilon_{t,\sigma}}(z_0)\right\}
	\end{equation*}
	is shown in Figure \ref{figuretrafo3}. Moreover, on each connected component of $\mathbb{C}\setminus\Sigma_{\bf R}$ there is a continuous extension of ${\bf R}(z)$ to the closure of the same component.
	\item The continuous limiting values ${\bf R}_{\pm}(z)$ on $\Sigma_{\bf R}\ni z$ obey the constraint ${\bf R}_+(z)={\bf R}_-(z){\bf G}_{\bf R}(z;t,\sigma,z_0)$ where, for $z\in[z_0+\frac{1}{2}\epsilon_{t,\sigma},\infty)$,
	\begin{align*}
		{\bf G}_{\bf R}(z;t,\sigma,z_0)={\bf P}^{(\infty)}(z)\begin{bmatrix}1 & (1-\Phi(z-\omega))\e^{-s(2g(z)-\ell)}\smallskip\\ 0 & 1\end{bmatrix}\big({\bf P}^{(\infty)}(z)\big)^{-1},
	\end{align*}
	and for $z\in\e^{-\im\frac{\pi}{4}}(-\infty,z_0-\frac{1}{2}\epsilon_{t,\sigma}]\cup\e^{\im\frac{\pi}{4}}(-\infty,z_0-\frac{1}{2}\epsilon_{t,\sigma}]$,
	\begin{equation*}
		{\bf G}_{\bf R}(z;t,\sigma,z_0)={\bf P}^{(\infty)}(z)\begin{bmatrix}1 & 0\smallskip\\ (1-\Phi(z-\omega))^{-1}\e^{s(2g(z)-\ell)} & 1\end{bmatrix}\big({\bf P}^{(\infty)}(z)\big)^{-1}.
	\end{equation*}
	Additionally, for $z\in\partial\mathbb{D}_{\epsilon_{t,\sigma}}(z_0)$,
	\begin{equation*}
		{\bf G}_{\bf R}(z;t,\sigma,z_0)={\bf P}^{(z_0)}(z)\big({\bf P}^{(\infty)}(z)\big)^{-1}.
	\end{equation*}
	\item As $z\rightarrow\infty$ and $z\notin\Sigma_{\bf R}$,
	\begin{equation}\label{e82}
		{\bf R}(z)=\mathbb{I}+{\bf R}_1z^{-1}+{\bf R}_2z^{-2}+\mathcal{O}\big(z^{-3}\big),
	\end{equation}
	with ${\bf R}_1,{\bf R}_2$ equal to, compare \eqref{e76},\eqref{e74a},
	\begin{align*}
		 {\bf R}_1={\bf R}_1(t,\sigma)=&\,\,{\bf S}_1(t,\sigma)-\frac{z_0}{4}\sigma_3,\\
		 {\bf R}_2={\bf R}_2(t,\sigma)=&\,\,{\bf S}_2(t,\sigma)+\frac{z_0^2}{32}\begin{bmatrix}-3 & 0\\ 0 & 5\end{bmatrix}-\frac{z_0}{4}{\bf S}_1(t,\sigma)\sigma_3.
	\end{align*}
\end{enumerate}
\end{problem}
	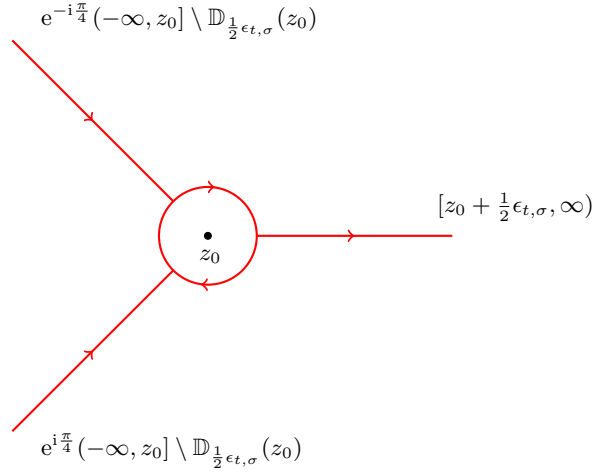
\begin{figure}[tbh]
	\begin{tikzpicture}[xscale=0.65,yscale=0.65]
	\draw [thick, color=red, decoration={markings, mark=at position 0.5 with {\arrow{>}}}, postaction={decorate}] (1,0) -- (5,0);
	\draw [thick, color=red, decoration={markings, mark=at position 0.25 with {\arrow{<}}}, decoration={markings, mark=at position 0.75 with {\arrow{<}}}, postaction={decorate}] (0,0) circle [radius=1];
\node [below] at (0.05,-0.1) {{\small $z_0$}};
\node [right] at (4.5,0.6) {{\small $[z_0+\frac{1}{2}\epsilon_{t,\sigma},\infty)$}};
\draw [thick, color=red, decoration={markings, mark=at position 0.5 with {\arrow{>}}}, postaction={decorate}] (-4,4) -- (-0.7071067810,0.7071067810);
\draw [thick, color=red, decoration={markings, mark=at position 0.5 with {\arrow{>}}}, postaction={decorate}] (-4,-4) -- (-0.7071067810,-0.7071067810);
\node [right] at (-3.6,4.4) {{\small $\e^{-\im\frac{\pi}{4}}(-\infty,z_0]\setminus\mathbb{D}_{\frac{1}{2}\epsilon_{t,\sigma}}(z_0)$}};
\node [right] at (-3.6,-4.4) {{\small $\e^{\im\frac{\pi}{4}}(-\infty,z_0]\setminus\mathbb{D}_{\frac{1}{2}\epsilon_{t,\sigma}}(z_0)$}};
\draw [fill, color=black] (0,0) circle [radius=0.07];
\end{tikzpicture}
\caption{The oriented jump contours for ${\bf R}(z)$ in the complex $z$-plane.}
\label{figuretrafo3}
\end{figure}
\begin{rem}\label{rem8} By construction, compare RHP \ref{para1} and \ref{para2}, ${\bf R}(z)$ is analytic in the disk $\mathbb{D}_{\frac{1}{2}\epsilon_{t,\sigma}}(z_0)$ and on the half ray $\Gamma_3$.
\end{rem}
By Proposition \ref{prop9} and \eqref{e78}, as soon as $t,\sigma\rightarrow+\infty$, the jump matrix in the ${\bf R}$-RHP is close to the identity matrix on all of $\Sigma_{\bf R}$. In detail, we record the following norm estimates.
\begin{prop}\label{prop10} There exist $t_0,\sigma_0\geq 1$ and $c>0$ such that
\begin{equation}\label{e83}
	\|{\bf G}_{\bf R}(\cdot;t,\sigma,z_0)-\mathbb{I}\|_{L^{\infty}(\Sigma_{\bf R})}\leq\frac{c}{\sigma\epsilon_{t,\sigma}},\ \ \ \ \|{\bf G}_{\bf R}(\cdot;t,\sigma,z_0)-\mathbb{I}\|_{L^2(\Sigma_{\bf R})}\leq\frac{c}{\sigma\sqrt{\epsilon_{t,\sigma}}}
\end{equation}
for all $t\geq t_0$ and $\sigma\geq \sigma_0$.
\end{prop}
\begin{proof} The three rays extending towards infinity yield exponentially small contributions to the identity only. For instance, on the real line extending to $+\infty$,
\begin{equation*}
	{\bf G}_{\bf R}(z;t,\sigma,z_0)=\mathbb{I}-\frac{\im}{2}\big(1-\Phi(z-\omega)\big)\e^{-s(2g(z)-\ell)}\begin{bmatrix}1 & \frac{1}{\sqrt{\sigma(z-z_0)}}\smallskip\\
	-\sqrt{\sigma(z-z_0)} & -1\end{bmatrix},\ \ z\geq z_0+\frac{1}{2}\epsilon_{t,\sigma},
\end{equation*}
i.e. with \eqref{e66b} we are led to the optimization of $z\mapsto\sqrt{\sigma(z-z_0)}\exp[-\frac{4}{3}s(z-z_0)^{3/2}]$ on $[z_0+\frac{1}{2}\epsilon_{t,\sigma},\infty)$. This function's critical point is located at $z=z_0+1/(4^{2/3}\sigma)$ and choosing $t,\sigma$ sufficiently large we can ensure that it does not lie in $[z_0+\frac{1}{2}\epsilon_{t,\sigma},\infty)$. Consequently, we obtain exponentially small sub-leading contributions in both, $L^2$ and $L^{\infty}$ sense on the three rays,
\begin{equation*}
	\|{\bf G}_{\bf R}(\cdot;t,\sigma,z_0)-\mathbb{I}\|_{L^{\infty}(\Sigma_{\bf R}\setminus\partial\mathbb{D}_{\frac{1}{2}\epsilon_{t,\sigma}}(z_0))}\leq c_1\sqrt{\sigma\epsilon_{t,\sigma}}\,\e^{-c_2(\sigma\epsilon_{t,\sigma})^{\frac{3}{2}}},\ \ \ c_k>0,\ \ t\geq t_0,\ \sigma\geq\sigma_0,
\end{equation*}
and
\begin{equation}\label{e83a}
	\|{\bf G}_{\bf R}(\cdot;t,\sigma,z_0)-\mathbb{I}\|_{L^2(\Sigma_{\bf R}\setminus\partial\mathbb{D}_{\frac{1}{2}\epsilon_{t,\sigma}}(z_0))}\leq c_3\sqrt[4]{\frac{\epsilon_{t,\sigma}}{\sigma}}\,\e^{-c_4(\sigma\epsilon_{t,\sigma})^{\frac{3}{2}}},\ \ \ c_k>0,\ \ t\geq t_0,\ \sigma\geq\sigma_0.
\end{equation}
On the circle boundary $\partial\mathbb{D}_{\frac{1}{2}\epsilon_{t,\sigma}}(z_0)$ we use \eqref{e78} instead and \eqref{e83} follows at once.
\end{proof}
In turn, by \cite{DZ}, RHP \ref{trafo5} is asymptotically solvable:
\begin{theo}\label{theo4} There exist $t_0,\sigma_0\geq 1$ such that the RHP for ${\bf R}(z)$ defined in \eqref{e81} is uniquely solvable in $L^2(\Sigma_{\bf R})$ for all $t\geq t_0$ and all $\sigma\geq\sigma_0$. We can compute the solution of the same problem iteratively via the integral equation
\begin{equation}\label{e84}
	{\bf R}(z;t,\sigma,z_0)=\mathbb{I}+\frac{1}{2\pi\im}\int_{\Sigma_{\bf R}}{\bf R}_-(\lambda;t,\sigma,z_0)\big({\bf G}_{\bf R}(\lambda;t,\sigma,z_0)-\mathbb{I}\big)\frac{\d\lambda}{\lambda-z},\ \ \ z\in\mathbb{C}\setminus\Sigma_{\bf R},
\end{equation}
using that, for all $t\geq t_0$ and $\sigma\geq\sigma_0$, with $c>0$,
\begin{equation}\label{e85}
	\|{\bf R}_-(\cdot;t,\sigma,z_0)-\mathbb{I}\|_{L^2(\Sigma_{\bf R})}\leq\frac{c}{\sigma\sqrt{\epsilon_{t,\sigma}}}.
\end{equation}
\end{theo}
Equipped with Theorem \ref{theo4}, we can now derive Theorem \ref{fimain3}.
\begin{proof}[Proof of Theorem \ref{fimain3}]
Let $t\geq t_0,\sigma\geq\sigma_0$ as in Theorem \ref{theo4}. We begin by recalling the explicit and invertible sequence of transformations
\begin{equation*}
	{\bf X}(z)\stackrel{\eqref{e58}}{\mapsto}{\bf Y}(z)\stackrel{\eqref{e60}}{\mapsto}{\bf T}(z)\stackrel{\eqref{e75}}{\mapsto}{\bf S}(z)\stackrel{\eqref{e81}}{\mapsto}{\bf R}(z),\ \ z\notin(\Sigma_{\bf R}\cup\mathbb{R}),
\end{equation*}
which converts the differential identities \eqref{e52},\eqref{e52aa} into the following three exact formul\ae,
\begin{equation}\label{e88}
	\frac{\partial}{\partial t}\ln F_{\sigma}(t)=-g_1(t,\sigma)\sigma^2-\sigma R_1^{12}(t,\sigma),
\end{equation}
followed by
\begin{align}\label{e89}
	\frac{\partial}{\partial\sigma}\ln F_{\sigma}(t)=g_1(t,\sigma)t\sigma-3g_2(t,\sigma)\sigma^2+\frac{1}{16\sigma}-R_1^{21}(t,\sigma)-2\sigma R_2^{12}(t,\sigma)+\left(t+\frac{\sigma}{2}z_0(t,\sigma)\right)R_1^{12}(t,\sigma),
\end{align}
and concluding with
\begin{align}\label{e89a}
	\frac{\d}{\d\alpha}\ln F_{\alpha}(\alpha)=-3g_2(\alpha,\alpha)\alpha^2+\frac{1}{16\alpha}-R_1^{21}(\alpha,\alpha)-2\alpha R_2^{12}(\alpha,\alpha)+\frac{\alpha}{2}z_0(\alpha,\alpha)R_1^{12}(\alpha,\alpha).
\end{align}
Here we have used that $\det{\bf R}(z)\equiv 1$ and so $\tr{\bf R}_1=0$. Moreover, $g_j=g_j(t,\sigma)$ and $z_0=z_0(t,\sigma)$ are as in \eqref{e66a},\eqref{e66aa} and \eqref{e64}, and the scalar entries $R_k^{mn}(t,\sigma)$ of ${\bf R}_k$ are computable via
\begin{equation*}
	{\bf R}_k(t,\sigma)=\frac{\im}{2\pi}\int_{\Sigma_{\bf R}}{\bf R}_-(\lambda;t,\sigma,z_0)\big({\bf G}_{\bf R}(\lambda;t,\sigma,z_0)-\mathbb{I}\big)\lambda^{k-1}\,\d\lambda,\ \ \ k\in\mathbb{Z}_{\geq 1},
\end{equation*}
compare \eqref{e84}. Assuming now temporarily that $t_0\leq t\leq\sigma$, we have the exact identity
\begin{equation*}
	\ln F_{\sigma}(t)=\ln F_{\sigma}(\sigma)-\int_t^{\sigma}\frac{\partial}{\partial x}\ln F_{\sigma}(x)\,\d x
\end{equation*}
and can thus use \eqref{e88} and \eqref{e89a} in its asymptotic evaluation. First, by Corollary \ref{appcor1}, we have for all $\sigma\geq\sigma_0$ and all $t>0$, using again $\omega=\frac{t}{\sigma}$, compare \eqref{rej2},
\begin{equation}\label{e90}
	-g_1(t,\sigma)\sigma^2=\frac{\partial}{\partial t}\left\{\sigma^{\frac{3}{2}}B(\omega)\right\}-\frac{1}{4\sigma}\left\{\frac{\d}{\d\omega}C(\omega)\right\}^2-\frac{\sigma^{-\frac{3}{2}}}{12\pi^{\frac{9}{2}}}\frac{\partial}{\partial t}\left\{\int_{-\infty}^0\frac{\e^{-(\lambda-\omega)^2}}{1-\Phi(\lambda-\omega)}\frac{\d\lambda}{\sqrt{-\lambda}}\right\}^3-r_2(t,\sigma)\sigma^2,
\end{equation}
with the error term $r_2(t,\sigma)$ satisfying $\int_t^{\infty}|r_2(x,\sigma)|\,\d x\leq c\sigma^{-5}$ in the same parameter regime and $C(\omega)$ as in \eqref{rej3}. On the other hand, by \eqref{e78} and a residue computation, with $\Gamma:=\partial\mathbb{D}_{\frac{1}{2}\epsilon_{t,\sigma}}(z_0)$, for all $t\geq t_0,\sigma\geq\sigma_0$,
\begin{equation*}
	\frac{\im}{2\pi}\ointclockwise_{\Gamma}\big({\bf G}_{\bf R}(\lambda;t,\sigma,z_0)-\mathbb{I}\big)\,\d\lambda\sim-\frac{7}{48}\begin{bmatrix}0 & 0\\ 1 & 0\end{bmatrix}\frac{1}{\sigma\eta(z_0)}-\frac{1}{16}\begin{bmatrix}0 & 1\\ 0 & 0\end{bmatrix}\frac{\eta'(z_0)}{(\sigma\eta(z_0))^2}+\sum_{n=3}^{\infty}\sum_{m=1}^{\infty}\frac{{\bf J}_{mn}(\omega)}{\sigma^n(\pi\sigma)^{\frac{3m}{2}}},
\end{equation*}
with $z\mapsto\eta(z)=\eta(z;t,\sigma,z_0)$ as in Proposition \ref{prop9} and a matrix-valued function ${\bf J}_{mn}$ which satisfies
\begin{equation*}
	\int_t^{\infty}\left\|{\bf J}_{mn}\left(\frac{s}{\sigma}\right)\right\|\d s\leq c_{mn}\sigma,\ \ \ \ \ \ \int_{\sigma}^{\infty}\left\|{\bf J}_{mn}\left(\frac{t}{s}\right)\right\|\frac{\d s}{s^{n+\frac{3m}{2}}}\leq c_{mn}t^{1-n-\frac{3m}{2}},\ \ \ c_{mn}>0,\ \ t\geq t_0,\ \sigma\geq \sigma_0.
\end{equation*}
Also, using the integral equation, 
\begin{equation*}
	{\bf R}_-(z;t,\sigma,z_0)=\mathbb{I}+\frac{1}{2\pi\im}\int_{\Sigma_{\bf R}}{\bf R}_-(\lambda;t,\sigma,z_0)\big({\bf G}_{\bf R}(\lambda;t,\sigma,z_0)-\mathbb{I}\big)\frac{\d\lambda}{\lambda-z_-},\ \ \ z\in\Sigma_{\bf R},
\end{equation*}
see \eqref{e84}, with \eqref{e85} in place, we obtain iteratively through the structure in \eqref{e78}, for all $t\geq t_0,\sigma\geq\sigma_0$,
\begin{equation*}
	\frac{\im}{2\pi}\ointclockwise_{\Gamma}\big({\bf R}_-(\lambda;t,\sigma,z_0)-\mathbb{I}\big)\big({\bf G}_{\bf R}(\lambda;t,\sigma,z_0)-\mathbb{I}\big)\,\d\lambda\sim\sum_{n=3}^{\infty}\sum_{m=1}^{\infty}\frac{{\bf K}_{mn}(\omega)}{\sigma^n(\pi\sigma)^{\frac{3m}{2}}}
\end{equation*}
with a matrix-valued function ${\bf K}_{mn}$ that satisfies
\begin{equation*}
	\int_t^{\infty}\left\|{\bf K}_{mn}\left(\frac{s}{\sigma}\right)\right\|\d s\leq c_{mn}\sigma,\ \ \ \ \ \ \int_{\sigma}^{\infty}\left\|{\bf K}_{mn}\left(\frac{t}{s}\right)\right\|\frac{\d s}{s^{n+\frac{3m}{2}}}\leq c_{mn}t^{1-n-\frac{3m}{2}},\ \ \ c_{mn}>0,\ \ t\geq t_0,\ \sigma\geq \sigma_0.
\end{equation*}
But also, for any $t_0\leq t\leq\sigma$, by \eqref{e83a} and \eqref{e85},
\begin{equation*}
	\left\|\frac{\im}{2\pi}\int_{\Sigma_{\bf R}\setminus\Gamma}{\bf R}_-(\lambda;t,\sigma,z_0)\big({\bf G}_{\bf R}(\lambda;t,\sigma,z_0)-\mathbb{I}\big)\,\d\lambda\right\|\leq\left(\frac{c}{\sigma}\right)\e^{-c_2t^{\frac{3}{2}}},\ \ \ c>0,
\end{equation*}
and thus together, for all $t_0\leq t\leq\sigma$,
\begin{equation}\label{e91}
	{\bf R}_1(t,\sigma)=-\frac{7}{48}\begin{bmatrix}0 & 0\\ 1 & 0\end{bmatrix}\frac{1}{\sigma\eta(z_0)}-\frac{1}{16}\begin{bmatrix}0 & 1\\ 0 & 0\end{bmatrix}\frac{\eta'(z_0)}{(\sigma\eta(z_0))^2}+\sum_{n=3}^{\infty}\sum_{m=1}^{\infty}\frac{{\bf L}_{mn}(\omega)}{\sigma^n(\pi\sigma)^{\frac{3m}{2}}}+\mathcal{O}\Big(\sigma^{-1}\e^{-ct^{\frac{3}{2}}}\Big)
\end{equation}
where the matrix-valued function ${\bf L}_{mn}$ satisfies
\begin{equation*}
	\int_t^{\infty}\left\|{\bf L}_{mn}\left(\frac{s}{\sigma}\right)\right\|\d s\leq c_{mn}\sigma,\ \ \ \ \ \ \int_{\sigma}^{\infty}\left\|{\bf L}_{mn}\left(\frac{t}{s}\right)\right\|\frac{\d s}{s^{n+\frac{3m}{2}}}\leq c_{mn}t^{1-n-\frac{3m}{2}},\ \ \ c_{mn}>0,\ \ t\geq t_0,\ \sigma\geq \sigma_0.
\end{equation*}
Consequently, combining \eqref{e90} and \eqref{e91} we find
\begin{align}\label{e92}
	\int_t^{\sigma}\frac{\partial}{\partial x}&\,\ln F_{\sigma}(x)\,\d x\stackrel{\eqref{e88}}{=}-\int_t^{\sigma}\big(g_1(x,\sigma)\sigma^2+\sigma R_1^{12}(x,\sigma)\big)\,\d x=\sigma^{\frac{3}{2}}C(1)-\sigma^{\frac{3}{2}}C(\omega)\nonumber\\
	&\,-\frac{1}{4}\int_{\omega}^1\left\{\frac{\d}{\d u}D(u)\right\}^2\d u+\mathcal{O}\big(\sigma^{-\frac{3}{2}}\big),\ \ \ \ t_0\leq t\leq\sigma,
\end{align}
with $D(u)$ as in \eqref{rej3}. Second, by Corollary \ref{appcor2}, for all $\alpha\geq\alpha_0$,
\begin{equation}\label{e93}
	-3g_2(\alpha,\alpha)\alpha^2=\frac{\partial}{\partial\alpha}\left\{\alpha^{\frac{3}{2}}C(1)\right\}+\mathcal{O}\big(\alpha^{-\frac{5}{2}}\big).
\end{equation}
Moreover, for all $t\geq t_0,\sigma\geq\sigma_0$ by \eqref{e78} and via iteration,
\begin{align*}
	\frac{\im}{2\pi}\ointclockwise_{\Gamma}{\bf R}_-(\lambda;t,\sigma,z_0)\big({\bf G}_{\bf R}(\lambda;t,\sigma,z_0)-\mathbb{I}\big)\lambda\,\d\lambda&\,\sim-\frac{7}{48}\begin{bmatrix}0 & 0\\ 1 & 0\end{bmatrix}\frac{z_0}{\sigma\eta(z_0)}\\
	+\frac{5}{48}&\left(\frac{1}{\sigma^2\eta(z_0)}-\frac{3}{5}\frac{z_0\eta'(z_0)}{(\sigma\eta(z_0))^2}\right)\begin{bmatrix}0 & 1\\ 0 & 0\end{bmatrix}+\sum_{n=3}^{\infty}\sum_{m=0}^{\infty}\frac{{\bf M}_{mn}(\omega)}{\sigma^n(\pi\sigma)^{\frac{3m}{2}}},
\end{align*}
where the matrix-valued function ${\bf M}_{mn}$ satisfies
\begin{equation*}
	\int_t^{\infty}\left\|{\bf M}_{mn}\left(\frac{s}{\sigma}\right)\right\|\d s\leq c_{mn}\sigma,\ \ \ \ \int_{\sigma}^{\infty}\left\|{\bf M}_{mn}\left(\frac{t}{s}\right)\right\|\frac{\d s}{s^{n+\frac{3m}{2}}}\leq c_{mn}t^{1-n-\frac{3m}{2}},\ \ \ c_{mn}>0,\ \ t\geq t_0,\ \sigma\geq \sigma_0.
\end{equation*}
Furthermore, for any $\sigma_0\leq\sigma\leq t$, by \eqref{e83a} and \eqref{e85},
\begin{equation*}
	\left\|\frac{\im}{2\pi}\int_{\Sigma_{\bf R}\setminus\Gamma}{\bf R}_-(\lambda;t,\sigma,z_0)\big({\bf G}_{\bf R}(\lambda;t,\sigma,z_0)-\mathbb{I}\big)\lambda\,\d\lambda\right\|\leq\left(\frac{c}{\sigma}\right)\e^{-c_2\sigma^{\frac{3}{2}}},\ \ \ c>0,
\end{equation*}
and so all together, for all $\sigma_0\leq\sigma\leq t$,
\begin{align}\label{e94}
	{\bf R}_2(t,\sigma)\sim-\frac{7}{48}\begin{bmatrix}0 & 0\\ 1 & 0\end{bmatrix}\frac{z_0}{\sigma\eta(z_0)}
	+&\,\frac{5}{48}\left(\frac{1}{\sigma^2\eta(z_0)}-\frac{3}{5}\frac{z_0\eta'(z_0)}{(\sigma\eta(z_0))^2}\right)\begin{bmatrix}0 & 1\\ 0 & 0\end{bmatrix}
	+\sum_{n=3}^{\infty}\sum_{m=0}^{\infty}\frac{{\bf M}_{mn}(\omega)}{\sigma^n(\pi\sigma)^{\frac{3m}{2}}}\nonumber\\
	&\,+\mathcal{O}\Big(\sigma^{-1}\e^{-c\sigma^{\frac{3}{2}}}\Big).
\end{align}
Consequently, combining \eqref{e93},\eqref{e94} and \eqref{e91},
\begin{eqnarray}
	\ln F_{\sigma}(\sigma)\!\!\!&=&\!\!\!\ln F_{\sigma_0}(\sigma_0)+\int_{\sigma_0}^{\sigma}\frac{\d}{\d\alpha}\ln F_{\alpha}(\alpha)\,\d\alpha\nonumber\\
	&\stackrel{\eqref{e89a}}{=}&\!\!\!\ln F_{\sigma_0}(\sigma_0)-\int_{\sigma_0}^{\sigma}\big(3g_2(\alpha,\alpha)\alpha^2-\frac{1}{16\alpha}+R_1^{21}(\alpha,\alpha)+2\alpha R_2^{12}(\alpha,\alpha)-\frac{\alpha}{2}z_0(\alpha,\alpha)R_1^{12}(\alpha,\alpha)\Big)\,\d\alpha\nonumber\\
	&=&\sigma^{\frac{3}{2}}C(1)+\eta_0+\mathcal{O}\big(\sigma^{-1}\big),\ \ \ \sigma\geq\sigma_0,\label{e95}
\end{eqnarray}
where $\eta_0$ is a numerical constant, independent of $(t,\sigma)$. Together with \eqref{e92}, \eqref{e95} yields thus
\begin{align}
	\ln F_{\sigma}(t)=&\,\ln F_{\sigma}(\sigma)-\int_t^{\sigma}\frac{\partial}{\partial x}\ln F_{\sigma}(x)\,\d x\nonumber\\
	=&\,\sigma^{\frac{3}{2}}C(\omega)+\frac{1}{4}\int_{\omega}^1\left\{\frac{\d}{\d u}D(u)\right\}^2\d u+\eta_1+\mathcal{O}\big(\sigma^{-1}\big),\ \ \ \ t_0\leq t\leq\sigma,\label{e96}
\end{align}
with another numerical constant $\eta_1$. Moving ahead, we now consider the complementary regime $\sigma_0\leq\sigma\leq t$ and use the identity
\begin{equation*}
	\ln F_{\sigma}(t)=\ln F_t(t)-\int_{\sigma}^t\frac{\partial}{\partial y}\ln F_y(t)\,\d y.
\end{equation*}
By Corollary \ref{appcor1} and \ref{appcor2}, for all $\sigma\geq\sigma_0$ and all $t>0$,
\begin{align}\label{e97}
	g_1(t,\sigma)t\sigma-3g_2(t,\sigma)\sigma^2+\frac{1}{16\sigma}=\frac{\partial}{\partial\sigma}\left\{\sigma^{\frac{3}{2}}C(\omega)\right\}+\frac{t}{4\sigma^2}\left\{\frac{\d}{\d\omega}D(\omega)\right\}^2+\frac{1}{16\sigma}+\overline{r}(t,\sigma)
\end{align}
and the error term $\overline{r}(t,\sigma)$ satisfies $\int_{\sigma}^{\infty}|\overline{r}(t,y)|\,\d y\leq c \sigma^{-\frac{3}{2}}$ in the same parameter regime. Moving ahead, \eqref{e91} gets replaced by
\begin{equation*}
	{\bf R}_1(t,\sigma)\sim-\frac{7}{48}\begin{bmatrix}0 & 0\\ 1 & 0\end{bmatrix}\frac{1}{\sigma\eta(z_0)}-\frac{1}{16}\begin{bmatrix}0 & 1\\ 0 & 0\end{bmatrix}\frac{\eta'(z_0)}{(\sigma\eta(z_0))^2}+\sum_{n=3}^{\infty}\sum_{m=1}^{\infty}\frac{{\bf L}_{mn}(\omega)}{\sigma^n(\pi\sigma)^{\frac{3m}{2}}}+\mathcal{O}\Big(\sigma^{-1}\e^{-c\sigma^{\frac{3}{2}}}\Big)
\end{equation*}
once $\sigma_0\leq\sigma\leq t$ and thus we find together with \eqref{e94} that
\begin{equation}\label{e98}
	R_1^{21}(t,\sigma)+2\sigma R_2^{12}(t,\sigma)-\left(t+\frac{\sigma}{2}z_0(t,\sigma)\right)R_1^{12}(t,\sigma)=\frac{1}{16\sigma}+\hat{r}(t,\sigma),\ \ \ \sigma_0\leq\sigma\leq t,
\end{equation}
where $\hat{r}(t,\sigma)$ is such that $\int_{\sigma}^{\infty}|\hat{r}(t,y)|\,\d y\leq c\sigma^{-\frac{3}{2}}$. Hence, combining \eqref{e97} and \eqref{e98},
\begin{equation*}
	\int_{\sigma}^t\frac{\partial}{\partial y}\ln F_y(t)\,\d y\stackrel{\eqref{e89}}{=}t^{\frac{3}{2}}C(1)-\sigma^{\frac{3}{2}}C(\omega)+\frac{1}{4}\int_1^{\omega}\left\{\frac{\d}{\d u}D(u)\right\}^2\d u+\mathcal{O}\big(\sigma^{-\frac{3}{2}}\big),\ \ \ \ \sigma_0\leq\sigma\leq t,
\end{equation*}
and so with \eqref{e95},
\begin{equation}\label{e99}
	\ln F_{\sigma}(t)=\ln F_t(t)-\int_{\sigma}^t\frac{\partial}{\partial y}\ln F_y(t)\,\d y=\sigma^{\frac{3}{2}}C(\omega)-\frac{1}{4}\int_1^{\omega}\left\{\frac{\d}{\d u}D(u)\right\}^2\d u+\eta_2+\mathcal{O}\big(t^{-1}\big)+\mathcal{O}\big(\sigma^{-\frac{3}{2}}\big),
\end{equation}
valid for all $\sigma_0\leq\sigma\leq t$, with another numerical constant $\eta_2$. It now remains to combine \eqref{e99} and \eqref{e96} into the single estimate
\begin{equation}\label{e100}
	\ln F_{\sigma}(t)=\sigma^{\frac{3}{2}}C(\omega)+\frac{1}{4}\int_{\omega}^{\infty}\left\{\frac{\d}{\d u}D(u)\right\}^2\d u+\eta+\mathcal{O}\big(\max\{t^{-1},\sigma^{-1}\}\big),\ \ t\geq t_0,\ \ \sigma\geq\sigma_0,
\end{equation}
where we use that $\frac{\d}{\d x}D(x)$ is square integrable on $\mathbb{R}_+$ and $\eta$ denotes an outstanding numerical constant. In order to determine it, we simply fix $\sigma\geq\sigma_0$ and let $t\rightarrow+\infty$, i.e. we let $\omega\rightarrow+\infty$. In this limit, both leading terms in \eqref{e100} approach zero super-exponentially fast, on the other hand \eqref{fi14} yields $F_{\sigma}(t)\sim 1$ in the same limit, so we must have $\eta=0$. This completes our proof of Theorem \ref{fimain3}.
\end{proof}
\section{Proof of Corollary \ref{fimain4}}\label{sec6}
We now return to our starting point in RHP \ref{master} but set out to investigate the same problem as $t\rightarrow-\infty$. En route we will use the same notations and symbols as in our previous $t\rightarrow+\infty$ analysis.
\subsection{Asymptotics, part $3$ ($t\rightarrow-\infty$ and $0<\sigma t^2\leq c$ with $c>0$)} Inspired by \cite{CCR} we consider the transformation,
\begin{equation}\label{e102}
	{\bf Y}(z;t,\sigma):={\bf X}\left(\frac{z-t}{\sigma};t,\sigma\right){\bf M}^{\textnormal{Ai}}\left(z;t,\frac{\pi}{4}\right),\ \ \ \ z\in\mathbb{C}\setminus\Sigma_{\bf Y},
\end{equation}
with the Airy parametrix ${\bf M}^{\textnormal{Ai}}$ in \eqref{app2} and the contour $\Sigma_{\bf Y}$ in RHP \ref{trafo6}. Using RHP \ref{AiryRHP}, the initial RHP \ref{master} is transformed to the one below:
\begin{problem}\label{trafo6} Let $(t,\sigma)\in(-\infty,0)\times(0,\infty)$. The function ${\bf Y}(z)={\bf Y}(z;t,\sigma)\in\mathbb{C}^{2\times 2}$ defined in \eqref{e102} is uniquely determined by the following properties:
\begin{enumerate}
	\item[(1)] ${\bf Y}(z)$ is analytic for $z\in\mathbb{C}\setminus\Sigma_{\bf Y}$ where $\Sigma_{\bf Y}:=\bigcup_{j=1}^4\Gamma_j\cup\{t\}$ with
	\begin{equation*}
		\Gamma_1:=(t,\infty),\ \ \ \ \ \Gamma_3:=(-\infty,t),\ \ \ \ \ \Gamma_2:=\e^{-\im\frac{\pi}{4}}(-\infty,t),\ \ \ \ \ \Gamma_4:=\e^{\im\frac{\pi}{4}}(-\infty,t),
	\end{equation*}
	denotes the contour shown in Figure \ref{figuretrafo5}. In addition, on each connected component of $\mathbb{C}\setminus\Sigma_{\bf Y}$ there is a continuous extension of ${\bf Y}(z)$ to the closure of the same component.
%
%
%
	\item[(2)] The continuous limiting values ${\bf Y}_{\pm}(z)$ on $\Gamma_j\ni z$ satisfy the jump condition ${\bf Y}_+(z)={\bf Y}_-(z){\bf G}_{\bf Y}(z;t,\sigma)$ where the jump matrix ${\bf G}_{\bf Y}(z;t,\sigma)$ is piecewise given by
	\begin{equation*}
		{\bf G}_{\bf Y}(z;t,\sigma)=\begin{bmatrix}1 & 1-\Phi\big(\frac{z-t}{\sigma}\big)\\ 0 & 1\end{bmatrix},\ \ z\in\Gamma_1;\ \ \ \ {\bf G}_{\bf Y}(z;t,\sigma)=\begin{bmatrix}1 & 0\\ 1 & 1\end{bmatrix},\ \ z\in\Gamma_2\cup\Gamma_4;
	\end{equation*}
	and
	\begin{equation*}
		{\bf G}_{\bf Y}(z;t,\sigma)=\begin{bmatrix}\Phi\big(\frac{z-t}{\sigma}\big) & 1-\Phi\big(\frac{z-t}{\sigma}\big)\smallskip\\ -1-\Phi\big(\frac{z-t}{\sigma}\big) & \Phi\big(\frac{z-t}{\sigma}\big)\end{bmatrix},\ \ z\in\Gamma_3.
	\end{equation*}
	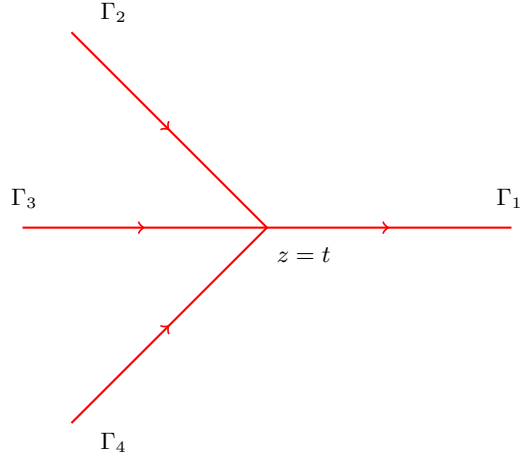
\begin{figure}[tbh]
	\begin{tikzpicture}[xscale=0.65,yscale=0.65]
	\draw [thick, color=red, decoration={markings, mark=at position 0.25 with {\arrow{>}}}, decoration={markings, mark=at position 0.75 with {\arrow{>}}}, postaction={decorate}] (-5,0) -- (5,0);
\node [below] at (0.75,-0.2) {{\small $z=t$}};
\node [right] at (4.5,0.6) {{\small $\Gamma_1$}};
\node [left] at (-4.5,0.6) {{\small $\Gamma_3$}};
\draw [thick, color=red, decoration={markings, mark=at position 0.5 with {\arrow{>}}}, postaction={decorate}] (-4,4) -- (0,0);
\draw [thick, color=red, decoration={markings, mark=at position 0.5 with {\arrow{>}}}, postaction={decorate}] (-4,-4) -- (0,0);
\node [right] at (-3.6,4.4) {{\small $\Gamma_2$}};
\node [right] at (-3.6,-4.4) {{\small $\Gamma_4$}};
\end{tikzpicture}
\caption{The oriented jump contours for ${\bf Y}(z)$ in the complex $z$-plane.}
\label{figuretrafo5}
\end{figure}
	\item[(3)] ${\bf Y}(z)$ is bounded in a neighbourhood of $z=t$.
	\item[(4)] As $z\rightarrow\infty$, provided $z\notin\Sigma_{\bf Y}$,
	\begin{equation}\label{e103}
		{\bf Y}(z)=\Big\{\mathbb{I}+{\bf Y}_1z^{-1}+{\bf Y}_2z^{-2}+\mathcal{O}\big(z^{-3}\big)\Big\}z^{-\frac{1}{4}\sigma_3}\frac{1}{\sqrt{2}}\begin{bmatrix}1 & 1\\ -1 & 1\end{bmatrix}\e^{-\im\frac{\pi}{4}\sigma_3}\e^{-\frac{2}{3}z^{\frac{3}{2}}\sigma_3},
	\end{equation}
	where we choose principal branches for all fractional exponents and the coefficients ${\bf Y}_1,{\bf Y}_2$ equal
	\begin{align*}
		{\bf Y}_1={\bf Y}_1(t,\sigma)=&\,\,\sigma{\bf X}_1(t,\sigma)-\frac{7}{48}\begin{bmatrix}0 & 0\\ 1 & 0\end{bmatrix},\\ 
		{\bf Y}_2={\bf Y}_2(t,\sigma)=&\,\,\sigma^2{\bf X}_2(t,\sigma)+t\sigma{\bf X}_1(t,\sigma)
		-\frac{7\sigma}{48}{\bf X}_1(t,\sigma)\begin{bmatrix}0 & 0\\ 1 & 0\end{bmatrix}+\frac{5}{48}\begin{bmatrix}0 & 1\\ 0 & 0\end{bmatrix}.
	\end{align*}
\end{enumerate}
\end{problem}
Next, with \eqref{e42} in mind, we center RHP \ref{trafo6} at $z=0$. In detail, we consider the transformation
\begin{equation}\label{e104}
	{\bf T}(z;t,\sigma):={\bf Y}(z+t;t,\sigma),\ \ z\in\mathbb{C}\setminus(\Sigma_{\bf Y}-t),
\end{equation}
and thus transform RHP \ref{trafo6} to the one below:
\begin{problem}\label{trafo7} Let $(t,\sigma)\in(-\infty,0)\times(0,\infty)$. The function ${\bf T}(z)={\bf T}(z;t,\sigma)\in\mathbb{C}^{2\times 2}$ defined in \eqref{e104} is uniquely determined by the following four properties:
\begin{enumerate}
	\item[(1)] ${\bf T}(z)$ is analytic for $z\in\mathbb{C}\setminus\Sigma_{\bf T}$. The contour $\Sigma_{\bf T}:=\Sigma_{\bf Y}-t$ is shown in Figure \ref{figuretrafo6} and
	on each connected component of $\mathbb{C}\setminus\Sigma_{\bf T}$ there is a continuous extension of ${\bf T}(z)$ to the closure of the same component.
	\item[(2)] The continuous limiting values ${\bf T}_{\pm}(z)$ on $\Sigma_{\bf T}\ni z$ satisfy ${\bf T}_+(z)={\bf T}_-(z){\bf G}_{\bf T}(z;t,\sigma)$ where the jump matrix ${\bf G}_{\bf T}(z;t,\sigma)$ is of the form, 
	\begin{align*}
		{\bf G}_{\bf T}(z;t,\sigma)=&\,\begin{bmatrix} 1 & 1-\Phi(\frac{z}{\sigma})\\ 0 & 1\end{bmatrix},\ \ z\in\Sigma_4;\ \ \ \ \ {\bf G}_{\bf T}(z;t,\sigma)=\begin{bmatrix}\Phi(\frac{z}{\sigma}) & 1-\Phi(\frac{z}{\sigma})\smallskip\\
		-1-\Phi(\frac{z}{\sigma}) & \Phi(\frac{z}{\sigma})\end{bmatrix},\ \ z\in\Sigma_1;\\
		&\hspace{1cm}{\bf G}_{\bf T}(z;t,\sigma)=\begin{bmatrix}1 & 0\\ 1 & 1\end{bmatrix},\ \ z\in\Sigma_2\cup\Sigma_3.
	\end{align*}
	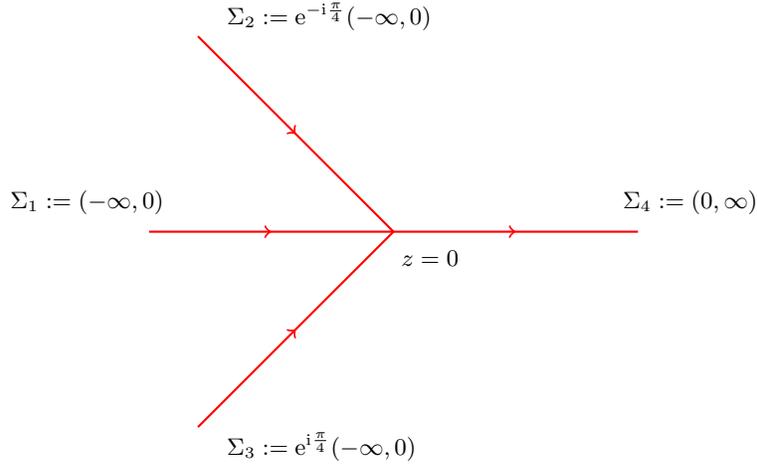
\begin{figure}[tbh]
	\begin{tikzpicture}[xscale=0.65,yscale=0.65]
	\draw [thick, color=red, decoration={markings, mark=at position 0.25 with {\arrow{>}}}, decoration={markings, mark=at position 0.75 with {\arrow{>}}}, postaction={decorate}] (-5,0) -- (5,0);
\node [below] at (0.75,-0.2) {{\small $z=0$}};
\node [right] at (4.5,0.6) {{\small $\Sigma_4:=(0,\infty)$}};
\node [left] at (-4.5,0.6) {{\small $\Sigma_1:=(-\infty,0)$}};
\draw [thick, color=red, decoration={markings, mark=at position 0.5 with {\arrow{>}}}, postaction={decorate}] (-4,4) -- (0,0);
\draw [thick, color=red, decoration={markings, mark=at position 0.5 with {\arrow{>}}}, postaction={decorate}] (-4,-4) -- (0,0);
\node [right] at (-3.6,4.4) {{\small $\Sigma_2:=\e^{-\im\frac{\pi}{4}}(-\infty,0)$}};
\node [right] at (-3.6,-4.4) {{\small $\Sigma_3:=\e^{\im\frac{\pi}{4}}(-\infty,0)$}};
\end{tikzpicture}
\caption{The oriented jump contour $\Sigma_{\bf T}$ in the complex $z$-plane.}
\label{figuretrafo6}
\end{figure}
	\item[(3)] ${\bf T}(z)$ is bounded in a neighbourhood of $z=0$.
	\item[(4)] As $z\rightarrow\infty$, valid in a full vicinity of infinity off the jump contour $\Sigma_{\bf T}$,
	\begin{equation}\label{e105}
		{\bf T}(z)=\Big\{\mathbb{I}+{\bf T}_1z^{-1}+{\bf T}_2z^{-2}+\mathcal{O}\big(z^{-3}\big)\Big\}z^{-\frac{1}{4}\sigma_3}\frac{1}{\sqrt{2}}\begin{bmatrix}1 & 1\\ -1 & 1\end{bmatrix}\e^{-\im\frac{\pi}{4}\sigma_3}\e^{-\frac{2}{3}(z+t)^{\frac{3}{2}}\sigma_3},
	\end{equation}
	using principal branches throughout and where the coefficients ${\bf T}_1,{\bf T}_2$, equal, compare \eqref{e103},
	\begin{align*}
		{\bf T}_1=&\,{\bf T}_1(t,\sigma)={\bf Y}_1(t,\sigma)-\frac{t}{4}\sigma_3,\\
		{\bf T}_2=&\,{\bf T}_2(t,\sigma)={\bf Y}_2(t,\sigma)-t{\bf Y}_1(t,\sigma)-\frac{t}{4}{\bf Y}_1(t,\sigma)\sigma_3+\frac{t^2}{32}\begin{bmatrix}5 & 0\\ 0 & -3\end{bmatrix}.
	\end{align*}
\end{enumerate}
\end{problem}
Moving ahead, we employ a $g$-function transformation similar to the one from \cite[$(3.8)$]{CIK} with $k=0$ or, equivalently, the one from \cite[$(3.1)$]{Bo1} with $V=0$: let 
\begin{equation*}
	g(z)=g(z;t):=\frac{2}{3}z^{\frac{1}{2}}\left(z+\frac{3t}{2}\right),\ \ z\in\mathbb{C}\setminus(-\infty,0],
\end{equation*}
be defined and analytic for $z\in\mathbb{C}\setminus(-\infty,0]$ such that $z^{\frac{1}{2}}=\sqrt{z}>0$ for $z>0$. Now set
\begin{equation}\label{e106}
	{\bf S}(z;t,\sigma):=\begin{bmatrix}1 & 0\\
	-\frac{1}{4}t^2& 1\end{bmatrix}{\bf T}(z;t,\sigma)\e^{g(z;t)\sigma_3},\ \ z\in\mathbb{C}\setminus\Sigma_{\bf T},
\end{equation}
and summarize the defining properties of ${\bf S}(z)$ below:
\begin{problem}\label{trafo8} Let $(t,\sigma)\in(-\infty,0)\times(0,\infty)$. The function ${\bf S}(z)={\bf S}(z;t,\sigma)\in\mathbb{C}^{2\times 2}$ defined in \eqref{e106} is uniquely determined by the following four properties:
\begin{enumerate}
	\item[(1)] ${\bf S}(z)$ is analytic for $z\in\mathbb{C}\setminus\Sigma_{\bf T}$ and on each connected component of $\mathbb{C}\setminus\Sigma_{\bf T}$ there is a continuous extension of ${\bf S}(z)$ to the closure of the same component, compare Figure \ref{figuretrafo6}.
%
%
	\item[(2)] The continuous limiting values ${\bf S}_{\pm}(z)$ on $\Sigma_{\bf T}\ni z$ satisfy ${\bf S}_+(z)={\bf S}_-(z){\bf G}_{\bf S}(z;t,\sigma)$ with ${\bf G}_{\bf S}(z;t,\sigma)$ equal to
	\begin{align*}
		{\bf G}_{\bf S}(z;t,\sigma)=&\,\begin{bmatrix}1 & (1-\Phi(\frac{z}{\sigma}))\e^{-2g(z;t)}\\ 0 & 1\end{bmatrix},\ \ z\in\Sigma_4;\\
		{\bf G}_{\bf S}(z;t,\sigma)=&\,\begin{bmatrix}\Phi(\frac{z}{\sigma})\e^{\Pi(z;t)} & 1-\Phi(\frac{z}{\sigma})\smallskip\\
		-1-\Phi(\frac{z}{\sigma}) & \Phi(\frac{z}{\sigma})\e^{-\Pi(z;t)}\end{bmatrix},\ \ z\in\Sigma_1;\\
		{\bf G}_{\bf S}(z;t,\sigma)=&\,\begin{bmatrix}1 & 0\\ \e^{2g(z;t)} & 1\end{bmatrix},\ \ z\in\Sigma_2\cup\Sigma_3;
	\end{align*}
	where we abbreviate 
	\begin{equation*}
		\Pi(z;t):=g_+(z;t)-g_-(z;t)=\frac{4\im}{3}\sqrt{|z|}\left(z+\frac{3t}{2}\right),\ \ \ z<0.
	\end{equation*}
	\item[(3)] ${\bf S}(z)$ is bounded in a neighbourhood of $z=0$.
	\item[(4)] As $z\rightarrow\infty$, valid in a full vicinity of infinity off the jump contour $\Sigma_{\bf T}$
	\begin{equation}\label{e107}
		{\bf S}(z)=\Big\{\mathbb{I}+{\bf S}_1z^{-1}+{\bf S}_2z^{-2}+\mathcal{O}\big(z^{-3}\big)\Big\}z^{-\frac{1}{4}\sigma_3}\frac{1}{\sqrt{2}}\begin{bmatrix}1 & 1\\ -1 & 1\end{bmatrix}\e^{-\im\frac{\pi}{4}\sigma_3},
	\end{equation}
	with principal branches throughout and coefficients ${\bf S}_1,{\bf S}_2$ given by, see \eqref{e105},
	\begin{align*}
		{\bf S}_1={\bf S}_1(t,\sigma)=&\,\begin{bmatrix}1 & 0\smallskip\\ -\frac{1}{4}t^2 & 1\end{bmatrix}{\bf T}_1(t,\sigma)\begin{bmatrix}1 & 0\smallskip\\ \frac{1}{4}t^2 & 1\end{bmatrix}+\begin{bmatrix}\frac{1}{32}t^4 & \frac{1}{4}t^2\smallskip\\ -\frac{1}{24}t^3-\frac{1}{192}t^6 & -\frac{1}{32}t^4\end{bmatrix}\\
		{\bf S}_2={\bf S}_2(t,\sigma)=&\,\begin{bmatrix}1 & 0\smallskip\\ -\frac{1}{4}t^2 & 1\end{bmatrix}{\bf T}_2(t,\sigma)\begin{bmatrix}1 & 0\smallskip\\ \frac{1}{4}t^2 & 1\end{bmatrix}+\begin{bmatrix}1 & 0\smallskip\\ -\frac{1}{4}t^2 & 1\end{bmatrix}{\bf T}_1(t,\sigma)\begin{bmatrix}\frac{1}{32}t^4 & \frac{1}{4}t^2\smallskip\\ -\frac{1}{24}t^3+\frac{1}{384}t^6 & \frac{1}{32}t^4\end{bmatrix}\\
		&\hspace{2cm}+\begin{bmatrix}-\frac{1}{96}(t^5-\frac{1}{64}t^8) & -\frac{1}{24}(t^3-\frac{1}{16}t^6)\smallskip\\
		\frac{1}{64}(t^4+\frac{1}{12}t^7-\frac{1}{480}t^{10}) & -\frac{1}{2048}t^8\end{bmatrix}.
	\end{align*}
\end{enumerate}
\end{problem}
We proceed with the following inequalities for the nontrivial entries of ${\bf G}_{\bf S}(z;t,\sigma)$.
\begin{prop}\label{prop11} We have for any $t<0$ and $\sigma>0$, 
\begin{equation}\label{e108}
	\Big|\e^{2g(z;t)}\Big|\leq \e^{\frac{1}{2}t\sqrt{|z|}},\ \ z\in\Sigma_2\cup\Sigma_3;\ \ \ \ \ \ \big|\Phi\Big(\frac{z}{\sigma}\Big)\big|\leq\frac{1}{2}\e^{-(z/\sigma)^2},\ \ z\in\Sigma_1;
\end{equation}
followed by
\begin{equation}\label{e109}
	\Big|\big(1-\Phi\Big(\frac{z}{\sigma}\Big)\big)\e^{-2g(z;t)}\Big|\leq\frac{1}{2}\exp\left[-\frac{\sqrt{z}}{\sigma^2}\left(z^{\frac{3}{2}}+\frac{4}{3}z\sigma^2+2t\sigma^2\right)\right],\ \ z\in\Sigma_4.
\end{equation}
\end{prop}
\begin{proof} All three estimates are easy consequences of the explicit formula for $g(z;t)$ and estimate \eqref{e42}.
\end{proof}
By \eqref{e108} and \eqref{e109}, ${\bf G}_{\bf S}$ in RHP \ref{trafo8} is asymptotically localized near $z=0$ and along the line segment $\Sigma_1$, as $t\rightarrow-\infty$ and $0<\sigma\leq c|t|^{-\frac{1}{2}-\epsilon}$ for any $c,\epsilon>0$ (so in particular for the values of $\sigma$ with $0<\sigma t^2\leq c$ for any fixed $c>0$).
We thus proceed with the necessary local analysis and the following two model problems:
\begin{problem}\label{para3} Find ${\bf P}^{(\infty)}(z)\in\mathbb{C}^{2\times 2}$ such that
\begin{enumerate}
	\item[(1)] ${\bf P}^{(\infty)}(z)$ is analytic for $z\in\mathbb{C}\setminus(-\infty,0]$.
	\item[(2)] ${\bf P}^{(\infty)}(z)$ attains square-integrable limiting values on $(-\infty,0]$ and those are related by the jump condition
	\begin{equation*}
		{\bf P}^{(\infty)}_+(z)={\bf P}^{(\infty)}_-(z)\begin{bmatrix}0 & 1\\ -1 & 0\end{bmatrix},\ \ \ z\in(-\infty,0).
	\end{equation*}
	\item[(3)] As $z\rightarrow\infty$,
	\begin{equation*}
		{\bf P}^{(\infty)}(z)=\Big\{\mathbb{I}+\mathcal{O}\big(z^{-\infty}\big)\Big\}z^{-\frac{1}{4}\sigma_3}\frac{1}{\sqrt{2}}\begin{bmatrix}1 & 1\\ -1 & 1\end{bmatrix}\e^{-\im\frac{\pi}{4}\sigma_3},
	\end{equation*}
	with the principal branch for $z^{\alpha}:\mathbb{C}\setminus(-\infty,0]\rightarrow\mathbb{C}$ such that $z^{\alpha}>0$ when $z>0$.
\end{enumerate}
\end{problem}
As in \eqref{e77}, a direct computation verifies that
\begin{equation}\label{e110}
	{\bf P}^{(\infty)}(z)=z^{-\frac{1}{4}\sigma_3}\frac{1}{\sqrt{2}}\begin{bmatrix}1 & 1\\ -1 & 1\end{bmatrix}\e^{-\im\frac{\pi}{4}\sigma_3},\ \ z\in\mathbb{C}\setminus(-\infty,0],
\end{equation}
has the properties listed in RHP \ref{para3}. Next, in a vicinity of $z=0$, we require a solution to the following model problem.
\begin{problem}\label{para4} Let $-t\geq t_0$ be sufficiently large and $\sigma>0$ such that $0<\sigma t^2\leq c$ for any fixed $c>0$. Now find ${\bf P}^{(0)}(z)={\bf P}^{(0)}(z;t,\sigma)\in\mathbb{C}^{2\times 2}$ such that
\begin{enumerate}
	\item[(1)] ${\bf P}^{(0)}(z)$ is analytic for $z\in\mathbb{D}_1(0)\setminus\Sigma_{\bf T}$ with the open disk $\mathbb{D}_r(z_0):=\{z\in\mathbb{C}:\,|z-z_0|<r\}$.
	\item[(2)] ${\bf P}^{(0)}(z)$ has the following local jump behavior, see Figure \ref{figuretrafo6} for contour orientations:
	\begin{align*}
		{\bf P}_+^{(0)}(z)=&\,{\bf P}_-^{(0)}(z)\begin{bmatrix}1 & (1-\Phi(\frac{z}{\sigma}))\e^{-2g(z;t)}\smallskip\\ 0 & 1\end{bmatrix},\ \ z\in\Sigma_4\cap\mathbb{D}_1(0);\\
		{\bf P}_+^{(0)}(z)=&\,{\bf P}_-^{(0)}(z)\begin{bmatrix}\Phi(\frac{z}{\sigma})\e^{\Pi(z;t)} & 1-\Phi(\frac{z}{\sigma})\smallskip\\
		-1-\Phi(\frac{z}{\sigma}) & \Phi(\frac{z}{\sigma})\e^{-\Pi(z;t)}\end{bmatrix},\ \ z\in\Sigma_1\cap\mathbb{D}_1(0);\\
		{\bf P}_+^{(0)}(z)=&\,{\bf P}_-^{(0)}(z)\begin{bmatrix}1 & 0\smallskip\\ \e^{2g(z;t)} & 1\end{bmatrix},\ \ z\in(\Sigma_2\cup\Sigma_3)\cap\mathbb{D}_1(0).
	\end{align*}
	\item[(3)] When $-t\geq t_0$ and $\sigma>0$ such that $0<\sigma t^2\leq c$ for any fixed $c>0$, we have the following asymptotic matching between ${\bf P}^{(0)}(z)$ and ${\bf P}^{(\infty)}(z)$,
	\begin{equation}\label{e111}
		{\bf P}^{(0)}(z)=\Big\{\mathbb{I}+\mathcal{O}\big(\sigma t\big)+\mathcal{O}\big(t^{-1}\big)\Big\}{\bf P}^{(\infty)}(z),
	\end{equation}
	which holds uniformly for $0<\frac{1}{4}\leq|z|\leq\frac{3}{4}<1$. Here, ${\bf M}_1^{\textnormal{PV}}(x)$ is as in RHP \ref{genBesselRHP}, condition $(4)$.
\end{enumerate}
\end{problem}
In order to solve Problem \ref{para4} we consider the conformal change of coordinates
\begin{equation*}
	\zeta(z):=\frac{z}{\sigma},\ \ \ z\in\mathbb{D}_1(0),
\end{equation*}
and in turn the function
\begin{equation}\label{e112}
	{\bf P}^{(0)}(z)=\sigma^{-\frac{1}{4}\sigma_3}{\bf M}^{\textnormal{PV}}\left(\zeta(z);-t\sqrt{\sigma}-\frac{2}{3}z\sqrt{\sigma}\right)\e^{g(z;t)\sigma_3},\ \ \ z\in\mathbb{D}_1(0)\setminus\Sigma_{\bf T},
\end{equation}
defined in terms of the generalized Bessel parametrix ${\bf M}^{\textnormal{PV}}(\zeta;x)$, the unique solution of RHP \ref{genBesselRHP}. Note that \eqref{e112} is well-defined in the indicated domain for all $-t$ sufficiently large and $\sigma>0$ such that $0<\sigma t^2\leq c$, compare Lemma \ref{genBessellem}, condition (i). Moreover, the same function is analytic in $z\in\mathbb{D}_1(0)\setminus\Sigma_{\bf T}$ by RHP \ref{genBesselRHP} and it attains continuous limiting values on $\Sigma_{\bf T}$. Using RHP \ref{genBesselRHP}, condition $(2)$, one then checks that these limiting values are as desired, so we are left to establish that \eqref{e112} obeys \eqref{e111}. To that end, if $-t\geq t_0$ and $\sigma>0$ are such that $0<\delta\leq\sigma t^2\leq c$ is bounded away from zero, then \eqref{e111} is immediate from \eqref{e110}, \eqref{e112} and \eqref{app4a}, given that \eqref{app4a} holds uniformly in any small disk centered on the positive half ray and given that $-t\sqrt{\sigma}$ is bounded away from zero, compare Lemma \ref{genBessellem}. In fact, for those values of $(t,\sigma)$ we have
\begin{equation}\label{e113}
		{\bf P}^{(0)}(z)=\Big\{\mathbb{I}+\sigma^{-\frac{1}{4}\sigma_3}{\bf M}_1^{\textnormal{PV}}\left(-t\sqrt{\sigma}-\frac{2}{3}z\sqrt{\sigma}\right)\sigma^{\frac{1}{4}\sigma_3}\Big(\frac{z}{\sigma}\Big)^{-1}+\mathcal{O}\big(\sigma^{\frac{3}{2}}\big)\Big\}{\bf P}^{(\infty)}(z),
\end{equation}
uniformly for $0<\frac{1}{4}\leq|z|\leq\frac{3}{4}<1$. If however $-t\geq t_0$ and $\sigma>0$ are such that $0<\sigma t^2\leq c$ tends to zero, then \eqref{app4a} is no longer applicable since $-t\sqrt{\sigma}\downarrow 0$. Instead we now use \eqref{app4b} and RHP \ref{BesselRHP}, condition $(4)$, and obtain
\begin{align}
	&\,{\bf P}^{(0)}(z)=\bigg\{\mathbb{I}+\frac{\sigma}{2z}\left(-t-\frac{2z}{3}\right)^{\frac{1}{2}\sigma_3}\begin{bmatrix}\frac{3}{8} & 1\smallskip\\ -\frac{9}{64} & -\frac{3}{8}\end{bmatrix}\left(-t-\frac{2z}{3}\right)^{-\frac{1}{2}\sigma_3}\int_{-\infty}^{\infty}\big(\chi_{[0,\infty)}(u)-\Phi(u)\big)\,\d u\label{e114}\\
	&\,+\left(-t-\frac{2z}{3}\right)^{\frac{1}{2}\sigma_3}\begin{bmatrix}\frac{9}{128} & -\frac{1}{8}\smallskip\\ -\frac{3}{1024} & -\frac{9}{128}\end{bmatrix}\left(-t-\frac{2z}{3}\right)^{-\frac{1}{2}\sigma_3}\frac{1}{z}\left(-t-\frac{2z}{3}\right)^{-2}+\mathcal{O}\big(t^3\sigma^2\big)+\mathcal{O}(\sigma)\bigg\}{\bf P}^{(\infty)}(z),\nonumber
\end{align}
uniformly for $0<\frac{1}{4}\leq|z|\leq\frac{3}{4}<1$. Combining \eqref{e113} and \eqref{e114} we arrive at \eqref{e111}, i.e. \eqref{e112} constitutes a solution of RHP \ref{para4}. Having completed the necessary local analysis with \eqref{e110} and \eqref{e112} we now compare both functions to ${\bf S}(z)$ in the following way. Define
\begin{equation}\label{e115}
	{\bf R}(z;t,\sigma):={\bf S}(z;t,\sigma)\begin{cases}\big({\bf P}^{(0)}(z;t,\sigma)\big)^{-1},&z\in\mathbb{D}_{\frac{1}{2}}(0)\setminus\Sigma_{\bf T}\smallskip\\ \big({\bf P}^{(\infty)}(z)\big)^{-1},&z\notin\overline{\mathbb{D}_{\frac{1}{2}}(0)}\setminus\Sigma_{\bf T}\end{cases},
\end{equation}
and recall RHP \ref{trafo8}, \ref{para3} and \ref{para4}. The defining properties of ${\bf R}(z)$ are as follows:
\begin{problem}\label{trafo9} Let $-t\geq t_0$ and $\sigma>0$ such that $0<\sigma t^2\leq c$. The function ${\bf R}(z)={\bf R}(z;t,\sigma)\in\mathbb{C}^{2\times 2}$ defined in \eqref{trafo9} is uniquely determined by the following properties:
\begin{enumerate}
	\item[(1)] ${\bf R}(z)$ is analytic for $z\in\mathbb{C}\setminus\Sigma_{\bf R}$ where
	\begin{equation*}
		\Sigma_{\bf R}:=\partial\mathbb{D}_{\frac{1}{2}}(0)\cup\left[\frac{1}{2},\infty\right)\cup\left(-\infty,-\frac{1}{2}\right]\cup\big(\Sigma_2\setminus\mathbb{D}_{\frac{1}{2}}(0)\big)\cup\big(\Sigma_3\setminus\mathbb{D}_{\frac{1}{2}}(0)\big)
	\end{equation*}
	is shown in Figure \ref{figuretrafo7}. Moreover, on each connected component of $\mathbb{C}\setminus\Sigma_{\bf R}$ there is a continuous extension of ${\bf R}(z)$ to the closure of the same component.
%
%
	\item[(2)] The continuous limiting values ${\bf R}_{\pm}(z)$ on $\Sigma_{\bf R}\ni z$ obey the constraint ${\bf R}_+(z)={\bf R}_-(z){\bf G}_{\bf R}(z;t,\sigma)$ where, for $z\in[\frac{1}{2},\infty)$,
	\begin{equation*}
		{\bf G}_{\bf R}(z;t,\sigma)={\bf P}^{(\infty)}(z)\begin{bmatrix}1 & (1-\Phi(\frac{z}{\sigma}))\e^{-2g(z;t)}\smallskip\\ 0 & 1\end{bmatrix}\big({\bf P}^{(\infty)}(z)\big)^{-1},
	\end{equation*}
	and for $z\in(-\infty,-\frac{1}{2}]$,
	\begin{equation*}
		{\bf G}_{\bf R}(z;t,\sigma)={\bf P}^{(\infty)}_-(z)\begin{bmatrix}1-\Phi(\frac{z}{\sigma}) & -\Phi(\frac{z}{\sigma})\e^{\Pi(z;t)}\smallskip\\ \Phi(\frac{z}{\sigma})\e^{-\Pi(z;t)} & 1+\Phi(\frac{z}{\sigma})\end{bmatrix}\big({\bf P}^{(\infty)}_-(z)\big)^{-1}.
	\end{equation*}
	Additionally, for $z\in(\Sigma_2\cup\Sigma_3)\setminus\mathbb{D}_{\frac{1}{2}}(0)$,
	\begin{equation*}
		{\bf G}_{\bf R}(z;t,\sigma)={\bf P}^{(\infty)}(z)\begin{bmatrix}1 & 0\smallskip\\ \e^{2g(z;t)} & 1\end{bmatrix}\big({\bf P}^{(\infty)}(z)\big)^{-1},
	\end{equation*}
	and for $z\in\partial\mathbb{D}_{\frac{1}{2}}(0)$,
	\begin{equation*}
		{\bf G}_{\bf R}(z;t,\sigma)={\bf P}^{(0)}(z;t,\sigma)\big({\bf P}^{(\infty)}(z)\big)^{-1}.
	\end{equation*}
	\item[(3)] As $z\rightarrow\infty$ and $z\notin\Sigma_{\bf R}$,
	\begin{equation*}
		{\bf R}(z)=\mathbb{I}+{\bf R}_1z^{-1}+{\bf R}_2z^{-2}+\mathcal{O}\big(z^{-3}\big),
	\end{equation*}
	with ${\bf R}_1,{\bf R}_2$ equal to, compare \eqref{e107} and \eqref{e110},
	\begin{equation*}
		{\bf R}_1={\bf R}_1(t,\sigma)={\bf S}_1(t,\sigma),\ \ \ \ {\bf R}_2={\bf R}_2(t,\sigma)={\bf S}_2(t,\sigma).
	\end{equation*}
\end{enumerate}
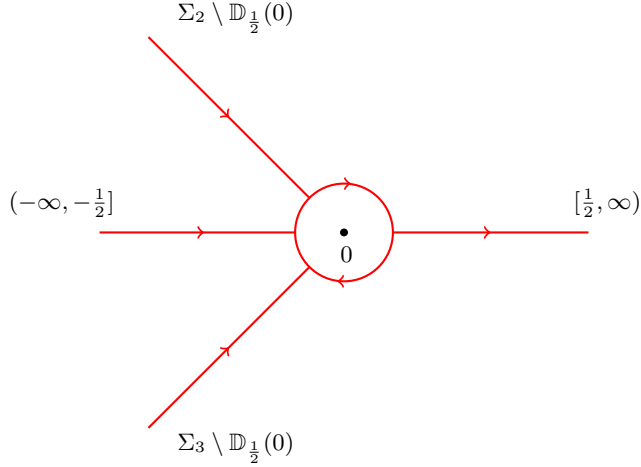
\begin{figure}[tbh]
	\begin{tikzpicture}[xscale=0.65,yscale=0.65]
	\draw [thick, color=red, decoration={markings, mark=at position 0.5 with {\arrow{>}}}, postaction={decorate}] (1,0) -- (5,0);
	\draw [thick, color=red, decoration={markings, mark=at position 0.5 with {\arrow{<}}}, postaction={decorate}] (-1,0) -- (-5,0);
	\draw [thick, color=red, decoration={markings, mark=at position 0.25 with {\arrow{<}}}, decoration={markings, mark=at position 0.75 with {\arrow{<}}}, postaction={decorate}] (0,0) circle [radius=1];
\node [below] at (0.05,-0.1) {{\small $0$}};
\node [right] at (4.5,0.6) {{\small $[\frac{1}{2},\infty)$}};
\node [left] at (-4.5,0.6) {{\small $(-\infty,-\frac{1}{2}]$}};
\draw [thick, color=red, decoration={markings, mark=at position 0.5 with {\arrow{>}}}, postaction={decorate}] (-4,4) -- (-0.7071067810,0.7071067810);
\draw [thick, color=red, decoration={markings, mark=at position 0.5 with {\arrow{>}}}, postaction={decorate}] (-4,-4) -- (-0.7071067810,-0.7071067810);
\node [right] at (-3.6,4.4) {{\small $\Sigma_2\setminus\mathbb{D}_{\frac{1}{2}}(0)$}};
\node [right] at (-3.6,-4.4) {{\small $\Sigma_3\setminus\mathbb{D}_{\frac{1}{2}}(0)$}};
\draw [fill, color=black] (0,0) circle [radius=0.07];
\end{tikzpicture}
\caption{The oriented jump contours for ${\bf R}(z)$ in the complex $z$-plane.}
\label{figuretrafo7}
\end{figure}
\end{problem}
\begin{rem}\label{rem10} By construction, compare RHP \ref{para4}, ${\bf R}(z)$ is analytic in the disk $\mathbb{D}_{\frac{1}{2}}(0)$.
\end{rem}
By Proposition \ref{prop11} and \eqref{e111}, as soon as $t\rightarrow-\infty$ and $\sigma>0$ is such that $0<\sigma t^2\leq c$, the jump matrix in the ${\bf R}$-RHP is close to the identity matrix on all of $\Sigma_{\bf R}$. In detail, we have the following small norm estimates.
\begin{prop}\label{prop12} For any $c>0$ there exist $t_0,d>0$ such that
\begin{equation}\label{e116}
	\|{\bf G}_{\bf R}(\cdot;t,\sigma)-\mathbb{I}\|_{L^{\infty}(\Sigma_{\bf R})}\leq d\max\big\{\sigma t,t^{-1}\big\},\ \ \ \ \|{\bf G}_{\bf R}(\cdot;t,\sigma)-\mathbb{I}\|_{L^2(\Sigma_{\bf R})}\leq d\max\big\{\sigma t,t^{-1}\big\}
\end{equation}
for all $-t\geq t_0$ and all $\sigma>0$ such that $0<\sigma t^2\leq c$.
\end{prop}
\begin{proof} Since ${\bf P}_-^{(\infty)}(z)$ is bounded on the four rays extending to infinity and independent of $(t,\sigma)$, the same four rays yield exponentially small contributions of ${\bf G}_{\bf R}(z;t,\sigma)$ to the identity by \eqref{e42} and \eqref{e108},\eqref{e109}. Hence, \eqref{e116} is a direct consequence of \eqref{e111}.
\end{proof}
In turn, \eqref{e116} implies that RHP \ref{trafo9} is asymptotically solvable, cf. \cite{DZ}:
\begin{theo}\label{theo5} Let $c>0$. There exist $t_0,d>0$ such that the RHP for ${\bf R}(z)$ defined in \eqref{e115} is uniquely solvable in $L^2(\Sigma_{\bf R})$ for all $-t\geq t_0$ and all $\sigma>0$ such that $0<\sigma t^2\leq c$. We can compute the solution of the same problem iteratively via the integral equation
\begin{equation}\label{e117}
	{\bf R}(z;t,\sigma)=\mathbb{I}+\frac{1}{2\pi\im}\int_{\Sigma_{\bf R}}{\bf R}_-(\lambda;t,\sigma)\big({\bf G}_{\bf R}(\lambda;t,\sigma)-\mathbb{I}\big)\frac{\d\lambda}{\lambda-z},\ \ \ \ z\in\mathbb{C}\setminus\Sigma_{\bf R},
\end{equation}
using that, for all $-t\geq t_0$ and $\sigma>0$ such that $0<\sigma t^2\leq c$, we have
\begin{equation}\label{e118}
	\|{\bf R}_-(\cdot;t,\sigma)-\mathbb{I}\|_{L^2(\Sigma_{\bf R})}\leq d\max\big\{\sigma t,t^{-1}\big\}.
\end{equation}
\end{theo}
At this point we are prepared to prove Corollary \ref{fimain4}.
\begin{proof}[Proof of Corollary \ref{fimain4}] We choose $c>0$ and $-t\geq t_0$ together with $0<\sigma t^2\leq c$ as in Theorem \ref{theo5}. After that we recall the explicit and invertible sequence of transformations
\begin{equation*}
	{\bf X}(z)\stackrel{\eqref{e102}}{\mapsto}{\bf Y}(z)\stackrel{\eqref{e104}}{\mapsto}{\bf T}(z)\stackrel{\eqref{e106}}{\mapsto}{\bf S}(z)\stackrel{\eqref{e115}}{\mapsto}{\bf R}(z),
\end{equation*}
that converts the differential identities in \eqref{e52} into the following two exact formul\ae,
\begin{equation}\label{e120}
	\frac{\partial}{\partial t}\ln F_{\sigma}(t)=\frac{1}{4}t^2-R_1^{12}(t,\sigma),
\end{equation}
and, using en route that $\tr{\bf R}_1=0$ because $\det{\bf R}(z)\equiv 1$,
\begin{equation}\label{e121}
	\frac{\partial}{\partial\sigma}\ln F_{\sigma}(t)=\frac{1}{16\sigma}-\frac{1}{\sigma}R_1^{21}(t,\sigma)-\frac{t}{2\sigma}R_1^{12}(t,\sigma)-\frac{2}{\sigma}R_2^{12}(t,\sigma).
\end{equation}
Next, by \eqref{e117} with $c>0$,
\begin{equation*}
	{\bf R}_k(t,\sigma)=\frac{\im}{2\pi}\ointclockwise_{|\lambda|=\frac{1}{2}}{\bf R}_-(\lambda;t,\sigma)\big({\bf G}_{\bf R}(\lambda;t,\sigma)-\mathbb{I}\big)\,\lambda^{k-1}\,\d\lambda+\mathcal{O}\left(\e^{-c|t|}\right),\ \ \ \ \ k\in\{1,2\},
\end{equation*}
so in particular by \eqref{e114} and \eqref{e118}, as $t\rightarrow-\infty$ and $\sigma>0$ is such that $-t\sqrt{\sigma}\downarrow 0$,
\begin{align*}
	{\bf R}_1(t,\sigma)=&\,\frac{\sigma}{2}(-t)^{\frac{1}{2}\sigma_3}\begin{bmatrix}\frac{3}{8} & 1\smallskip\\ -\frac{9}{64} & -\frac{3}{8}\end{bmatrix}(-t)^{-\frac{1}{2}\sigma_3}\int_{-\infty}^{\infty}\big(\chi_{[0,\infty)}(y)-\Phi(y)\big)\,\d y\\
	&\hspace{1cm}+(-t)^{\frac{1}{2}\sigma_3}\begin{bmatrix}\frac{9}{128} & -\frac{1}{8}\smallskip\\ -\frac{3}{1024} & -\frac{9}{128}\end{bmatrix}(-t)^{-\frac{1}{2}\sigma_3}t^{-2}+\mathcal{O}\big(t^3\sigma^2\big)+\mathcal{O}(\sigma)+\mathcal{O}\big(\max\big\{(\sigma t)^2,t^{-2}\big\}\big),\\
	{\bf R}_2(t,\sigma)=&\,\mathcal{O}\big(t^3\sigma^2\big)+\mathcal{O}(\sigma)+\mathcal{O}\big(\max\big\{(\sigma t)^2,t^{-2}\big\}\big).
\end{align*}
In turn, by \eqref{e120} and \eqref{e121} after indefinite integration, as $t\rightarrow-\infty$, with a numerical constant $\eta_1\in\mathbb{R}$,
\begin{equation}\label{e122}
	\ln F_{\sigma}(t)=\frac{t^3}{12}-\frac{1}{8}\ln|t|+\eta_1+\mathcal{O}\big(\max\big\{\sigma t^2,t^{-1}\big\}\big),\ \ t\rightarrow-\infty,\ \sigma>0:\ -t\sqrt{\sigma}\downarrow 0.
\end{equation}
On the other hand, by \eqref{e114} and \eqref{e118}, as $t\rightarrow-\infty$ and $\sigma>0$ is such that $0<\delta\leq-t\sqrt{\sigma}\leq c$ is bounded away from zero yet finite,
\begin{align*}
	{\bf R}_1(t,\sigma)=&\,\sigma^{-\frac{1}{4}\sigma_3}{\bf M}_1^{\textnormal{PV}}(-t\sqrt{\sigma})\sigma^{\frac{1}{4}\sigma_3}\sigma+\mathcal{O}\big(\sigma^{\frac{3}{2}}\big)+\mathcal{O}\big(\max\big\{(\sigma t)^2,t^{-2}\big\}\big),\\
	{\bf R}_2(t,\sigma)=&\,\mathcal{O}\big(\sigma^{\frac{3}{2}}\big)+\mathcal{O}\big(\max\big\{(\sigma t)^2,t^{-2}\big\}\big).
\end{align*}
Moreover, from the workings in Lemma \ref{genBessellem}, leading to \eqref{app4b}, as $x\downarrow 0$,
\begin{align}
	{\bf M}_1^{\textnormal{PV}}(x)=\frac{1}{2}x^{\frac{1}{2}\sigma_3}\begin{bmatrix}\frac{3}{8} & 1\smallskip\\ -\frac{9}{64} & -\frac{3}{8}\end{bmatrix}&\,x^{-\frac{1}{2}\sigma_3}\int_{-\infty}^{\infty}\big(\chi_{[0,\infty)}(y)-\Phi(y)\big)\,\d y\nonumber\\
	&\,+x^{\frac{1}{2}\sigma_3}\begin{bmatrix}\frac{9}{128} & -\frac{1}{8}\smallskip\\ -\frac{3}{1024} & -\frac{9}{128}\end{bmatrix}x^{-\frac{1}{2}\sigma_3}x^{-2}+x^{\frac{1}{2}\sigma_3}\mathcal{O}\big(x^2\big)x^{-\frac{1}{2}\sigma_3},\label{e123}
\end{align}
where $\mathcal{O}(x^2)$ is to be understood entrywise. Thus, setting
\begin{equation*}
	{\bf M}_1^{\textnormal{PV}}(x)=:\begin{bmatrix}p(x) & q(x)\\ r(x) & -p(x)\end{bmatrix},\ \ \ \ \ \ p,q,r:(0,\infty)\rightarrow\mathbb{R},
\end{equation*}
we derive from \eqref{e120} and \eqref{e121} by indefinite integration, as $t\rightarrow-\infty$, with a numerical constant $\eta_2\in\mathbb{R}$,
\begin{equation}\label{e124}
	\ln F_{\sigma}(t)=\frac{t^3}{12}-\frac{1}{8}\ln|t|+\int_0^{-t\sqrt{\sigma}}\left(q(x)+\frac{1}{8x}\right)\d x+\eta_2+\mathcal{O}\big(t^{-1}\big),\ \ \sigma>0:\ 0<\delta\leq-t\sqrt{\sigma}\leq c,
\end{equation}
using en route that
\begin{equation}\label{e125}
	q(x)+\frac{1}{8x}=\frac{x}{2}\int_{-\infty}^{\infty}\big(\chi_{[0,\infty)}(y)-\Phi(y)\big)\,\d y+\mathcal{O}\big(x^3\big),\ \ x\downarrow 0,
\end{equation}
is integrable at $x=0$. Consequently, combining \eqref{e122} and \eqref{e124} we arrive at \eqref{fi25}, modulo the integral term, after fixing $-t\geq t_0$ and then letting $\sigma\downarrow 0$, while using \eqref{e125},\eqref{fi13} and the numerical constant
\begin{equation*}
	\lim_{t\rightarrow-\infty}\left[\lim_{\sigma\downarrow 0}F_{\sigma}(t)-\frac{t^3}{12}+\frac{1}{8}\ln|t|\right]=\frac{1}{24}\ln 2+\zeta'(-1),
\end{equation*}
known from \cite[$(9)$]{BBD} and \cite[$(3)$]{DIK}. We are now left to establish that the integral term in \eqref{e124} is indeed of the type as in \eqref{fi25} with the asymptotics \eqref{fi26} in place. To this end, set $u(x):=2p(x)-q^2(x)$ and obtain from Lemma \ref{genBessellem}, part (iii) and \eqref{e123} that
\begin{equation}\label{rej1}
	u(x)=q'(x),\ x>0;\ \ \ \ \ u(x)=\frac{1}{8x^2}+\frac{1}{2}\int_{-\infty}^{\infty}\big(\chi_{[0,\infty)}(y)-\Phi(y)\big)\,\d y+\mathcal{O}\big(x^2\big),\ \ x\downarrow 0.
\end{equation}
Consequently,
\begin{equation*}
	\int_0^y\left(q(x)+\frac{1}{8x}\right)\d x=\int_0^y(y-x)\left(u(x)-\frac{1}{8x^2}\right)\d x,\ \ y>0,
\end{equation*}
and this completes our proof of \eqref{fi25}.
\end{proof}
\begin{appendix}
\section{Useful model functions}\label{appA}
We assemble a few model functions that are useful in our asymptotic analysis. Most of these are very well-known in nonlinear steepest descent literature.
\subsection{The Airy parametrix} Let $w=\textnormal{Ai}(z),z\in\mathbb{C}$ denote the Airy function, cf. \cite[$9.2.2$]{NIST}. Define the unimodular Wronskian matrix
\begin{equation}\label{app1}
	{\bf M}_0(\zeta):=\sqrt{2\pi}\,\e^{-\im\frac{\pi}{4}}\begin{bmatrix}\textnormal{Ai}(\zeta) & \e^{\im\frac{\pi}{3}}\textnormal{Ai}(\e^{-\im\frac{2\pi}{3}}\zeta)\smallskip\\
	\textnormal{Ai}'(\zeta) & \e^{-\im\frac{\pi}{3}}\textnormal{Ai}'(\e^{-\im\frac{2\pi}{3}}\zeta) \end{bmatrix},\ \ \ \zeta\in\mathbb{C},
\end{equation}
which is an entire function, and in turn the Airy parametrix ${\bf M}^{\textnormal{Ai}}(\zeta)={\bf M}^{\textnormal{Ai}}(\zeta;\zeta_0)$ for any fixed $\zeta_0\in\mathbb{R}$ and $\alpha\in(0,\frac{\pi}{2})$:
\begin{equation}\label{app2}
	{\bf M}^{\textnormal{Ai}}(\zeta;\zeta_0,\alpha):={\bf M}_0(\zeta)\begin{cases}\mathbb{I},&\textnormal{arg}\,(\zeta-\zeta_0)\in(0,\pi-\alpha)\\ \bigl[\begin{smallmatrix} 1 & 0\\ -1 & 1\end{smallmatrix}\bigr],&\textnormal{arg}\,(\zeta-\zeta_0)\in(\pi-\alpha,\pi)\\ \bigl[\begin{smallmatrix}1&-1\\ 0 & 1\end{smallmatrix}\bigr],&\textnormal{arg}\,(\zeta-\zeta_0)\in(-\pi+\alpha,0)\\
	\bigl[\begin{smallmatrix}1&-1\\ 0 & 1\end{smallmatrix}\bigr]\bigl[\begin{smallmatrix} 1 & 0\\ 1 & 1\end{smallmatrix}\bigr],&\textnormal{arg}\,(\zeta-\zeta_0)\in(-\pi,-\pi+\alpha)\end{cases}.
\end{equation}
The relevant analytic and asymptotic properties of ${\bf M}^{\textnormal{Ai}}$ are summarized below.
\begin{problem}\label{AiryRHP} The Airy parametrix ${\bf M}^{\textnormal{Ai}}(\zeta)={\bf M}^{\textnormal{Ai}}(\zeta;\zeta_0,\alpha)$ in \eqref{app2} has the following properties.
\begin{enumerate}
	\item[(1)] ${\bf M}^{\textnormal{Ai}}(\zeta)$ is analytic for $\zeta\in\mathbb{C}\setminus\big(\bigcup_{j=1}^4\Gamma_j\cup\{\zeta_0\}\big)$ with
	\begin{equation*}
		\Gamma_1:=(\zeta_0,\infty),\ \ \ \ \Gamma_3:=(-\infty,\zeta_0),\ \ \ \ \Gamma_2:=\e^{-\im\alpha}(-\infty,\zeta_0),\ \ \ \ \Gamma_4:=\e^{\im\alpha}(-\infty,\zeta_0),
	\end{equation*}
	and all four rays are oriented ``from left to right'' as shown in Figure \ref{figureAiry}.
	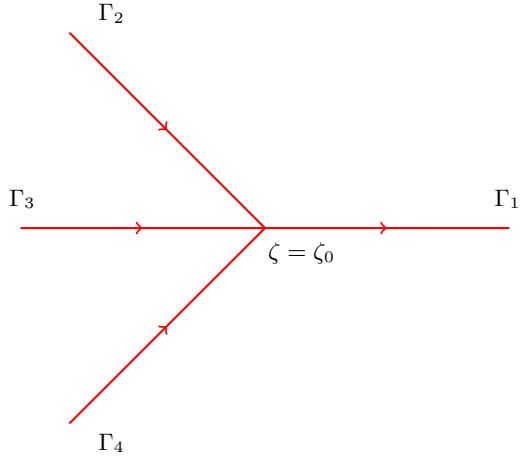
\begin{figure}[tbh]
\begin{tikzpicture}[xscale=0.65,yscale=0.65]
\draw [thick, color=red, decoration={markings, mark=at position 0.25 with {\arrow{>}}}, decoration={markings, mark=at position 0.75 with {\arrow{>}}}, postaction={decorate}] (-5,0) -- (5,0);
\node [below] at (0.75,-0.1) {{\small $\zeta=\zeta_0$}};
\node [right] at (4.5,0.6) {{\small $\Gamma_1$}};
\node [left] at (-4.5,0.6) {{\small $\Gamma_3$}};
\draw [thick, color=red, decoration={markings, mark=at position 0.5 with {\arrow{>}}}, postaction={decorate}] (-4,4) -- (0,0);
\draw [thick, color=red, decoration={markings, mark=at position 0.5 with {\arrow{>}}}, postaction={decorate}] (-4,-4) -- (0,0);
\node [right] at (-3.6,4.4) {{\small $\Gamma_2$}};
\node [right] at (-3.6,-4.4) {{\small $\Gamma_4$}};
\end{tikzpicture}
\caption{The oriented jump contours for the Airy parametrix ${\bf M}^{\textnormal{Ai}}(\zeta)$ in the complex $\zeta$-plane.}
\label{figureAiry}
\end{figure}
	\item[(2)] The limiting values from from either side of the jump contours $\Gamma_j$ are related via the jump conditions
	\begin{equation*}
		\begin{cases}  
		{\bf M}_+^{\textnormal{Ai}}(\zeta)={\bf M}_-^{\textnormal{Ai}}(\zeta)\bigl[\begin{smallmatrix}1&1\\ 0 & 1\end{smallmatrix}\bigr],& \zeta\in\Gamma_1;\smallskip\\
		{\bf M}_+^{\textnormal{Ai}}(\zeta)={\bf M}_-^{\textnormal{Ai}}(\zeta)\bigl[\begin{smallmatrix}1 & 0\\ 1 & 1\end{smallmatrix}\bigr],& \zeta\in\Gamma_2\cup\Gamma_4;
		\end{cases}\ \ \ \ \ \ \ {\bf M}_+^{\textnormal{Ai}}(\zeta)={\bf M}_-^{\textnormal{Ai}}(\zeta)\bigl[\begin{smallmatrix}0 & 1\\ -1 & 0\end{smallmatrix}\bigr],\ \ \zeta\in\Gamma_3.
	\end{equation*}
	\item[(3)] ${\bf M}^{\textnormal{Ai}}(\zeta)$ is bounded in a neighbourhood of $\zeta=\zeta_0$.
	\item[(4)] As $\zeta\rightarrow\infty$, valid in a full neighborhood of infinity off the jump contours,
	\begin{align*}
		{\bf M}^{\textnormal{Ai}}(\zeta)\sim&\,\left\{\mathbb{I}+\sum_{m=0}^{\infty}\begin{bmatrix}0 & 0\\ b_{2m+1} & 0\end{bmatrix}\zeta^{-3m-1}+\sum_{m=0}^{\infty}\begin{bmatrix}0 & a_{2m+1}\\ 0 & 0\end{bmatrix}\zeta^{-3m-2}+\sum_{m=1}^{\infty}\begin{bmatrix}a_{2m}& 0\\ 0 & b_{2m}\end{bmatrix}\zeta^{-3m}\right\}\\
		&\hspace{1cm}\times\zeta^{-\frac{1}{4}\sigma_3}\frac{1}{\sqrt{2}}\begin{bmatrix}1 & 1\\ -1 & 1\end{bmatrix}\e^{-\im\frac{\pi}{4}\sigma_3}\e^{-\frac{2}{3}\zeta^{\frac{3}{2}}\sigma_3},
	\end{align*}
	with $\zeta^{\alpha}:\mathbb{C}\setminus(-\infty,0]\rightarrow\mathbb{C}$ denoting the principal branch such that $\zeta^{\alpha}>0$ for $\zeta>0$ and 
	\begin{equation*}
		\sigma_3:=\bigl[\begin{smallmatrix}1&0\\ 0 & -1\end{smallmatrix}\bigr]\ \ \ \textnormal{as well as}\ \  a_m:=\left(\frac{3}{2}\right)^mu_m,\ \  b_m:=\left(\frac{3}{2}\right)^mv_m
	\end{equation*}
	with the coefficients $\{u_m,v_m\}_{m=1}^{\infty}$ from \cite[$9.7.2$]{NIST}. In particular, $b_1=-\frac{7}{48}$ and $a_1=\frac{5}{48}$.
\end{enumerate}
\end{problem}
\subsection{The Bessel parametrix} Let $w=I_0(z),z\in\mathbb{C}$ and $w=K_0(z),z\in\mathbb{C}\setminus(-\infty,0]$ denote the modified Bessel functions of order zero, cf. \cite[$10.25.2,10.25.3$]{NIST}. Define the unimodular Wronskian matrix
\begin{equation}\label{app3}
	{\bf N}_0(\zeta):=\e^{-\im\frac{\pi}{4}\sigma_3}\pi^{\frac{1}{2}\sigma_3}\begin{bmatrix}I_0(\zeta^{\frac{1}{2}}) & \frac{\im}{\pi}K_0(\zeta^{\frac{1}{2}})\smallskip\\
	\im\pi\zeta^{\frac{1}{2}}I_0'(\zeta^{\frac{1}{2}}) & -\zeta^{\frac{1}{2}}K_0'(\zeta^{\frac{1}{2}})\end{bmatrix},\ \ \ \zeta\in\mathbb{C}\setminus(-\infty,0],
\end{equation}
with principal branches throughout, in particular for $\zeta^{\frac{1}{2}}:\mathbb{C}\setminus(-\infty,0]\rightarrow\mathbb{C}$ such that $\zeta^{\frac{1}{2}}=\sqrt{\zeta}>0$ when $\zeta>0$, and in turn the Bessel parametrix ${\bf M}^{\textnormal{Be}}(\zeta)$:
\begin{equation}\label{app4}
	{\bf M}^{\textnormal{Be}}(\zeta):=\begin{bmatrix}1 & 0\\ -\frac{3}{8} & 1\end{bmatrix}{\bf N}_0(\zeta)\begin{cases}\mathbb{I},&\textnormal{arg}\,\zeta\in(-\frac{3\pi}{4},\frac{3\pi}{4})\\
	\bigl[\begin{smallmatrix}1 & 0\\ -1& 1\end{smallmatrix}\bigr],&\textnormal{arg}\,\zeta\in(\frac{3\pi}{4},\pi)\\
	\bigl[\begin{smallmatrix}1 & 0\\ 1& 1\end{smallmatrix}\bigr],&\textnormal{arg}\,\zeta\in(-\pi,-\frac{3\pi}{4})
	\end{cases}.
\end{equation}
The relevant analytic and asymptotic properties of ${\bf M}^{\textnormal{Be}}$ are summarized below.
\begin{problem}\label{BesselRHP} The Bessel parametrix ${\bf M}^{\textnormal{Be}}(\zeta)$ in \eqref{app4} has the following properties.
\begin{enumerate}
	\item[(1)] ${\bf M}^{\textnormal{Be}}(\zeta)$ is analytic for $\zeta\in\mathbb{C}\setminus\big(\bigcup_{j=1}^3\Sigma_j\cup\{0\}\big)$ with
	\begin{equation*}
		\Sigma_1:=(-\infty,0),\ \ \ \ \ \ \Sigma_2:=\e^{-\im\frac{\pi}{4}}(-\infty,0),\ \ \ \ \ \ \Sigma_3:=\e^{\im\frac{\pi}{4}}(-\infty,0),
	\end{equation*}
	and all three rays are oriented ``from left to right'' as shown in Figure \ref{figureBessel}.
	\begin{figure}[tbh]
\begin{tikzpicture}[xscale=0.65,yscale=0.65]
\draw [thick, color=red, decoration={markings, mark=at position 0.25 with {\arrow{>}}}, decoration={markings, mark=at position 0.75 with {\arrow{>}}}, postaction={decorate}] (-5,0) -- (0,0);
\node [below] at (0.75,-0.1) {{\small $\zeta=0$}};
\node [left] at (-4.5,0.6) {{\small $\Sigma_1$}};
\draw [thick, color=red, decoration={markings, mark=at position 0.5 with {\arrow{>}}}, postaction={decorate}] (-4,4) -- (0,0);
\draw [thick, color=red, decoration={markings, mark=at position 0.5 with {\arrow{>}}}, postaction={decorate}] (-4,-4) -- (0,0);
\node [right] at (-3.6,4.4) {{\small $\Sigma_2$}};
\node [right] at (-3.6,-4.4) {{\small $\Sigma_3$}};
\end{tikzpicture}
\caption{The oriented jump contours for the Bessel parametrix ${\bf M}^{\textnormal{Be}}(\zeta)$ in the complex $\zeta$-plane.}
\label{figureBessel}
\end{figure}
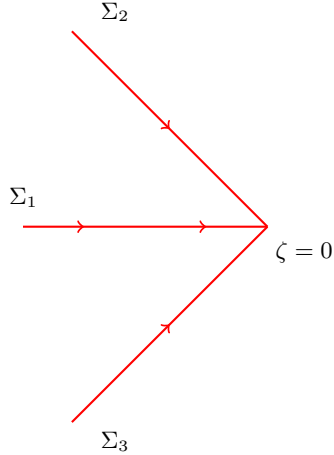
	\item[(2)] The limiting values from either side of the jump contours $\Sigma_j$ are related via the jump conditions
	\begin{equation*}
		{\bf M}_+^{\textnormal{Be}}(\zeta)={\bf M}_-^{\textnormal{Be}}(\zeta)\bigl[\begin{smallmatrix}0 & 1\\ -1 & 0\end{smallmatrix}\bigr],\ \ \ \zeta\in\Sigma_1;\ \ \ \ \ \ \ \ \ {\bf M}_+^{\textnormal{Be}}(\zeta)={\bf M}_-^{\textnormal{Be}}(\zeta)\bigl[\begin{smallmatrix}1 & 0\\ 1 & 1\end{smallmatrix}\bigr],\ \ \ \zeta\in\Sigma_2\cup\Sigma_3.
	\end{equation*}
	\item[(3)] As $\zeta\rightarrow 0$,
	\begin{equation*}
		{\bf M}^{\textnormal{Be}}(\zeta)=\begin{bmatrix}1 & 0\\ -\frac{3}{8}&1\end{bmatrix}\widehat{{\bf M}^{\textnormal{Be}}}(\zeta)\begin{bmatrix} 1 & \frac{1}{2\pi\im}\ln\zeta\\ 0 & 1\end{bmatrix}\begin{cases}\mathbb{I},&\textnormal{arg}\,\zeta\in(-\frac{3\pi}{4},\frac{3\pi}{4})\\
	\bigl[\begin{smallmatrix}1 & 0\\ -1& 1\end{smallmatrix}\bigr],&\textnormal{arg}\,\zeta\in(\frac{3\pi}{4},\pi)\\
	\bigl[\begin{smallmatrix}1 & 0\\ 1& 1\end{smallmatrix}\bigr],&\textnormal{arg}\,\zeta\in(-\pi,-\frac{3\pi}{4})
	\end{cases}
	\end{equation*}
	with the principal branch for the logarithm $\ln:\mathbb{C}\setminus(-\infty,0]\rightarrow\mathbb{C}$ and where $\zeta\mapsto\widehat{{\bf M}^{\textnormal{Be}}}(\zeta)$ is analytic at $\zeta=0$. In detail, for $|\zeta|<r$,
	\begin{equation*}
		\widehat{{\bf M}^{\textnormal{Be}}}(\zeta)=\e^{-\im\frac{\pi}{4}\sigma_3}\pi^{\frac{1}{2}\sigma_3}\begin{bmatrix}I_0(\zeta^{\frac{1}{2}}) & \frac{\im}{\pi}\Big(I_0(\zeta^{\frac{1}{2}})\ln 2+\sum_{k=0}^{\infty}\psi(k+1)\frac{(\frac{1}{4}\zeta)^k}{(k!)^2}\Big)\smallskip\\ \im\pi\zeta^{\frac{1}{2}}I_0'(\zeta^{\frac{1}{2}}) & 
		-\Big(\zeta^{\frac{1}{2}}I_0'(\zeta^{\frac{1}{2}})\ln 2-I_0(\zeta^{\frac{1}{2}})+\sum_{k=1}^{\infty}2k\,\psi(k+1)\frac{(\frac{1}{4}\zeta)^k}{(k!)^2}\Big)\end{bmatrix},
	\end{equation*}
	with
	\begin{equation*}
		I_0(\zeta^{\frac{1}{2}})=\sum_{k=0}^{\infty}\frac{(\frac{1}{4}\zeta)^k}{(k!)^2};\ \ \ \ \zeta^{\frac{1}{2}}I_0'(\zeta^{\frac{1}{2}})=\sum_{k=1}^{\infty}2k\frac{(\frac{1}{4}\zeta)^k}{(k!)^2},\ \ \zeta\in\mathbb{C};\ \ \ \psi(\zeta):=\frac{\Gamma'(\zeta)}{\Gamma(\zeta)},\ \ \zeta\in\mathbb{C}\setminus\{0,-1,-2,\ldots\}.
	\end{equation*}
	\item[(4)] As $\zeta\rightarrow\infty$, valid in a full neighbhorhood of infinity off the jump contours,
	\begin{align*}
		{\bf M}^{\textnormal{Be}}(\zeta)\sim&\,\left\{\mathbb{I}+\sum_{m=1}^{\infty}\begin{bmatrix}a_{2m}(0) & a_{2m-1}(0)\smallskip\\
		b_{2m+1}(0)-\frac{3}{8}a_{2m}(0) & b_{2m}(0)-\frac{3}{8}a_{2m-1}(0)\end{bmatrix}\zeta^{-m}\right\}\\
		&\hspace{1cm}\times\zeta^{-\frac{1}{4}\sigma_3}\frac{1}{\sqrt{2}}\begin{bmatrix}1 & 1\\ -1 & 1\end{bmatrix}\e^{-\im\frac{\pi}{4}\sigma_3}\e^{\zeta^{\frac{1}{2}}\sigma_3},
	\end{align*}
	using the principal branch for $\zeta^{\alpha}:\mathbb{C}\setminus(-\infty,0]\rightarrow\mathbb{C}$ such that $\zeta^{\alpha}>0$ for $\zeta>0$ and where $\{a_m(\nu),b_m(\nu)\}_{m=1}^{\infty}$ denote the coefficients in \cite[$10.17.1,10.17.8$]{NIST}. In particular, $a_1(0)=-\frac{1}{8},a_2(0)=\frac{9}{128},b_2(0)=-\frac{15}{128}$ and $b_3(0)=\frac{105}{1024}$.
\end{enumerate}
\end{problem}
\subsection{A generalized Bessel parametrix} The following generalization of RHP \ref{BesselRHP} was first studied in \cite[Section $7.1$]{CCR}, albeit in slightly disguised and broader fashion. Let $w=\Phi(z),z\in\mathbb{C}$ be as in \eqref{e41}.
\begin{problem}\label{genBesselRHP} For any $x>0$, determine ${\bf M}^{\textnormal{PV}}(\zeta)={\bf M}^{\textnormal{PV}}(\zeta;x)\in\mathbb{C}^{2\times 2}$\footnote{The superscript PV is used to reiterate Remark \ref{rejrem}: The solution to RHP \ref{genBesselRHP} relates to an integro-differential generalization of the Painlev\'e-V, PV, equation.} such that
\begin{enumerate}
	\item[(1)] ${\bf M}^{\textnormal{PV}}(\zeta)$ is analytic for $\zeta\in\mathbb{C}\setminus\big(\bigcup_{j=1}^4\Sigma_j\cup\{0\}\big)$ with
	\begin{equation*}
		\Sigma_1:=(-\infty,0),\ \ \ \ \ \ \Sigma_2:=\e^{-\im\frac{\pi}{4}}(-\infty,0),\ \ \ \ \ \ \Sigma_3:=\e^{\im\frac{\pi}{4}}(-\infty,0),\ \ \ \ \ \ \ \Sigma_4:=(0,\infty),
	\end{equation*}
	and all four rays are oriented ``from left to right'' as shown in Figure \ref{figuregenBessel}. ${\bf M}^{\textnormal{PV}}(\zeta)$ attains continuous limiting values ${\bf M}_{\pm}^{\textnormal{PV}}(\zeta)$ on $\Sigma_j\ni\zeta$.
	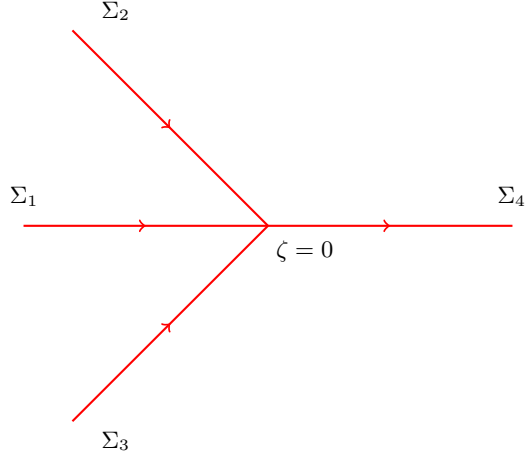
\begin{figure}[tbh]
\begin{tikzpicture}[xscale=0.65,yscale=0.65]
\draw [thick, color=red, decoration={markings, mark=at position 0.25 with {\arrow{>}}}, decoration={markings, mark=at position 0.75 with {\arrow{>}}}, postaction={decorate}] (-5,0) -- (5,0);
\node [below] at (0.75,-0.1) {{\small $\zeta=0$}};
\node [left] at (-4.5,0.6) {{\small $\Sigma_1$}};
\node [right] at (4.5,0.6) {{\small $\Sigma_4$}};
\draw [thick, color=red, decoration={markings, mark=at position 0.5 with {\arrow{>}}}, postaction={decorate}] (-4,4) -- (0,0);
\draw [thick, color=red, decoration={markings, mark=at position 0.5 with {\arrow{>}}}, postaction={decorate}] (-4,-4) -- (0,0);
\node [right] at (-3.6,4.4) {{\small $\Sigma_2$}};
\node [right] at (-3.6,-4.4) {{\small $\Sigma_3$}};
\end{tikzpicture}
\caption{The oriented jump contours for the generalized Bessel parametrix ${\bf M}^{\textnormal{PV}}(\zeta)$ in the complex $\zeta$-plane.}
\label{figuregenBessel}
\end{figure}
	\item[(2)] The limiting values from either side of the jump contours $\Sigma_j$ are related via the jump conditions
	\begin{equation*}
		{\bf M}_+^{\textnormal{PV}}(\zeta)={\bf M}_-^{\textnormal{PV}}(\zeta)\Bigl[\begin{smallmatrix}\Phi(\zeta) & 1-\Phi(\zeta)\smallskip\\ -1-\Phi(\zeta) & \Phi(\zeta)\end{smallmatrix}\Bigr],\ \ \ \zeta\in\Sigma_1;\ \ \ \ \ \ {\bf M}_+^{\textnormal{PV}}(\zeta)={\bf M}_-^{\textnormal{PV}}(\zeta)\bigl[\begin{smallmatrix}1 & 0\\ 1 & 1\end{smallmatrix}\bigr],\ \ \ \zeta\in\Sigma_2\cup\Sigma_3.
	\end{equation*}
	\begin{equation*}
		{\bf M}_+^{\textnormal{PV}}(\zeta)={\bf M}_-^{\textnormal{PV}}(\zeta)\Bigl[\begin{smallmatrix}1& 1-\Phi(\zeta)\smallskip\\
		0 & 1\end{smallmatrix}\Bigr],\ \ \ \zeta\in\Sigma_4.
	\end{equation*}
	\item[(3)] ${\bf M}^{\textnormal{PV}}(\zeta)$ is bounded in a neighbourhood of $\zeta=0$.
	\item[(4)] As $\zeta\rightarrow\infty$, valid in a full neighborhood of infinity off the jump contours,
	\begin{equation}\label{app4a}
		{\bf M}^{\textnormal{PV}}(\zeta)=\Big\{\mathbb{I}+{\bf M}_1^{\textnormal{PV}}\zeta^{-1}+\mathcal{O}\big(\zeta^{-2}\big)\Big\}\zeta^{-\frac{1}{4}\sigma_3}\frac{1}{\sqrt{2}}\begin{bmatrix}1 & 1\\ -1 & 1\end{bmatrix}\e^{-\im\frac{\pi}{4}\sigma_3}\e^{x\zeta^{\frac{1}{2}}\sigma_3},
	\end{equation}
	with $\zeta^{\alpha}:\mathbb{C}\setminus(-\infty,0]\rightarrow\mathbb{C}$ denoting the principal branch such that $\zeta^{\alpha}>0$ for $\zeta>0$ and where the scalar entries of ${\bf M}_1^{\textnormal{PV}}={\bf M}_1^{\textnormal{PV}}(x)$ are $\zeta$-independent.
\end{enumerate}
\end{problem}
The unique solvability of RHP \ref{genBesselRHP}, for any $x>0$, can be established as in \cite[Appendix A]{CCR}, after collapsing the rays $\Sigma_2$ and $\Sigma_3$ with $\Sigma_1$. We therefore take the existence of ${\bf M}^{\textnormal{PV}}(\zeta)={\bf M}^{\textnormal{PV}}(\zeta;x)$ for granted and concern ourselves instead with some of its properties:
\begin{lem}\label{genBessellem} Let ${\bf M}^{\textnormal{PV}}(\zeta;x)$ denote the unique solution of RHP \ref{genBesselRHP} defined for $\zeta\in\mathbb{C}\setminus(\bigcup_{j=1}^4\Sigma_j\cup\{0\})$ and $x>0$.
\begin{enumerate}
	\item[(i)] For any $x>0$, there exists $\epsilon>0$ so that ${\bf M}^{\textnormal{PV}}(\zeta;z)$ exists for all $z\in\mathbb{D}_{\epsilon}(x)=\{z\in\mathbb{C}:\,|z-x|<\epsilon\}$ and solves the same RHP \ref{genBesselRHP}. In particular, condition $(3)$ in the same problem holds uniformly for the same $z\in\mathbb{D}_{\epsilon}(x)$ and $\mathbb{D}_{\epsilon}(x)\ni z\mapsto {\bf M}_1^{\textnormal{PV}}(z)$ is analytic.
	\item[(ii)] As $x\downarrow 0$,
	\begin{equation}\label{app4b}
		{\bf M}^{\textnormal{PV}}(\zeta;x)=x^{\frac{1}{2}\sigma_3}\left\{\mathbb{I}+\frac{1}{2\zeta}\begin{bmatrix}\frac{3}{8} & 1\smallskip\\ -\frac{9}{64} & -\frac{3}{8}\end{bmatrix}\int_{-\infty}^{\infty}\big(\chi_{[0,\infty)}(y)-\Phi(y)\big)\,\d y+\mathcal{O}\left(\frac{x^2}{|\zeta|}\right)\right\}{\bf M}^{\textnormal{Be}}(\zeta x^2),
	\end{equation}
	uniformly for $|\zeta x^2|>\frac{1}{4}$ and $\zeta\notin\Sigma_1\cup\Sigma_2\cup\Sigma_3\cup\Sigma_4$, compare Figure \ref{figureBessel}.
	\item[(iii)] Setting
	\begin{equation}\label{app4bb}
		{\bf M}_1^{\textnormal{PV}}(x)=:\begin{bmatrix}p(x) & q(x)\\ r(x) & -p(x)\end{bmatrix},\ \ \ \ \ p,q,r:\mathbb{R}_+\rightarrow\mathbb{R},
	\end{equation}
	we have that $q'(x)=2p(x)-q^2(x),x>0$ and 
	\begin{equation}\label{app4c}
		\Psi(\zeta;x):=\begin{bmatrix}1 & 0\\ q(x) & 1\end{bmatrix}{\bf M}^{\textnormal{PV}}(\zeta;x),\ \ \zeta\in\mathbb{C}\setminus\big(\bigcup_{j=1}^4\Sigma_j\cup\{0\}\big),\ \ x>0,
	\end{equation}
	satisfies
	\begin{equation}\label{app4d}
		\frac{\partial}{\partial x}\Psi(\zeta;x)=\begin{bmatrix}0 & -1\\ 4p(x)-2q^2(x)-\zeta & 0\end{bmatrix}\Psi(\zeta;x).
	\end{equation}
\end{enumerate}
\end{lem}
\begin{proof} Property (i) is a consequence of the analytic Fredholm theorem used in the solvability proof of RHP \ref{genBesselRHP}, see for instance \cite{Z}. For property (ii), we argue very much as in \cite[$7.2$]{CCR}, so we first fix $\zeta\in\mathbb{D}_1(0)\setminus(\bigcup_{j=1}^4\Sigma_j\cup\{0\})$ and define
\begin{equation}\label{app5}
	{\bf P}^{(0)}(\zeta;x):=\begin{bmatrix}1 & 0\\ -\frac{3}{8}&1\end{bmatrix}{\bf N}_0(\zeta)\begin{bmatrix}1 & m(\zeta;x)\\ 0 & 1\end{bmatrix}\begin{cases}\mathbb{I},&\textnormal{arg}\,\zeta\in(-\frac{3\pi}{4},\frac{3\pi}{4})\\
	\bigl[\begin{smallmatrix}1 & 0\\ -1& 1\end{smallmatrix}\bigr],&\textnormal{arg}\,\zeta\in(\frac{3\pi}{4},\pi)\\
	\bigl[\begin{smallmatrix}1 & 0\\ 1& 1\end{smallmatrix}\bigr],&\textnormal{arg}\,\zeta\in(-\pi,-\frac{3\pi}{4})
	\end{cases}
\end{equation}
in terms of \eqref{app3} and the function
\begin{equation*}
	m(\zeta;x):=\frac{1}{2\pi\im}\int_{-\infty}^{\infty}\frac{\chi_{[0,\infty)}(\lambda x^{-2})-\Phi(\lambda x^{-2})}{\lambda-\zeta}\,\d\lambda,
	\ \ \zeta\in\mathbb{C}\setminus\mathbb{R}.
\end{equation*}
Observe that ${\bf P}^{(0)}(\zeta;x)$ satisfies properties $(1),(2)$ and $(3)$ in RHP \ref{genBesselRHP}, up to the rescaling $\zeta\mapsto\zeta x^{-2}$. Hence, with \eqref{e42} and \eqref{app4} in mind, we then study the ratio function
\begin{equation}\label{app5a}
	{\bf J}(\zeta;x):=x^{-\frac{1}{2}\sigma_3}{\bf M}^{\textnormal{PV}}(\zeta x^{-2};x)\begin{cases}\big({\bf M}^{\textnormal{Be}}(\zeta)\big)^{-1},&\zeta\in(\mathbb{C}\setminus\overline{\mathbb{D}_{\frac{1}{4}}(0)})\setminus\Sigma_{\bf T}\\
	({\bf P}^{(0)}(\zeta;x))^{-1},&\zeta\in\mathbb{D}_{\frac{1}{4}}(0)\setminus\Sigma_{\bf T}\end{cases}
\end{equation}
where $\Sigma_{\bf T}$ is shown in Figure \ref{figuretrafo6}. This function, for any $x>0$, is analytic in $\zeta\in\mathbb{C}\setminus((-\infty,-\frac{1}{4}]\cup\partial\mathbb{D}_{\frac{1}{4}}(0)\cup[\frac{1}{4},\infty))$, see Figure \ref{figuregenBesselratio} for the corresponding jump contour $\Sigma_{\bf J}$, and it obeys
\begin{equation*}
	{\bf J}_+(\zeta;x)={\bf J}_-(\zeta;x){\bf M}^{\textnormal{Be}}(\zeta)\begin{bmatrix}1 & 1-\Phi(\zeta x^{-2})\\ 0 & 1\end{bmatrix}\big({\bf M}^{\textnormal{Be}}(\zeta)\big)^{-1},\ \ \zeta>\frac{1}{4},
\end{equation*}
as well as
\begin{equation*}
	{\bf J}_+(\zeta;x)={\bf J}_-(\zeta;x){\bf M}^{\textnormal{Be}}_-(\zeta)\begin{bmatrix}1-\Phi(\zeta x^{-2}) & -\Phi(\zeta x^{-2})\\ \Phi(\zeta x^{-2}) & 1+\Phi(\zeta x^{-2})\end{bmatrix}\big({\bf M}_-^{\textnormal{Be}}(\zeta)\big)^{-1},\ \ \zeta<-\frac{1}{4},
\end{equation*}
and
\begin{equation}\label{app6}
	{\bf J}_+(\zeta;x)={\bf J}_-(\zeta;x){\bf P}^{(0)}(\zeta;x)\big({\bf M}^{\textnormal{Be}}(\zeta)\big)^{-1},\ \ \ |\zeta|=\frac{1}{4}.
\end{equation}
Moreover, by RHP \ref{BesselRHP}, condition $(4)$, and RHP \ref{genBesselRHP}, condition $(4)$, we have that
\begin{equation*}
	{\bf J}(\zeta;x)=\mathbb{I}+\mathcal{O}\big(\zeta^{-1}\big),\ \ \ \zeta\rightarrow\infty,\ \zeta\notin\mathbb{R}.
\end{equation*}
Thus, if ${\bf G}_{\bf J}(\zeta;x)$ denotes the underlying jump matrix of ${\bf J}(\zeta;x)$, then by \eqref{e42} and RHP \ref{BesselRHP}, with some $c>0$,
\begin{align*}
	\big\|{\bf G}_{\bf J}(\zeta;x)-\mathbb{I}\|\leq&\, c\sqrt{|\zeta|}\e^{-\zeta^2/x^4+\sqrt{|\zeta|}},\ \ \ \ \ \zeta\in\left(-\infty,-\frac{1}{4}\right)\cup\left(\frac{1}{4},\infty\right);\\
	 \big\|{\bf G}_{\bf J}(\zeta;x)-\mathbb{I}\|\leq &\,c\big|m(\zeta;x)\big|,\ \ \ \ \ \ |\zeta|=\frac{1}{4},
\end{align*}
where
\begin{equation*}
	\big|m(\zeta;x)\big|=\frac{1}{2\pi}\left|\int_{-\infty}^{\infty}\frac{\chi_{[0,\infty)}(\lambda x^{-2})-\Phi(\lambda x^{-2})}{\lambda-\zeta}\,\d\lambda\right|\stackrel{\eqref{e42}}{\leq}\frac{cx^2}{|\Im\zeta|},\ \ |\zeta|=\frac{1}{4},\ \ |\Im\zeta|\geq\delta>0.
\end{equation*}
But $\frac{\d}{\d x}\Phi(x)$ decays super exponentially fast at $\pm\infty$, so the last estimate also holds near $\zeta=\pm \frac{1}{4}$, compare the reasoning in \cite[Section $6.3$]{CCR}, i.e. we have all together, with some $c,x_0>0$,
\begin{equation}\label{app7}
	\|{\bf G}_{\bf J}(\cdot;x)-\mathbb{I}\|_{L^2(\Sigma_{\bf J})}\leq cx^2,\ \ \ \ \ \|{\bf G}_{\bf J}(\cdot;x)-\mathbb{I}\|_{L^{\infty}(\Sigma_{\bf J})}\leq cx^2\ \ \ \forall\,\,0<x\leq x_0.
\end{equation}
These small norm estimates show that ${\bf J}(\zeta;x)$ is computable via the integral equation
\begin{equation*}
	{\bf J}(\zeta;x)=\mathbb{I}+\frac{1}{2\pi\im}\int_{\Sigma_{\bf J}}{\bf J}_-(\lambda;x)\big({\bf G}_{\bf J}(\lambda;x)-\mathbb{I}\big)\frac{\d\lambda}{\lambda-\zeta},\ \ \ \zeta\in\mathbb{C}\setminus\Sigma_{\bf J},
\end{equation*}
for all $0<x\leq x_0$, cf. \cite{DZ}, using that in the same parameter regime
\begin{equation}\label{app8}
	\|{\bf J}_-(\cdot;x)-\mathbb{I}\|_{L^2(\Sigma_{\bf J})}\leq cx^2.
\end{equation}
Consequently, 
\begin{equation*}
	{\bf J}(\zeta;x)\stackrel[\eqref{app8}]{\eqref{app7}}{=}\mathbb{I}+\frac{1}{2\pi\im}\ointclockwise_{|\lambda|=\frac{1}{4}}\big({\bf G}_{{\bf J}}(\lambda;x)-\mathbb{I}\big)\frac{\d\lambda}{\lambda-\zeta}+\mathcal{O}\left(\frac{x^4}{1+|\zeta|}\right),\ \ x\downarrow 0,
\end{equation*}
which holds uniformly in $\zeta\in\mathbb{C}\setminus\Sigma_{\bf J}$. Since by \eqref{app6},
\begin{equation*}
	{\bf G}_{\bf J}(\lambda;x)-\mathbb{I}=\begin{bmatrix}1&0\\ -\frac{3}{8} & 1\end{bmatrix}{\bf N}_0(\lambda)\begin{bmatrix}0 & 1\\ 0 & 0\end{bmatrix}\big({\bf N}_0(\lambda)\big)^{-1}\begin{bmatrix}1 & 0\\ \frac{3}{8} & 1\end{bmatrix}m(\lambda;x),\ \ \ \ |\lambda|=\frac{1}{4},
\end{equation*}
and also
\begin{equation*}
	m(\lambda;x)=-\frac{x^2}{2\pi\im\lambda}\int_{-\infty}^{\infty}\big(\chi_{[0,\infty)}(y)-\Phi(y)\big)\,\d y+\mathcal{O}\big(x^4\big),\ \ \ \ \ x\downarrow 0,
\end{equation*}
uniformly in $|\lambda|=\frac{1}{4}$, an explicit residue computation yields for $|\zeta|>\frac{1}{4}$, while using property $(3)$ in RHP \ref{BesselRHP},
\begin{align}
	{\bf J}(\zeta;x)=&\,\,\mathbb{I}-\frac{x^2}{2\pi\im\zeta}\begin{bmatrix}1 & 0\\ -\frac{3}{8} & 1\end{bmatrix}\widehat{{\bf M}^{\textnormal{Be}}}(0)\begin{bmatrix}0 & 1\\ 0 & 0\end{bmatrix}\big(\widehat{{\bf M}^{\textnormal{Be}}}(0)\big)^{-1}\begin{bmatrix}1 & 0\\ \frac{3}{8} & 1\end{bmatrix}\int_{-\infty}^{\infty}\big(\chi_{[0,\infty)}(y)-\Phi(y)\big)\,\d y+\mathcal{O}\left(\frac{x^4}{|\zeta|}\right)\nonumber\\
	=&\,\,\mathbb{I}+\frac{x^2}{2\zeta}\begin{bmatrix}\frac{3}{8} & 1\smallskip\\ -\frac{9}{64} & -\frac{3}{8}\end{bmatrix}\int_{-\infty}^{\infty}\big(\chi_{[0,\infty)}(y)-\Phi(y)\big)\,\d y+\mathcal{O}\left(\frac{x^4}{|\zeta|}\right),\ \ \ x\downarrow 0.\label{app9}
\end{align}
It now remains to recall \eqref{app5a}, i.e. \eqref{app4b} follows at once from \eqref{app9} after changing variables.\bigskip
	\begin{figure}[tbh]
\begin{tikzpicture}[xscale=0.65,yscale=0.65]
\draw [thick, color=red, decoration={markings, mark=at position 0.5 with {\arrow{>}}}, postaction={decorate}] (-5,0) -- (-1,0);
\draw [thick, color=red, decoration={markings, mark=at position 0.5 with {\arrow{>}}}, postaction={decorate}] (1,0) -- (5,0);
\node [left] at (-4.5,0.6) {{\small $(-\infty,-\frac{1}{2}]$}};
\node [right] at (4.5,0.6) {{\small $[\frac{1}{2},\infty)$}};
\draw [thick, color=red, decoration={markings, mark=at position 0.25 with {\arrow{<}}}, decoration={markings, mark=at position 0.75 with {\arrow{<}}}, postaction={decorate}] (0,0) circle [radius=1];
\draw [fill, color=black] (0,0) circle [radius=0.07];
\node [below] at (0.05,-0.1) {{\small $0$}};
\end{tikzpicture}
\caption{The oriented jump contours for the function ${\bf J}(\zeta;x)$ in the complex $\zeta$-plane.}
\label{figuregenBesselratio}
\end{figure}
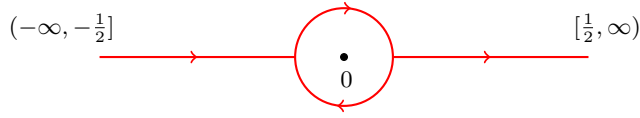

Finally, for property (iii) we can adapt the arguments in \cite[Section $7.5$]{CCR}, i.e. we first notice that
\begin{equation*}
	{\bf A}(\zeta;x):=\frac{\partial}{\partial x}{\bf M}^{\textnormal{PV}}(\zeta;x)\big({\bf M}^{\textnormal{PV}}(\zeta;x)\big)^{-1},\ \ \ (\zeta,x)\in\mathbb{C}\times\mathbb{R}_+
\end{equation*}
is entire in $\zeta$ by RHP \ref{genBesselRHP}, in fact by \eqref{app4a}, with some $x$-dependent functions $\ast$,
\begin{equation*}
	{\bf A}(\zeta;x)=\begin{bmatrix}-q(x) & -1\\ 2p(x)-\zeta & q(x)\end{bmatrix}+\frac{1}{\zeta}\begin{bmatrix}\ast & q'(x)+q^2(x)-2p(x)\\ \ast & \ast\end{bmatrix}+\mathcal{O}\big(\zeta^{-2}\big),\ \ \ \zeta\rightarrow\infty,
\end{equation*}
and so for $\zeta$ away from the jump contour shown in Figure \ref{figuregenBessel}, with $x>0$,
\begin{equation*}
	\frac{\partial}{\partial x}{\bf M}^{\textnormal{PV}}(\zeta;x)=\begin{bmatrix}-q(x) & -1\\ 2p(x)-\zeta & q(x)\end{bmatrix}{\bf M}^{\textnormal{PV}}(\zeta;x),\ \ \ \ q'(x)=2p(x)-q^2(x).
\end{equation*}
The last system shows that \eqref{app4c} satisfies \eqref{app4d} and this completes our proof of the Lemma.
\end{proof}
%

\section{Auxiliary results}\label{appB}
We collect a series of auxiliary results which are used in the main body of the text.
\subsection{Isometry property of $Q_t$} The linear transformation $Q_t:L^2(0,\infty)\rightarrow L^2(\mathbb{R})$ defined in \eqref{e38aa} is likely known to be an isometry, we now supplement a short proof for the same fact. To this end, let $\mathcal{S}(\mathbb{R})$ denote the Schwartz space on $\mathbb{R}$, i.e. the set of all smooth functions $f$ so that $f$ and all its derivatives are rapidly decreasing. The Fourier transform of a function $f\in\mathcal{S}(\mathbb{R})$ is defined by
\begin{equation*}
	\hat{f}(\xi):=\int_{-\infty}^{\infty}f(x)\e^{-\im\xi x}\,\d x,\ \ \xi\in\mathbb{R},
\end{equation*}
and we recall that $\hat{f}\in\mathcal{S}(\mathbb{R})$ for any $f\in\mathcal{S}(\mathbb{R})$, cf. \cite[Theorem $1.3$]{St}. Also, if $f\in\mathcal{S}(\mathbb{R})$ then 
\begin{equation}\label{appC1}
	\|\hat{f}\|_{L^2(\mathbb{R})}=\sqrt{2\pi}\,\|f\|_{L^2(\mathbb{R})},
\end{equation}
see \cite[Theorem $1.12$]{St}, and we let $C_0^{\infty}(\mathbb{R})$ denote the set of compactly supported, smooth functions. Note that $C_0^{\infty}(\mathbb{R})$ is dense in $L^2(\mathbb{R})$.
\begin{lem}\label{isolem} Define
\begin{equation}\label{appC2}
	(Q_tf)(x):=\int_{-\infty}^{\infty}\textnormal{Ai}(x+y+t)f(y)\,\d y
\end{equation}
for $f\in C_0^{\infty}(\mathbb{R})$ and fixed $t\in\mathbb{R}$. Then $\|Q_tf\|_{L^2(\mathbb{R})}=\|f\|_{L^2(\mathbb{R})}$.
\end{lem}
\begin{proof} When $f\in C_0^{\infty}(\mathbb{R})$, then $f\in\mathcal{S}(\mathbb{R})$ with $\textnormal{supp}(f)\subset[-a,a]$, say. Consequently, \eqref{appC2} is well-defined and we have by Fubini's theorem and \eqref{k3},
\begin{equation}\label{appC3}	
	(Q_tf)(x)=\int_{-\infty}^{\infty}\left[\frac{1}{2\pi}\int_{\mathbb{R}+\im\delta}\e^{\im(\frac{1}{3}\lambda^3+[x+y+t]\lambda)}\,\d\lambda\right]f(y)\,\d y=\frac{1}{2\pi}\int_{\mathbb{R}+\im\delta}\e^{\im(\frac{1}{3}\lambda^3+[x+t]\lambda)}\hat{f}(-\lambda)\,\d\lambda.
\end{equation}
However, by \cite[Theorem $11.1.3$]{Sim2}, $\hat{f}$ is an entire function so that for all $n\in\mathbb{Z}_{\geq 1}$, there is $c_n>0$ with
\begin{equation*}
	\big|\hat{f}(x+\im y)\big|\leq c_n(1+|x|+|y|)^{-n}\e^{a|y|}.
\end{equation*}
Thus, we can use Cauchy's theorem to deform $\mathbb{R}+\im\delta$ to $\mathbb{R}$ in \eqref{appC3} resulting in the identity
\begin{equation*}
	(Q_tf)(x)=\frac{1}{2\pi}\int_{-\infty}^{\infty}\e^{-\im(\frac{1}{3}\lambda^3+t\lambda)}\hat{f}(\lambda)\e^{-\im\lambda x}\,\d\lambda=\frac{1}{2\pi}\,\hat{g}(x),\ \ \ x\in\mathbb{R},
\end{equation*}
where $g(\lambda):=\e^{-\im(\frac{1}{3}\lambda^3+t\lambda)}\hat{f}(\lambda)$ is entire with good decay properties at $\pm\infty$, in fact $g\in\mathcal{S}(\mathbb{R})$. Thus,
\begin{equation*}
	\|Q_tf\|_{L^2(\mathbb{R})}=\frac{1}{2\pi}\|\hat{g}\|_{L^2(\mathbb{R})}\stackrel{\eqref{appC1}}{=}\frac{1}{\sqrt{2\pi}}\|g\|_{L^2(\mathbb{R})}=\frac{1}{\sqrt{2\pi}}\|\hat{f}\|_{L^2(\mathbb{R})}\stackrel{\eqref{appC1}}{=}\|f\|_{L^2(\mathbb{R})},\ \ t\in\mathbb{R},
\end{equation*}
as claimed above.
\end{proof}
\begin{cor} If $f\in L^2(\mathbb{R})$, we define
\begin{equation}\label{appC4}
	y_t(f):=\lim_{\lambda\rightarrow\infty}\int_{-\lambda}^{\lambda}\textnormal{Ai}(x+y+t)f(y)\,\d y,
\end{equation}
where the limit is understood in the $L^2(\mathbb{R})$ norm. Then, with $Q_t:L^2(\mathbb{R})\rightarrow L^2(\mathbb{R})$ given by $f\mapsto y_t(f)$, we have $\|Q_tf\|_{L^2(\mathbb{R})}=\|f\|_{L^2(\mathbb{R})}$, so in particular 
\begin{equation}\label{appC5}
	\|Q_tf\|_{L^2(\mathbb{R})}=\|f\|_{L^2(0,\infty)}\ \ \ \forall\, f\in L^2(0,\infty).
\end{equation}
\end{cor}
\begin{proof} The limit in the right hand side of \eqref{appC4} exists in $L^2(\mathbb{R})$ by the fact that $\|Q_tf\|_{L^2(\mathbb{R})}=\|f\|_{L^2(\mathbb{R})}$ for all $f\in C_0^{\infty}(\mathbb{R})$, see Lemma \ref{isolem}, with $Q_tf$ as in \eqref{appC2}. But since $C_0^{\infty}(\mathbb{R})$ is dense in $L^2(\mathbb{R})$, we find in turn $\|Q_tf\|_{L^2(\mathbb{R})}=\|f\|_{L^2(\mathbb{R})}$ for all $f\in L^2(\mathbb{R})$, and thus also the restriction of the same norm equality to $f\in L^2(0,\infty)$ holds true. The proof is complete.
\end{proof}
\subsection{Monotonicity of the complementary error function} It is well-known that $w=\textnormal{erfc}(z)$ constitutes a monotonically decreasing function on $\mathbb{R}\ni z$, cf. \cite[$7.2.2$]{NIST}. Below we state a certain generalization of this behavior for the modulus of the complementary error function in a sector of the complex plane.
\begin{lem}\label{applem1} For any $x\leq 0,r\geq 0$ and $\theta\in[\frac{3\pi}{4},\frac{5\pi}{4}]$,
\begin{equation}\label{appB1}
	\Big|\textnormal{erfc}\big(x+r\e^{\im\theta}\big)\Big|\geq \textnormal{erfc}(x).
\end{equation}
\end{lem}
\begin{proof} We have for any $x,r,\theta\in\mathbb{R}$,
\begin{eqnarray}
	\textnormal{erfc}\big(x+r\e^{\im\theta}\big)&=&\frac{2}{\sqrt{\pi}}\int_{x+r\e^{\im\theta}}^{\infty}\e^{-t^2}\,\d t=2-\frac{2}{\sqrt{\pi}}\int_{-\infty}^{x+r\e^{\im\theta}}\e^{-t^2}\,\d t=2-\frac{2}{\sqrt{\pi}}\int_{-x-r\e^{\im\theta}}^{\infty}\e^{-t^2}\,\d t\nonumber\\
	&=&2-\frac{2}{\sqrt{\pi}}\int_{-x}^{\infty}\e^{-(t-r\e^{\im\theta})^2}\,\d t,\label{appB2}
\end{eqnarray}
with the integration path along $\mathbb{R}$ in the last integral. However, once $x\leq 0,r\geq 0$ and $\theta\in[\frac{3\pi}{4},\frac{5\pi}{4}]$,
\begin{equation*}
	\left|\int_{-x}^{\infty}\e^{-(t-r\e^{\im\theta})^2}\,\d t\right|\leq \e^{-r^2\cos(2\theta)}\int_{-x}^{\infty}\e^{-t^2+2tr\cos\theta}\,\d t\leq\int_{-x}^{\infty}\e^{-t^2}\,\d t,
\end{equation*}
and so with the triangle inequality back in \eqref{appB2}, for the same $x,r,\theta$ as in the last estimate,
\begin{equation*}
	\Big|\textnormal{erfc}\big(x+r\e^{\im\theta}\big)\Big|\geq 2-\frac{2}{\sqrt{\pi}}\left|\int_{-x}^{\infty}\e^{-(t-r\e^{\im\theta})^2}\,\d t\right|\geq 2-\frac{2}{\sqrt{\pi}}\int_{-x}^{\infty}\e^{-t^2}\,\d t=\frac{2}{\sqrt{\pi}}\int_x^{\infty}\e^{-t^2}\,\d t=\textnormal{erfc}(x).
\end{equation*}
This concludes our proof of \eqref{appB1}.
\end{proof}
\subsection{Large $\sigma$ expansion of $g_j=g_j(t,\sigma)$ in \eqref{e66a} and \eqref{e66aa}} Given $t,\sigma>0$, the equation
\begin{equation}\label{appB3}
	z_0=-\frac{1}{(\pi\sigma)^{\frac{3}{2}}}\int_{-\infty}^{z_0}\frac{\e^{-(\lambda-\omega)^2}}{1-\Phi(\lambda-\omega)}\frac{\d\lambda}{\sqrt{z_0-\lambda}},\ \ \ \ \omega=\frac{t}{\sigma},
\end{equation}
admits a unique real-valued solution $z_0=z_0(t,\sigma)<0$, compare Proposition \ref{prop8}. More is true, as $\sigma\rightarrow+\infty$, uniformly in $t>0$, we have $z_0(t,\sigma)=o(1)$ as made precise in the following asymptotic expansion:
\begin{lem}\label{applem2} There exist $c,\sigma_0>0$ so that
\begin{equation}\label{appB4}
	z_0(t,\sigma)=-\frac{f_0(\omega)}{(\pi\sigma)^{\frac{3}{2}}}+\frac{f_0(\omega)f_1(\omega)}{(\pi\sigma)^3}-\frac{f_0(\omega)f_1^2(\omega)+f_0^2(\omega)f_2(\omega)}{(\pi\sigma)^{\frac{9}{2}}}+r_1(t,\sigma)
\end{equation}
for $\sigma\geq\sigma_0$ and all $t>0$. The coefficients $f_n$ equal 
\begin{equation*}
	f_n(\omega):=\frac{1}{n!}\int_{-\infty}^0\frac{\partial^n}{\partial\lambda^n}\left\{\frac{\e^{-(\lambda-\omega)^2}}{1-\Phi(\lambda-\omega)}\right\}\frac{\d\lambda}{\sqrt{-\lambda}},\ \ \ n\in\mathbb{Z}_{\geq 0},
\end{equation*}
and the error term $r_1(t,\sigma)$ is such that
\begin{equation*}
	\int_t^{\infty}\big|r_1(s,\sigma)\big|\,\d s\leq c\sigma^{-5},\ \ \ \ \ \ \int_{\sigma}^{\infty}\big|r_1(t,s)\big|\,\d s\leq ct^{-5}\ \ \ \forall\,\sigma\geq\sigma_0,\ \ t>0.
\end{equation*}
\end{lem}
\begin{proof} If
\begin{equation*}
	(-\infty,0)\ni x\mapsto f\left(x,y\right):=\int_{-\infty}^x\frac{\e^{-(\lambda-y)^2}}{1-\Phi(\lambda-y)}\frac{\d\lambda}{\sqrt{x-\lambda}},\ \ \ y>0,
\end{equation*}
then by Taylor expansion, as $x\uparrow 0$, uniformly in $y>0$,
\begin{equation}\label{appB5}
	f(x,y)\sim\sum_{n=0}^{\infty}f_n(y)x^n,\ \ \ \ \ f_n(y)=\frac{1}{n!}\int_{-\infty}^0\frac{\partial^n}{\partial\lambda^n}\left\{\frac{\e^{-(\lambda-y)^2}}{1-\Phi(\lambda-y)}\right\}\frac{\d\lambda}{\sqrt{-\lambda}}.
\end{equation}
On the other hand, $\mathbb{R}_+\ni\sigma\mapsto z_0(t,\sigma),t>0$ is increasing for $\sigma\geq\sigma_0$ and so $\lim_{\sigma\rightarrow+\infty}z_0(t,\sigma)$ exists uniformly in $t>0$. In fact, from \eqref{appB3} we find at once 
\begin{equation*}
	\lim_{\sigma\rightarrow+\infty}z_0(t,\sigma)=0,
\end{equation*}
uniformly in $t>0$, and so upon iteration of \eqref{appB5}, as $\sigma\rightarrow+\infty$,
\begin{equation*}
	z_0(t,\sigma)\sim-\frac{1}{(\pi\sigma)^{\frac{3}{2}}}\sum_{n=0}^{\infty}f_n(\omega)z_0^n(t,\sigma)=-\frac{f_0(\omega)}{(\pi\sigma)^{\frac{3}{2}}}+\frac{f_0(\omega)f_1(\omega)}{(\pi\sigma)^3}-\frac{f_0(\omega)f_1^2(\omega)+f_0^2(\omega)f_2(\omega)}{(\pi\sigma)^{\frac{9}{2}}}+\sum_{n=4}^{\infty}\frac{F_n(\omega)}{\sigma^{\frac{3n}{2}}}.
\end{equation*}
Here, $F_n$ is a multivariate polynomial in the variables $\{f_j(\omega)\}_{j=0}^{n-1}$ of degree $n$ and with $1-\Phi(x)\geq\frac{1}{2}$ for all $x\leq 0$ we obtain at once the coarse estimates
\begin{equation*}
	\int_t^{\infty}\left|F_n\left(\frac{s}{\sigma}\right)\right|\d s\leq c_n\sigma,\ \ \ \ \ \ \int_{\sigma}^{\infty}\left|F_n\Big(\frac{t}{s}\Big)\right|\frac{\d s}{s^{\frac{3n}{2}}}\leq c_n t^{1-\frac{3n}{2}},\ \ \ \ \ c_n>0,
\end{equation*}
valid for all $t,\sigma>0$ and $n\in\mathbb{Z}_{\geq 2}$. This completes the proof of \eqref{appB4}.
\end{proof}
Equipped with \eqref{appB4} we now return to \eqref{e66a}.
\begin{cor}\label{appcor1} There exist $c,\sigma_0>0$ so that
\begin{equation}\label{appB6}
	g_1(t,\sigma)=-\frac{h_0(\omega)}{(\pi\sigma)^{\frac{3}{2}}}+\frac{h_1^2(\omega)}{(\pi\sigma)^3}-\frac{f_0^2(\omega)f_1(\omega)}{4(\pi\sigma)^{\frac{9}{2}}}+r_2(t,\sigma)
\end{equation}
for $\sigma\geq\sigma_0$ and all $t>0$. The coefficients $f_n(\omega)$ are as in Lemma \ref{applem2}, the coefficients $h_n(\omega)$ equal
\begin{equation*}
	h_n(\omega):=\frac{1}{n!}\int_{-\infty}^0\frac{\partial^n}{\partial\lambda^n}\left\{\frac{\e^{-(\lambda-\omega)^2}}{1-\Phi(\lambda-\omega)}\right\}\sqrt{-\lambda}\,\d\lambda,\ \ \ n\in\mathbb{Z}_{\geq 0},
\end{equation*}
and the error term $r_2(t,\sigma)$ satisfies
\begin{equation*}
	\int_t^{\infty}\big|r_2(s,\sigma)\big|\,\d s\leq c\sigma^{-5},\ \ \ \ \ \ \int_{\sigma}^{\infty}\big|r_2(t,s)\big|\,\d s\leq ct^{-5}\ \ \ \forall\,\sigma\geq\sigma_0,\ \ t>0.
\end{equation*}
\end{cor}
\begin{proof} Similar to the workings in Lemma \ref{applem2}, we have by Taylor expansion, as $\sigma\rightarrow+\infty$, uniformly in $t>0$,
\begin{align*}
	g_1(t,\sigma)\stackrel{\eqref{e66a}}{=}-\frac{1}{(\pi\sigma)^{\frac{3}{2}}}\int_{-\infty}^{z_0}\frac{\e^{-(\lambda-\omega)^2}}{1-\Phi(\lambda-\omega)}\sqrt{z_0-\lambda}\,\d\lambda-\frac{z_0^2}{4}\sim-\frac{1}{(\pi\sigma)^{\frac{3}{2}}}\sum_{n=0}^{\infty}h_n(\omega)z_0^n(t,\sigma)-\frac{1}{4}z_0^2(t,\sigma),
\end{align*}
in terms $h_n(\omega)$ as in the formulation of the Corollary. Hence, with \eqref{appB4}, as $\sigma\rightarrow+\infty$, uniformly in $t>0$,
\begin{equation*}
	g_1(t,\sigma)\sim-\frac{h_0(\omega)}{(\pi\sigma)^{\frac{3}{2}}}+\frac{h_1(\omega)f_0(\omega)-\frac{1}{4}f_0^2(\omega)}{(\pi\sigma)^3}-\frac{h_1(\omega)f_0(\omega)f_1(\omega)+h_2(\omega)f_0^2(\omega)-\frac{1}{2}f_0^2(\omega)f_1(\omega)}{(\pi\sigma)^{\frac{9}{2}}}+\sum_{n=4}^{\infty}\frac{H_n(\omega)}{\sigma^{\frac{3n}{2}}},
\end{equation*}
where $H_n$ is a multivariate polynomial in the variables $\{h_j(\omega)\}_{j=1}^{n-1}\cup\{f_j(\omega)\}_{j=0}^{n-2}$ of degree $n$. Noting that, via integration by parts, $h_n(\omega)=\frac{1}{2n}f_{n-1}(\omega)$ for all $n\in\mathbb{Z}_{\geq 1}$, the leading terms simplify to
\begin{equation*}
	g_1(t,\sigma)\sim-\frac{h_0(\omega)}{(\pi\sigma)^{\frac{3}{2}}}+\frac{h_1^2(\omega)}{(\pi\sigma)^3}-\frac{f_0^2(\omega)f_1(\omega)}{4(\pi\sigma)^{\frac{9}{2}}}+\sum_{n=4}^{\infty}\frac{H_n(\omega)}{\sigma^{\frac{3n}{2}}},
\end{equation*}
and we record the coarse estimates
\begin{equation*}
	\int_t^{\infty}\left|H_n\left(\frac{s}{\sigma}\right)\right|\d s\leq c_n\sigma,\ \ \ \ \ \ \ \int_{\sigma}^{\infty}\left|H_n\Big(\frac{t}{s}\Big)\right|\frac{\d s}{s^{\frac{3n}{2}}}\leq c_n t^{1-\frac{3n}{2}},\ \ \ \ c_n>0,
\end{equation*}
valid for all $t,\sigma>0$ and $n\in\mathbb{Z}_{\geq 2}$. Estimate \eqref{appB6} follows at once.
\end{proof}
Next, we have the following result for \eqref{e66aa}.
\begin{cor}\label{appcor2} There exist $c,\sigma_0>0$ so that
\begin{equation}\label{appB7}
	g_2(t,\sigma)=\frac{k_0(\omega)}{3(\pi\sigma)^{\frac{3}{2}}}-\frac{f_0^2(\omega)h_1(\omega)-\frac{5}{6}f_0^3(\omega)}{4(\pi\sigma)^{\frac{9}{2}}}+r_3(t,\sigma)
\end{equation}
for $\sigma\geq\sigma_0$ and all $t>0$. The coefficients $f_n(\omega)$ and $h_n(\omega)$ are as in Lemma \ref{applem2} and \ref{appcor1}, the coefficients $k_n(\omega)$ equal
\begin{equation*}
	k_n(\omega):=\frac{1}{n!}\int_{-\infty}^0\frac{\partial^n}{\partial\lambda^n}\left\{\frac{\e^{-(\lambda-\omega)^2}}{1-\Phi(\lambda-\omega)}\right\}(-\lambda)^{\frac{3}{2}}\,\d\lambda,\ \ n\in\mathbb{Z}_{\geq 0},
\end{equation*}
and the error term $r_3(t,\sigma)$ satisfies
\begin{equation*}
	\int_t^{\infty}\big|r_3(s,\sigma)\big|\,\d s\leq c\sigma^{-5},\ \ \ \ \ \ \int_{\sigma}^{\infty}\big|r_3(t,s)\big|\,\d s\leq ct^{-5}\ \ \ \forall\,\sigma\geq\sigma_0,\ \ t>0.
\end{equation*}
\end{cor}
\begin{proof} By Taylor expansion, as $\sigma\rightarrow+\infty$, uniformly in $t>0$,
\begin{align*}
	g_2(t,\sigma)\stackrel{\eqref{e66aa}}{\sim}\frac{1}{3(\pi\sigma)^{\frac{3}{2}}}\sum_{n=0}^{\infty}k_n(\omega)z_0^n(t,\sigma)-\frac{z_0(t,\sigma)}{2(\pi\sigma)^{\frac{3}{2}}}\sum_{n=0}^{\infty}h_n(\omega)z_0^n(t,\sigma)-\frac{1}{12}z_0^3(t,\sigma),
\end{align*}
and hence with \eqref{appB4}, in the same limit,
\begin{align*}
	g_2&\,(t,\sigma)\sim\frac{k_0(\omega)}{3(\pi\sigma)^{\frac{3}{2}}}-\frac{\frac{1}{3}f_0(\omega)k_1(\omega)-\frac{1}{2}f_0(\omega)h_0(\omega)}{(\pi\sigma)^3}\\
	&\,+\frac{\frac{1}{3}(k_1(\omega)f_0(\omega)f_1(\omega)+f_0^2(\omega)k_2(\omega))-\frac{1}{2}(f_0^2(\omega)h_1(\omega)+f_0(\omega)f_1(\omega)h_0(\omega))+\frac{5}{24}f_0^3(\omega)}{(\pi\sigma)^{\frac{9}{2}}}+\sum_{n=4}^{\infty}\frac{K_n(\omega)}{\sigma^{\frac{3n}{2}}}
\end{align*}
with a multivariate polynomial $K_n$ in the variables $\{f_j(\omega)\}_{j=0}^{n-2}\cup\{h_j(\omega)\}_{j=0}^{n-3}\cup\{k_j(\omega)\}_{j=1}^{n-1}$ of degree $n$. However, via integration by parts, $k_n(\omega)=\frac{3}{2n}h_{n-1}(\omega)$ for all $n\in\mathbb{Z}_{\geq 1}$, and so after simplification
\begin{align*}
	g_2(t,\sigma)\sim\frac{k_0(\omega)}{3(\pi\sigma)^{\frac{3}{2}}}-\frac{f_0^2(\omega)h_1(\omega)-\frac{5}{6}f_0^3(\omega)}{4(\pi\sigma)^{\frac{9}{2}}}+\sum_{n=4}^{\infty}\frac{K_n(\omega)}{\sigma^{\frac{3n}{2}}}.
\end{align*}
Given the coarse estimates
\begin{equation*}
	\int_t^{\infty}\left|K_n\left(\frac{s}{\sigma}\right)\right|\d s\leq c_n\sigma,\ \ \ \ \ \ \ \int_{\sigma}^{\infty}\left|K_n\Big(\frac{t}{s}\Big)\right|\frac{\d s}{s^{\frac{3n}{2}}}\leq c_n t^{1-\frac{3n}{2}},\ \ \ \ c_n>0,
\end{equation*}
valid for all $t,\sigma>0$ and $n\in\mathbb{Z}_{\geq 2}$, we arrive at the above expansion for $g_2(t,\sigma)$.
\end{proof}
\end{appendix}


\begin{bibsection}
\begin{biblist}

\bib{AB}{article}{
AUTHOR = {Akemann, G.},
author={Bender, M.},
     TITLE = {Interpolation between {A}iry and {P}oisson statistics for
              unitary chiral non-{H}ermitian random matrix ensembles},
   JOURNAL = {J. Math. Phys.},
  FJOURNAL = {Journal of Mathematical Physics},
    VOLUME = {51},
      YEAR = {2010},
    NUMBER = {10},
     PAGES = {103524, 21},
      ISSN = {0022-2488},
   MRCLASS = {60B20 (33C10 33C45 82B41)},
  MRNUMBER = {2761338},
MRREVIEWER = {Razvan Teodorescu},
       DOI = {10.1063/1.3496899},
       URL = {https://doi-org.bris.idm.oclc.org/10.1063/1.3496899},
}

\bib{AP}{article}{
AUTHOR = {Akemann, G.}
author={Phillips, M. J.},
     TITLE = {The interpolating {A}iry kernels for the {$\beta=1$} and
              {$\beta=4$} elliptic {G}inibre ensembles},
   JOURNAL = {J. Stat. Phys.},
  FJOURNAL = {Journal of Statistical Physics},
    VOLUME = {155},
      YEAR = {2014},
    NUMBER = {3},
     PAGES = {421--465},
      ISSN = {0022-4715},
   MRCLASS = {82B41 (60B20)},
  MRNUMBER = {3192169},
       DOI = {10.1007/s10955-014-0962-6},
       URL = {https://doi-org.bris.idm.oclc.org/10.1007/s10955-014-0962-6},
}

\bib{AP2}{article}{
AUTHOR = {Akemann, G.},
author={Phillips, M. J.},
     TITLE = {Universality conjecture for all {A}iry, sine and {B}essel
              kernels in the complex plane},
 BOOKTITLE = {Random matrix theory, interacting particle systems, and
              integrable systems},
    SERIES = {Math. Sci. Res. Inst. Publ.},
    VOLUME = {65},
     PAGES = {1--23},
 PUBLISHER = {Cambridge Univ. Press, New York},
      YEAR = {2014},
   MRCLASS = {15B52 (30C40 60B20)},
  MRNUMBER = {3380679},
MRREVIEWER = {Beno\^{\i}t Collins},
}

\bib{ACV}{article}{
AUTHOR = {Akemann, Gernot},
author={Cikovic, Milan},
author={Venker, Martin},
     TITLE = {Universality at weak and strong non-{H}ermiticity beyond the
              elliptic {G}inibre ensemble},
   JOURNAL = {Comm. Math. Phys.},
  FJOURNAL = {Communications in Mathematical Physics},
    VOLUME = {362},
      YEAR = {2018},
    NUMBER = {3},
     PAGES = {1111--1141},
      ISSN = {0010-3616},
   MRCLASS = {60B20 (82B05)},
  MRNUMBER = {3845296},
MRREVIEWER = {Dominique L\'{e}pingle},
       DOI = {10.1007/s00220-018-3201-1},
       URL = {https://doi-org.bris.idm.oclc.org/10.1007/s00220-018-3201-1},
}

\bib{ADM}{article}{
AUTHOR = {Akemann, G.},
author={Duits, M.},
author={Molag, L. D.},
     TITLE = {The elliptic {G}inibre ensemble: {A} unifying approach to
              local and global statistics for higher dimensions},
   JOURNAL = {J. Math. Phys.},
  FJOURNAL = {Journal of Mathematical Physics},
    VOLUME = {64},
      YEAR = {2023},
    NUMBER = {2},
     PAGES = {Paper No. 023503, 39},
      ISSN = {0022-2488},
   MRCLASS = {60 (30A05 41 82D05)},
  MRNUMBER = {4544540},
       DOI = {10.1063/5.0089789},
       URL = {https://doi-org.bris.idm.oclc.org/10.1063/5.0089789},
}



\bib{ASS}{article}{
  title = {Transition to Chaos in Random Networks with Cell-Type-Specific Connectivity},
  author = {Aljadeff, Johnatan},
  author={Stern, Merav},
  author={Sharpee, Tatyana},
  journal = {Phys. Rev. Lett.},
  volume = {114},
  issue = {8},
  pages = {088101},
  numpages = {5},
  year = {2015},
  month = {Feb},
  publisher = {American Physical Society},
  doi = {10.1103/PhysRevLett.114.088101},
  url = {https://link.aps.org/doi/10.1103/PhysRevLett.114.088101}
}

\bib{AGBTAM}{article}{
title={Predicting the stability of large structured food webs},
author={Allesina, Stefano},
author={Grilli, Jacopo},
author={Barabás, György},
author={Tang, Si},
author={Aljadeff, Johnatan},
author={Maritan, Amos},
journal={Nature Communications},
volume={6},
issue={1},
pages={7842},
year={2015},
month={Jul},
doi={10.1038/ncomms8842},
url={https://doi.org/10.1038/ncomms8842},
}

\bib{AT}{article}{
title={The stability–complexity relationship at age 40: a random matrix perspective},
author={Allesina, Stefano},
author={Tang, Si},
journal={Population Ecology},
volume={57},
issue={1},
pages={63-75},
year={2015},
doi={10.1007/s10144-014-0471-0},
url={https://doi.org/10.1007/s10144-014-0471-0},
}

\bib{ACQ}{article}{
AUTHOR = {Amir, Gideon},
author={Corwin, Ivan},
author={Quastel, Jeremy},
     TITLE = {Probability distribution of the free energy of the continuum
              directed random polymer in {$1+1$} dimensions},
   JOURNAL = {Comm. Pure Appl. Math.},
  FJOURNAL = {Communications on Pure and Applied Mathematics},
    VOLUME = {64},
      YEAR = {2011},
    NUMBER = {4},
     PAGES = {466--537},
      ISSN = {0010-3640},
   MRCLASS = {60K35 (60B20 60F05 60H15 82C22 82C44)},
  MRNUMBER = {2796514},
MRREVIEWER = {Timo Sepp\"{a}l\"{a}inen},
       DOI = {10.1002/cpa.20347},
       URL = {https://doi-org.bris.idm.oclc.org/10.1002/cpa.20347},
}

\bib{BBD}{article}{
AUTHOR = {Baik, Jinho},
author={Buckingham, Robert},
author={DiFranco, Jeffery},
     TITLE = {Asymptotics of {T}racy-{W}idom distributions and the total
              integral of a {P}ainlev\'{e} {II} function},
   JOURNAL = {Comm. Math. Phys.},
  FJOURNAL = {Communications in Mathematical Physics},
    VOLUME = {280},
      YEAR = {2008},
    NUMBER = {2},
     PAGES = {463--497},
      ISSN = {0010-3616},
   MRCLASS = {33E17 (15A52 34M55 47B35 60F99 82B44)},
  MRNUMBER = {2395479},
MRREVIEWER = {Andrei A. Kapaev},
       DOI = {10.1007/s00220-008-0433-5},
       URL = {https://doi.org/10.1007/s00220-008-0433-5},
}

\bib{BDT}{book}{
AUTHOR = {Baik, Jinho},
AUTHOR = {Deift, Percy},
AUTHOR = {Suidan, Toufic},
     TITLE = {Combinatorics and random matrix theory},
    SERIES = {Graduate Studies in Mathematics},
    VOLUME = {172},
 PUBLISHER = {American Mathematical Society, Providence, RI},
      YEAR = {2016},
     PAGES = {xi+461},
      ISBN = {978-0-8218-4841-8},
   MRCLASS = {60B20 (30E25 33E17 41A60 47B35 82C23)},
  MRNUMBER = {3468920},
MRREVIEWER = {Terence Tao},
       DOI = {10.1090/gsm/172},
       URL = {https://doi-org.bris.idm.oclc.org/10.1090/gsm/172},
}

\bib{Ben}{article}{
AUTHOR = {Bender, Martin},
     TITLE = {Edge scaling limits for a family of non-{H}ermitian random
              matrix ensembles},
   JOURNAL = {Probab. Theory Related Fields},
  FJOURNAL = {Probability Theory and Related Fields},
    VOLUME = {147},
      YEAR = {2010},
    NUMBER = {1-2},
     PAGES = {241--271},
      ISSN = {0178-8051},
   MRCLASS = {60B20 (60G55 60G70)},
  MRNUMBER = {2594353},
       DOI = {10.1007/s00440-009-0207-9},
       URL = {https://doi-org.bris.idm.oclc.org/10.1007/s00440-009-0207-9},
}

\bib{BB}{article}{
AUTHOR = {Betea, Dan},
author={Bouttier, J\'{e}r\'{e}mie},
     TITLE = {The periodic {S}chur process and free fermions at finite
              temperature},
   JOURNAL = {Math. Phys. Anal. Geom.},
  FJOURNAL = {Mathematical Physics, Analysis and Geometry. An International
              Journal Devoted to the Theory and Applications of Analysis and
              Geometry to Physics},
    VOLUME = {22},
      YEAR = {2019},
    NUMBER = {1},
     PAGES = {Paper No. 3, 47},
      ISSN = {1385-0172},
   MRCLASS = {82C23 (05E05 60G55 60K35)},
  MRNUMBER = {3903828},
       DOI = {10.1007/s11040-018-9299-8},
       URL = {https://doi-org.bris.idm.oclc.org/10.1007/s11040-018-9299-8},
}

\bib{BBW}{article}{
AUTHOR = {Betea, Dan},
author={Bouttier, J\'{e}r\'{e}mie},
author={Walsh, Harriet},
     TITLE = {Multicritical random partitions},
   JOURNAL = {S\'{e}m. Lothar. Combin.},
  FJOURNAL = {S\'{e}minaire Lotharingien de Combinatoire},
    VOLUME = {85B},
      YEAR = {2021},
     PAGES = {Art. 33, 12},
   MRCLASS = {60C05 (05A17 60B20)},
  MRNUMBER = {4311914},
}

\bib{BoDei}{article}{
AUTHOR = {Borodin, Alexei},
author={Deift, Percy},
     TITLE = {Fredholm determinants, {J}imbo-{M}iwa-{U}eno
              {$\tau$}-functions, and representation theory},
   JOURNAL = {Comm. Pure Appl. Math.},
  FJOURNAL = {Communications on Pure and Applied Mathematics},
    VOLUME = {55},
      YEAR = {2002},
    NUMBER = {9},
     PAGES = {1160--1230},
      ISSN = {0010-3640},
   MRCLASS = {35Q15 (22E30 33C05 33E17 34M55 37K20 43A75)},
  MRNUMBER = {1908746},
MRREVIEWER = {Nicholas S. Witte},
       DOI = {10.1002/cpa.10042},
       URL = {https://doi-org.bris.idm.oclc.org/10.1002/cpa.10042},
}

\bib{Bo1}{article}{
AUTHOR = {Bothner, T.},
     TITLE = {From gap probabilities in random matrix theory to eigenvalue
              expansions},
   JOURNAL = {J. Phys. A},
  FJOURNAL = {Journal of Physics. A. Mathematical and Theoretical},
    VOLUME = {49},
      YEAR = {2016},
    NUMBER = {7},
     PAGES = {075204, 77},
      ISSN = {1751-8113},
   MRCLASS = {45C05 (33C10 33C45 45M05 60B20 82B26)},
  MRNUMBER = {3462295},
MRREVIEWER = {Christopher Steven Goodrich},
       DOI = {10.1088/1751-8113/49/7/075204},
       URL = {https://doi-org.bris.idm.oclc.org/10.1088/1751-8113/49/7/075204},
}

\bib{Bo0}{article}{
AUTHOR = {Bothner, T.},
     TITLE = {On the origins of {R}iemann-{H}ilbert problems in mathematics},
   JOURNAL = {Nonlinearity},
  FJOURNAL = {Nonlinearity},
    VOLUME = {34},
      YEAR = {2021},
    NUMBER = {4},
     PAGES = {R1--R73},
      ISSN = {0951-7715},
   MRCLASS = {30E25 (01A60 45M05 60B20)},
  MRNUMBER = {4246443},
       DOI = {10.1088/1361-6544/abb543},
       URL = {https://doi-org.bris.idm.oclc.org/10.1088/1361-6544/abb543},
}


\bib{BCT}{article}{
AUTHOR = {Bothner, Thomas},
author={Cafasso, Mattia},
author={Tarricone, Sofia},
     TITLE = {Momenta spacing distributions in anharmonic oscillators and
              the higher order finite temperature {A}iry kernel},
   JOURNAL = {Ann. Inst. Henri Poincar\'{e} Probab. Stat.},
  FJOURNAL = {Annales de l'Institut Henri Poincar\'{e} Probabilit\'{e}s et
              Statistiques},
    VOLUME = {58},
      YEAR = {2022},
    NUMBER = {3},
     PAGES = {1505--1546},
      ISSN = {0246-0203},
   MRCLASS = {45J05 (30E25 33C10 35J10 42A38 81V70)},
  MRNUMBER = {4452641},
       DOI = {10.1214/21-aihp1211},
       URL = {https://doi-org.bris.idm.oclc.org/10.1214/21-aihp1211},
}

\bib{Bo}{article}{
AUTHOR = {Bothner, T.},
TITLE = {A Riemann-Hilbert approach to Fredholm determinants of Hankel composition operators: scalar-valued kernels},
YEAR = {2022},
eprint={https://arxiv.org/abs/2205.15007},
      archivePrefix={arXiv},
      primaryClass={math-ph},
}

\bib{CC}{article}{
AUTHOR = {Cafasso, Mattia},
author={Claeys, Tom},
     TITLE = {A {R}iemann-{H}ilbert approach to the lower tail of the
              {K}ardar-{P}arisi-{Z}hang equation},
   JOURNAL = {Comm. Pure Appl. Math.},
  FJOURNAL = {Communications on Pure and Applied Mathematics},
    VOLUME = {75},
      YEAR = {2022},
    NUMBER = {3},
     PAGES = {493--540},
      ISSN = {0010-3640},
   MRCLASS = {60B20 (60F10 60G55)},
  MRNUMBER = {4373176},
       DOI = {10.1002/cpa.21978},
       URL = {https://doi-org.bris.idm.oclc.org/10.1002/cpa.21978},
}

\bib{CCG}{article}{
AUTHOR = {Cafasso, Mattia},
author={Claeys, Tom},
author={Girotti, Manuela},
     TITLE = {Fredholm determinant solutions of the {P}ainlev\'{e} {II}
              hierarchy and gap probabilities of determinantal point
              processes},
   JOURNAL = {Int. Math. Res. Not. IMRN},
  FJOURNAL = {International Mathematics Research Notices. IMRN},
      YEAR = {2021},
    NUMBER = {4},
     PAGES = {2437--2478},
      ISSN = {1073-7928},
   MRCLASS = {47G10 (33E17 34M55 37J65 60G55)},
  MRNUMBER = {4218326},
MRREVIEWER = {Nizar Demni},
       DOI = {10.1093/imrn/rnz168},
       URL = {https://doi-org.bris.idm.oclc.org/10.1093/imrn/rnz168},
}

\bib{CCR}{article}{
title={Airy Kernel Determinant Solutions to the KdV Equation and Integro-Differential Painlevé Equations},
   volume={386},
   ISSN={1432-0916},
   url={http://dx.doi.org/10.1007/s00220-021-04108-9},
   DOI={10.1007/s00220-021-04108-9},
   number={2},
   journal={Communications in Mathematical Physics},
   publisher={Springer Science and Business Media LLC},
   author={Cafasso, M.},
   author={Claeys, T.},
   author={Ruzza, G.},
   year={2021},
   month={Jun},
   pages={1107–1153},
 }
 
  \bib{ChCR}{article}{
 AUTHOR = {Charlier, Christophe},
 author={Claeys, Tom},
 author={Ruzza, Giulio},
     TITLE = {Uniform tail asymptotics for {A}iry kernel determinant
              solutions to {K}d{V} and for the narrow wedge solution to
              {KPZ}},
   JOURNAL = {J. Funct. Anal.},
  FJOURNAL = {Journal of Functional Analysis},
    VOLUME = {283},
      YEAR = {2022},
    NUMBER = {8},
     PAGES = {Paper No. 109608, 54},
      ISSN = {0022-1236},
   MRCLASS = {41A60 (35Q53 60G55 60H35)},
  MRNUMBER = {4452069},
MRREVIEWER = {Shuaixia Xu},
       DOI = {10.1016/j.jfa.2022.109608},
       URL = {https://doi-org.bris.idm.oclc.org/10.1016/j.jfa.2022.109608},
}
 
 
 
 \bib{CESX}{article}{
 AUTHOR = {Cipolloni, Giorgio},
 author={Erd\"os, L\'{a}szl\'{o}},
 author={Schr\"{o}der, Dominik},
 author={Xu, Yuanyuan},
     TITLE = {Directional extremal statistics for {G}inibre eigenvalues},
   JOURNAL = {J. Math. Phys.},
  FJOURNAL = {Journal of Mathematical Physics},
    VOLUME = {63},
      YEAR = {2022},
    NUMBER = {10},
     PAGES = {Paper No. 103303, 11},
      ISSN = {0022-2488},
   MRCLASS = {60B20 (15A18 15B52)},
  MRNUMBER = {4496015},
       DOI = {10.1063/5.0104290},
       URL = {https://doi-org.bris.idm.oclc.org/10.1063/5.0104290},
}
 

\bib{CESX2}{article}{
AUTHOR={Cipolloni, G.},
AUTHOR={Erd\"os, L.},
AUTHOR={Schr\"oder, D.},
AUTHOR={Xu, Y.},
TITLE={On the rightmost eigenvalue of non-Hermitian random matrices},
YEAR={2022},
eprint={https://arxiv.org/abs/2206.04448},
      archivePrefix={arXiv},
      primaryClass={math.PR},
}



\bib{CIK}{article}{
AUTHOR = {Claeys, T.},
AUTHOR={Its, A.},
AUTHOR={Krasovsky, I.},
     TITLE = {Higher-order analogues of the {T}racy-{W}idom distribution and
              the {P}ainlev\'{e} {II} hierarchy},
   JOURNAL = {Comm. Pure Appl. Math.},
  FJOURNAL = {Communications on Pure and Applied Mathematics},
    VOLUME = {63},
      YEAR = {2010},
    NUMBER = {3},
     PAGES = {362--412},
      ISSN = {0010-3640},
   MRCLASS = {34M55 (33E17 34M50 37K15 47B10 60B20 82C05)},
  MRNUMBER = {2599459},
       DOI = {10.1002/cpa.20284},
       URL = {https://doi-org.bris.idm.oclc.org/10.1002/cpa.20284},
}


\bib{DDMS}{article}{
  title = {Noninteracting fermions at finite temperature in a $d$-dimensional trap: Universal correlations},
  author = {Dean, David S.},
  author={Le Doussal, Pierre},
  author={Majumdar, Satya N.},
  author={Schehr, Gr\'egory},
  journal = {Phys. Rev. A},
  volume = {94},
  issue = {6},
  pages = {063622},
  numpages = {41},
  year = {2016},
  month = {Dec},
  publisher = {American Physical Society},
  doi = {10.1103/PhysRevA.94.063622},
  url = {https://link.aps.org/doi/10.1103/PhysRevA.94.063622}
}



\bib{DIK}{article}{
AUTHOR = {Deift, P.},
author={Its, A.},
author={Krasovsky, I.},
     TITLE = {Asymptotics of the {A}iry-kernel determinant},
   JOURNAL = {Comm. Math. Phys.},
  FJOURNAL = {Communications in Mathematical Physics},
    VOLUME = {278},
      YEAR = {2008},
    NUMBER = {3},
     PAGES = {643--678},
      ISSN = {0010-3616},
   MRCLASS = {47G10 (33C10 47A53 47B35 47N30 60F99 82B05 82B44)},
  MRNUMBER = {2373439},
MRREVIEWER = {A. B\"{o}ttcher},
       DOI = {10.1007/s00220-007-0409-x},
       URL = {https://doi.org/10.1007/s00220-007-0409-x},
}

\bib{DT}{article}{
AUTHOR = {Deift, P.},
author={Trubowitz, E.},
     TITLE = {Inverse scattering on the line},
   JOURNAL = {Comm. Pure Appl. Math.},
  FJOURNAL = {Communications on Pure and Applied Mathematics},
    VOLUME = {32},
      YEAR = {1979},
    NUMBER = {2},
     PAGES = {121--251},
      ISSN = {0010-3640},
   MRCLASS = {34B25 (35P25 58F07)},
  MRNUMBER = {512420},
MRREVIEWER = {R. C. Gilbert},
       DOI = {10.1002/cpa.3160320202},
       URL = {https://doi-org.bris.idm.oclc.org/10.1002/cpa.3160320202},
}

\bib{DZ}{article}{
AUTHOR = {Deift, P.},
AUTHOR={Zhou, X.},
     TITLE = {A steepest descent method for oscillatory {R}iemann-{H}ilbert
              problems. {A}symptotics for the {MK}d{V} equation},
   JOURNAL = {Ann. of Math. (2)},
  FJOURNAL = {Annals of Mathematics. Second Series},
    VOLUME = {137},
      YEAR = {1993},
    NUMBER = {2},
     PAGES = {295--368},
      ISSN = {0003-486X},
   MRCLASS = {35Q53 (34A55 34L25 35Q15 35Q55)},
  MRNUMBER = {1207209},
MRREVIEWER = {Alexey V. Samokhin},
       DOI = {10.2307/2946540},
       URL = {https://doi-org.bris.idm.oclc.org/10.2307/2946540},
}


\bib{DGIL}{article}{
AUTHOR = {Di Francesco, P.},
author={Gaudin, M.},
author={Itzykson, C.},
author={Lesage, F.},
     TITLE = {Laughlin's wave functions, {C}oulomb gases and expansions of
              the discriminant},
   JOURNAL = {Internat. J. Modern Phys. A},
  FJOURNAL = {International Journal of Modern Physics A. Particles and
              Fields. Gravitation. Cosmology},
    VOLUME = {9},
      YEAR = {1994},
    NUMBER = {24},
     PAGES = {4257--4351},
      ISSN = {0217-751X},
   MRCLASS = {81V70 (05E10 22E70 52B11 82D10)},
  MRNUMBER = {1289574},
MRREVIEWER = {Peter N. Zhevandrov},
       DOI = {10.1142/S0217751X94001734},
       URL = {https://doi-org.bris.idm.oclc.org/10.1142/S0217751X94001734},
}

\bib{DMRS}{article}{
  title = {Exact Short-Time Height Distribution in the One-Dimensional Kardar-Parisi-Zhang Equation and Edge Fermions at High Temperature},
  author = {Le Doussal, Pierre},
  author={Majumdar, Satya N.},
  author={Rosso, Alberto},
  author={Schehr, Gr\'egory},
  journal = {Phys. Rev. Lett.},
  volume = {117},
  issue = {7},
  pages = {070403},
  numpages = {5},
  year = {2016},
  month = {Aug},
  publisher = {American Physical Society},
  doi = {10.1103/PhysRevLett.117.070403},
  url = {https://link.aps.org/doi/10.1103/PhysRevLett.117.070403}
}

\bib{D}{article}{
author={Le Doussal, Pierre},
title = {Large deviations for the KPZ equation from the KP equation},
YEAR={2019},
eprint={https://arxiv.org/abs/1910.03671},
      archivePrefix={arXiv},
      primaryClass={cond-mat.dis-nn},
}


\bib{F2}{article}{
AUTHOR = {Forrester, P. J.},
     TITLE = {The spectrum edge of random matrix ensembles},
   JOURNAL = {Nuclear Phys. B},
  FJOURNAL = {Nuclear Physics. B. Theoretical, Phenomenological, and
              Experimental High Energy Physics. Quantum Field Theory and
              Statistical Systems},
    VOLUME = {402},
      YEAR = {1993},
    NUMBER = {3},
     PAGES = {709--728},
      ISSN = {0550-3213},
   MRCLASS = {82B41 (15A18 15A52 15A90 82B05)},
  MRNUMBER = {1236195},
MRREVIEWER = {Pawel S. Kurzepa},
       DOI = {10.1016/0550-3213(93)90126-A},
       URL = {https://doi.org/10.1016/0550-3213(93)90126-A},
}

\bib{F1}{book}{
AUTHOR = {Forrester, P. J.},
     TITLE = {Log-gases and random matrices},
    SERIES = {London Mathematical Society Monographs Series},
    VOLUME = {34},
 PUBLISHER = {Princeton University Press, Princeton, NJ},
      YEAR = {2010},
     PAGES = {xiv+791},
      ISBN = {978-0-691-12829-0},
   MRCLASS = {82-02 (33C45 60B20 82B05 82B41 82B44)},
  MRNUMBER = {2641363},
MRREVIEWER = {Steven Joel Miller},
       DOI = {10.1515/9781400835416},
       URL = {https://doi-org.bris.idm.oclc.org/10.1515/9781400835416},
}

\bib{FS1}{article}{
AUTHOR = {Fyodorov, Yan V.},
author={Sommers, H.-J.},
     TITLE = {Random matrices close to {H}ermitian or unitary: overview of
              methods and results},
      NOTE = {Random matrix theory},
   JOURNAL = {J. Phys. A},
  FJOURNAL = {Journal of Physics. A. Mathematical and General},
    VOLUME = {36},
      YEAR = {2003},
    NUMBER = {12},
     PAGES = {3303--3347},
      ISSN = {0305-4470},
   MRCLASS = {82B44 (15A52 82-02 82B31)},
  MRNUMBER = {1986421},
       DOI = {10.1088/0305-4470/36/12/326},
       URL = {https://doi-org.bris.idm.oclc.org/10.1088/0305-4470/36/12/326},
}

\bib{FKS}{article}{
AUTHOR = {Fyodorov, Yan V.},
author={Khoruzhenko, Boris A.},
author={Sommers, Hans-J\"{u}rgen},
     TITLE = {Almost {H}ermitian random matrices: crossover from
              {W}igner-{D}yson to {G}inibre eigenvalue statistics},
   JOURNAL = {Phys. Rev. Lett.},
  FJOURNAL = {Physical Review Letters},
    VOLUME = {79},
      YEAR = {1997},
    NUMBER = {4},
     PAGES = {557--560},
      ISSN = {0031-9007},
   MRCLASS = {82B41},
  MRNUMBER = {1459918},
       DOI = {10.1103/PhysRevLett.79.557},
       URL = {https://doi-org.bris.idm.oclc.org/10.1103/PhysRevLett.79.557},
}

\bib{FSK}{article}{
AUTHOR = {Fyodorov, Yan V.},
author={Sommers, Hans-J\"{u}rgen},
author={Khoruzhenko, Boris A.},
     TITLE = {Universality in the random matrix spectra in the regime of
              weak non-{H}ermiticity},
      NOTE = {Classical and quantum chaos},
   JOURNAL = {Ann. Inst. H. Poincar\'{e} Phys. Th\'{e}or.},
  FJOURNAL = {Annales de l'Institut Henri Poincar\'{e}. Physique Th\'{e}orique},
    VOLUME = {68},
      YEAR = {1998},
    NUMBER = {4},
     PAGES = {449--489},
      ISSN = {0246-0211},
   MRCLASS = {60F99 (15A52 60B15 82B41)},
  MRNUMBER = {1634312},
MRREVIEWER = {Oleksiy Khorunzhiy},
       URL = {http://www.numdam.org/item?id=AIHPA_1998__68_4_449_0},
}



\bib{Ga}{book}{
AUTHOR = {Gakhov, F. D.},
     TITLE = {Boundary value problems},
      NOTE = {Translated from the Russian,
              Reprint of the 1966 translation},
 PUBLISHER = {Dover Publications, Inc., New York},
      YEAR = {1990},
     PAGES = {xxii+561},
      ISBN = {0-486-66275-6},
   MRCLASS = {45E05},
  MRNUMBER = {1106850},
}

\bib{GA}{article}{
title={Connectance of large dynamic (cybernetic) systems: critical values for stability},
  author={Gardner, Mark R},
  author={Ashby, W Ross},
  journal={Nature},
  volume={228},
  number={5273},
  pages={784--784},
  year={1970},
  publisher={Nature Publishing Group}
}

\bib{GS}{article}{
AUTHOR = {Ghosal, Promit},
author={Silva, Guilherme L. F.},
     TITLE = {Universality for {M}ultiplicative {S}tatistics of {H}ermitian
              {R}andom {M}atrices and the {I}ntegro-{D}ifferential
              {P}ainlev\'{e} {II} {E}quation},
   JOURNAL = {Comm. Math. Phys.},
  FJOURNAL = {Communications in Mathematical Physics},
    VOLUME = {397},
      YEAR = {2023},
    NUMBER = {3},
     PAGES = {1237--1307},
      ISSN = {0010-3616},
   MRCLASS = {60B20 (15B52 34M55 37)},
  MRNUMBER = {4541921},
       DOI = {10.1007/s00220-022-04518-3},
       URL = {https://doi-org.bris.idm.oclc.org/10.1007/s00220-022-04518-3},
}




\bib{Gi}{article}{
AUTHOR = {Ginibre, Jean},
     TITLE = {Statistical ensembles of complex, quaternion, and real
              matrices},
   JOURNAL = {J. Mathematical Phys.},
  FJOURNAL = {Journal of Mathematical Physics},
    VOLUME = {6},
      YEAR = {1965},
     PAGES = {440--449},
      ISSN = {0022-2488},
   MRCLASS = {22.60 (53.90)},
  MRNUMBER = {173726},
MRREVIEWER = {J. Dieudonn\'{e}},
       DOI = {10.1063/1.1704292},
       URL = {https://doi-org.bris.idm.oclc.org/10.1063/1.1704292},
}

\bib{Gir}{article}{
AUTHOR = {Girko, V. L.},
     TITLE = {The elliptic law},
   JOURNAL = {Teor. Veroyatnost. i Primenen.},
  FJOURNAL = {Akademiya Nauk SSSR. Teoriya Veroyatnoste\u{\i} i ee Primeneniya},
    VOLUME = {30},
      YEAR = {1985},
    NUMBER = {4},
     PAGES = {640--651},
      ISSN = {0040-361X},
   MRCLASS = {60F99 (81G45 82A31)},
  MRNUMBER = {816278},
MRREVIEWER = {Nina B. Maslova},
}

\bib{Gir2}{article}{
AUTHOR = {Girko, V.L.},
     TITLE = {The generalized elliptic law},
   JOURNAL = {Random Oper. Stoch. Equ.},
  FJOURNAL = {Random Operators and Stochastic Equations},
    VOLUME = {21},
      YEAR = {2013},
    NUMBER = {2},
     PAGES = {191--215},
      ISSN = {0926-6364},
   MRCLASS = {60B20 (15B52 65F15)},
  MRNUMBER = {3068415},
       DOI = {10.1515/rose-2013-0010},
       URL = {https://doi-org.bris.idm.oclc.org/10.1515/rose-2013-0010},
}

\bib{GGK}{book}{
AUTHOR = {Gohberg, Israel},
author={Goldberg, Seymour},
author={Krupnik, Nahum},
     TITLE = {Traces and determinants of linear operators},
    SERIES = {Operator Theory: Advances and Applications},
    VOLUME = {116},
 PUBLISHER = {Birkh\"{a}user Verlag, Basel},
      YEAR = {2000},
     PAGES = {x+258},
      ISBN = {3-7643-6177-8},
   MRCLASS = {47B10 (45B05 45P05 47A53 47G10 47L10)},
  MRNUMBER = {1744872},
MRREVIEWER = {Hermann K\"{o}nig},
       DOI = {10.1007/978-3-0348-8401-3},
       URL = {https://doi-org.bris.idm.oclc.org/10.1007/978-3-0348-8401-3},
}

\bib{HM}{article}{
AUTHOR = {Hastings, S. P.},
AUTHOR={McLeod, J. B.},
     TITLE = {A boundary value problem associated with the second {P}ainlev\'{e}
              transcendent and the {K}orteweg-de\thinspace {V}ries equation},
   JOURNAL = {Arch. Rational Mech. Anal.},
  FJOURNAL = {Archive for Rational Mechanics and Analysis},
    VOLUME = {73},
      YEAR = {1980},
    NUMBER = {1},
     PAGES = {31--51},
      ISSN = {0003-9527},
   MRCLASS = {34B30 (35Q20)},
  MRNUMBER = {555581},
MRREVIEWER = {Richard Brown},
       DOI = {10.1007/BF00283254},
       URL = {https://doi-org.bris.idm.oclc.org/10.1007/BF00283254},
}

\bib{IIKS}{article}{
AUTHOR = {Its, A. R.},
AUTHOR={Izergin, A. G.},
AUTHOR={Korepin, V. E.},
AUTHOR={ Slavnov, N. A.},
     TITLE = {Differential equations for quantum correlation functions},
 BOOKTITLE = {Proceedings of the {C}onference on {Y}ang-{B}axter
              {E}quations, {C}onformal {I}nvariance and {I}ntegrability in
              {S}tatistical {M}echanics and {F}ield {T}heory},
   JOURNAL = {Internat. J. Modern Phys. B},
  FJOURNAL = {International Journal of Modern Physics B},
    VOLUME = {4},
      YEAR = {1990},
    NUMBER = {5},
     PAGES = {1003--1037},
      ISSN = {0217-9792},
   MRCLASS = {82B10 (35Q40 58G40 82C10)},
  MRNUMBER = {1064758},
MRREVIEWER = {Anatoliy Prykarpatsky},
       DOI = {10.1142/S0217979290000504},
       URL = {https://doi-org.bris.idm.oclc.org/10.1142/S0217979290000504},
}

\bib{JMMS}{article}{
AUTHOR = {Jimbo, Michio},
author={Miwa, Tetsuji},
author={M\^{o}ri, Yasuko},
author={Sato, Mikio},
     TITLE = {Density matrix of an impenetrable {B}ose gas and the fifth
              {P}ainlev\'{e} transcendent},
   JOURNAL = {Phys. D},
  FJOURNAL = {Physica D. Nonlinear Phenomena},
    VOLUME = {1},
      YEAR = {1980},
    NUMBER = {1},
     PAGES = {80--158},
      ISSN = {0167-2789},
   MRCLASS = {82A15 (14D05 58F07)},
  MRNUMBER = {573370},
       DOI = {10.1016/0167-2789(80)90006-8},
       URL = {https://doi-org.bris.idm.oclc.org/10.1016/0167-2789(80)90006-8},
}

\bib{Joh}{article}{
AUTHOR = {Johansson, K.},
     TITLE = {From {G}umbel to {T}racy-{W}idom},
   JOURNAL = {Probab. Theory Related Fields},
  FJOURNAL = {Probability Theory and Related Fields},
    VOLUME = {138},
      YEAR = {2007},
    NUMBER = {1-2},
     PAGES = {75--112},
      ISSN = {0178-8051},
   MRCLASS = {60G70 (15A52 60G07 62G32 82B41)},
  MRNUMBER = {2288065},
MRREVIEWER = {Alexander Roitershtein},
       DOI = {10.1007/s00440-006-0012-7},
       URL = {https://doi-org.bris.idm.oclc.org/10.1007/s00440-006-0012-7},
}

\bib{KZ}{article}{
AUTHOR = {Kimura, Taro},
author={Zahabi, Ali},
     TITLE = {Universal edge scaling in random partitions},
   JOURNAL = {Lett. Math. Phys.},
  FJOURNAL = {Letters in Mathematical Physics},
    VOLUME = {111},
      YEAR = {2021},
    NUMBER = {2},
     PAGES = {Paper No. 48, 16},
      ISSN = {0377-9017},
   MRCLASS = {60B20 (05E10 60G20 82B27)},
  MRNUMBER = {4244925},
       DOI = {10.1007/s11005-021-01389-y},
       URL = {https://doi-org.bris.idm.oclc.org/10.1007/s11005-021-01389-y},
}

\bib{KD}{article}{
  title = {Exact short-time height distribution in the one-dimensional Kardar-Parisi-Zhang equation with Brownian initial condition},
  author = {Krajenbrink, Alexandre},
  author={Le Doussal, Pierre},
  journal = {Phys. Rev. E},
  volume = {96},
  issue = {2},
  pages = {020102},
  numpages = {6},
  year = {2017},
  month = {Aug},
  publisher = {American Physical Society},
  doi = {10.1103/PhysRevE.96.020102},
  url = {https://link.aps.org/doi/10.1103/PhysRevE.96.020102}
}

\bib{Kra}{article}{
AUTHOR = {Krajenbrink, Alexandre},
     TITLE = {From {P}ainlev\'{e} to {Z}akharov-{S}habat and beyond: {F}redholm
              determinants and integro-differential hierarchies},
   JOURNAL = {J. Phys. A},
  FJOURNAL = {Journal of Physics. A. Mathematical and Theoretical},
    VOLUME = {54},
      YEAR = {2021},
    NUMBER = {3},
     PAGES = {Paper No. 035001, 51},
      ISSN = {1751-8113},
   MRCLASS = {37K10 (34M55 37J65)},
  MRNUMBER = {4209129},
       DOI = {10.1088/1751-8121/abd078},
       URL = {https://doi-org.bris.idm.oclc.org/10.1088/1751-8121/abd078},
}


\bib{LMS}{article}{
  title = {Multicritical Edge Statistics for the Momenta of Fermions in Nonharmonic Traps},
  author = {Le Doussal, Pierre},
  author={Majumdar, Satya N.},
  author={Schehr, Gr\'egory},
  journal = {Phys. Rev. Lett.},
  volume = {121},
  issue = {3},
  pages = {030603},
  numpages = {7},
  year = {2018},
  month = {Jul},
  publisher = {American Physical Society},
  doi = {10.1103/PhysRevLett.121.030603},
  url = {https://link.aps.org/doi/10.1103/PhysRevLett.121.030603}
}

\bib{LW}{article}{
AUTHOR = {Liechty, Karl},
author={Wang, Dong},
     TITLE = {Asymptotics of free fermions in a quadratic well at finite
              temperature and the {M}oshe-{N}euberger-{S}hapiro random
              matrix model},
   JOURNAL = {Ann. Inst. Henri Poincar\'{e} Probab. Stat.},
  FJOURNAL = {Annales de l'Institut Henri Poincar\'{e} Probabilit\'{e}s et
              Statistiques},
    VOLUME = {56},
      YEAR = {2020},
    NUMBER = {2},
     PAGES = {1072--1098},
      ISSN = {0246-0203},
   MRCLASS = {60B20 (15B52 82B23)},
  MRNUMBER = {4076776},
       DOI = {10.1214/19-AIHP994},
       URL = {https://doi-org.bris.idm.oclc.org/10.1214/19-AIHP994},
}

\bib{M}{article}{
AUTHOR={May, R.},
TITLE={Will a Large Complex System be Stable?},
JOURNAL={Nature},
year={1972},
volume={238},
number={5364},
pages={413--414},
doi={10.1038/238413a0},
url={https://doi.org/10.1038/238413a0},
}




\bib{NIST}{book}{
TITLE = {N{IST} handbook of mathematical functions},
    EDITOR = {Olver, Frank W. J.}
    editor={Lozier, Daniel W.}
    editor={Boisvert, Ronald F.}
    editor={Clark, Charles W.},
 PUBLISHER = {U.S. Department of Commerce, National Institute of Standards
              and Technology, Washington, DC; Cambridge University Press,
              Cambridge},
      YEAR = {2010},
     PAGES = {xvi+951},
      ISBN = {978-0-521-14063-8},
   MRCLASS = {33-00 (00A20 65-00)},
  MRNUMBER = {2723248},
}

\bib{O}{article}{
AUTHOR = {Okounkov, Andrei},
     TITLE = {Generating functions for intersection numbers on moduli spaces
              of curves},
   JOURNAL = {Int. Math. Res. Not.},
  FJOURNAL = {International Mathematics Research Notices},
      YEAR = {2002},
    NUMBER = {18},
     PAGES = {933--957},
      ISSN = {1073-7928},
   MRCLASS = {14H10 (14H70 37K20)},
  MRNUMBER = {1902297},
MRREVIEWER = {Gilberto Bini},
       DOI = {10.1155/S1073792802110099},
       URL = {https://doi-org.bris.idm.oclc.org/10.1155/S1073792802110099},
}

\bib{Olv}{book}{
AUTHOR = {Olver, F. W. J.},
     TITLE = {Asymptotics and special functions},
    SERIES = {Computer Science and Applied Mathematics},
 PUBLISHER = {Academic Press [Harcourt Brace Jovanovich, Publishers], New
              York-London},
      YEAR = {1974},
     PAGES = {xvi+572},
   MRCLASS = {41A60},
  MRNUMBER = {0435697},
MRREVIEWER = {Norman Bleistein},
}

\bib{Po}{book}{
author={Porter, C.E.},
title={Fluctuations of quantal spectra},
series={Statistical Theories of Spectra: Fluctuations},
PUBLISHER={Academic Press, New York},
YEAR={1965},
}

\bib{RA}{article}{
  title = {Eigenvalue Spectra of Random Matrices for Neural Networks},
  author = {Rajan, Kanaka},
  author={Abbott, L. F.},
  journal = {Phys. Rev. Lett.},
  volume = {97},
  issue = {18},
  pages = {188104},
  numpages = {4},
  year = {2006},
  month = {Nov},
  publisher = {American Physical Society},
  doi = {10.1103/PhysRevLett.97.188104},
  url = {https://link.aps.org/doi/10.1103/PhysRevLett.97.188104}
}

\bib{Ri}{article}{
AUTHOR = {Rider, B.},
     TITLE = {A limit theorem at the edge of a non-{H}ermitian random matrix
              ensemble},
      NOTE = {Random matrix theory},
   JOURNAL = {J. Phys. A},
  FJOURNAL = {Journal of Physics. A. Mathematical and General},
    VOLUME = {36},
      YEAR = {2003},
    NUMBER = {12},
     PAGES = {3401--3409},
      ISSN = {0305-4470},
   MRCLASS = {15A52},
  MRNUMBER = {1986426},
MRREVIEWER = {Oleksiy Khorunzhiy},
       DOI = {10.1088/0305-4470/36/12/331},
       URL = {https://doi-org.bris.idm.oclc.org/10.1088/0305-4470/36/12/331},
}



\bib{Sim}{book}{
AUTHOR = {Simon, Barry},
     TITLE = {Trace ideals and their applications},
    SERIES = {Mathematical Surveys and Monographs},
    VOLUME = {120},
   EDITION = {Second},
 PUBLISHER = {American Mathematical Society, Providence, RI},
      YEAR = {2005},
     PAGES = {viii+150},
      ISBN = {0-8218-3581-5},
   MRCLASS = {47L20 (47A40 47A55 47B10 47B36 47E05 81Q15 81U99)},
  MRNUMBER = {2154153},
MRREVIEWER = {Pavel B. Kurasov},
       DOI = {10.1090/surv/120},
       URL = {https://doi-org.bris.idm.oclc.org/10.1090/surv/120},
}

\bib{Sim2}{book}{
AUTHOR = {Simon, Barry},
     TITLE = {Basic complex analysis},
    SERIES = {A Comprehensive Course in Analysis, Part 2A},
 PUBLISHER = {American Mathematical Society, Providence, RI},
      YEAR = {2015},
     PAGES = {xviii+641},
      ISBN = {978-1-4704-1100-8},
   MRCLASS = {30-01 (33-01 34-01 40-01 41-01 44-01)},
  MRNUMBER = {3443339},
MRREVIEWER = {Fritz Gesztesy},
       DOI = {10.1090/simon/002.1},
       URL = {https://doi-org.bris.idm.oclc.org/10.1090/simon/002.1},
}

\bib{SCSS}{article}{
AUTHOR = {Sommers, H.-J.},
author={Crisanti, A.},
author={Sompolinsky, H.},
author={Stein, Y.},
     TITLE = {Spectrum of large random asymmetric matrices},
   JOURNAL = {Phys. Rev. Lett.},
  FJOURNAL = {Physical Review Letters},
    VOLUME = {60},
      YEAR = {1988},
    NUMBER = {19},
     PAGES = {1895--1898},
      ISSN = {0031-9007},
   MRCLASS = {82A31 (60F99)},
  MRNUMBER = {948613},
       DOI = {10.1103/PhysRevLett.60.1895},
       URL = {https://doi-org.bris.idm.oclc.org/10.1103/PhysRevLett.60.1895},
}

\bib{SCS}{article}{
  title = {Chaos in Random Neural Networks},
  author = {Sompolinsky, H.},
  author={Crisanti, A.},
  author={Sommers, H. J.},
  journal = {Phys. Rev. Lett.},
  volume = {61},
  issue = {3},
  pages = {259--262},
  numpages = {0},
  year = {1988},
  month = {Jul},
  publisher = {American Physical Society},
  doi = {10.1103/PhysRevLett.61.259},
  url = {https://link.aps.org/doi/10.1103/PhysRevLett.61.259}
}

\bib{Sos}{article}{
AUTHOR = {Soshnikov, Alexander},
     TITLE = {Universality at the edge of the spectrum in {W}igner random
              matrices},
   JOURNAL = {Comm. Math. Phys.},
  FJOURNAL = {Communications in Mathematical Physics},
    VOLUME = {207},
      YEAR = {1999},
    NUMBER = {3},
     PAGES = {697--733},
      ISSN = {0010-3616},
   MRCLASS = {82B41 (15A52 60F99 82B44)},
  MRNUMBER = {1727234},
MRREVIEWER = {Boris A. Khoruzhenko},
       DOI = {10.1007/s002200050743},
       URL = {https://doi-org.bris.idm.oclc.org/10.1007/s002200050743},
}

\bib{St}{book}{
AUTHOR = {Stein, Elias M.},
author={Shakarchi, Rami},
     TITLE = {Fourier analysis},
    SERIES = {Princeton Lectures in Analysis},
    VOLUME = {1},
      NOTE = {An introduction},
 PUBLISHER = {Princeton University Press, Princeton, NJ},
      YEAR = {2003},
     PAGES = {xvi+311},
      ISBN = {0-691-11384-X},
   MRCLASS = {42-01},
  MRNUMBER = {1970295},
MRREVIEWER = {Steven George Krantz},
}

\bib{TW}{article}{
AUTHOR = {Tracy, Craig A.}
author={Widom, Harold},
     TITLE = {Level-spacing distributions and the {A}iry kernel},
   JOURNAL = {Comm. Math. Phys.},
  FJOURNAL = {Communications in Mathematical Physics},
    VOLUME = {159},
      YEAR = {1994},
    NUMBER = {1},
     PAGES = {151--174},
      ISSN = {0010-3616},
   MRCLASS = {82B05 (33C90 47A75 47G10 47N55 82B10)},
  MRNUMBER = {1257246},
MRREVIEWER = {Estelle L. Basor},
       URL = {http://projecteuclid.org/euclid.cmp/1104254495},
}

\bib{TW2}{article}{
AUTHOR = {Tracy, Craig A.},
author={Widom, Harold},
     TITLE = {Fredholm determinants, differential equations and matrix
              models},
   JOURNAL = {Comm. Math. Phys.},
  FJOURNAL = {Communications in Mathematical Physics},
    VOLUME = {163},
      YEAR = {1994},
    NUMBER = {1},
     PAGES = {33--72},
      ISSN = {0010-3616},
   MRCLASS = {82B05 (33C90 47A75 47G10 47N55 82B10)},
  MRNUMBER = {1277933},
MRREVIEWER = {Peter J. Forrester},
       URL = {http://projecteuclid.org.bris.idm.oclc.org/euclid.cmp/1104270379},
}

\bib{W}{article}{
AUTHOR = {Widom, Harold},
     TITLE = {Integral operators in random matrix theory},
 BOOKTITLE = {Random matrices, random processes and integrable systems},
    SERIES = {CRM Ser. Math. Phys.},
     PAGES = {229--249},
 PUBLISHER = {Springer, New York},
      YEAR = {2011},
   MRCLASS = {60B20 (15B52 33E17 47G10 60-02)},
  MRNUMBER = {2858437},
       DOI = {10.1007/978-1-4419-9514-8\_3},
       URL = {https://doi-org.bris.idm.oclc.org/10.1007/978-1-4419-9514-8_3},
}




\bib{Z}{article}{
AUTHOR = {Zhou, Xin},
     TITLE = {The {R}iemann-{H}ilbert problem and inverse scattering},
   JOURNAL = {SIAM J. Math. Anal.},
  FJOURNAL = {SIAM Journal on Mathematical Analysis},
    VOLUME = {20},
      YEAR = {1989},
    NUMBER = {4},
     PAGES = {966--986},
      ISSN = {0036-1410},
   MRCLASS = {34B25 (35G15 45F15 45P05)},
  MRNUMBER = {1000732},
MRREVIEWER = {David J. Kaup},
       DOI = {10.1137/0520065},
       URL = {https://doi-org.bris.idm.oclc.org/10.1137/0520065},
}



\end{biblist}
\end{bibsection}
\end{document}